\documentclass[10pt,a4paper,twoside]{amsart}
\usepackage{bm}
\usepackage{amsmath, latexsym, amsfonts, amssymb, amsthm,stmaryrd,mathabx}
\usepackage{graphicx, color,hyperref,dsfont,caption,subcaption,wrapfig}
\usepackage{datetime}
\hypersetup{
    linktoc=page,
    linkcolor=red,          
    citecolor=blue,        
    filecolor=blue,      
     urlcolor=cyan,
    colorlinks=true           
}

\numberwithin{equation}{section}

\newtheorem{theorem}{Theorem}[section]
\newtheorem{lemma}[theorem]{Lemma}
\newtheorem{proposition}[theorem]{Proposition}
\newtheorem{corollary}[theorem]{Corollary}

\newtheorem{remark}[theorem]{Remark}
\newtheorem{definition}[theorem]{Definition}

\theoremstyle{definition}

\definecolor{remi}{rgb}{0,0,0}
    
\renewcommand{\tilde}{\widetilde}          
\DeclareMathSymbol{\leqslant}{\mathalpha}{AMSa}{"36} 
\DeclareMathSymbol{\geqslant}{\mathalpha}{AMSa}{"3E} 
\DeclareMathSymbol{\eset}{\mathalpha}{AMSb}{"3F}     
\renewcommand{\leq}{\;\leqslant\;}                   
\renewcommand{\geq}{\;\geqslant\;}                   
\newcommand{\dd}{\text{\rm d}}             

\newcommand{\mc}{\mathcal}

\newcommand{\cc}{\mathbb{C}}

\newcommand{\la}{\lambda}
\newcommand{\eps}{\epsilon}

\newcommand{\pl}{\partial}
\newcommand{\x}{\times}

\newcommand{\bbar}{\overline}

\newcommand{\cjd}{\rangle}
\newcommand{\cjg}{\langle}

\newcommand{\demi}{\tfrac{1}{2}}

\renewcommand{\H}{\mathbb{H}}
\newcommand{\D}{\mathbb{D}}
\newcommand{\R}{\mathbb{R}}
\newcommand{\Z}{\mathbb{Z}}
\newcommand{\N}{\mathbb{N}}

\newcommand{\E}{\mathds{E}}
\renewcommand{\P}{\mathds{P}}

\newcommand{\caD}{{\mathcal D}}

\newcommand{\caT}{{\mathcal T}}

\newcommand{\C}{\mathbb{C}}
\renewcommand{\H}{\mathbb{H}}

\renewcommand{\P}{\mathds{P}}

\newcommand{\hf}{\frac{_1}{^2}}

\usepackage{xparse}

\def\eps{\varepsilon}

\def\bi{\begin{itemize}}
\def\ei{\end{itemize}}

\def\bnum{\begin{enumerate}}
\def\enum{\end{enumerate}}
\def\<#1{\langle #1 \rangle}

\begin{document}

\title[$\mathbb{H}^3$-WZW on Riemann surfaces and Liouville correspondence ]{Probabilistic construction of the $\mathbb{H}^3$-Wess-Zumino-Witten conformal field theory and correspondence with Liouville theory}

\author[C. Guillarmou]{Colin Guillarmou$^{1}$}
\address{Universit\'e Paris-Saclay, CNRS,  Laboratoire de math\'ematiques d'Orsay, 91405, Orsay, France.}
\email{colin.guillarmou@universite-paris-saclay.fr}

 \author[A. Kupiainen]{ Antti Kupiainen$^{2}$}
\address{University of Helsinki, Department of Mathematics and Statistics}
\email{antti.kupiainen@helsinki.fi}

\author[R. Rhodes]{R\'emi Rhodes$^{3}$}
\address{Aix Marseille Univ, CNRS, I2M, Marseille, France}
\email{remi.rhodes@univ-amu.fr }

\footnotetext[1]{Universit\'e Paris-Saclay, CNRS,  Laboratoire de math\'ematiques d'Orsay, 91405, Orsay, France.}

\footnotetext[2]{University of Helsinki, Department of Mathematics and Statistics.}
\footnotetext[3]{Aix Marseille Univ, CNRS, I2M, Marseille, France.
}


 \begin{abstract}
 
 Wess-Zumino-Witten (WZW) models are among the most basic and  most studied Conformal Field Theories (CFT). 
They have had a huge influence not only in physics but also in mathematics, in representation theory and geometry. However their rigorous probabilistic construction and analysis starting from the path integral  is still missing and all their properties have been obtained algebraically from their postulated  affine Lie algebra symmetry. Initially considered as taking values in  a compact semi simple Lie Group G, the WZW model  has also a "dual" formulation where the group $G$ is replaced by the homogenous space $G^\C/G$, where $G^\C$ is the complexification of $G$,  and it has been argued that the former can be (re-) constructed from the latter. For $G={\rm SU}(2)$, the space  ${\rm SL}(2,\C)/{\rm SU}(2)$ can be identified with the three dimensional hyperbolic space $\H^3$ and, in physics, the corresponding CFT has been studied as the simplest example of the AdS/CFT correspondence. Furthermore, a surprising correspondence between the $\H^3$-WZW  CFT and the Liouville CFT was found by Ribault and Teschner and later generalised by Hikida and Shomerus. This correspondence has been dubbed by Gaiotto and Teschner as the "quantum analytic Langlands correspondence" since the analytic Langlands correspondence of Etingof, Frenkel and Kazhdan seems to emerge in its formal semi classical limit. 
In this paper we give a rigorous  construction of the  path integral  for the $\H^3$-WZW model  
on a  closed Riemann surface $\Sigma$ and   twisted by 
an arbitrary smooth gauge field on $\Sigma$. Using the probabilistic path integral we prove a correspondence between the correlation functions of the primary fields of the $\H^3$ model and those of Liouville CFT extending the expressions  proposed  by Ribault and Teschner and by Hikida and Schomerus to this general setup. 

 \end{abstract}

\maketitle
\setcounter{tocdepth}{2}





\normalsize

\section{Introduction}


Following the pioneering work \cite{BPZ84}, two dimensional conformal field theories (CFT) have been a source of inspiration in mathematics. 
In  \cite{BPZ84},  CFTs were studied via the axiomatic  conformal bootstrap approach. In this approach, the  basic objects of the CFT are the so-called  correlation functions
\begin{align*}
\langle V_{\alpha_1}(x_1)V_{\alpha_2}(x_2)\dots V_{\alpha_n}(x_n)\rangle
\end{align*}
where $\langle -\rangle$ refers to a probabilistic expectation and $V_\alpha$ are (random) fields defined on the two dimensional space (Riemann surface in general).  A set of axioms was postulated in  \cite{BPZ84} for the correlation functions whose consistency led to stringent conditions allowing  to derive/guess exact solutions.

 A more traditional approach to (Euclidean) quantum field theory (and hence also to CFT) is  the path integral approach. Here one starts with an action functional $S(\varphi)$ defined on field configurations $\varphi: \Sigma\to M$, where $\Sigma$ and $M$ are manifolds, the space-time and target space, respectively. The correlation functions are subsequently expressed in terms of a  path integral
\begin{align}\label{path}
\langle V_{\alpha_1}(x_1)V_{\alpha_2}(x_2)\dots V_{\alpha_n}(x_n)\rangle=\int \prod_{i=1}^n W_{\alpha_i}(\varphi,x_i)e^{-S(\varphi)}D\varphi
\end{align}
where $W_{\alpha}$ are suitable  functionals of the jets of the field $\varphi$ at $x_i$ and the integral is formally over some set of maps $\varphi: \Sigma\to M$. 
Giving a rigorous mathematical definition to the path integral and reconciling it with the bootstrap axioms and explicit solutions  poses a challenge for mathematicians. This challenge was addressed by the present authors and collaborators   \cite{DKRV16,GRVIHES,KRV_DOZZ, GKRV20_bootstrap,GKRV21_Segal}  
in the context of the Liouville CFT, which was given a rigorous probabilistic construction from which the bootstrap solution was derived. For a review see \cite{guillarmou2024}.

 In this paper, the probabilistic formulation of CFT of another major class of CFTs, the Wess-Zumino-Witten (WZW) CFTs, is addressed.  These theories are parametrised by a Lie group $G$ and they are examples of CFTs with an {\it extended symmetry}, i.e. a symmetry algebra extending the Virasoro algebra that expresses the conformal symmetry of a CFT. In the WZW model, this algebra is the affine Lie algebra related to $G$. The subsequent discussion will provide a concise review of this context. 
  
\subsection{Sigma models} The WZW model is a particular instance of a more extensive class of quantum field theories (QFTs), namely the {\it  sigma models}. To provide a concise definition,  let $\Sigma$ be a closed oriented surface equipped with a Riemannian metric $g$ and  $M$ a smooth manifold with Riemannian metric $m$. The action functional for the sigma model is
\begin{align*}
S(\varphi,g)= \int_\Sigma \|d\varphi(x)\|^2_{T_{\varphi(x)}M\otimes T^*_x\Sigma}{\rm dv}_g(x)
\end{align*}
 where ${\rm v}_g$ is the Riemannian volume and $\|d\varphi(x)\|^2_{T_{\varphi(x)}M\otimes T^*_x\Sigma}$ is the squared norm 
 of  $d\varphi(x)$ viewed as element in $T_{\varphi(x)}M\otimes T^*_x\Sigma$ and  using the metric dual to $g$ on $T^*_x\Sigma$ and $m$ on $T_{\varphi(x)}M$. In local coordinates, it reads
 \[  \|d\varphi(x)\|^2_{T_{\varphi(x)}M\otimes T^*_x\Sigma}=\sum_{i,j=1}^2 g^{ij}(x)m(\pl_{x_i}\varphi,\pl_{x_j}\varphi)\]
 if $(g^{ij})_{i,j}$ is the inverse  of the matrix  $g$ in the system of coordinates $(x_i)_{i=1,2}$. The action 
 $S(\varphi,g)$ is conformally invariant with respect to $g$.
The formal path integral for the sigma model is given by the  expression \eqref{path} with $D\varphi$ being the formal volume measure $\prod_x \dd {\rm v}_m(\varphi(x))$, with ${\rm v}_m$ the Riemannian volume of the metric $m$. Note that for $M=\R$ equipped with the Euclidean metric, this action is the Dirichlet energy of the map $\varphi:\Sigma\to\R$ and the path integral can be defined in terms of the Gaussian Free Field.
  
  To make sense out of such an expression, one has to resort to a renormalisation procedure where the path integral \eqref{path} is replaced by an expression in terms of a mollified field $\varphi_\epsilon$ with  mollification in scales smaller than $\epsilon$ and a mollified measure $D\varphi_\epsilon$ so that the resulting expectation $ \langle-\rangle$ is well defined, and then study the limit as $\epsilon\to 0$. In the process of doing this, it has been found that in general one needs to renormalise the action functional as well, i.e. to replace $S(\varphi_\epsilon)$ by $S_\epsilon(\varphi_\epsilon)$. Then, as $\epsilon\to 0$ and if the limit exists, the limit  is not expected in general to be conformally invariant. 
 A prominent illustration of this phenomenon can be observed in the Heisenberg model, wherein   $M$ is  the $n$-dimensional sphere $\mathbb{S}^n$ and $\Sigma=\R^2$. Perturbative arguments \cite{POLYAKOV197579} suggest that $S_\epsilon=\frac{1}{T_\epsilon} S$ where the  temperature parameter $T_\epsilon$ satisfies $\lim_{\epsilon\to 0}T_\epsilon =0$. The correlation functions of the limiting theory are expected to exhibit exponential decay in large spatial scales. A rigorous proof of these statements is a longstanding challenge in mathematical physics.
 
\subsection{WZW model}  The WZW model is a modification of a sigma model. For $M=G$ a Lie group, there is a natural bilinear form on the Lie algebra $\mathfrak{g}$ of $G$ which is invariant under the left and right actions of $G$, namely the Killing form. For a matrix group, i.e. $G\subset {\rm GL}(n,\C)$, the Killing form is given by $B(X,Y)={\rm Tr} ({\rm ad}X\circ {\rm ad}Y)$ for $X,Y\in \mathfrak{g}$, which reduces to $4{\rm Tr}(XY)$ in the case of $G={\rm SU}(2)$ if ${\rm Tr}$ is the trace on ${\rm GL}(n,\C)$. 
For a  compact semi-simple  Lie group $G$, this form is  negative definite and it defines (minus) a Riemannian metric on $G$.  
The Wess-Zumino model is then the corresponding sigma model and its action functional is given, using the decomposition $d=\bar{\pl}+\pl$ of the de Rham differential induced by the complex structure of $\Sigma$, by
\begin{align*}
S_{\textrm{WZ}}(h)=\frac{1}{4\pi i}\int_\Sigma {\rm Tr}(h^{-1}\partial h\wedge h^{-1}\bar\partial h)  
\end{align*}
defined on maps $h:\Sigma\to G$. This integral does only depend on the complex structure of $\Sigma$ and does not need a choice of Riemannian metric.As in the case of the Heisenberg model perturbative arguments suggest that, due to renormalisation, the corresponding path integral does not result in a Conformal Field Theory (CFT).
 
It was observed by Witten \cite{wittennonabelian} that it is possible to add a further term to the action $S_{\textrm{WZ}}$, which should then give rise to a CFT. Such term\footnote{The so called B-term in physics terminology.} can be introduced in the context of the sigma model by picking a two form $B$ on the target space $M$ and setting for $\varphi:\Sigma\to M$ smooth
\begin{align*}
S_B(\varphi)= \int_\Sigma \varphi^*B.
\end{align*}
which is trivially conformally invariant as it depends only on the differentiable structure of $\Sigma$ and $(\varphi,B)$. The WZW model considered in this paper is in fact a sigma model with a $B$ term. Let us describe now  the term considered by Witten. Let $X$ be a compact $3$-dimensional manifold with  boundary $\pl X=\Sigma$ and take $G={\rm SU}(2)$. Given a smooth map $h:\Sigma\to G$, let $\tilde{h}:X\to G$ be a smooth extension of $h$ to $X$ and set
\begin{align*}
S_{\textrm{top}}(\tilde{h})=\frac{1}{12\pi i}\int_{\tilde\Sigma}{\rm Tr}(\tilde{h}^{-1}d\tilde{h}\wedge\tilde{h}^{-1}d\tilde{h}\wedge \tilde{h}^{-1}d\tilde{h}).
\end{align*}
Given two such extensions $\tilde{h},\tilde{h}'$, one can show $S_{\textrm{top}}(\tilde{h})-S_{\textrm{top}}(\tilde{h}')\in 2\pi i\Z$ (see subsection \ref{WZW_coset}). Hence, for $k\in\Z$, the expression $e^{kS_{\textrm{top}}(\tilde{h})}$ is independent of the extension and, by defining  
\begin{align*}
S_{\textrm{WZW}}(\tilde{h})=S_{\textrm{WZ}}(h)+S_{\textrm{top}}(\tilde{h}),
\end{align*}
the integrand  in the formal path integral 
\begin{align}\label{wzwpath}
\cjg F\cjd = \int F(h)e^{ -kS_{\textrm{WZW}}(\tilde{h})}Dh
\end{align}
only depends on $h$ (here $Dh$ is the formal Haar measure on maps $h:\Sigma\to G$). In \cite{wittennonabelian}, Witten argued that, for each $k\in\N$, this path integral should give rise to a CFT. 
While, for ${\rm SU}(2)$ (or more generally a compact group $G$), a rigorous mathematical definition of this path integral is still lacking, it was observed in \cite{GK} that   replacing $G$ by the coset space $G^\C/G$  leads to a probabilistic formulation in terms of Gaussian integrals, where $G^\C={\rm SL}(2,\C)$ is the complexification of $G$. Furthermore, these $G^\C/G$-WZW models are of independent interest.  To mention just a few, they appeared in 
\begin{description}
\item[(i)]  a path integral description of the so called {\it coset construction of CFT}, for each pair $(G,H)$ with $G$ compact Lie group and $H$ a Lie subgroup \cite{GK},
\item[(ii)] in the construction of a scalar product in the quantum Hilbert space of the three dimensional {\it Chern-Simons gauge theory} \cite{Gawedzki:1989rr},
\item[(iii)] in the simplest {\it AdS/CFT correspondence} with $G={\rm SU}(2)$  \cite{Maldacena_Ooguri},
\item[(iv)] in  {\it quantum Langlands correspondence} \cite{Teschner:2017djr,Gaiotto:2024tpl} for  $G={\rm SU}(2)$. 
\end{description}
In this paper we give a rigorous probabilistic construction of the WZW path integral for the coset space ${\rm SL}(2,\C)/{\rm SU}(2)$, which identifies with the $3$-dimensional hyperbolic space $\H^3$.

\subsection{Hyperbolic space sigma models} 

The coset space $G^\C/G$ can be identified with the set of positive definite elements $h$ in $G^\C$ and represented under the form $h=qq^\ast$ with $q\in G^\C$ (for a matrix group $q^*$ is the adjoint matrix of $q$). Taking $G={\rm SU}(2)$
we can parametrise  ${\rm SL}(2,\C)/{\rm SU}(2)$ by the matrices
\begin{align}\label{hmatrix}
q=\left(
  \begin{array}{cc}
   e^{\phi/2} & e^{-\phi/2}\gamma \\
   0 & e^{-\phi/2}
  \end{array} \right)
\end{align}
with $(\phi,\gamma)\in \R\times\C$ since every  $h\in {\rm SL}(2,\C)$ can uniquely be written as $h=qu$ for some $u\in {\rm SU}(2)$.
Using this representation the WZ action functional takes the form
\begin{align}\label{action0}
S_{\textrm{WZ}}(qq^\ast)
=\frac{1}{4\pi i}\int_\Sigma 2\partial\phi\wedge \bar\partial\phi+e^{-2\phi}(\partial\gamma\wedge \bar\partial \bar\gamma+\partial\bar\gamma\wedge \bar\partial \gamma)
\end{align}
which we recognize as a sigma model with target $\R^3$ equipped with the metric $|d\phi|^2+e^{-2\phi}|d\gamma|^2$ i.e. the hyperbolic space $\H^3$.

For general semi simple $G$, the space $G^\C/G$ is contractible and one can then show that the topological term $S_{\textrm{top}}(qq^\ast)$  can be expressed in terms of a local action functional on $\Sigma$ instead of the $3$-dimensional extension $X$: for all $q\in G^\C$
\begin{align}\label{hyper}
S_{\textrm{top}}(qq^\ast)=\frac{i}{4\pi}\int_\Sigma {\rm Tr}( q^{-1}dq\wedge (q^{-1}dq)^\ast)
\end{align}
which is a ``B-term" as discussed above with $B={\rm Tr}(\omega_{{\rm MC}}\wedge (\omega_{{\rm MC}})^*)$ where $\omega_{{\rm MC}}$ is the Maurer-Cartan $1$-form on $G$.  Specializing to $G={\rm SU}(2)$  the WZW action then becomes (see section \ref{sec:General_case})
\begin{align}\label{action1}
S_{\textrm{WZW}}(qq^\ast)=\frac{1}{2\pi i}\int_\Sigma( \partial\phi\wedge \bar\partial\phi+e^{-2\phi}\partial\bar\gamma\wedge \bar\partial \gamma)=-\frac{1}{4\pi}\int_\Sigma (|d\phi|^2_g+4e^{-2\phi}|\bar{\pl}\gamma|_g^2){\rm dv}_g.
\end{align}
where in the last identity we have picked a Riemannian metric $g$ compatible with the complex structure, and ${\rm v}_g$ is the Riemannian measure.
The group ${\rm SL}(2,\C)$ acts on positive definite matrices $h$ by $h_0.h=h_0hh_0^\ast$ with  $h_0\in {\rm SL}(2,\C)$ and the 
${\rm SL}(2,\C)$-invariant measure on $\H^3$ is  given by $\dd {\rm v}_{\H^3}:=e^{-2\phi}\dd\phi\, \dd \gamma$. The formal path integral of the $\H^3$ model is then\footnote{Note the sign in the exponent: the Killing form is negative on the positive matrices so  we flip the sign of $k$ in \eqref{wzwpath} for stability.}
\begin{align}\label{pathi}
\big\langle F\big\rangle_\Sigma^{\H^3}=\int F(qq^*)e^{kS_{\textrm{WZW}}(qq^\ast)}Dq
\end{align}
with the formal measure $Dq:=\prod_x \dd  {\rm v}_{\H^3}(q(x))$. 

We will give a probabilistic definition and construction of this expression and prove its conformal invariance. Note that, in contrast to the path integral in the case of the compact groups $G$, there is no need to restrict $k$ to integers and indeed our construction holds for all real  $k>3$. 

\subsection{Probabilistic definition of the $\H^3$-WZW  path integral}
We specialize here to the Riemann sphere $\Sigma=\hat\C$ equipped with a smooth  metric $g=g(z)|dz|^2$ conformal to the Euclidean metric for simplicity. 
For fixed smooth $\phi$, the action \eqref{action1} is quadratic in $\gamma$. The quadratic form is given by
\begin{align*}
\frac{1}{2\pi i}\int_{\hat{\C}} e^{-2\phi}\partial\bar\gamma\wedge \bar\partial \gamma=-\frac{1}{\pi} \cjg \gamma,\caD_\phi \gamma\cjd_{L^2}
\end{align*}
where the scalar product is in $L^2(\hat{\C}, {\rm v}_g)$ and the operator $\caD_\phi =\partial_{\bar z}^\ast e^{-2\phi}\partial_{\bar z}$ with the adjoint taken with respect to ${\rm v}_g$. Hence,  conditionally on $\phi$, the $\gamma$-variable in  the path integral \eqref{pathi} should be interpreted as a complex Gaussian with covariance given by the inverse  of $\frac{k}{\pi}\caD_\phi $. This operator has a kernel,  namely the constant functions, so we set $\gamma=\gamma_0+\gamma_g$ where $\gamma_g$ is orthogonal to constants in $L^2(\Sigma, \dd {\rm v}_g)$ and $\gamma_0\in \C$.  Then, for smooth $\phi$ the covariance kernel of $\gamma_g$  is simply given by
\begin{align}\label{gammacov}
\E[ \gamma_g(z)\bar \gamma_g(z')]=\frac{1}{k\pi}\int_{\C} (\frac{1}{z-w} -\kappa_g(w))(\frac{1}{\bar z'-\bar w}-\overline{\kappa_g(w)})e^{2\phi(w)}\dd w,
\end{align}
together with $\E[ \gamma_g(z)\gamma_g(z')]=0$, where the function $\kappa_g$ comes from $\ker \mc{D}_\phi$ (see \eqref{covKsphere}). We call $\gamma_g$ a Witten field, as its covariance is related to the Witten type 
Laplacian $(e^{-\phi}\bar{\pl}_ze^{\phi})^*(e^{-\phi}\bar{\pl}_ze^{\phi})$.

However, since the action \eqref{action1} has a Dirichlet energy $\frac{k}{4\pi }\int_{\hat{\C}}|d\phi|^2_g{\rm dv}_g$, 
$\phi$ should be interpreted as a multiple of the Gaussian Free Field (GFF in short), which is a random distribution with a logarithmic singularity in its covariance at coinciding points. Hence the factor $e^{2\phi(w)}$ in the covariance of $\gamma_g$ needs renormalisation and has to be interpreted in terms of the Gaussian Multiplicative Chaos measure (GMC in short), a random measure defined by Kahane \cite{Kahane85}. Let us explain the construction.
For $k>2$, we denote 
\[
b:=(k-2)^{-\hf}
\]
and assume furthermore that $b\in (0,1)$ i.e.  $k>3$. Then we let $X_g$ be the Gaussian Free Field associated to the metric $g$,  with covariance given $\E[X_g(x)X_g(x')]=
G_g(x,x')$ if $G_g$ is the Green's function of the Laplacian $(2\pi)^{-1}\Delta_g$ on $(\hat{\C},g)$, and $X_{g,\eps}(x)$ is its regularisation by averaging the random distribution $X_g$ on the $g$-geodesic circle of center $x$ and small radius $\eps>0$. Let $c\in \R$ and define the field 
$\gamma_g$, conditionally on $X_g$, to be the complex Gaussian random variable with covariance kernel \eqref{gammacov} where $\phi_g=b(c+X_g)$, i.e. the term $e^{2\phi}\dd w$ in \eqref{gammacov}  is replaced by the GMC measure $e^{2bc}M_{2b}(X_g,\dd w):=e^{2bc} \lim_{\eps\to 0}\eps^{2b^2}e^{2b X_{g,\eps}(w)}\dd w$:
\[
\E[ \gamma_g(z)\bar \gamma_g(z')]=\frac{1}{k\pi}\int_{\C} (\frac{1}{z-w} -\kappa_g(w))(\frac{1}{\bar z'-\bar w}-\overline{\kappa_g(w)})e^{2bc}M_{2b}(X_g,\dd w).
\]
We call the pair $(\phi_g,\gamma_g)$ of random variables a \emph{Witten pair}.
Our definition for the path integral \eqref{pathi} for $\Sigma=\hat{\C}$ is then:  for $F:\mc{D}'(\hat{\C})\times \mc{D}'(\hat{\C},\C)\to \R^+$ continuous\footnote{Here $\mc{D}'(\hat{\C})$ is the set of real valued distributions on the sphere and $\mc{D}'(\hat{\C},\C)$ the set of complex valued distributions.}
\begin{equation}\label{defPI_intro}
\big\langle F(\phi,\gamma) \big\rangle^{\H^3}_{\hat{\C},g}:=\big(\frac{k-2}{32\pi}\big)^{\frac{1}{3}}\big(\frac{\pi}{k}\big)^{-\frac{2}{3}}\Big(\frac{\det(\Delta_g)}{{\rm v}_g(\hat{\C})}\Big)^{-\frac{3}{2}}\int_{\R\times\C}\E\Big[F(b(c+X_g), \gamma_0+\gamma_g)e^{-\frac{b}{4\pi}\int_\Sigma K_g(c+X_g)\dd {\rm v}_g}\Big]\,\dd c\, \dd \gamma_0.
\end{equation}
Here $K_g$ is the scalar curvature (twice the Gauss curvature) of the metric $g$ and the corresponding term in \eqref{defPI_intro} stems from the fact that the path integral \eqref{pathi} is intended to be {\it unnormalised}, which means that when defining the $\gamma$ integral in terms of Gaussian expectation we need to restore the ``partition function" of the Gaussian measure \eqref{gammacov}, which is a determinant of a Witten Laplacian in this case. The $K_g$ term appears from a Polyakov anomaly in this determinant, which also contains an extra 
$(-2)\times \frac{1}{4\pi}\int_{\hat{\C}} |d\phi|_g^2 {\rm dv}_g$, that explains the choice $\phi_g=(k-2)^{-1/2}(c+X_g)$ 
rather than $k^{-1/2}(c+X_g)$.

The heuristic justification of this definition is explained in more details  in Section \ref{s:heuristic_sphere}. We show in Proposition \ref{prop:weyl} that this definition satisfies the following conformal covariance: if $g'=e^{\omega}g$ for $\omega \in C^\infty(\hat{\C})$, 
\begin{equation}\label{PI_Weyl_Anomaly}
\big\langle F(\phi,\gamma) \big\rangle^{\H^3}_{\Sigma,g'}=e^{ \frac{{\bf c}(k)}{96\pi}\int_{\hat{\C}}(|d\omega|^2_g+2K_g\omega){\rm dv}_g}\big\langle F(\phi-\tfrac{b^2}{2}\omega,\gamma) \big\rangle^{\H^3}_{\Sigma,g}
\end{equation}
with ${\bf c}(k):=\frac{3k}{k-2}=3+6b^2$ the central charge of the $\H^3$-WZW model.
The vertex operators, also called primary fields, for this theory are 
\begin{equation}\label{Vjmuz}
V_{j,\mu}(z):=e^{2(j+1)\phi(z)}e^{\mu \gamma(z)-\bar\mu \bar\gamma(z)}, \ \ \mu\in \C, \, j< -\hf
\end{equation}
and the $n$-point correlations correspond to taking $F(\phi,\gamma)=\prod_{\ell=1}^m V_{j_\ell,\mu_\ell}(z_\ell)$ for distinct points $z_1,\dots,z_m$
 on the sphere. We show the following (see Propositions \ref{existcorrel} and \ref{confweight} for precise statements):
 \begin{theorem}\label{intro:th_correl}
For $k>3$ with $b=(k-2)^{-1/2}\in (0,1)$, the path integral \eqref{defPI_intro} defines a CFT with central charge 
$\mathbf{c}(k):=\frac{3k}{k-2}$.
For ${\bf z}=(z_1,\dots,z_m)$ distinct points  in $\hat{\C}$ and $\boldsymbol{j}=(j_1,\dots,j_m)$ some associated weights in $(-\infty,-1/2)$ such that $s_0:=-1+\sum_{\ell=1}^m(j_\ell+1)>0$, 
the $n$-point correlation functions admit a probabilistic definition using renormalisation of \eqref{Vjmuz}, and are  distributions in the 
variable $\boldsymbol{\mu}=(\mu_1,\dots,\mu_m)\in \C^m$. Moreover they are 
 equal to  
\[ \Big\cjg \prod_{\ell=1}^m V_{j_\ell,\mu_\ell}(z_\ell) \Big \cjd_{\hat{\C},g}^{\H^3}=C_{g}({\bf z},\boldsymbol{j})
\E\Big[\Big(\int_{\C}\Big|\sum_{\ell=1}^m\frac{\mu_\ell}{z_\ell-z}\Big|^2e^{2bu(z)}M_{2b}^g(X_g,\dd z)\Big)^{-s_0}\Big] \delta_{V^0}\]
where $\delta_{V^0}$ is the uniform measure supported on the hyperplane $V^0=\{ \boldsymbol{\mu}\in \C^m\,|\, \sum_{\ell=1}^m\mu_\ell=0\}$, see \eqref{deltaV^0}, the function $u$ is given by 
\[u(z):= \sum_{\ell=1}^m 2b^2(j_\ell+1)G_{g}(z,z_\ell)-\frac{b^2}{4\pi} \int_\Sigma G_{g}(z,z')K_g(z'){\rm dv}_g(z')\]
and $C_{g}({\bf z},\boldsymbol{j})$ is an explicit deterministic constant given in Proposition \ref{existcorrel}. For $g'=e^\omega g$ a metric conformal to $g$, the following conformal covariance holds
\[
\Big\langle\prod_{\ell=1}^mV_{j_\ell,\mu_\ell}(z_\ell) \Big \rangle^{\H^3}_{\hat{\C},g'}=  e^{ \frac{{\bf c}(k)}{96\pi}\int_{\hat{\C}}(|d\omega|^2_g+2K_g\omega){\rm dv}_g-\sum_{\ell=1}^m \triangle_{j_\ell}\omega(z_\ell)}\Big\langle\prod_{\ell=1}^mV_{j_\ell,\mu_\ell}(z_\ell) \Big\rangle^{\H^3}_{\hat{\C},g}
\]
where $\triangle_j:=-\frac{j(j+1)}{k-2}$ is the conformal weight   of the primary field $V_{j,\mu}$. 
\end{theorem}

The case of general surfaces with positive genus has a more complicated structure due to the fact that the operator 
$\bar\partial$ has a nontrivial co-kernel. Yet, the final definition of the path integral has a similar flavour, see subsection \ref{general}. We remark the similarities with Liouville theory \cite{DKRV16,GRVIHES} where the correlation functions are expressed in terms of negative moments of integrals of the
Gaussian Multiplicative Chaos.

\subsection{The $\mathbb{H}^3$-WZW  -- Liouville correspondence}

Liouville CFT on a Riemannian surface $(\Sigma,g)$ has the classical action functional\footnote{For simplicity, we take the cosmological constant to be $\mu=1$ in the conventions of \cite{DKRV16,GRVIHES}.}
\begin{align*}
S_{\rm L}(\phi,g)=\frac{1}{4\pi}\int_\Sigma( |d\phi|^2_g+Q K_g\phi+ e^{2b\phi})\dd {\rm v}_g
\end{align*}
with $Q=\frac{1}{b}+b$.
 Its path integral is defined probabilistically as
\begin{align*}
\big\langle F(\phi)\big\rangle^{\rm L}_{\Sigma,g}=\Big(\frac{\det(\Delta_g)}{{\rm v}_g(\Sigma)}\Big)^{-1/2}\int_\Sigma \E\big[F(c+X_g)e^{-\frac{Q}{4\pi}\int_\Sigma (c+X_g) K_g\dd {\rm v}_g+ e^{2bc} M^g_{2b}(X_g,\Sigma)}\Big]\,\dd c
\end{align*}
with $M^g_{2b}(X_g,\dd z)$ the GMC measure with respect to the GFF $X_g$ and the primary fields are given by (suitably renormalised) $V_\alpha(z)=e^{\alpha\phi(z)}$ for $\alpha\in\C$. A  surprising connection between the $\mathbb{H}^3$ model and the LCFT correlation functions  was observed in physics by Ribault and Teschner \cite{Ribault:2005wp}. A  formal path integral derivation was later given by Hikeda and Schomerus \cite{Hikida:2007tq}. The present paper can be seen as a mathematical proof, and also a  generalisation, of   this correspondence. For $\Sigma=\hat\C$, this correspondence takes the form (see Theorem \ref{Liouvillesphere2} for the precise statement):

\begin{theorem}
Let $m\geq 3$. Under the assumptions of Theorem \ref{intro:th_correl}, for each $\boldsymbol{\mu}\in \C^m$ there exist a 
set ${\bf y}=(y_1,\dots,y_p)$ of disjoint points in the sphere $\hat\C$ and integers $n_1,\dots,n_p\in \N\setminus \{0\}$, depending on ${\bf z}$ and $\boldsymbol{\mu}$, with $ n_1+\dots+n_p=m-2$, such that the following identity holds in the distribution sense in $\boldsymbol{\mu}\in \C^m$
\[ \Big\langle \prod_{\ell=1}^mV_{j_\ell,\mu_\ell}(z_\ell )\Big\rangle_{\hat{\C},g}^{\mathbb{H}^3}=\mu_{\rm L}^{-s_0} \mc{F}_{{\bf z},\boldsymbol{j}}({\bf y})\,  \Big\cjg \prod_{\ell=1}^{m} V_{\alpha_\ell}(z_\ell) \prod_{j=1}^{p} V_{-\frac{n_j}{b}}(y_j)\Big\cjd_{\hat{\C},g}^{\rm L} \delta_{V^0}\]
where $\mu_{\rm L}:=\frac{ |\sum_{\ell=1}^m\mu_\ell z_\ell^r |^2}{\pi k}$with $r$ the smallest integer such that $\sum_{\ell=1}^m\mu_\ell z_\ell^r \not=0$, $s_0:=-1+\sum_{\ell=1}^m(j_\ell+1)>0$, 
the function $\mc{F}_{{\bf z},\boldsymbol{j}}({\bf y})$ is explicitly given by \eqref{calculFdeg}, $\delta_{V^0}$ is the measure 
\eqref{deltaV^0} and $\alpha_\ell=2b(j_\ell+1)+\frac{1}{b}$.
\end{theorem}
The weights $-\frac{n_j}{b}$ corresponding to the points $y_j$ are ``degenerate weights'' in the CFT terminology. In this correspondence, 
the new insertions $y_\ell$ appearing in the Liouville side are the zeros of the meromorphic $1$-form $\sum_{\ell=1}^m\frac{\mu_\ell}{z-z_\ell}dz$.
A similar correspondence holds also on surfaces with genus larger than zero with an extra integral over the   degenerate field insertion points $y_\ell$. We will address this question in a broader context, namely that of a gauge-twisted version of the $\H^3$ model.

\subsection{Gauged WZW model and holomorphic vector bundles}

Conformal field theories have an algebra of conformal symmetries called Virasoro algebra. They appear by considering the
variation of the correlation functions under deformation of the complex structure on a compact surface $\Sigma$, which 
can be done by varying the metric. 
As mentioned above, the WZW models are expected to have a   rich symmetry structure, extending these conformal symmetries.
The deformations generating these extra symmetries, called Kac-Moody symmetries, need the introduction of a $1$-form $A$ on $\Sigma$ taking values in the Lie algebra $\mathfrak{g}^{\C}$ of $G^\C={\rm SL}(2,\C)$, giving rise to a connection on the trivial bundle $\Sigma\times \C^2$. We choose $A$ such that $A^*=-A$, that is the connection is assumed unitary.
This connection splits into 
$A=A^{1,0}+A^{0,1}$ where $A^{0,1}\in C^\infty(\Sigma;(T^*\Sigma)^{1,0})$ and $A^{1,0}=-(A^{0,1})^*\in C^\infty(\Sigma;(T^*\Sigma)^{0,1})$ (in local complex coordinates $A^{1,0}(z)=a(z)dz$ and $A^{0,1}(z)=-\bar{a}(z)d\bar{z}$ for some trace free matrix $a$). 
The Dolbeault operator $\bar{\pl}+A^{0,1}$ gives rise to a structure of holomorphic vector bundle on the bundle $\Sigma \times \C^2$ 
with local holomorphic trivialisations $(s_1,s_2)$ over a chart $U\subset \Sigma$ with  $s_i: U \to \C^2$ solving $(\bar{\pl}+A^{0,1})s_i=0$ for $i=1,2$. We denote this holomorphic vector bundle by $(E,\bar{\pl}_E:=\bar{\pl}+A^{0,1})$. Since we work with $A^{0,1}$ taking values in $\mathfrak{g}^\C={\rm sl}_2(\C)$, there is an extra structure on $E$, namely its determinant bundle is trivial. 
The gauge group  of such holomorphic rank-$2$ vector bundles with trivial determinant is ${\bf G}:=C^\infty(\Sigma,{\rm SL}(2,\C))$ and it acts on such bundles by conjugating the Dolbeault operator $h\circ \bar{\pl}_{E}\circ h^{-1}$, and the associated 
connection form changes by 
\[A^{0,1}_{h}=hA^{0,1}h^{-1}+h\bar\partial h^{-1}.\]
Changing $A^{0,1}$ amounts to deforming the holomorphic structure on $E$ and the action of these changes on the correlations should generate the Kac-Moody symmetries.
The  WZW action twisted by $A$ is associated to the holomorphic structure  $\bar{\pl}_E=\bar{\pl}+A^{0,1}$, and is defined by
\[
S_{\textrm{WZW}}(h,A)=S_{\textrm{WZW}}(h)+\frac{i}{2\pi}\int_\Sigma {\rm Tr}[A^{1,0}\wedge h^{-1}\bar\partial h+h\partial h^{-1} \wedge A^{0,1}+hA^{1,0}h^{-1} \wedge A^{0,1}].
\]
For fixed $A$, we can restrict this action to positive definite elements $h=qq^*$ with $q$ given by \eqref{hmatrix} as before: geometrically, 
$h$ represents a Hermitian metric on the holomorphic vector bundle $(E=\Sigma \times \C^2,\bar{\pl}_E)$. The formal path integral 
\begin{align}\label{pathiA}
\big\langle F\big\rangle^{\H^3}_{\rm A}=\int F(qq^*)e^{kS_{\textrm{WZW}}(qq^*,A)}Dq
\end{align} 
is thus a sum over all possible Hermitian metrics on $E$. Each such metric $h$ on $E$ induces a 
 unique $h$-unitary connection whose $(0,1)$-part is given by $A^{0,1}$, i.e. the holomorphic structure induced by this connection 
 is $\bar{\pl}_E$, it is 
explicitly given by $A^{1,0}_h+A^{0,1}$ with $A^{1,0}_h:=hA^{1,0}h^{-1}+h\partial h^{-1}$. The critical points of the action 
$h\mapsto S_{\textrm{WZW}}(h,A)$ are exactly the metrics $h$ such that the unitary connection $A^{1,0}_h+A^{0,1}$ is flat, i.e. has zero curvature (see Section \ref{s:WZWaction_E}), so the path integral describes fluctuations around flat unitary ${\rm SL}(2,\C)$-connections. 

The  gauge group ${\bf G}$ acts on the Hermitian metrics $h$ and on the connection forms $A=A^{1,0}+A^{0,1}$ 
by 
\[h_0.h=h_0hh^*_0, \quad   h_0.A:=A_{h_0}:=A_{h_0}^{0,1}+A^{1,0}_{{h_0^{-1}}^*}\] 
for $h_0\in {\bf G}$: $A_{h_0}=-A_{h_0}^*$ is the unitary connection associated to the conjugated operator  $h_0\bar{\pl}_Eh_0^{-1}$ and to the canonical Hermitian metric on the trivial bundle $E=\Sigma \times \C^2$ over $\Sigma$.
A computation gives (Lemma \ref{l:invarianceS})
\begin{align}\label{gauge}
S_{\textrm{WZW}}(h_0.h, h_0.A)=S_{\textrm{WZW}}(h,A)-S_{\textrm{WZW}}(h_0^{*{-1}},A^{1,0})-S_{\textrm{WZW}}(h_0^{-1},A^{0,1}).
\end{align}
This identity and a heuristic change of variables in the path integral suggest that the following should be true for each $h_0\in {\bf G}$
\begin{equation}\label{intro_gauge_inv}
\big\langle F\big\rangle^{\H^3}_A =e^{kS_{\textrm{WZW}}(h_0^{*-1},A^{1,0})+kS_{\textrm{WZW}}(h_0^{-1},A^{0,1})}\big\langle \caT_{h_0^{-1}}F\big\rangle^{\H^3}_{ A_{h_0}}
\end{equation}
and  $\caT_{h_0}F(h)=F(h_0.h)$. Thus e.g. for $\caT$-invariant observables $\caT_{h_0}F=F$, $\langle F\rangle^{\H^3}_A$ is constant up to a scalar factor on the gauge orbit $\{  A_{h_0}\,|\,  h_0\in {\bf G}\}$. Continuing along such heuristic lines and considering observables transforming under representations of $G^\C$ leads to the so called Ward identities for the correlation functions and ultimately to representations of the affine Lie algebra (see \cite{Gaw1} for a lucid discussion).
 
\subsection{Probabilistic definition, general case.} For practical purposes, we replace 
$S_{\textrm{WZW}}(h,A)$ by 
\[\tilde{S}_{\rm WZW}(h,A):=S_{\textrm{WZW}}(h,A)-S_{\textrm{WZW}}({\rm Id},A)\] 
which has the advantage of being invariant by $h_0\in {\rm SU}(2)$, $\tilde{S}_{\textrm{WZW}}(h_0.h, h_0.A)=\tilde{S}_{\textrm{WZW}}(h, A)$ (see Corollary \ref{gauge_inv}) and only add a deterministic term in the path integral. We will take as before 
\[  k>3 , \quad b=(k-2)^{-1/2}\in (0,1).\]
The probabilistic definition of the gauged model is more involved, since the action $S_{\textrm{WZW}}(h,A)$ 
contains non linearities in $(\phi,\gamma)$ that look harder to renormalise if $(\phi,\gamma)=(\phi_g,\gamma_g)$ 
is a Witten pair as defined before. Instead, we follow the strategy of Gawedzki in \cite{Gaw1} and write the holomorphic 
bundle $(E=\Sigma\times \C^2,\bar{\pl}_E=\bar{\pl}+A^{0,1})$ in terms of extension of a holomophic 
line bundle: every rank $2$ holomorphic bundle with trivial determinant and Hermitian metric admits
 an orthogonal (non-holomorphic) splitting $E=\mc{L}^{-1}\oplus \mc{L}$ where $\mc{L}$ is a complex line bundle equipped with a holomorphic structure $\bar{\pl}_{\mc{L}}$, $\mc{L}^{-1}$ its dual equipped with the dual holomorphic structure $\bar{\pl}_{\mc{L}^{-1}}$ 
 and the $\bar{\pl}_E$ operator has a upper triangular form in this splitting
 \begin{equation}\label{dbar_Esetup_intro}
  \bar{\pl}_E= \left(\begin{array}{cc} 
\bar{\pl}_{\mc{L}^{-1}} & \beta \\
0 & \bar{\pl}_{\mc{L}}
\end{array}\right) \end{equation}
for some $\mc{L}^{-2}$-valued $1$-form $\beta\in C^\infty(\Sigma;\mc{L}^{-2}\otimes (T^*\Sigma)^{0,1})$. This means that $(\mc{L}^{-1},\bar{\pl}_{\mc{L}^{-1}})$ is a holomorphic subbundle of $E$ and one has a short exact sequence of holomorphic bundles \eqref{short_exact_seq}. 
We call the pair  $(\mc{L},\beta)$
the parameters of the extension describing $(E,\bar{\pl}_E=\bar{\pl}+A^{0,1})$ and we notice that such a splitting is not unique in general 
and $\mc{L}$ is not always topologically trivial. 
Each 
$h\in C^\infty(\Sigma, {\rm SL}(2,\C))$ 
can be represented as an automorphism of $E=\mc{L}^{-1}\oplus \mc{L}$ and written as $h=qq^*$ with $q$ having the form \eqref{hmatrix} but now in the splitting $\mc{L}^{-1}\oplus \mc{L}$ rather than $\C\oplus \C$, with $\phi:\Sigma\to \R$ and 
$\gamma:\Sigma\to \mc{L}^{-2}$ is a section of the holomorphic line bundle $\mc{L}^{-2}=\mc{L}^{-1}\otimes \mc{L}^{-1}$ identified with 
the endomorphisms from $\mc{L}$ to $\mc{L}^{-1}$. The advantage of this description is that the WZW action now has the form (see Lemma \ref{lem:S(h,A)_formula_stable})
\begin{equation}\label{actionLbeta_intro} 
\tilde{S}_{\rm WZW}(h,A)=-\frac{1}{\pi }\int_\Sigma \big(\frac{1}{4}|d\phi|^2_g +e^{-2\phi}|\bar{\pl}_{\mc{L}^{-2}}\gamma+ \beta|^2_{\mc{L}^{-2},g}-|\beta|^2_{\mc{L}^{-2},g}\big){\rm dv}_g +\frac{1}{\pi i}\int_{\Sigma }F_{\mc{L}}\phi
\end{equation}
where  $(\mc{L},\beta)$ contains the information about the unitary connection form $A$, $F_{\mc{L}}$ 
is the curvature of $(\mc{L},\bar{\pl}_{\mc{L}})$ equipped with the canonical Hermitian metric on $E=\Sigma\times \C^2$ (which induces a 
metric on $\mc{L},\mc{L}^{-1}$). The action has a somehow similar form as \eqref{action1}, but 
with an additional term involving $\beta$ and $\bar{\pl}_{\mc{L}^{-2}}$ replacing the usual $\bar{\pl}$ operator on the trivial bundle. 
For generic $A$, the parameters $(\mc{L},\beta)$ representing $E$ as an extension are such that the $0$th-cohomology space\footnote{i.e. the space of holomorphic sections of $\mc{L}^{-2}$} 
\[ H^0(\Sigma,\mc{L}^{-2}):=\ker \bar{\pl}_{\mc{L}^{-2}}=0\]
is trivial. In that case, the operator $\mc{D}_\phi:=\bar{\pl}_{\mc{L}^{-2}}^*e^{-2\phi}\bar{\pl}_{\mc{L}^{-2}}$ describes the quadratic form
\[ \int_\Sigma e^{-2\phi}|\bar{\pl}_{\mc{L}^{-2}}\gamma|^2{\rm dv}_g=\cjg \mc{D}_\phi \gamma,\gamma\cjd_{2}\]
where $\bar{\pl}_{\mc{L}^{-2}}\gamma\in \mc{L}^{-2}\otimes (T^*\Sigma)^{0,1}$ and 
we use the Hermitian metric on $\mc{L}^{-2}$ and $g$ on $ (T^*\Sigma)^{0,1}$ to compute the pointwise norm in the integrand. 
The inverse of $\mc{D}_\phi$ has the form (Lemma \ref{inverseDphi})
\[ \cjg \mc{D}_\phi^{-1} f,f\cjd_{L^2(\Sigma,{\rm v}_g)}= \int_\Sigma |(1-\Pi_\phi)T_0f|^2e^{2\phi}{\rm dv}_g\]
with $T_0$ the operator satisfying $\bar{\pl}_{\mc{L}^{-2}}T_0={\rm Id}$ and mapping to $(\ker \bar{\pl}_{\mc{L}^{-2}}^*)^\perp$  and $\Pi_\phi$
the orthogonal projection on the 1st cohomology space\footnote{The space of holomorphic $1$-forms with values in $\mc{L}^{-2}$} $\ker \bar{\pl}_{\mc{L}^{-2}}^*\simeq H^1(\Sigma,\mc{L}^{-2})$ with respect to the measure $e^{2\phi}{\rm dv}_g$ . 
 We thus define $\gamma=\gamma_g$ to be a random variable, which is, conditionally on $\phi_g=b(X_g+c)$, 
 a centered Gaussian with (distributional) values in the complex line bundle $\mc{L}^{-2}$ with covariance kernel  
\[ \E[ \cjg \gamma_g,f\cjd_2 \overline{\cjg \gamma_g,f'\cjd_2}] =     \frac{\pi}{k}  \cjg (1-\Pi^\theta)T_0f',
(1-\Pi^\theta)T_0 f\cjd_{L^2(\Sigma,\theta)},\qquad \E[\cjg \gamma_g,f\cjd_2 \cjg \gamma_g,f'\cjd_2]=0\]
for any pair of smooth sections $f,f'\in C^\infty(\Sigma;\mc{L}^{-2})$ where $\theta=M_{2b}^g(\phi_g,\dd x)$ is the GMC measure, $\Pi^{\theta}: L^2(\Sigma;\mc{L}^{-2}\otimes \Lambda^{0,1}\Sigma)\to \ker (\bar{\pl}_{\mc{L}^{-2}}^*)$ is the orthogonal 
projection for the $L^2(\Sigma,\theta)$ pairing, and the $\cjg \cdot,\cdot\cjd_2$ pairing is with respect to the Riemannian volume measure ${\rm v}_g$. Using this Witten pair $(\phi_g,\gamma_g)$ we give in Definition \ref{defH3sigma} a probabilistic definition for the formal path integral 
\[\cjg F\cjd^{\H^3}_{\Sigma,g,\mc{L},\beta}=\int F(qq^*) e^{k\tilde{S}_{\rm WZW}(qq^*,A)}Dq\] 
associated to the action \eqref{actionLbeta_intro}, similar to \eqref{defPI_intro}, but involving also $\beta$. A new difficulty that appears is the cokernel $\ker \bar{\pl}_{\mc{L}^{-2}}^*$ that needs to be taken into account, and for non-generic connection form 
$A$, $H^0(\Sigma,\mc{L}^{-2})=\ker \bar{\pl}_{\mc{L}^{-2}}\not=0$, meaning that the 
$\gamma$ field has ``zero modes'' of dimension  $\dim H^0(\Sigma,\mc{L}^{-2})$ as it was the case for $A=0$ on $\Sigma=\hat{\C}$ (where 
$\ker \bar{\pl}=\C$) discussed above.  
We prove that the probabilistic path integral satisfies:
\begin{itemize}
\item \textbf{Weyl anomaly}  (Proposition \ref{prop:weyl_surface}): with ${\bf c}(k)=3k/(k-2)$ and $\omega\in C^\infty(\Sigma)$ 
\[\big\langle F(\phi,\gamma) \big\rangle^{\H^3}_{\Sigma,e^{\omega}g,\mc{L},\beta}=e^{ \frac{{\bf c}(k)}{96\pi}\int_{\hat{\C}}(|d\omega|^2_g+2K_g\omega){\rm dv}_g}\big\langle F(\phi-\tfrac{b^2}{2}\omega,\gamma) \big\rangle^{\H^3}_{\Sigma,g,\mc{L},\beta}\]
\item \textbf{Diffeomorphism invariance} (Proposition \ref{prop:weyl_surface}): for $\psi$ diffeomorphism,  $\psi^*g$ and $\psi^*\mc{L}$ 
the pullback metric and bundle
\[ \big\langle F(\psi^*\phi ,\psi^*\gamma) \big\rangle^{\H^3}_{\Sigma,\psi^*g,\psi^*\mc{L},\psi^*\beta}=\big \langle F(\phi ,\gamma) \big\rangle_{\Sigma,g,\mc{L},\beta}^{\H^3} .\]
\item \textbf{Gauge invariance} (Proposition \ref{gauge_invariance}): for upper triangular gauge (in the splitting $E=\mc{L}^{-1}\oplus \mc{L}$),
\[h_0^{-1}:= \left(\begin{array}{cc} 
e^{\frac{u}{2}} & e^{-\frac{u}{2}}v \\
0 & e^{-\frac{u}{2}} 
\end{array}\right), \quad \big\langle  F(\phi,\gamma) \big\rangle^{\H^3}_{\Sigma,g,\mc{L}_u,\beta_{u,v}}=e^{k\mc{W}(u,\beta,v)}\Big\langle  F(\phi-{\rm Re}(u),e^{-u}(\gamma+v)) \big\rangle^{\H^3}_{\Sigma,g,\mc{L},\beta}\]
with $u\in C^\infty(\Sigma,\C)$, $v\in C^\infty(\Sigma;\mc{L}^{-2})$,  $\mc{L}_u$ being $\mc{L}$ equipped with $\bar{\pl}_{\mc{L}_u}:=e^{u/2}\bar{\pl}_{\mc{L}}e^{-u/2}$, $\mc{L}_u^{-1}$ its dual, $\beta_{u,v}=e^{-u}(\beta+\bar{\pl}_{\mc{L}^{-2}}(v))$ and $\mc{W}(u,\beta,v)$ defined in \eqref{Wubetav}. This corresponds to  the invariance \eqref{intro_gauge_inv}. 
\end{itemize}

For ${\bf z}=(z_1,\dots,z_m)$ disjoint points in $\Sigma$ 
and $\boldsymbol{j}=(j_1,\dots,j_m)$ some real weights, the primary fields appearing in the correlations are also given in this case 
by \eqref{Vjmuz}, except that  $\boldsymbol{\mu}=(\mu_1,\dots,\mu_m)\in \mc{L}^2_{z_1}\times \dots \times \mc{L}^2_{z_m}$ so that $\mu_\ell \gamma(z_\ell)\in \C$ is well defined if $\mc{L}^2$ is the dual  line bundle of $\mc{L}^{-2}$. We prove, as for the case of the sphere with trivial bundle (i.e. $A=0$), that the correlation functions can be defined and are expressed in terms of moments of GMC measures.
In this introduction, we focus on the case of generic bundles represented as an extension with parameters $(\mc{L},\beta)$, 
but the non-generic bundle case is also dealt with in Proposition \ref{exists_correl_Sigma_NG}. 
The admissibility bounds for the weights in order to have well-defined probabilistic correlations are 
\begin{equation}\label{intro_admissibility}
\forall \ell=1,\dots,m, \,\, j_\ell <-1/2 , \quad -k\, {\rm deg}(\mc{L})+\sum_{\ell}(j_{\ell}+1)>\dim H^0(\Sigma,\mc{L}^{-2})
\end{equation} 
with $\chi(\Sigma)$ the Euler characteristic and ${\rm deg}(\mc{L})$ the degree\footnote{The degree classifies the topological type of $\mc{L}$, see Section \ref{sec:holo_vect_bundle}.} of $\mc{L}$.
There is a particular $\mc{L}^{-2}$-valued $(0,1)$-form that plays an important role 
\[\Gamma_{\boldsymbol{\mu},{\bf z}}:=\sum_{\ell=1}^m T_0(\bar{\mu}_\ell \delta_{z_\ell}) \quad \textrm{ solving }\,\, \bar{\pl}^*_{\mc{L}^{-2}}\Gamma_{\boldsymbol{\mu},{\bf z}}=\sum_{\ell=1}^m\bar{\mu}_\ell \delta_{z_\ell}\] 
with $\delta_{z_\ell}$ the Dirac mass at $z_\ell$. We have (see also Proposition \ref{exists_correl_Sigma}): 
\begin{proposition}\label{intro_correl}
Assume that $H^0(\Sigma,\mc{L}^{-2})=\ker \bar{\pl}_{\mc{L}^{-2}}=0$, and that $\boldsymbol{j}=(j_1,\dots,j_m)$ satisfy the admissibility bounds \eqref{intro_admissibility}, and let 
$s_0:= -(k-2){\rm deg}(\mc{L})-\frac{\chi(\Sigma)}{2}+\sum_{\ell}(j_{\ell}+1)>0$.
For each  $\boldsymbol{\mu}=(\mu_1,\dots,\mu_m)\in \mc{L}^{2}_{z_1}\times \dots \times \mc{L}^{2}_{z_m}$ non-zero, 
the $n$-point correlation function admits a probabilistic definition using renormalisation of \eqref{Vjmuz} and is equal to
\[\begin{split} 
\Big\langle \prod_{\ell=1}^mV_{j_\ell,\mu_\ell}(z_\ell)\Big\rangle^{\H^3}_{\Sigma,g,\mc{L},\beta}=&\frac{C_{k,\mc{L},\beta}(s_0)}{\det(\mc{D}_0)}\Big(\frac{\det(\Delta_{g})}{{\rm v}_g(\Sigma)}\Big)^{-\frac{1}{2}}e^{B_{\H^3}({\bf z},\boldsymbol{j})} \int_{W_1}e^{-2i{\rm Im}\cjg \Gamma_{\boldsymbol{\mu},{\bf z}}+ \omega,\beta \cjd_2} 
 \E\Big[ \|\Gamma_{\boldsymbol{\mu},{\bf z}}+\omega\|_{L^2(\theta_u)}^{-2s_0} \Big] \dd \omega.
\end{split}\]
with $\dd \omega$ the volume measure associated to the Hermitian metric given by the $L^2(\Sigma,{\rm v}_g)$ pairing on $W_1=\ker \bar{\pl}^*_{\mc{L}^{-2}}\simeq H^1(\Sigma,\mc{L}^{-2})$, $u$ and $\theta_u$ are given by 
\[
u(x):= \sum_{\ell=1}^m 2b^2(j_\ell+1)G_{g}(x,z_\ell)-b^2 \int_\Sigma G_g(x,y)\mc{K}_{\mc{L}}(y), \qquad \theta_u:= e^{2u}M^g_{2b}(X_g,\dd x)
\]
and $B_{\H^3}({\bf z},\boldsymbol{j})$, $C_{k,\mc{L},\beta}(s_0)$ are explicit constants defined in  \eqref{Cklbeta} and \eqref{Bzj}, $\mc{K}_{\mc{L}}:=\frac{K_g}{4\pi}{\rm v}_g-\frac{k-2}{\pi i}F_{\mc{L}}
$.
\end{proposition}
The correlation function is an $\mc{O}(|\boldsymbol{\mu}|^{2(s_0-N)})$ near $\boldsymbol{\mu}=0$ if $N=\dim H^1(\Sigma,\mc{L}^{-2})$, and extends as a distribution on $\mc{L}^2_{z_1}\times \dots \times \mc{L}^2_{z_m}\simeq \C^N$, see Proposition \ref{exists_correl_Sigma}.
A similar statement is proven in the distribution sense in $\boldsymbol{\mu}$ when $H^0(\Sigma,\mc{L}^{-2})\not=0$, see Proposition \ref{exists_correl_Sigma_NG}.

\subsection{General $\mathbb{H}^3$-WZW -- LCFT correspondence}\label{general}

In the case of a general holomorphic bundle, which means a general choice of unitary gauge form $A$, we also have a correspondence with Liouville theory. Let us state the result in the generic case, i.e. when $H^0(\Sigma,\mc{L}^{-2})=0$, the non-generic case can be found in Theorem \ref{Th_corresp_NG}. In the generic case, for each $\omega\in \ker \bar{\pl}_{\mc{L}^{-2}}^*$, the $1$-form $\Gamma_{\boldsymbol{\mu},{\bf z}}+\omega$ has finitely many zeros $y_1(\omega),\dots,y_{m(\omega)}(\omega)$ with multiplicities $n_1(\omega),\dots,n_{m(\omega)}(\omega)$ obeying $\sum_{\ell=1}^{m(\omega)} 
n_\ell(\omega)=m-\chi(\Sigma)+2{\rm deg}(\mc{L})$ with ${\rm deg}(\mc{L})$ the degree of $\mc{L}$ (see  Corollary \ref{expression_norm_Gamma}). 

\begin{theorem}\label{intro_generic_corresp}
Consider a holomorphic rank $2$ bundle $E$ with trivial determinant represented as an extension with parameters $(\mc{L},\beta)$ and assume that $\ker \bar{\pl}_{\mc{L}^{-2}}=H^0(\Sigma,\mc{L}^{-2})=0$.
 Let ${\bf z}=(z_1,\dots,z_m)$ be $m$ distinct marked points, $\boldsymbol{j}=(j_1,\dots,j_m)$ be some weights 
satisfying the admissiblity bounds \eqref{intro_admissibility}.
For $\boldsymbol{\mu}=(\mu_1,\dots,\mu_m)\in \mc{L}^{2}_{z_1}\times \dots \times \mc{L}^{2}_{z_m}$ the following WZW-Liouville correspondence holds
\[
\Big\langle \prod_{\ell=1}^mV_{j_\ell,\mu_\ell}(z_\ell)  \Big\rangle^{\mathbb{H}^3}_{\Sigma,g,\mc{L},\beta}=
\int_{W_1}e^{-2i{\rm Im}\cjg \Gamma_{\boldsymbol{\mu},{\bf z}}+\omega,\beta \cjd_2}e^{-s_0C(\omega)}\mc{F}_{{\bf z},\boldsymbol{j}}({\bf y}(\omega))\, \Big\cjg \prod_{\ell=1}^m V_{\alpha_\ell}(z_\ell)\prod_{\ell=1}^{m(\omega)}V_{-\frac{n_\ell(\omega)}{b}}(y_\ell(\omega))\Big\cjd_{\Sigma,g}^{\rm L} \dd\omega
\]
with $W_1=\ker \bar{\pl}_{\mc{L}^{-2}}^*\simeq H^{1}(\Sigma,\mc{L}^{-2})$, $\dd \omega$ is the volume measure associated to the $L^2$-product on $W_1$, with $\alpha_\ell:=2b(j_\ell+1)+\frac{1}{b}$ and ${\bf y}(\omega)=(y_1(\omega),\dots,y_{m(\omega)}(\omega))$ 
is the set of zeros of $\Gamma_{\boldsymbol{\mu},{\bf z}}+\omega$ with multplicities $n_\ell(\omega)$  and 
$\mc{F}_{{\bf z},\boldsymbol{j}}({\bf y})$ is an explicit function given by \eqref{formula_F_lastTh} and $C(\omega)=\int_\Sigma \log | \Gamma_{\boldsymbol{\mu},{\bf z}}+\omega|^2{\rm dv}_g$.
\end{theorem}

We emphasize that our construction of the path integral and correlation functions for a general connection $A$ depends a priori on the representation of the bundle $E$ as an extension of a line bundle $\mc{L}^{-1}$, and we are able in Proposition \ref{gauge_invariance} to check gauge covariance corresponding to equivalent extensions. It is an open problem which, at the moment, does appear difficult from the probabilistic point of view, to prove that the  correlation functions defined as above no not depend on the choice of extension. This problem is equivalent to proving gauge covariance of the path integral when the gauge transform
maps une extension of $E$ to another non-equivalent one. We plan to investigate this aspect in future work, in relation to the construction of the Kac-Moody algebra in our probabilistic model.\\

\textbf{Acknowledgements.} We thank particularly J\"org Teschner for many discussions, explanations and inspirations that motivated 
this work. We benefited also from helpful explanations by Federico Ambrosino, Laurent Charles and Emanuele Macri about holomorphic vector bundles.   R. Rhodes is partially supported by the Institut Universitaire de France (IUF) and acknowledges the support of the ANR-21-CE40-0003.  A. Kupianen is partially supported by the Academy of Finland.


\section{Preliminary background}
\subsection{Geometric background}
We consider a closed oriented Riemann surface $(\Sigma,g)$ and we denote by ${\rm v}_g$ its volume form (and view it as a measure using the orientation of $\Sigma$) and $K_g$ its scalar curvature.
The metric $g$ induces a conformal class $[g]:=\{ e^{\rho}g\, |\,\rho\in C^\infty(\Sigma)\}$, which is equivalent to a complex structure, i.e. a field $J\in C^\infty(\Sigma,{\rm End}(T\Sigma))$ of endomorphisms of the tangent bundle such that $J^2=-{\rm Id}$. We say that $(\Sigma,J)$ is a closed Riemann surface.
 For a metric $\hat{g}=e^{\rho}g$, one has the relation 
\begin{equation}\label{curv} 
K_{\hat{g}}=e^{-\rho}(\Delta_g \rho+K_{g})
\end{equation}
where $\Delta_g=d^*d$ is the non-negative Laplacian (here $d$ is exterior derivative and $d^*$ its adjoint with respect to ${\rm dv}_g$). 

\subsection{Determinant of Laplacians} \label{detoflap}

For a Riemannian metric $g$ on  $\Sigma$, the non-negative Laplacian $\Delta_g=d^*d$ has discrete spectrum
${\rm Sp}(\Delta_g)=(\la_j)_{j\in \N_0}$ with $\la_0=0$ and $\la_j\to +\infty$ and we shall denote $e_j$ the associated orthonormalized eigenfunctions. We can define the determinant of $\Delta_g$ by 
\[ {\det} (\Delta_g)=\exp(-\pl_s\zeta(s)|_{s=0})\]
where $\zeta(s):=\sum_{j=1}^\infty \la_j^{-s}$ is the spectral zeta function of $\Delta_g$, which admits a meromorphic continuation from ${\rm Re}(s)\gg 1$ to $s\in \cc$ and is holomorphic at $s=0$. We recall that if $\hat{g}=e^{\varphi}g$ for some $\varphi\in C^\infty(\Sigma)$, one has the so-called Polyakov formula (see \cite[eq. (1.13)]{OsgoodPS88}) 
\begin{equation}\label{detpolyakov} 
\frac{{\det}(\Delta_{\hat{g}})}{{\rm v}_{\hat{g}}(\Sigma)}= \frac{{\det}(\Delta_g)}{{\rm v}_{g}(\Sigma)}e^{-\frac{1}{48\pi}S^0_{\rm L}(g,\varphi)}
\end{equation}
where $K_g$ is the scalar curvature of $g$ as above and $S^0_{\rm L}$ is the Liouville action
\begin{equation}\label{LiouvilleS0}
S^0_{\rm L}(g,\varphi):=\int_\Sigma(  |d\varphi|_g^2+2K_g\varphi){\rm dv}_g.
\end{equation} 

\subsection{Green's function and resolvent of the Laplacian}\label{sec:Green}
Each closed Riemannian surface $(\Sigma,g)$ has a (non-negative) Laplace operator $\Delta_g=d^*d$ where the adjoint is taken with respect to ${\rm v}_g$.  The Green function $G_g$ on  $\Sigma$  is defined 
to be the integral kernel of the resolvent operator $R_g:L^2(\Sigma)\to L^2(\Sigma)$ satisfying $\Delta_{g}R_g=2\pi ({\rm Id}-P_0)$, $R_g^*=R_g$ and $R_g1=0$, where $P_0$ is the orthogonal projection in $L^2(\Sigma,{\rm v}_g)$ on $\ker \Delta_{g}$ (the constants). By integral kernel, we mean that for each $f\in L^2(\Sigma,{\rm v}_g)$
\[ R_gf(x)=\int_{\Sigma} G_g(x,x')f(x'){\rm dv}_g(x').\] 
The Laplacian $\Delta_{g}$ has an orthonormal basis of real valued eigenfunctions $(e_j)_{j\in \N_0}$ in $L^2(\Sigma,{\rm v}_g)$ with associated eigenvalues $\la_j\geq 0$; we set $\la_0=0$ and $\varphi_0=({\rm v}_g(\Sigma))^{-1/2}$.  The Green function then admits the following Mercer's representation in $L^2(\Sigma\times\Sigma, {\rm v}_g \otimes {\rm v}_g)$
\begin{equation}
G_g(x,x')=2\pi \sum_{j\geq 1}\frac{1}{\lambda_j}e_j(x)e_j(x').
\end{equation}

\subsection{Gaussian Free Field} 
On a Riemann surface without boundary, the Gaussian Free Field (GFF in short) is defined as follows. Let $(a_j)_j$ be a sequence of i.i.d. real Gaussians   $\mc{N}(0,1)$ with mean $0$ and variance $1$, defined on some probability space   $(\Omega,\mc{F},\mathbb{P})$,  and define  the Gaussian Free Field with vanishing mean in the metric $g$ by the random distribution (recall that $(e_j)_j$ is an orthonormal basis of eigenfunctions for $\Delta_g$ with eigenvalues $\la_j$)
\begin{equation}\label{GFFong}
X_g:= \sqrt{2\pi}\sum_{j\geq 1}a_j\frac{e_j}{\sqrt{\la_j}} 
\end{equation}
 where  the sum converges almost surely in the Sobolev space  $\mc{H}^{s}(\Sigma)$ for $s<0$ defined by
\begin{equation}\label{sobolev}
 \mc{H}^{s}(\Sigma):=\{f=\sum_{j\geq 0}f_je_j\,|\, \|f\|_{s}^2:=|f_0|^2+\sum_{j\geq 1}\lambda_j^{s}|f_j|^2<+\infty\}.
 \end{equation}
This Hilbert space is independent of $g$, only its norm depends on a choice of $g$. When it is important, we shall write $\mc{H}^s(\Sigma,\R)$ 
and $\mc{H}^s(\Sigma,\C)$ if the distribution is real or complex valued. 

The covariance $X_g$ is then the Green function when viewed as a distribution, which we will write with a slight abuse of notation
\[\mathbb{E}[X_g(x)X_g(x')]= \,G_g(x,x').\]
 
 \subsection{Metric regularisation of the GFF} \label{greg}
 As Gaussian Free Fields  are distributions on $\Sigma$, we need to regularize them using a procedure that we will call \emph{$g$-regularization}. If $d_g$ is Riemannian distance on $\Sigma$ and  $h$ a random distribution on $\Sigma$, for $\eps>0$ small we define a regularization $h_{\eps}$ of $h$ by averaging on geodesic circles of radius $\eps>0$: let $x\in \Sigma$ and let $\mc{C}_g(x,\eps)$ be the geodesic circle of center $x$ and radius $\eps>0$, and let $(f^n_{x,\eps})_{n\in \N} \in C^\infty(\Sigma)$ be a sequence with $||f^n_{x,\eps}||_{L^1}=1$ 
which is given by $f_{x,\eps}^n=\theta^n(d_g(x,\cdot)/\eps)$ where $\theta^n\in C_c^\infty((0,2))$ is non-negative, 
supported near $r=1$ and such that $f^n_{x,\eps}{\rm dv}_g$ 
converges in $\mc{D}'(\Sigma)$ to the uniform probability measure 
$\mu_{x,\eps}$
on $\mc{C}_g(x,\eps)$ as $n\to \infty$. If the pairing $\langle h, f_{x,\eps}^n\rangle$ converges almost surely towards a random variable $h_\epsilon(x)$ that has a modification which is continuous in the parameters $(x,\epsilon)$, we will say that $h$ admits a $g$-regularization $(h_\epsilon)_\epsilon$. This is the case for the GFF $X_g $ defined in \eqref{GFFong} and we denote by $X_{g,\epsilon}$ its  $g$-regularization.
   
\subsection{Gaussian multiplicative chaos} \label{sub:GMC}
 Given a random distribution $h$ on $\Sigma$ with a $g$-regularization $h_{g,\epsilon}$, we consider, for $b\in\R$ and $\sigma(\dd w)$ a Radon measure on $\Sigma$, the following measure on $\Sigma$
\begin{equation}\label{GMC}
M^g_{2b}(h,\sigma,\dd x):=\lim_{\epsilon\to 0}\epsilon^{2b^2}e^{2b h_{g,\epsilon}(x)}\, \sigma(\dd x)
\end{equation}
whenever the limit holds in probability in the sense of weak convergence of measures. If $\sigma$ admits a smooth Radon-Nikodym derivative with respect to the volume form ${\rm v}_g$ on $\Sigma$ and if $h=X_g$, then the above limit holds indeed in probability and    the limiting object is non trivial as long as the Radon-Nikodym derivative does not identically vanish and the parameter $b$ obeys $b\in (0,1)$. The reader may consult \cite[Section 3]{GRVIHES} for a pedagogical introduction to the construction of these random measures for non-probabilists.

 Also, from  \cite[Lemma 3.2]{GRVIHES},  we recall that there exist $W_g$   such that uniformly on $\Sigma$ 
\begin{equation}\label{varYg}
\lim_{\eps \to 0}\E[X^2_{g,\eps}(x)]+\ln\eps=W_g(x)  .
\end{equation}
 In particular, considering a metric $g'=e^{\omega}g$ conformal to the  metric $g$, we obtain the relation
\begin{equation}\label{relationentrenorm} 
M^{  g'}_{2b}(X_g,\sigma,\dd x)=e^{b^2\omega(x)} M^{g}_{2b}(X_g,\sigma,\dd x).
\end{equation}

\section{$\H^3$-WZW model on the Riemann sphere}
We first explain the construction of the model on the Riemann sphere to avoid any (heavy) geometrical background.  As we plan to write here an easy-access section, we explain in the first subsection only informally how the WZW leads to a   path integral that can  be given sense probabilistically. Detailed computations will be explained more systematically in Section \ref{H3surface} and require more geometrical background. Then the next subsections describe the probabilistic approach and the correspondence with Liouville theory. Finally we compute the structure constants of the $\H^3$-WZW model, i.e. the 3 point correlation functions on the Riemann sphere.

\subsection{Heuristic derivation for the probabilistic definition of the path integral}\label{s:heuristic_sphere}
We warn the reader that the following discussion is not fully rigorous mathematically but it explains and motivates the rigorous mathematical definition of the path integral given in the next section.
The Riemannian manifold $\Sigma$ we consider is the Riemann sphere equipped with a smooth Riemannian metric $g$. Actually we will specialize to the case when the Riemann sphere is realized as the extended complex plane $\hat\C$ equipped with a metric conformal to the Euclidean metric, thus $g=g(z)|dz|^2$ for some smooth functions $g(z)$ in $\C$. 
Being a smooth metric on the Riemann sphere means that   $\psi^*g$ is smooth on $\C$ if $\psi(z):=1/z$, which means that 
$g(1/z)|z|^{-4}$ extends smoothly near $z=0$. For example one could choose 
$g(z)=(1+|z|^2)^{-2}$, in which case the metric $g$ has constant positive curvature. We recall that the Dirichlet energy is conformally invariant, thus for $\phi\in C^\infty(\hat{\C})$, we have 
\[ \int_{\hat{\C}} |d\phi(z)|_g^2 {\rm dv}_g(z)=\int_{\C} |d\phi(z)|^2 \dd z\]
where $|d\phi|$ is the Euclidean norm of $d\phi$ and $|d\phi|_g=g(z)^{-1/2}|d\phi|$ is the norm of $d\phi$ with respect to $g$. Here and below, we use complex variable $z$ in $\C$ and we denote by $\dd z:=\dd {\rm Re}(z)\dd {\rm Im}(z)$ the Lebesgue measure in $\C$.
Later we will see that fixing such a choice of metric is not a restriction since the theory is invariant under diffeomorphism and 
every smooth metric on the sphere is diffeomorphic to such a $g$ for some function $g(z)$ by the uniformisation theorem.

As explained in the introduction, the $\H^3$-WZW path integral (for the case of the sphere $\hat{\C}$) 
is a formal integral over the space of functions on $\hat{\C}$ with values in positive definite matrices in ${\rm SL}(2,\C)$. Such fields of matrices can always be represented under the form 
\begin{equation}\label{paramh}
h(z)=\left(\begin{array}{cc}
e^{\phi(z)}(1+|\nu(z)|^2) & \nu(z) \\
\bar{\nu}(z) &e^{-\phi(z)}
\end{array}\right)
\end{equation}
for some functions $\phi: \hat{\C}\to \R$ and $\nu :\hat{\C}\to \C$. The pair $(e^{\phi},\nu)$ takes values in the $3$-dimensional hyperbolic space
$\H^3:=\{(e^{\phi},\nu)\,|\, \phi\in \R, \nu\in \C\}$, which also represents the coset space ${\rm SL}(2,\C)/{\rm SU}(2)$.
The WZW action $S_\Sigma(h)$, defined in general by the expression \eqref{WZW_def}, can be computed to be (the norm $|\cdot|$ below is simply the Euclidean norm)
\begin{equation}\label{actionS(h)sphere}
 \begin{split}
 S_\Sigma(h)
 &= \frac{1}{4\pi} \int_\C (|d\phi|^2+ 4|\pl_{\bar{z}}\nu+(\pl_{\bar{z}}\phi) \nu|^2)  \dd z.
 \end{split}
 \end{equation}

The path integral we have to focus on is thus
\begin{align}\label{formalPI0}
F\mapsto \int F( \phi, \nu) e^{-\frac{k}{4\pi }   \int_{\C} |d\phi|^2 + 4|\pl_{\bar{z}}\nu+(\pl_{\bar{z}}\phi)\nu|^2\dd z} D\phi D\nu,
\end{align}
where $D\phi$ and $D\nu$ respectively stand for the formal Lebesgue measure on maps $\phi:\hat{\C}\to \R$ and   $\nu:\hat{\C}\to\C$ (note that $D\phi D\nu$ is just the formal Haar measure on maps $h$ of the form \eqref{paramh} under the ${\rm SL}(2,\C)$ action). Restricting the path integral to the first term in the action, i.e.   $e^{-\frac{k}{4\pi } \int_\C |d\phi|^2\dd z} D\phi $, gives rise to a standard Gaussian Free Field on the sphere. One crucial point to understand how we will derive the probabilistic formulation is to see that, conditionally on $\phi$, the above path integral is a Gaussian integral in $\nu$. This is what we explain now.
So the first step is the probabilistic interpretation of the formal path integral
\begin{align}\label{formalPItilde}
\int F(   \nu) e^{-\frac{k}{\pi }   \int_{\C}|\pl_{\bar{z}}\nu+(\pl_{\bar{z}}\phi)\nu|^2\dd z}  D\nu.
\end{align}
Assume for a while that $\phi$ were a smooth bounded function.   Let   $L^2(\hat{\C},{\rm v}_g)$ be the $L^2$ space of complex valued functions, where ${\rm dv}_g=g(z)\dd z$ is the Riemannian measure associated to the metric $g$. 
We denote by $\langle\cdot,\cdot\rangle_2$ the sesquilinear product on  $L^2(\hat{\C},{\rm v}_g)$ and we consider the Witten-type Laplacian 
\[ \widetilde{\mc{D}}_{\phi}\nu :=   - g(z)^{-1}e^{-\phi}\pl_z (e^{-2\phi}\pl_{\bar{z}}(e^{\phi}\nu)). \]
It is straightforward to check that $\widetilde{\mc{D}}_{\phi}$ is self-adjoint with respect to the ${\rm v}_g$ measure
and that 
\[   \int_{\C}|\pl_{\bar{z}}\nu+(\pl_{\bar{z}}\phi)\nu|^2 \dd z= \int_{\C}e^{-2\phi}|\pl_{\bar{z}}(e^{\phi}\nu)|^2 \dd z  = \cjg \widetilde{\mc{D}}_{\phi}\nu,\nu\cjd_2.\]
The operator $\tilde{\mc{D}}_{\phi}$ is studied in a much more general setting in Lemma \ref{inverseDphi}, but in the present case we have, with the notations of Section \ref{sec:General_case}, that $\mc{L}$ is the trivial holomorphic bundle  $\hat{\C}\times \C$, $F_{\mc{L}}=0$, $\bar{\pl}f:=(\pl_{\bar{z}}f)d\bar{z}$ for each $f\in C^\infty(\hat{\C})$, $\ker \bar{\pl}^*=0$ and $\ker \bar{\pl}=\C$. 
From this lemma, one gets that $\tilde{\mc{D}}_{\phi}$ is Fredholm, has $1$-dimensional kernel given by $\ker \tilde{\mc{D}}_{\phi}:=\{Ce^{-\phi}\, |\, C\in \C\}$
 and there is a continuous operator $\widetilde{R}_{\phi}: C^\infty(\hat{\C})\to C^\infty(\hat{\C})$ such that 
\[ \widetilde{\mc{D}}_{\phi} \widetilde{R}_{\phi} ={\rm Id}-\tilde{P}_\phi, \quad \tilde{P}_\phi u:=  \frac{e^{-\phi}}{\|e^{-\phi}\|^2_2} \cjg u,e^{-\phi}\cjd_2\]
($\tilde{P}_\phi$ is the orthogonal projector on $\ker  \widetilde{\mc{D}}_{\phi}$).  
The operator $\widetilde{R}_{\phi}$ 
has  an integral kernel, which can be expressed in terms of Cauchy-transform on the Riemann sphere, given by
\begin{equation}\label{covKspheretilde}
\widetilde{  R}_{\phi}(z,z'):=\frac{e^{-\phi(z)-\phi(z')}}{\pi^2} \int_\C \big(\frac{1}{z-w}-\rho_\phi(w)\big)\big(\frac{1}{\bar{z}'-\bar{w}}-\bar{\rho}_\phi(w)\big)e^{2\phi(w)}\dd w
\end{equation}
with  
\begin{equation}\label{rhophi}
\rho_\phi(w):=\frac{1}{\|e^{-\phi}\|^2_2}\int_\C \frac{1}{z-w}  e^{-2\phi(z)} g(z)\dd z.
\end{equation}
Below, when the role of the background metric $g$ will be important, we shall add the subscript $g$ in the notations, namely we will write  $\widetilde{\mc{D}}_{\phi,g}$ and $\tilde{R}_{\phi,g}$ instead of $\widetilde{\mc{D}}_{\phi}$ and $\tilde{R}_{\phi}$. When $\phi$ is assumed to be smooth, we can consider a centered $\C$-valued Gaussian field $\nu_g$ with covariance $\tilde{R}_{\phi}$, which we will call a {\it Witten field}. Thus it obeys
\begin{equation}\label{gaussiancovtilde}
\E[\nu_g(z)\overline{\nu_g(z')}] =\frac{\pi}{k} \widetilde{R}_{\phi}(z,z') \quad \text{ and  }\quad  \E[\nu_g(z) {\nu_g(z')}]=0.
\end{equation}

We can thus  give the following probabilistic definition for the path integral interpretation of  \eqref{formalPItilde}:
\[
\int F(   \nu) e^{-\frac{k}{\pi }   \int_{\C}|\pl_{\bar{z}}\nu+(\pl_{\bar{z}}\phi)\nu|^2\dd z} D\nu =\det\Big(\frac{k}{\pi}\tilde{\mc{D}}_{\phi}\Big)^{-1}\int_{\C}\E\Big[F\Big(\nu_0 \frac{e^{-\phi}}{\|e^{-\phi}\|_{L^2}}+ \nu_g \big)\Big]\dd \nu_0.
\]
 The $\phi$ dependence in the regularized determinant could be problematic (in view of integrating latter in $\phi$) but  the determinant of our 
 Witten type Laplacian makes this  dependence explicit: applying Lemma \ref{lem:detwitten} (with $\mc{L}$ being the trivial bundle 
 $\hat{\C}\times \C$, $F_{\mc{L}}=0$, $\ker \bar{\pl}^*=0$ and $\ker \bar{\pl}=\C$)
\begin{equation}\label{regdetsphere}
\det\Big(\frac{k}{\pi}\tilde{\mc{D}}_{\phi}\Big) =e^{-\frac{1}{2\pi}\int_{\hat{\C}} (|d\phi|^2_g-\tfrac{1}{2}K_g\phi){\rm dv}_g} 
\det\big(\frac{k}{\pi}\tilde{\mc{D}}_{0}\Big)\frac{\|e^{-\phi}\|^2_2}{{\rm v}_g(\hat{\C})},
\end{equation}
where  $\det((k/\pi)\tilde{\mc{D}}_{0}) =(k/\pi)^{-\frac{2}{3}}\det( \tilde{\mc{D}}_{0})$ and $\tilde{\mc{D}}_{0}=\frac{1}{4}\Delta_g$ coincides with the standard Laplacian on functions, and $K_g$ is the scalar curvature of $g$, given explictly by $K_g(z)=g(z)^{-1}\Delta \log g(z)$ for $\Delta$ the positive Euclidean Laplacian on $\C=\R^2$.
  Finally, after a change of variables in the zero mode $\nu_0$, we end up with 
\begin{align*}
\int F(   \nu) e^{\frac{k}{\pi }   \int_{\C}|\pl_{\bar{z}}\nu+(\pl_{\bar{z}}\phi)\nu|^2\dd z} D\nu 
=(\frac{\pi}{k})^{-2/3}\Big(\frac{\det (\tilde{\mc{D}}_{0})}{{\rm v}_g(\hat{\C})}\Big)^{-1}e^{\frac{1}{2\pi}\int_{\hat{\C}} (|d\phi|^2_g-\tfrac{1}{2}K_g\phi){\rm dv}_g}\int_{\C}  \E\Big[F\big(\nu_0 e^{-\phi}+ \nu_g \big)\Big]\dd \nu_0.
\end{align*}
Later on, we will sample $\phi $ according to the GFF law, thus $\phi$ will be a real valued distribution in the Sobolev space $\mc{H}^{-s}(\hat{\C})$ for any $s>0$. In that respect, the expression we have obtained for the law of the field $\nu_g$ is not convenient because of
the factor $e^{-\phi(z)-\phi(z')}$ in the covariance kernel \eqref{covKspheretilde}, which cannot be evaluated pointwise if $\phi $ is a	 GFF. To get rid of this term, it is more reasonable to perform the   change of  field variables $(\phi,\nu)\mapsto(\phi,\gamma)$ with $\gamma:=e^\phi\nu$. 

To implement this concretely,  we consider the field $\gamma_g=e^{\phi}\nu_g-P_0(e^{\phi}\nu_g)$, called {\it rescaled Witten field},  where $P_0$ is the orthogonal projection on constants (the kernel of $\pl_{\bar{z}}$ on $\hat{\C}$).  
Conditionally on $\phi$, the rescaled Witten field  is thus a centered Gaussian distribution with covariance kernel (up to the factor $\frac{\pi}{k}$)
\begin{equation}\label{covKsphere}
  R_\phi (z,z'):=\frac{1}{\pi^2} \int_\C \big(\frac{1}{z-w}-\kappa_g(w)\big)\big(\frac{1}{\bar{z}'-\bar{w}}-\bar{\kappa}_g(w)\big)\theta(\dd w)
\end{equation}
with $\theta(\dd w)=e^{2\phi(w)}\dd w$ and  
\begin{equation}
\kappa_g(w):=\frac{1}{ {\rm v}_g(\hat{\C})}\int_\C \frac{1}{z-w}    {\rm dv}_g(z).
\end{equation}
The advantage of the expression \eqref{covKsphere} is that it perfectly makes sense in the case when $\phi$ is a GFF by means of  Gaussian multiplicative chaos theory, represented by the measure  $\theta(\dd w)=e^{2\phi}\dd w$. Notice also that (by Lemma \ref{inverseDphi})
$\cjg R_\phi f',f\cjd_2=\cjg T_0f',T_0f\cjd_{L^2(\hat{\C},\theta)}$ for each smooth pair $f,f'$ of functions on $\hat{\C}$, where $T_0: L^2(\hat{\C},{\rm v}_g)\to L^2(\C, \dd z)$ is the operator (with $\bar{\rho}_0$ given by \eqref{rhophi} with $\phi=0$)
\begin{equation}\label{Tphi}
T_0 f(w):=\frac{1}{\pi}\int_{\C} f(z) (\frac{1}{\bar{z}-\bar{w}}-\bar{\rho}_0(w)){\rm dv}_g(z),
\end{equation}
which satisfies for all $f\in L^2(\hat{\C},{\rm v}_g)$
\[- g(z)^{-1} \pl_z(T_0 f)(z)=f(z)-P_0 f.\]

 Here we emphasize that the background measure used to define $\theta$ appears as the Lebesgue measure $\dd z$ on $\C$, but if one thinks of $\hat{\C}$ as the Riemann sphere, $(T_0f)d\bar{z}$ is a smooth 
differential $1$-form on $\hat{\C}$ and the product above is actually equal to
\[ \cjg R_\phi f',f\cjd_2= \int_{\hat{\C}} \cjg (T_0f')d\bar{z},(T_0f)d\bar{z}\cjd_{g}e^{2\phi(z)}{\rm dv}_g(z)\]
where the product $\cjg\cdot,\cdot\cjd_g$ is the scalar product on $1$-forms on $\hat{\C}$ induced by $g$. The reader will find more details about this when we deal with general surfaces (where it is necessary to work with forms as there are no global coordinates), 
from Section \ref{sec:General_case} to the end of the paper.  This discussion motivates the following probabilistic formulation for the path integral of the field $\gamma$ (here $\gamma_0=({\rm v}_g(\hat{\C}))^{1/2}P_0(\gamma)$)
\begin{align}\label{PIcondphi}
\int & F( \gamma )e^{   -\frac{k}{\pi }   \int_{\C}|\pl_{\bar{z}}\nu+(\pl_{\bar{z}}\phi)\nu|^2\dd z}D\nu
\\&=(\frac{\pi}{k})^{-\frac{2}{3}}\Big(\frac{\det(\mc{D}_0)}{{\rm v}_g(\hat{\C})}\Big)^{-1}e^{\frac{1}{2\pi}\int_{\hat{\C}} (|d\phi|^2_g-\tfrac{1}{2}K_g\phi){\rm dv}_g}\int_{\C}  \E\Big[F\big(\gamma_0  + \gamma_g \big)\Big]\dd \gamma_0.\nonumber
\end{align}
Integrating then  this expression  against the term $e^{-\frac{k}{4\pi }    \int_{\hat{\C}} |d\phi|^2\dd z} D\phi$ produces the following definition
\begin{align}\label{PIcondphi2}
\int & F(\phi, \gamma )e^{-\frac{k}{4\pi}\int_{\hat{\C}} (|d\phi|^2 + 4|\pl_{\bar{z}}\nu+(\pl_{\bar{z}}\phi)\nu|^2)\dd z}  
D\phi D\nu  
\\
&:=(\frac{\pi}{k})^{-\frac{2}{3}}\Big(\frac{\det(\mc{D}_0)}{{\rm v}_g(\hat{\C})}\Big)^{-1}  \int e^{-\tfrac{k-2}{4\pi}\int_{\hat{\C}} |d\phi|^2\dd z-\tfrac{1}{4\pi}\int_{\hat{\C}} K_g\phi \, {\rm dv}_g}\int_{\C}  \E[F\big(\phi,\gamma_0  + \gamma_g \big)|\phi]\dd \phi\dd \gamma_0 \nonumber
\end{align}
for the path integral where $\E[\,\cdot\,|\phi]$ stands for conditional expectation given the field $\phi$. Assuming $k>3$, we introduce the parameter $b=(k-2)^{-1/2}$, which belongs to $(0,1)$. As usual, the $\phi$ integral corresponding to the term $\tfrac{k-2}{4\pi}\int_{\hat{\C}} |d\phi|^2_g{\rm dv}_g$ in the action  gives rise to $b$ times $c+X_g$ (with $X_g$ the standard GFF and $c\in \R$ distributed with Lebesgue measure $\dd c$), and has total mass $(\tfrac{k-2}{2\pi})^{\frac{1}{3}}(\frac{\det(\Delta_g)}{{\rm v}_g(\hat{\C})})^{-\frac{1}{2}}$. The curvature term $\int_{\hat{\C}} K_g\phi \, {\rm dv}_g$ is just a linear perturbation of this GFF. We will summarize all these considerations in Subsection \ref{sub:PIsphere} thereafter to produce the final expression for the path integral, see Definition \ref{defH3sphere}.

\subsection{Probabilistic setup}
We introduce here the probabilistic material to give sense to the  path integral explained  in the heuristics above.
\begin{definition}\label{wittensphere}
Let $\theta$ be a measure on the complex plane. A  rescaled Witten field $\gamma_g$ on the Riemann sphere associated to the pair $({\rm v}_g, \theta)$ is a centered Gaussian process on $\hat{\C}$ with values in complex valued  distributions on $\hat{\C}$, and  with covariance: for each $f,f'\in C^\infty(\hat{\C})$,
\begin{equation}\label{gaussiancov}
\E[ \cjg \gamma_g,f\cjd_2 \overline{\cjg \gamma_g,f'\cjd_2}] =     \frac{\pi}{k}  \cjg T_0f',
T_0 f\cjd_{L^2(\hat{\C},\theta)},\qquad  \E[\cjg \gamma_g,f\cjd_2 \cjg \gamma_g,f'\cjd_2]=0
\end{equation}
where $T_0$ is the operator \eqref{Tphi} with $\phi=0$.
\end{definition}

If $\theta$ is a measure on $\C$ (or $\hat{\C}$) and $u\in C^\infty(\hat{\C},\C)$, we shall denote 
\begin{equation}\label{notation_theta_u} 
\theta_u(\dd z):=e^{2b{\rm Re}(u)}\theta(\dd z).
\end{equation}
If $\theta$ is a measure as in Definition \ref{wittensphere}, we define the operator $R^\theta$ by 
\begin{equation}\label{def_of_Rtheta} 
\cjg R^\theta f',f\cjd_2= \int_{\hat{\C}} (T_0f').\bbar{(T_0f)}\theta(\dd z).
\end{equation}

 In what follows, we will be interested in the case when the parameter $\theta$ of the rescaled  Witten field is sampled according to a GMC measure (see Section \ref{sub:GMC}) associated to the GFF with an eventual smooth drift. One interesting aspect of the rescaled Witten field is that its kernel \eqref{covKsphere} is a measurable function of the integration measure $\theta$, which makes possible to sample it in a measurable way with respect to the GFF.  Yet, one aspect to be analyzed carefully is that the law (not only the kernel) of $\gamma_g$,  with covariance kernel \eqref{covKsphere}, can be produced from the law of  the GFF in a measurable way. We will focus on these aspects in the remaining part of this subsection. 
In the treatment of the Riemann sphere, for $h$ a GFF with a possible smooth drift, we will always denote by $M^g_{2b}(h, \dd z)$ the GMC measure $M^g_{2b}(h, \sigma,\dd z)$ with background measure $\sigma(\dd z)=\dd z$ being the Lebesgue measure on $\C$.  If $\theta(\dd z)=M^{g}_{2b}(X_g ,\dd z)$, one has $e^{2bc}\theta(\dd z)= M^{g}_{2b}(c+X_g ,\dd z)$ and 
we notice that $e^{2bc}R^{\theta}=R_\phi$ with $\phi=c+X_g$, if one uses the notation above (with a slight abuse of notation since now $\phi$ is not smooth anymore).
 First, we introduce the notion of GFF-Witten pair.
\begin{definition}[\textbf{GFF-Witten pair}]
Let $b\in (0,1)$. For $c\in \R$, a random distribution  $(X_g,\gamma_g)$ on $\mc{H}^{-s}(\hat{\C},\R)\times \mc{H}^{-s}(\hat{\C},\C)$, with $s>0$, such that $X_g$ is a GFF and, conditionally on $X_g$, $\gamma_g$ is a rescaled Witten field with respect to the pair $({\rm v}_g,e^{2bc}M^{g}_{2b}(X_g ,\dd z))$, is called a GFF-Witten pair with parameters $({\rm v}_g,M^{g}_{2b}(c+X_g ,\dd z))$. 
\end{definition}
Notice that if $\gamma_{g}^0$ is a rescaled Witten field with respect to the pair $({\rm v}_g,M^{g}_{2b}(X_g ,\dd z))$, then 
$\gamma_g:=e^{bc}\gamma_{g}^0$ is a rescaled Witten field with respect to the pair $({\rm v}_g,e^{2bc}M^{g}_{2b}(X_g ,\dd z))$. Thus, to construct the latter,
 it will be sufficient to consider the case $c=0$. 
 \begin{remark}\label{rem:integ}
For $\theta(\dd z)=M^{g}_{2b}(X_g ,\dd z)$, the integral kernel
\[
  R^{\theta} (z,z'):=\frac{1}{\pi^2} \int_\C \big(\frac{1}{z-w}-\kappa_g(w)\big)\big(\frac{1}{\bar{z}'-\bar{w}}-\bar{\kappa}_g(w)\big)\theta(\dd w)
\]
 of $R^{\theta}$  satisfies,  for all $z,z'$, $|R^{\theta} (z,z') |<\infty$   almost surely. However,   it is not true that almost surely, $\forall z,z'$ $|R^{\theta}(z,z') |<\infty$. 
 Indeed, since the kernel $R^{\theta}(z,z')$ is positive definite, by Cauchy Schwartz, it suffices to show
\[{R}^{\theta}(z,z)=\frac{1}{\pi^2} \int_\C \big|\frac{1}{z-w}-\kappa_g(w)\big|^2M^{g}_{2b}(X_g ,\dd w)<+\infty.\]
The finiteness of this integral follows from  the integrability of the singularity $|z-w|^{-2}$ with respect to $\theta$ since it satisfies, for fixed $z$ and any $\epsilon>0$, $M^{g}_{2b}(X_g ,B(z,r))\leq Cr^{2+2b^2-\epsilon}$ for some random constant $C$, see \cite[proof of Lemma 3.3]{DKRV16}, if $B(z,r)$ is the ball of center $z$ and radius $r$.
\end{remark}
 
 Now we study the existence in law of the GFF-Witten pair, and more precisely we will give a construction that is measurable with respect to a given GFF.
\begin{proposition}\label{existpair}
If $b\in (0,1)$, there exists a GFF-Witten pair $(X_g,\gamma_g^0)$ with parameters  $({\rm v}_g,M^{g}_{2b}(X_g ,\dd z))$, 
and thus a Witten pair $(X_g,\gamma_g=e^{bc}\gamma_g^0)$ with parameters  $({\rm v}_g,M^{g}_{2b}(c+X_g ,\dd z))$. Furthermore, given a GFF $X_g$, we can construct the law of the field $\gamma_g^0$ in a measurable way with respect to $X_g$.
\end{proposition}

\begin{proof} Consider a Witten pair $(X_g,\gamma_g)$ with parameters $({\rm v}_g,M^{g}_{2b}(X_g ,\dd z))$ and let $\theta(\dd z)=M_{2b}^g(X_g,\dd z)$. 
For $f\in L^2(\hat{\C},{\rm v}_g)$ real valued and $h,h'\in L^2(\hat{\C},{\rm v}_g)$ complex valued, we define (recall that $G_g$ is the Green's function of $\Delta_g$)
\[ u(z):=b\int_\Sigma G_g(z,z')f(z'){\rm dv}_g(z'),\quad \text{ and }\quad  \theta_{u}(\dd z)=e^{2u}\theta(\dd z)=M^{g}_{2b}(X_g+\frac{u}{b},\dd z).\]
Then the law of the Witten pair is characterized by
\[\E[e^{\langle X_g,f\rangle_2+\langle\gamma_g ,h\rangle_2+\langle \bar\gamma_g ,h'\rangle_2}]=e^{\frac{1}{2b}\langle f,u\rangle_2 }\E[e^{\frac{\pi}{k}\langle h,  R^{\theta_{u}} h'\rangle_2}]\]
for all $f\in \mc{H}^{s}(\hat{\C})$ real valued and $h,h'\in \mc{H}^{s}(\hat{\C})$ complex valued.

Let us now consider an orthonormal basis  $(f_n)_{n\geq 0}$ of real valued eigenfunctions of  $\Delta_g$ (which in turn is a basis of eigenfunction of $\mc{D}_0=\frac{1}{4}\Delta_g$) with increasing eigenfunctions $(\lambda_n)_{n\geq 0}$ (thus $\lambda_0=0$). Let us consider a sequence of real-valued i.i.d. Gaussian random variables $(\alpha_n)_{n\geq 1}$ and define (actually any GFF admits the following expansion)
\begin{equation}\label{hsxg}
X_g:=\sqrt{2\pi}\sum_{n\geq 1}\frac{\alpha_n}{\sqrt{\lambda_n}}f_n.
\end{equation}
Then the $\mc{H}^{-s}$-norm (defined using the metric $g$)  is easily evaluated in expectation 
$$\E[\|X_g\|^2_{\mc{H}^{-s}}]=2\pi \sum_{n\geq 1}\lambda_n^{-1-s}$$
which converges by the Weyl law for $s>0$. Hence the series \eqref{hsxg} converges almost surely. 
Let us now consider the GMC $\theta(\dd z)=M^{g}_{2b}(X_g,\dd z)$ with respect to this GFF, which  converges in probability.

We define a probability measure on $\R^{N}\times\C^{N}$, for all $N$, as the law of $(\alpha_1,\dots,\alpha_N,\beta_1,\dots, \beta_N)$ where $(\alpha_1,\dots,\alpha_N)$ are as above and, conditionally on $X_g$, the law of $(\beta_1,\dots, \beta_N)$ is that of complex Gaussians with covariances
\begin{equation}\label{betacov}
\E[\beta_n\bar{\beta}_m|X_g]=\int_\Sigma T_0 f_m(z) \bbar{T_0 f_n (z)}\theta(\dd z),\quad \E[\beta_n \beta_m|X_g]=0.
\end{equation}
 These probability laws are compatible so that, by Kolmogorov's theorem, there exists a random sequence $(\alpha_n,\beta_n)_{n\geq 1}\in \R^{\N^*}\times \C^{\N^*}$ with that law. From this sequence, we can define again $X_g$ as in \eqref{hsxg}, consider then $M^{g}_{2b}(X_g,\dd z)$ and, conditionally on $X_g$, the law of the sequence $(\beta_n)_{n\geq 1}$ is that of centered Gaussian random variables with covariance given by \eqref{betacov}. We set 
\[\gamma_g^0:=\Big(\frac{\pi}{k}\Big)^{1/2}\sum_{  n\geq 1} \beta_n f_n .\]
 The Sobolev norm of $\gamma^0_g$ (defined using the metric $g$) is then 
\[\|\gamma^0_g\|^2_{\mc{H}^{-s}}=\frac{\pi}{k}\sum_{n\geq 1}\lambda_n^{-s}|\beta_n|^2.\]
Therefore
\begin{align*}
\E[\|\gamma^0_g\|^2_{\mc{H}^{-s}}]
=&\frac{\pi}{k}\E\Big[\sum_{n\geq 1}\lambda_n^{-s}\E[|\beta_n|^2|X_g]\Big]
= \frac{\pi}{k}\E\Big[\sum_{n\geq 1}\lambda_n^{-s}\int_{\hat{\C}} |T_0 f_n(z)|^2\theta(\dd z)\Big]
\\
\leq & C\sum_{n\geq 1}\lambda_n^{-s}\int_{\hat{\C}} |T_0 f_n(z)|^2 \dd z=C \sum_{n\geq 1}\lambda_n^{-s-1 }
\end{align*}
where we used that $\E[M^{g}_{2b}(X_g,\dd z)]\leq C\dd z$ and that, when $\phi=0$, $\int_{\hat{\C}}|T_0 f_n(z)|^2 \dd z=\cjg R_{\phi=0}^{-1}f_n,f_n\cjd_2 =4\lambda_n^{-1}$. Therefore the last series converges for $s>0$ by the Weyl law.
Hence the process $\gamma_g^0$ converges in Sobolev spaces with negative index and it has the desired covariance.
\end{proof}

\noindent \textbf{Regularized GFF-Witten pair}.
To define later the correlation functions, we need to provide a regularisation of a GFF-Witten pair. We explain here how to proceed.  Let $(X_g,\gamma_g)$ be a GFF-Witten pair with parameters $({\rm v}_g,M^{g}_{2b}(X_g ,\dd z))$. Let $\rho$ be a smooth radial function with compact support on $\C$, nonnegative with $\int \rho(z)\,\dd z=1$. We set $\rho_{\epsilon,x}(y):=\epsilon^{-2}\rho((x-y)/\epsilon)$ for $y\in \C$ and
 \begin{equation}\label{regpair}
 \gamma^0_{g, \epsilon}(x):=\cjg \gamma^0_g,\rho_{\epsilon,x}\cjd_2.
 \end{equation}
 The mapping $x\mapsto \rho_{\epsilon,x}$ is continuous with values in $\mc{H}^{s}$ (for all $s>0$) so that the process $\gamma^0_{g,\epsilon}$ is continuous. The regularized GFF-Witten pair is the continuous process $(X_{g,\epsilon},\gamma^0_{g,\epsilon})$. 
 
 The conditional variance of $\gamma^0_{g,\epsilon}$ is
 \begin{equation}\label{exprtildeK}
\E[\gamma^0_{g,\epsilon}(z)\overline{\gamma^0_{g,\epsilon}}(z')|X_g]=\frac{1}{\pi k}  \int \big(\frac{\tilde{\rho}\big(\frac{|z-w|}{\epsilon}\big)}{z-w}-\kappa_g(w)\big)\big(\frac{\tilde{\rho}\big(\frac{|z-w|}{\epsilon}\big)}{\bar z'-\bar w}-\bar \kappa_g(w)\big)M^{g}_{2b}(X_g ,\dd w)
\end{equation}
with $\tilde{\rho}(u):=\int_0^\infty \mathbf{1}_{r<|u|}\rho(r)r\,\dd r$, where we have used the relation $\int_0^{2\pi}\frac{1}{z+\epsilon e^{i\theta}}\,\dd \theta=\frac{1}{z}\mathbf{1}_{|z|>\epsilon}$. Note that $\tilde{\rho}(u)=1$ when $|u|$ is large.

\subsection{Probabilistic formulation of the path integral}\label{sub:PIsphere}
The previous considerations lead to the following probabilistic expression for the  path integral \eqref{PIcondphi2}.  We set 
\[b=(k-2)^{-1/2}\in (0,1)\] 
and we consider a GFF-Witten pair $(X_g,\gamma_g)\in \mc{H}^{-s}(\hat{\C},\R)\times \mc{H}^{-s}(\hat{\C},\C)$ for some $s>0$,  with parameters  $({\rm v}_g, M^g_{2b}(c+X_g,\dd w))$, defined on some probability space $(\Omega,\mc{F},\P)$ (with expectation $\E$).

\begin{definition}\label{defH3sphere}
For $b=(k-2)^{-\frac{1}{2}}\in (0,1)$, the path integral for the $\H^3$-WZW model on $(\Sigma,g)$ is defined by 
\begin{equation}\label{rigPI}
\langle F(\phi,\gamma) \rangle_{\hat{\C},g}:=C_k\Big(\frac{\det(\Delta_g)}{{\rm v}_g(\hat{\C})}\Big)^{-\frac{1}{2}}\Big(\frac{\det(\mc{D}_0)}{{\rm v}_g(\hat{\C})}\Big)^{-1}\int_{\R\times\C}\E\Big[F(b(c+X_g), \gamma_0+\gamma_g)e^{-\frac{b}{4\pi}\int_\Sigma K_g(c+X_g)\dd {\rm v}_g}\Big]\,\dd c\, \dd \gamma_0
\end{equation}
for nonnegative measurable $F$ on $\mc{H}^{-s}(\Sigma,\R)\times \mc{H}^{-s}(\Sigma,\C)$. The constant $C_{k}$ is given by
\begin{equation}\label{defck}
C_k:=(\tfrac{k-2}{2\pi})^{\frac{1}{3}}(\tfrac{\pi}{k})^{-\frac{2}{3}}.
\end{equation}
\end{definition}
Notice that, since $\mc{D}_0=\frac{1}{4}\Delta_g$, then $\det(\mc{D}_0)=4^{2/3}\det(\Delta_g)$ (see \eqref{scalingdet}) so that \eqref{rigPI} can be rewritten as 
\[\langle F(\phi,\gamma) \rangle_{\hat{\C},g}:=(\tfrac{k-2}{32\pi})^{\frac{1}{3}}(\tfrac{\pi}{k})^{-\frac{2}{3}}\Big(\frac{\det(\Delta_g)}{{\rm v}_g(\hat{\C})}\Big)^{-\frac{3}{2}}\int_{\R\times\C}\E\Big[F(b(c+X_g), \gamma_0+\gamma_g)e^{-\frac{b}{4\pi}\int_\Sigma K_g(c+X_g)\dd {\rm v}_g}\Big]\,\dd c\, \dd \gamma_0.
\]
 
This path integral defines conjecturally a CFT. Though proving this statement requires more (like establishing Segal's gluing axioms), the central charge can already be  read off the variations of the conformal factor of the metric:
\begin{proposition}{\bf (Weyl anomaly and diffeormophism invariance)}\label{prop:weyl}  
\begin{description}
\item[ Weyl anomaly] Consider a metric $g'=e^{\omega}g$ conformal to the metric $g$ for some smooth $\omega\in C^\infty(\Sigma)$. Then, for nonnegative measurable $F$ on $\mc{H}^{-s}(\Sigma,\R)\times \mc{H}^{-s}(\Sigma,\C)$,
$$\langle F(\phi,\gamma) \rangle_{\Sigma,g'}=e^{ \frac{{\bf c}(k)}{96\pi}S^0_{\rm L}(g,\omega)}\langle F(\phi-\tfrac{b^2}{2}\omega,\gamma) \rangle_{\Sigma,g}$$
with ${\bf c}(k):=\frac{3k}{k-2}=3+6b^2$ the central charge of the $\H^3$-model and $S^0_{\rm L}$ the Liouville action \eqref{LiouvilleS0}.
\item[Diffeomorphism invariance] If $\psi:\Sigma\to\Sigma$ is a diffeomorphism and $\psi^*g$ the pullback metric then
\[\langle F(\psi^*\phi ,\psi^*\gamma) \rangle_{\Sigma,\psi^*g}= \langle F(\phi ,\gamma) \rangle_{\Sigma,g} .\]

\end{description}
\end{proposition}
The proof will be done in a more general setting in Proposition \ref{prop:weyl_surface}. 

\subsection{Correlation functions}
The   observables of the model we will consider in this manuscript are the operators  formally given, for $z\in \hat{\C}$, $j\in\R$ and $\mu\in\C$, by  
\begin{equation}
V_{j,\mu}(z):=e^{2(j+1)\phi(z)}e^{\mu \gamma(z)-\bar\mu \bar\gamma(z)}.
\end{equation}
Because the pair of fields $(X_g,\gamma_g)$  is irregular and cannot be evaluated pointwise, we need to regularize these  fields
\begin{equation}
V^{g, \epsilon}_{j,\mu}(z):=\epsilon^{2b^2(j+1)^2}e^{2(j+1)b\phi_{g,\epsilon}(z)}e^{\mu \gamma_{g,\eps}(z)-\bar\mu \bar\gamma_{g,\eps}(z)}
\end{equation}
where $(X_{g,\epsilon},\gamma_{g,\epsilon})$ is a regularized  GFF-Witten pair as in Subsection  \ref{regpair}, and $\phi_{g,\epsilon}=c+X_{g,\epsilon}$.

Another issue we have to deal with is the integrability in the zero mode $  \gamma_0$ so that the correlation functions will be given sense as distributions in the $\mu$-variables. For this, we fix $m\in \N$ and we introduce the   space of test functions $\mc{D}(\C^m)$ (smooth functions with compact support) equipped with the usual topology and $\mc{D}'(\C^m)$ its dual. We fix distinct points $z_1,\dots,z_m\in\hat{\C}$ that we gather in the notation ${\bf z}=(z_1,\dots,z_m)$ and weights $\boldsymbol{j}=(j_1,\dots,j_m) \in\R^m$. The correlation functions of the $\H^3$-model are seen as elements in $\mc{D}'(\C^m)$ with respect to the variable $\boldsymbol{\mu}:=(\mu_1,\dots,\mu_m)\in\C^m$, and are defined as the following limit in $\mc{D}'(\C^m)$, if it exists, 
\begin{equation}\label{correl}
\Big\langle \prod_{\ell=1}^mV_{j_\ell,\mu_\ell}(z_\ell) \Big\rangle_{\hat{\C},g}:=\lim_{T'\to \infty}\lim_{T\to \infty}\lim_{\epsilon\to 0}\Big\langle \chi_T(|\gamma_0|)\chi_{T'}(c)\prod_{\ell=1}^mV^{g, \epsilon}_{j_\ell,\mu_\ell}(z_\ell) \Big\rangle_{\hat{\C},g}
\end{equation}
where  $\chi_T(t)={\bf 1}_{|t|<T}$. 
The existence of this limit depends on the following condition
\begin{equation}\label{seib} 
\forall \ell=1,\dots,m,\quad j_\ell+1<1/2\quad \text{ and }\quad \sum_{\ell=1}^m(j_\ell+1)>1,
\end{equation}
similar to the Seiberg bounds for Liouville CFT (see \cite{DKRV16,GRVIHES} and Proposition \ref{limitcorel}).
From now on, we define the random measures, for $u\in C^\infty(\hat{\C}\setminus \{{\bf z}\},\C)$ satisfying $u(z)=\alpha_\ell \log |z-z_\ell|+\mc{O}(1)$ as $z\to z_\ell$ for some $\alpha_\ell\in \C$ satisfying $2{\rm Re}(\alpha_\ell)>-2-2b^2$ (this condition is explained in the Proposition below), by 
\begin{equation}\label{d:thetaMg} 
\theta(\dd z):= M_{2b}^g(X_g,\dd z), \qquad  \theta_u(\dd z):=e^{2{\rm Re}(u)}M_{2b}^g(X_g,\dd z).
\end{equation}
and we write $L^2(\theta_u)=L^2(\hat{\C},\theta_u)$ for the $L^2$-space with respect to $\theta_u$. The correlation functions will be understood in the distribution sense in $\C^m$ by using the measure $\delta_{V^0}$ induced by the Euclidean metric on the hyperplane  $V^0:=\{\boldsymbol{\mu}\in\C^m\,|\, \sum_{\ell=1}^m\mu_\ell=0\}$: 
\begin{equation}\label{deltaV^0}
 \delta_{V^0}(f)=\pi^2\int_{\C^{m-1}} f\Big(\mu_1,\mu_2,\dots,-\sum_{\ell=1}^{m-1}\mu_\ell \Big)\dd \mu_1\dots \dd \mu_{m-1}, \qquad \forall f\in C_c^\infty(\C^m).
 \end{equation}
 Notice that if $H\in C^\infty(\C^m\setminus\{0\})$ satisfies $|H(\boldsymbol{\mu})|=\mc{O}(|\boldsymbol{\mu}|^{-2s_0})$ with $2s_0<m-1$, then the product $H \delta_{V^0}$ makes sense as a distribution on $\C^m$, and if $H>0$, it is a 
 measure with infinite mass, but well-defined on compactly supported continuous functions. If $s_0=-1+\sum_{\ell=1}^m(j_\ell+1)$ and $j_\ell+1<1/2$ for all $\ell$, we have in particular $2s_0<m-2$. 

\begin{proposition}[\textbf{Correlations functions}]\label{existcorrel}
Let $s_0:=-1+\sum_{\ell=1}^m(j_\ell+1)$ and assume that the condition \eqref{seib} holds, thus $s_0>0$. Then the limit \eqref{correl} exists as distributions on $\C^m$ in the $\boldsymbol{\mu}$ variable and given by the measure supported on $V^0$
\[
\Big\langle \prod_{\ell=1}^mV_{j_\ell,\mu_\ell}(z_\ell) \Big\rangle_{\hat{\C},g}=\big(\frac{k-2}{32\pi}\big)^{\frac{1}{3}}(\frac{\pi}{k})^{-\frac{2}{3}-s_0}\Big(\frac{\det(\Delta_g)}{{\rm v}_g(\hat{\C})}\Big)^{-\frac{3}{2}}e^{B_{\H^3}({\bf z},\boldsymbol{j})}
\frac{\Gamma(s_0)}{2b} \E[ \|\Gamma_{\boldsymbol{\mu},{\bf z}}\|^{-2s_0}_{L^2(\theta_{u})}] \delta_{V^0}
\]
where $\Gamma(\cdot)$ is the Euler Gamma function, $\theta_u$ is the measure in \eqref{d:thetaMg} with 
\[
\Gamma_{\boldsymbol{\mu},{\bf z}}(z):=\frac{1}{\pi} \sum_{\ell=1}^m  \frac{\mu_\ell}{z_\ell-z}, \qquad u(z):= \sum_{\ell=1}^m 2b^2(j_\ell+1)G_{g}(z,z_\ell)-\frac{b^2}{4\pi} \int_\Sigma G_{g}(z,z')K_g(z'){\rm dv}_g(z'),\]
the expectation satisfies $ \E[ \|\Gamma_{\boldsymbol{\mu},{\bf z}}\|^{-2s_0}_{L^2(\theta_{u})}]=\mc{O}(|\boldsymbol{\mu}|^{-2s_0})=\mc{O}(|\boldsymbol{\mu}|^{-(m-2)})$ near $\boldsymbol{\mu}=0$ and the global factor is (recall \eqref{varYg})
\[\begin{split}
B_{\H^3}({\bf z},\boldsymbol{j}) := &2b^2\sum_{\ell=1}^m(j_\ell+1)^2W_g(z_\ell)+2b^2\sum_{\ell\not=\ell'}(j_\ell+1)(j_{\ell'}+1)G_g(z_\ell,z_{\ell'})\\
& +
\frac{b^2}{32\pi^2}\int_{\Sigma^2} G_g(z,z')K_g(z)K_g(z'){\rm dv}_g(z){\rm dv}_g(z') -2b^2\sum_{\ell=1}^m(j_\ell+1)\int_\Sigma G_g(z,z_\ell)K_g(z){\rm dv}_g(z).
\end{split}\]
\end{proposition}

\begin{proof} Let us introduce some further notations (recall \eqref{defck}):
\[ D_g:=C_k \Big(\frac{\det(\Delta_{g})}{{\rm v}_g(\Sigma)}\Big)^{-3/2},\qquad  \phi_g:=c+X_g, \qquad \phi_{g,\epsilon}:=c+X_{g,\epsilon},\]
Let us consider the limit in \eqref{correl} in $\mc{D}'(\C^m)$: for $f\in C_c^\infty(\C^m)$
 \begin{align}
  \int_{\C^m}&f(\boldsymbol{\mu})\Big\langle \chi_T(|\gamma_0|)\chi_{T'}(c) \prod_{\ell=1}^mV^{g, \epsilon}_{j_\ell,\mu_\ell}(z_\ell)\Big \rangle_{\Sigma,g}\dd\boldsymbol{\mu} \nonumber
 \\
 =&D_g \int_{\C} \int_{\C^m}f(\boldsymbol{\mu})e^{2i{\rm Im}(\sum_{\ell=1}^m \mu_\ell \gamma_0)} \chi_T(|\gamma_0|)\dd \gamma_0 \nonumber \\
 & \times 
 \int_{\R}\chi_{T'}(c)\E\Big[\prod_{\ell=1}^m\epsilon^{2b^2(j_\ell+1)^2}e^{2b(j_\ell+1)\phi_{g,\epsilon}(z_\ell)}e^{2i{\rm Im}(\mu_\ell  \gamma_{g,\eps}(z_\ell))}e^{-\frac{b}{4\pi}\int_\Sigma K_g\phi_{g,\eps}\dd {\rm v}_g}\Big]\,\dd c \,\dd\boldsymbol{\mu}. \label{expectation_term}
 \end{align}
First, for $\eps>0$, the expectation above is a continuous function of $\boldsymbol{\mu}$ and we will see below that the limit as $T\to \infty$ of the term in the last line of \eqref{expectation_term} is continuous in $\boldsymbol{\mu}\not=0$ and an $\mc{O}(|\boldsymbol{\mu}|^{-s_1})$ for some $s_1<m-1$ as $\boldsymbol{\mu}\to 0$. 
Suppose $\tilde{f}_{T'}\in C^0(\C^m)$ is a family of functions with uniform compact support with respect to $T'>1$ such that 
 $|\tilde{f}_{T'}(\boldsymbol{\mu})|=\mc{O}(|\boldsymbol{\mu}|^{-s_1})$ for $s_1<m-1$ near $\boldsymbol{\mu}=0$ uniformly in $T'>1$ and $\tilde{f}_{T'}\to \tilde{f}$ as $T'\to \infty$ pointwise in $\C^m\setminus\{0\}$ for some $\tilde{f}\in C^0(\C^m\setminus\{0\})$
satisfying  $|\tilde{f}(\boldsymbol{\mu})|=\mc{O}(|\boldsymbol{\mu}|^{-s_1})$ for $s_1<m-1$.
Then it is straightforward to show that
 \begin{equation}\label{deltamu}
 \begin{split}
\lim_{T'\to \infty} \lim_{T\to \infty}\int_{|\gamma_0|<T}\Big(\int_{\C^m}\tilde{f}_{T'}(\boldsymbol{\mu})e^{2i{\rm Im}(\sum_{\ell=1}^m\mu_\ell \gamma_0)}\dd\boldsymbol{\mu} \Big)\dd \gamma_0
=&\pi^2  \int_{\C^{m-1}} \tilde{f}\Big(\mu_1,\mu_2,\dots,-\sum_{\ell=1}^{m-1}\mu_\ell \Big)\dd \mu_1\dots \dd \mu_{m-1}\\
= &\delta_{V^0}(\tilde{f})
 \end{split}
 \end{equation}
where $\delta_{V^0}$ is the measure on the hyperplane $V^0$ induced by the Euclidean metric on $\C^m$.
 Next we compute the expectation in \eqref{expectation_term}. We apply the Cameron-Martin's theorem, applied conditionally on $X_g$, to shift the 
 $\gamma_{g,\eps}(z_\ell)$ terms: using \eqref{exprtildeK}, the expectation and $c$-integral in \eqref{expectation_term}  can be rewritten as 
 \begin{equation}\label{Expectation_part}
  \int_{\R}\chi_{T'}(c)\E\Big[\prod_{\ell=1}^m\epsilon^{2b^2(j_\ell+1)^2}e^{2b(j_\ell+1)\phi_{g,\epsilon}(z_\ell)} e^{-\frac{b}{4\pi}\int_\Sigma K_g\phi_g\dd {\rm v}_g}e^{-\frac{\pi}{k}\sum_{\ell,\ell'=1}^m\mu_\ell\bar\mu_{\ell'}R^\theta_\epsilon(z_\ell,z_{\ell'})}\Big]\,\dd c .
  \end{equation}
  where  $R^\theta_\epsilon(z,z')$ is given by 
\[   \frac{\pi}{k} R^\theta_\epsilon(z,z'):=\E[\gamma_{g,\epsilon}(z)\overline{\gamma_{g,\epsilon}}(z')|X_g].\]
Using  the expression \eqref{exprtildeK} for $R^\theta_\eps(z_\ell,z_{\ell'})$ and the fact that we can restrict to $\sum_{\ell=1}^m\mu_\ell=0$ by \eqref{deltamu}, we get (recall \eqref{d:thetaMg})
\begin{equation}
 \frac{\pi}{k} \sum_{\ell,\ell'=1}^m\mu_\ell\bar\mu_{\ell'}R^\theta_\epsilon(z_\ell,z_{\ell'})=  \frac{e^{2bc}}{\pi k}\int_\C\Big|\sum_\ell \mu_\ell\frac{\tilde{\rho}\big(\frac{|z_\ell-z|}{\epsilon}\big)}{z_\ell-z}\Big|^2 \theta(\dd z) .
\end{equation}
 Now we deal with the $X_{g,\epsilon}$ insertions. Using Cameron-Martin's theorem to shift the terms involving $X_{g,\eps}(z_\ell)$ and 
 $\int_{\C}K_g X_{g,\eps}{\rm dv}_g$, we obtain (recall \eqref{varYg}) that the expectation with $c$-integral in \eqref{expectation_term} is equal to (as $\eps\to 0$)
\begin{equation}\label{eq:to_bound}
(1+o(1))e^{B_{\H^3}({\bf z},\boldsymbol{j})}
\int_{\R}\chi_{T'}(c)e^{-2bc+\sum_{\ell=1}^m2b(j_\ell+1) c}\E\Big[e^{-\frac{e^{2bc}}{\pi k}\int_\C\big|\sum_\ell \mu_\ell\frac{\tilde{\rho}(|z_\ell-z|/\epsilon)}{z_\ell-z}\big|^2 \theta_{u_\eps}(\dd z)}\Big]\,\dd c 
\end{equation}
with $\theta_{u_\eps}$ defined as in \eqref{d:thetaMg} with 
  \[u_\eps(z):=\sum_{\ell=1}^m 2b^2(j_\ell+1)G_{g,\eps}(z,z_\ell)-\frac{b^2}{4\pi} \int_\Sigma G_{g,\eps,\eps}(z,z')K_g(z'){\rm dv}_g(z').\]
where $G_{g,\epsilon}$ (resp. $G_{g,\eps,\eps}$) is the $g$-regularization of the Green function 
 its second variable (resp. both variables). We have
 \[u(z)= \lim_{\epsilon\to 0}u_\eps(z)=\sum_{\ell=1}^m 2b^2(j_\ell+1)G_{g}(z,z_\ell)-\frac{b^2}{4\pi} \int_\Sigma G_{g}(z,z')K_g(z'){\rm dv}_g(z')\] 
 pointwise outside the points  $z_\ell$. Let us define 
 \[\Gamma^\eps_{\boldsymbol{\mu},{\bf z}}(z):=\frac{1}{\pi}\sum_{\ell=1}^m \mu_\ell\frac{\tilde{\rho}(|z_\ell-z|/\epsilon)}{z_\ell-z}.\]
Now we claim that the term $\|\Gamma^\eps_{\boldsymbol{\mu},{\bf z}}\|_{L^2(\theta_{u_\eps})}^2=\int_\C\big|\Gamma^\eps_{\boldsymbol{\mu},{\bf z}}\big|^2 \theta_{u_\eps}(\dd z)$ converges almost surely towards
\begin{equation}\label{integrand}
\|\Gamma_{\boldsymbol{\mu},{\bf z}}\|_{L^2(\theta_{u})}^2=\int_\C| \Gamma_{\boldsymbol{\mu},{\bf z}}(z)|^2 \theta_u(\dd z).
\end{equation}
Indeed, this follows from the dominated convergence theorem after noticing that the quantity \eqref{integrand} is finite, because the singularity at $z_\ell$ is of the  form $|z_\ell-z|^{-2-4b^2(j_\ell+1)}$, which is integrable with respect to the measure $\theta(\dd z)=M^{g}_{2b}(X_g  ,\dd z)$ provided that $2+4b^2(j_\ell+1)<2+2b^2$, namely $j_\ell<-1/2$ (see \cite[Lemma 3.3]{DKRV16}). Also, we have the following estimate that controls the behaviour near $z\to\infty$
\[  |z|\geq2 \max_\ell |z_\ell |\Rightarrow | \Gamma_{\boldsymbol{\mu},{\bf z}}(z)|^2\leq \frac{C}{|z|^4} \]
for some constant $C>0$ (here we used $\sum_\ell \mu_\ell=0$). In particular
\[\E\Big[\int_{ |z|\geq 2 \max_\ell |z_\ell |}| \Gamma_{\boldsymbol{\mu},{\bf z}}(z)|^2 \theta_{u}(\dd z)\Big]\leq C \int_{ |z|\geq2 \max_\ell |z_\ell |}    \frac{e^{2b^2W_g(z)}}{|z|^4}  \dd z\]
using the previous bound and the  fact that $u$ is bounded on  the set $ |w|\geq2 \max_\ell |z_\ell |$. Since $W_g$ is bounded on the whole plane, the expectation just above is finite, that it has a limit as $\eps\to 0$, that the limit is uniformly bounded with respect to $T'$ by $\mc{O}(|\boldsymbol{\mu}|^{-2s_0})$ and that it has a limit as $T'\to \infty$. 
Now we need to show that \eqref{Expectation_part} is finite and that its  limit as 
$T\to \infty$ exists. Since all the terms in \eqref{Expectation_part} are positive, we can bound above $\chi_{T'}(c)\leq 1$ and use the dominated convergence: by Fubini, we exchange the $c$ integral and the expectation and by a change of variable we see that
\[\int_{\R}e^{-2bc+\sum_{\ell=1}^m2b(j_\ell+1) c}e^{-\frac{e^{2bc}\pi}{k}\|\Gamma^\eps_{\boldsymbol{\mu},{\bf z}}\|_{L^2(\theta_{u_\eps})}^2} \dd c=\frac{(k/\pi)^{s_0}}{2b}  \|\Gamma^\eps_{\boldsymbol{\mu},{\bf z}}\|^{-2s_0}_{L^2(\theta_{u_\eps})}\Gamma(s_0)
\]
where $s_0=-1+\sum_{\ell}(j_\ell+1)>0$ and $\Gamma(s)$ is Euler gamma function. Now by (the proof of) \cite[Lemma 3.3]{DKRV16}, for each $\boldsymbol{\mu}=0$ and if $j_\ell+1<1/2$ for each $\ell$, we have
\[ \sup_{\eps\in (0,1)}\E[  \|\Gamma^\eps_{\boldsymbol{\mu},{\bf z}}\|^{-2s_0}_{L^2(\theta_{u_\eps})}]<\infty\] 
and its limit as $\eps\to 0$ is given by $\E[  \|\Gamma_{\boldsymbol{\mu},{\bf z}}\|^{-2s_0}_{L^2(\theta_{u})}]<\infty$.
This shows that 
\begin{align*} 
H_{T'}(\boldsymbol{\mu}):=&\lim_{\eps\to 0}\int_{\R}\chi_{T'}(c)e^{-2bc+\sum_{\ell=1}^m2b(j_\ell+1) c}\E\Big[e^{-\frac{e^{2bc}\pi}{k}\int_\C|\Gamma^\eps_{\boldsymbol{\mu},{\bf z}}(z)|^2 \theta_{u_\eps}(\dd z)}\Big]\,\dd c\\
=&\int_{\R}\chi_{T'}(c)e^{-2bc+\sum_{\ell=1}^m2b(j_\ell+1) c}\E\Big[e^{-\frac{e^{2bc}\pi}{k}\int_\C|\Gamma_{\boldsymbol{\mu},{\bf z}}(z)|^2 \theta_{u}(\dd z)}\Big]\,\dd c\leq \frac{(k/\pi)^{s_0}}{2b}  \|\Gamma_{\boldsymbol{\mu},{\bf z}}\|^{-2s_0}_{L^2(\theta_{u})}\Gamma(s_0).
\end{align*}
By dominated convergence, $\lim_{T'\to \infty}H_{T'}(\boldsymbol{\mu})=\frac{(k/\pi)^{s_0}}{2b}  \|\Gamma_{\boldsymbol{\mu},{\bf z}}\|^{-2s_0}_{L^2(\theta_{u})}\Gamma(s_0)$ for all $\boldsymbol{\mu}\not=0$.
Remark that since one can bound below $|\Gamma_{\boldsymbol{\mu},{\bf z}}(z)|>C_0|\boldsymbol{\mu}|$ 
for $z$ in a small ball where $\Gamma_{\boldsymbol{\mu},{\bf z}}$ has no zeros, with $C_0>0$ some constant, 
we obtain
\[ \E[  \|\Gamma_{\boldsymbol{\mu},{\bf z}}\|^{-2s_0}_{L^2(\theta_{u})}] \leq C_1|\boldsymbol{\mu}|^{-2s_0}\]
for some constant $C_1>0$ and thus $|H_{T'}(\boldsymbol{\mu})|=\mc{O}(|\boldsymbol{\mu}|^{-2s_0})$ uniformly in $T'$. We are in the situation where we can use \eqref{deltamu} with $\tilde{f}(\boldsymbol{\mu})=f(\boldsymbol{\mu})H_{T'}(\boldsymbol{\mu})$ since $2s_0<m-2$. This thus shows the desired result.
\end{proof}

Varying the metric in the correlation functions leads to identifying the conformal weights of the operators $V_{j,\mu}(z)$. The proof is contained in  Proposition \ref{confweightgeneral} written in a general framework.
\begin{proposition}\label{confweight}
Let $g'=e^\omega g$ be a metric conformal to $g$. Then
\begin{equation}\label{e:anomaly_WZW}
\Big\langle\prod_{\ell=1}^mV_{j_\ell,\mu_\ell}(z_\ell) \Big \rangle_{\hat{\C},g'}=  e^{ \frac{{\bf c}(k)}{96\pi}S^0_{\rm L}(g,\omega)-\sum_{\ell=1}^m \triangle_{j_\ell}\omega(z_\ell)}\Big\langle\prod_{\ell=1}^mV_{j_\ell,\mu_\ell}(z_\ell) \Big\rangle_{\hat{\C},g}
\end{equation}
where ${\bf c}(k)=\frac{3k}{k-2}=3+6b^2$ is   the central charge and $\triangle_{j_\ell} $ is the  conformal weight of $V_{j_\ell,\mu_\ell} $,  given by
\begin{equation}\label{CWH3}
\triangle_{j_\ell} =-\frac{j_\ell(j_\ell+1)}{k-2}=-b^2j_\ell(j_\ell+1).
\end{equation}
\end{proposition}

\section{Correspondence between  $\mathbb{H}^3$-WZW and Liouville theory on the sphere}

In the following statement, we will relate the correlation functions of the $\mathbb{H}^3$-model to the correlation functions of the Liouville CFT for the trivial holomorphic bundle.  To emphasize which model the correlation functions refer to, we will add the superscript $\cjg \cdot\cjd^{\mathbb{H}^3}$ or $\cjg \cdot\cjd^{\rm L}$ to the corresponding correlations. Before stating the main result, we first recall the construction of the correlation functions in the Liouville model.

\subsection{Liouville CFT path integral}
  Let  $\Sigma$ be a closed Riemann surface of genus $g(\Sigma)$, equipped with a Riemannian metric $g$ compatible with the complex structure and  consider the parameters 
 \[ b\in(0,1), \quad \mu_{\rm L}>0 , \quad Q=b+b^{-1}.\] 
  The Liouville field is the random variable $\phi_g=c+X_g\in \mc{H}^{s}(\Sigma)$ (for any $s<0$) with $c\in\R$ and $X_g$ the Gaussian Free Field on $(\Sigma,g)$. For $F:  \mc{H}^{s}(\Sigma)\to\R^+$ (with $s<0$) a bounded continuous functional, we set 
\begin{align}\label{def:pathintegralF}
\big\cjg F(\phi)\big\cjd_{\Sigma, g,\mu_{\rm L}}^{\rm L}:=& \big(\frac{\det(\Delta_g)}{{\rm v}_{g}(\Sigma)}\big)^{-\frac{1}{2}}  \int_\R  \E\Big[ F( \phi_g) \exp\Big( -\frac{Q}{4\pi}\int_{\Sigma}K_{g}\phi_g\,{\rm dv}_{g} - \mu_{\rm L}   M^g_{2b}(\phi_g,{\rm v}_g,\Sigma)  \Big) \Big]\,\dd c. 
\end{align}
By \cite[Proposition 4.1]{GRVIHES}, this quantity  defines a measure and moreover the partition function defined as  the total mass of this measure, i.e $\cjg 1 \cjd^{\rm L}_ {\Sigma, g,\mu_L}$, is finite iff the genus $g(\Sigma)\geq 2$. We warn the reader that in \cite{GRVIHES}, the parameters are denoted $(\gamma,Q,\mu)$ and correspond to $(2b,Q,\mu_{\rm L})$ in the notation of the present paper.

The regularized vertex operators, for fixed $\alpha\in\R$ (called {\it weight}) and $x\in \Sigma$, are defined by
\[
V_{\alpha,\eps}(x):=\eps^{\alpha^2/2}  e^{\alpha  \phi_{g,\epsilon}(x) } 
\]
with $\phi_{g,\epsilon}$ the $g$-regularisation of $\phi_g$. The correlation functions are defined by the limit
\begin{equation}\label{defcorrelg}
\Big\cjg \prod_i V_{\alpha_i}(x_i) \Big\cjd^{\rm L}_ {\Sigma, g,\mu_{\rm L}}:=\lim_{\eps \to 0} \: \Big\cjg \prod_i V_{\alpha_i,\eps}(x_i) \Big\cjd^{\rm L}_ {\Sigma, g,\mu_{\rm L}}
\end{equation}
where we have fixed $m$ distinct points $x_1,\dots,x_m$  on $\Sigma$ with respective associated weights $\alpha_1,\dots,\alpha_m\in\R$. By convention, we shall write for the case $\mu_L=1$
\begin{equation}\label{muL=1}
\Big\cjg \prod_i V_{\alpha_i}(x_i) \Big\cjd^{\rm L}_ {\Sigma, g}:=\Big\cjg \prod_i V_{\alpha_i}(x_i) \Big\cjd^{\rm L}_ {\Sigma, g,\mu_{\rm L}=1}.
\end{equation}
Non triviality of correlation functions are summarized in the following proposition (see \cite[Prop 4.4]{GRVIHES}):

 \begin{proposition}\label{limitcorel} Let $\mathbf{x}:=( x_1,\dots,x_m)\in\Sigma^m$ be distinct points on  $\Sigma$ and $\boldsymbol{\alpha}:= (\alpha_1,\dots,\alpha_m)\in\R^m$.
The limit \eqref{defcorrelg} exists and is non zero if and only if the weights $(\alpha_1,\dots,\alpha_m)$ obey the Seiberg bounds
 \begin{align}\label{seiberg}
 & s(\boldsymbol{\alpha}):=\sum_{i}\alpha_i + 2 Q (g(\Sigma)-1)>0, \qquad \forall i,\quad \alpha_i<Q.
 \end{align}
 They satisfy the following conformal covariance: if $g'=e^{\omega} g$ for $\omega \in C^\infty(\Sigma)$
 \begin{equation}\label{LCFT_anomaly}
\Big\cjg \prod_i V_{\alpha_i}(x_i) \Big\cjd^{\rm L}_ {\Sigma, g',\mu_{\rm L}}=e^{\frac{1+6Q^2}{96\pi}S_{\rm L}^0(g,\omega)-\sum_{i}\frac{\alpha_i}{2}(Q-\frac{\alpha_i}{2})\omega(x_i)}\Big\cjg \prod_i V_{\alpha_i}(x_i)\Big \cjd^{\rm L}_ {\Sigma, g,\mu_{\rm L}}.
 \end{equation}
 \end{proposition}

Under \eqref{seiberg}, the correlation functions  have the following probabilistic representations \cite[Prop 4.4]{GRVIHES}  \begin{equation}\label{zidaneilamarque}
\begin{split}
\Big\cjg \prod_i V_{\alpha_i}(x_i) \Big\cjd^{\rm L}_ {\Sigma, g,\mu_{\rm L}} =&  \big(\frac{\det(\Delta_{g})}{{\rm v}_{g}(\Sigma)}\big)^{-\frac{1}{2}}   e^{B_{\rm L}(\mathbf{x},\boldsymbol{\alpha})}\int_\R e^{s(\boldsymbol{\alpha})c}\E[e^{-\mu_{\rm L}e^{2bc}M^g_{2b}(X_g+u_{\boldsymbol{\alpha}},{\rm v}_g,\Sigma)}]\dd c\\
=& \big(\frac{\det(\Delta_{g})}{{\rm v}_{g}(\Sigma)}\big)^{-\frac{1}{2}}   e^{B_{\rm L}(\mathbf{x},\boldsymbol{\alpha})}     \:   \mu_{\rm L}^{-s(\boldsymbol{\alpha})} \frac{\Gamma( \frac{s(\boldsymbol{\alpha})}{2b})}{2b} \:   \E\Big[   M^g_{2b}(X_g+u_{\boldsymbol{\alpha}},{\rm v}_g,\Sigma)^{-\frac{s(\boldsymbol{\alpha})}{2b}} \Big] 
\end{split}
\end{equation} 
 where $\Gamma(z)$ is Euler gamma function and, if $W_g$ is the function appearing in   \eqref{varYg}, 
 \begin{equation}\label{Z0Cx}
u_{\boldsymbol{\alpha}}(x):=   \sum_{i=1}^m \alpha_i   G_{g}(x_i,x)-\frac{Q}{4\pi}\int_\Sigma G_g(x,y)K_g(y){\rm dv}_g(y),
 \end{equation}
\begin{equation}\label{BL(x,alpha)}
\begin{split}
B_{\rm L}(\mathbf{x},\boldsymbol{\alpha}) := \sum_i\frac{\alpha_i^2}{2}W_{g}(x_i)+ \sum_{i<j}\alpha_i\alpha_j G_{g}(x_i,x_j)\ -\frac{Q}{4\pi}\Big(\cjg \sum_{i}\alpha_iG_g(x_i,\cdot),K_g\cjd_2-\frac{Q}{8\pi}\cjg G_g,K_g\otimes K_g\cjd_{L^2(\Sigma^2)}\Big)
\end{split}
\end{equation} 
In particular in constant curvature, all terms involving $K_g$ in \eqref{Z0Cx} and \eqref{BL(x,alpha)} vanish.

 \subsection{DOZZ formula} 
  The DOZZ formula is an explicit expression, in terms of special functions, for the structure constants of Liouville CFT,  which correspond to   the three point correlation function on the Riemann sphere, thereafter identified  with the extended complex plane $\hat\C$ by stereographic projection. On the sphere, every metric $g$ is (up to diffeomorphism) conformal to the round metric 
 \begin{equation}\label{def_g_0}
  g_0:=\frac{4}{(1+|z|^2)^2}|dz|^2.
  \end{equation} 
  From Proposition \ref{limitcorel}, the $m$-point correlations exist and are non trivial   if and only if the  Seiberg bounds \eqref{seiberg} are satisfied, here with genus $g(\Sigma)=0$. This implies in particular that $m$ must be greater or equal to $3$. 
Using the conformal covariance of the correlation functions, the 3 point functions depend in an explicit way on the points $z_1,z_2,z_3$ and the metric $g$, and are thus uniquely determined up to some factor depending only on the insertion weights $\alpha_1,\alpha_2,\alpha_3$ (and the parameters $b,\mu$) called the {\it structure constant}. More explicitly, it was proven   in \cite{KRV_DOZZ}:
\begin{align}
\label{3pointDOZZ}
 \Big\langle  &V_{\alpha_1}(z_1)  V_{\alpha_2}(z_2) V_{\alpha_3}(z_3)  \Big\rangle_{\hat\C,g,\mu_{\rm L}}^{\rm L}\\
  =&|z_1-z_3|^{2(\Delta_{\alpha_2}-\Delta_{\alpha_1}-\Delta_{\alpha_3})}|z_2-z_3|^{2(\Delta_{\alpha_1}-\Delta_{\alpha_2}-\Delta_{\alpha_3})}|z_1-z_2|^{2(\Delta_{\alpha_3}-\Delta^{}_{\alpha_1}-\Delta_{\alpha_2})}
\Big(\prod_{i=1}^3g(z_i)^{-\Delta_{\alpha_i}}\Big)\nonumber \\
&\times \frac{1}{2}C_{2b,\mu_{\rm L}}^{{\rm DOZZ}} (\alpha_1,\alpha_2,\alpha_3 ) \big(\frac{{\rm v}_{g}(\hat\C)}{{\det}'(\Delta_{g})}\big)^\hf e^{-6Q^2 S_{\rm L}^0(\hat\C,g,\omega)}  \nonumber
 \end{align} 
where the conformal weights are given by $\Delta^{\rm L}_{\alpha}:=\frac{\alpha}{2}(Q-\frac{\alpha}{2})$ (they differ from those of the $\mathbb{H}^3$-model), the function $\omega:\hat\C\to \R$ is defined by $e^\omega g= |z|_+^{-4}|dz|^2$, $S_{\rm L}^0(\hat\C,g,\omega)$ is the Liouville action given by \eqref{LiouvilleS0} and $C_{2b,\mu_{\rm L}}^{{\rm DOZZ}} (\alpha_1,\alpha_2,\alpha_3 ) $ is an explicit  function whose expression  can be found in   \cite[Subsection 1.2]{KRV_DOZZ}.

\subsection{Correspondence between $\mathbb{H}^3$-model and Liouville CFT}
The main result of this subsection is a correspondence between $\H^3$-WZW and Liouville theory. 
It expresses a mapping between the correlation functions of the $\mathbb{H}^3$-model and those of the Liouville CFT, but in the Liouville side one gets new insertions coming from the zeros $(y_\ell)_{\ell}$ of the meromorphic $1$-form $z\mapsto \Gamma_{\boldsymbol{\mu},{\bf z}}(z)dz$ on $\hat{\C}$. Let $r\in [ 1,m-1 ]$ be the smallest nonnegative integer such that    $\sum_\ell\mu_\ell z_\ell^{r}\not=0$. There exists $p=p(r)$ and distinct numbers $y_1,\dots,y_{p}\in \C$ and $n_1,\dots,n_{p}\in\N\setminus\{0\}$ such that $n_1+\dots+n_p+r-1=m-2$ and the following factorisation holds (see the proof of Theorem \ref{Liouvillesphere2} for details) 
\begin{equation}\label{polyndeg}
\forall z\in\C,\quad \sum_\ell  \frac{\mu_\ell}{z_\ell-z}=(-1)^{r}\Big(\sum_{\ell=1}^m\mu_\ell z_\ell^r\Big)\frac{\prod_{\ell=1}^p(y_\ell-z)^{n_\ell}}{\prod_{\ell=1}^m(z_\ell-z)}.
\end{equation}
The points $y_1,\dots,y_p$ are the zeros of $z\mapsto \Gamma_{\boldsymbol{\mu},{\bf z}}(z)dz$ in $\C$ with multiplicities 
$(n_1,\dots,n_p)$ and $y_{p+1}=\infty$ is also a zero with multiplicity $n_{p+1}=r-1$ iff $r\geq 2$. We denote by 
$m'$ the number of distincts zeros, i.e. $m'=p$ if $r=1$ ($z=\infty$ is not a zero) and $m'=p+1$ if $r\geq 2$ ($z=\infty$ is a zero).
We call these zeros $(y_\ell)_{\ell=1,\dots,m'}$
the \textbf{new insertions}, and remark that they depend on both $z_\ell$ and $\mu_\ell$. To avoid confusions, from now on 
we shall denote $\Delta^{\rm L}_{\alpha}=\frac{\alpha}{2}(Q-\frac{\alpha}{2})$ the Liouville conformal weight and $\Delta_{j}=-b^2j(j+1)$ the $\H^3$-WZW conformal weight. As before, we take $b\in(0,1)$.

 \begin{theorem}\label{Liouvillesphere2}
Let $(\hat{\C},g)$ be the Riemann sphere equipped with a metric $g=e^{\omega}g_0$ conformal to the round metric \eqref{def_g_0} 
and ${\bf z}=(z_1,\dots,z_m)$ distincts points in $\C$.  Assume the conditions \eqref{seib} hold and let ${\bf y}=(y_1,\dots,y_{m'})$ 
be the new insertions as defined just above, i.e. the zeros of the meromorphic $1$-form $z\mapsto \Gamma_{\boldsymbol{\mu},{\bf z}}(z)dz$ in $\hat{\C}$, with multiplicities $(n_\ell)_{\ell=1,\dots,m'}$.
Then, with $\delta_{V^0}$ the measure \eqref{deltaV^0}, one has in $\mc{D}'(\C^m)$
\[\Big \langle \prod_{\ell=1}^mV_{j_\ell,\mu_\ell}(z_\ell) \Big\rangle_{\Sigma,g}^{\mathbb{H}^3}
=\mu_{\rm L}^{-s_0}
\mc{F}_{{\bf z},\boldsymbol{j}}({\bf y})\, \Big\cjg \prod_{\ell=1}^{m} V_{\alpha_\ell}(z_\ell)\prod_{\ell=1}^{m'} 
V_{-\frac{n_\ell}{b}}(y_\ell)\Big\cjd^{\rm L}_ {\Sigma, g} \delta_{V^0}\]
where $\alpha_\ell=2b(j_\ell+1)+\frac{1}{b}$,  $\mu_{\rm L}:=\frac{ |\sum_{\ell=1}^m\mu_\ell z_\ell^r |^2}{\pi k}$ and the function $\mc{F}_{{\bf z},\boldsymbol{j}}$ is given by 
\begin{equation}\label{calculFdeg}
\begin{split}
\mc{F}_{{\bf z},\boldsymbol{j}}({\bf y})  &= \big(\frac{k-2}{32\pi}\big)^{\frac{1}{3}}\big(\frac{k}{\pi}\big)^{\frac{2}{3}} 
\Big(\frac{\det(\Delta_{g})}{{\rm v}_g(\Sigma)}\Big)^{-1} 4^{\Delta^{{\rm L}}_{-\frac{r-1}{b}}}   e^{-2\chi(Q^2-b^2) }
 e^{\frac{b^2-Q^2}{16\pi}S^0_{\rm L}(g_0,\omega)}
\prod_{\ell=1}^mg(z_\ell)^{\frac{1}{2}+\frac{1}{4b^2}}\prod_{\ell=1}^pg(y_\ell)^{\Delta^{\rm L}_{-n_\ell/b}} \\
& \times \prod_{\ell\not=\ell'=1}^m |z_\ell-z_{\ell'}|^{j_\ell+j_{\ell'}+2+\frac{1}{2b^2}}
\prod_{\ell=1}^m\prod_{\ell'=1}^p |z_\ell-y_{\ell'}|^{-2n_{\ell'}(j_\ell+1)-\frac{n_{\ell'}}{b^2}}\prod_{\ell\not=\ell'=1}^p|y_\ell-y_{\ell'}|^{\frac{n_\ell n_{\ell'}}{2b^2}}.
\end{split}
\end{equation}
Here we recall that $y_\ell$ depends on $\boldsymbol{\mu}$, $\chi=\log 2-1/2$ and $Q=b+b^{-1}$, the Liouville action $S^0_{\rm L}(g_0,\omega)$ is defined by  \eqref{LiouvilleS0} and $\Delta^{\rm L}_{\alpha}=\frac{\alpha}{2}(Q-\frac{\alpha}{2})$ is the Liouville conformal weight
\end{theorem}
\begin{proof} The starting point is the expression for correlation functions in Proposition \ref{existcorrel}.
The Weyl covariance  (see Prop. \ref{confweight} and \cite[Prop 4.6]{GRVIHES} for Liouville) allows us to choose the metric as we please, and we will work with $g$ the round metric on the Riemann sphere as it simplifies a few aspects of this expression. First, with $g_0$ the round metric \eqref{def_g_0}, the scalar curvature is  $K_{g_0}=2$ so that
\[\int_{\hat{\C}} K_{g_0}X_{g_0}\dd {\rm v}_{g_0}=0\quad \text{ and }\quad \int_{\hat{\C}} K_{g_0}(z')G_{g_0}(z,z') \dd {\rm v}_{g_0}(z')=0.\]
Second, recall now the expression for the Green function $G_{g_0}$ in the round metric
\begin{equation}\label{greenS}
G_{g_0}(z,z')=-\log|z-z'|-\frac{1}{4}\log g_0(z)-\frac{1}{4}\log g_0(z')+\chi
\end{equation}
where we have identified $g_0=g_0(z)|dz|^2$, and with $\chi=\log 2-1/2$. In particular, the Robin mass is given by $W_g=\chi$.  The expression for correlation functions in Proposition \ref{existcorrel} thus becomes
\begin{align*}
\big(\frac{k-2}{32\pi}\big)^{\frac{1}{3}}(\frac{\pi}{k})^{-\frac{2}{3}-s_0}\Big(\frac{\det(\Delta_{g_0})}{{\rm v}_{g_0}(\hat{\C})}\Big)^{-\frac{3}{2}}e^{B_{\H^3}({\bf z},\boldsymbol{j})}
\frac{1}{2b}\Gamma(s_0) \E[ \|\Gamma_{\boldsymbol{\mu},{\bf z}}\|^{-2s_0}_{L^2(\theta_{u})}] \,\delta_{V^0}
\end{align*}
with  $B_{\H^3}({\bf z},\boldsymbol{j})= \sum_{\ell=1}^m 2b^2(j_\ell+1)^2\chi+ \sum_{\ell\not=\ell'=1}^m 2b^2(j_\ell+1)(j_{\ell'}+1)G_{g_0}(z_\ell,z_{\ell'})$ and $u(z):=\sum_{\ell=1}^m 2b^2(j_\ell+1)G_{g_0}(z,z_\ell)$. Using \eqref{greenS}, we rewrite $u$ as 
\[u (z)=\sum_{\ell=1}^m 2b^2(j_\ell+1)\log\frac{1}{|z-z_\ell|}-\sum_{\ell=1}^m 2b^2(j_\ell+1)(\tfrac{1}{4}\log g_0(z_\ell)-\chi)-\tfrac{1}{4}\log g_0(z)\sum_{\ell=1}^m 2b^2(j_\ell+1).\]
 We rewrite the correlations as 
 \[ \begin{split}
\Big \langle \prod_{\ell=1}^mV_{j_\ell,\mu_\ell}(z_\ell) \Big\rangle_{\hat{\C},g_0}^{\mathbb{H}^3}
=&
\big(\frac{k-2}{32\pi}\big)^{\frac{1}{3}}(\frac{\pi}{k})^{-\frac{2}{3}-s_0}\Big(\frac{\det(\Delta_{g_0})}{{\rm v}_{g_0}(\hat{\C})}\Big)^{-\frac{3}{2}}e^{B_{\H^3}({\bf z},\boldsymbol{j})}
\frac{\Gamma(s_0)}{2b}\E[ \|\Gamma_{\boldsymbol{\mu},{\bf z}}\|^{-2s_0}_{L^2(\theta_{\tilde{u}})}] 
\\
& \times e^{-\chi (\sum_{\ell=1}^m 2b(j_\ell+1))2bs_0}\Big(\prod_{\ell=1}^mg_0(z_\ell)^{\frac{b(j_\ell+1)}{2}}\Big)^{2bs_0}\delta_{V^0}
\end{split}\]
 with $\tilde u(z):=\sum_{\ell=1}^m 2b^2(j_\ell+1)\log\frac{1}{|z-z_\ell|} -\tfrac{1}{4}\log g_0(z)\sum_{\ell=1}^m 2b^2(j_\ell+1)$. We proceed similarly with the term $B_{\H^3}({\bf z},\boldsymbol{j})$, replacing the Green functions involved in this expression by using \eqref{greenS}. Combining, the above expression becomes (recall that $s_0+1=\sum_{\ell}(j_{\ell}+1)$)
 \begin{equation}\label{H3S}
 \begin{split}
 \Big\langle   \prod_{\ell=1}^mV_{j_\ell,\mu_\ell}(z_\ell) \Big\rangle_{\hat{\C},g}^{\mathbb{H}^3}
 =&
\big(\frac{k-2}{32\pi}\big)^{\frac{1}{3}}(\frac{\pi}{k})^{-\frac{2}{3}-s_0}\Big(\frac{\det(\Delta_{g_0})}{{\rm v}_{g_0}(\hat{\C})}\Big)^{-\frac{3}{2}} 
\frac{\Gamma(s_0)}{2b}  e^{-2b^2\chi (s_0^2-1) }\Big(\prod_{\ell=1}^mg_0(z_\ell)^{-  \Delta_{j_\ell}}\Big) 
 \\
&\times \prod_{\ell\not=\ell'} |z_\ell-z_{\ell'}|^{-2b^2(j_\ell+1)(j_{\ell'}+1)}\E[ \|\Gamma_{\boldsymbol{\mu},{\bf z}}\|^{-2s_0}_{L^2(\theta_{\tilde{u}})}] \delta_{V^0}
\end{split}\end{equation}
where $\Delta_{j_\ell}$ is given by \eqref{CWH3}. For $(x_i)_{i=1,\dots,m+p}$ disjoint points in $\C$ and weights $\alpha_i$, 
the same computation can be done on Liouville correlation functions and this yields (in the round metric)
\begin{equation} \label{Lioucorrcorr}
\begin{split}
\Big\cjg \prod_{i=1}^{m+p} V_{\alpha_i}(x_i) \Big\cjd^{\rm L}_ {\hat{\C}, g_0,\mu_{\rm L}}=& 
 \big(\frac{\det(\Delta_{g_0})}{{\rm v}_{g_0}(\hat{\C})}\big)^{-\frac{1}{2}} e^{ -\frac{\chi}{2} ( s(\boldsymbol{\alpha})^2-4Q^2) } 
\Big(\prod_{i=1}^{m+p}g_0(x_i)^{ -\Delta^{\rm L}_{\alpha_i}}\Big)\\
& \times \mu_{{\rm L}}^{-s(\boldsymbol{\alpha})} \prod_{i\not= i'}|x_i-x_{i'}|^{-\frac{\alpha_i\alpha_{i'}}{2}} \frac{\Gamma( \frac{s(\boldsymbol{\alpha})}{2b})}{2b}
   \E[M^{g_0}_{2b}(X_{g_0}+\kappa,\hat{\C})^{-\frac{ s(\boldsymbol{\alpha})}{2b}}]
   \end{split}
   \end{equation}
with\footnote{Here, we recall that $M_{2b}^{g_0}(X_{g_0},\dd z)=M_{2b}^g(X_g,\sigma,\dd z)$ uses the Lebesgue measure  $\sigma(\dd z)=\dd z$ on $\C$ as background measure, which is the reason of the term $\frac{1}{2b}\log g_0(z)$ in $\kappa(z)$.} $\kappa(z) :=(\frac{1}{2b}-\frac{1}{4}\sum_{i=1}^{m+p}\alpha_i)\log g_0(z)+\sum_{i=1}^{m+p}\alpha_i\log \frac{1}{|z-x_i|}$, $ s(\boldsymbol{\alpha}):=\sum_{i=1}^{m+p}\alpha_i-2Q$ and the conformal weight in Liouville CFT given by $\Delta_{\alpha}^{\rm L}=\frac{\alpha}{2}(Q-\frac{\alpha}{2})$. What we need now  is to get an expression for the Liouville correlations with a point at $\infty$. For this we use formula \eqref{Lioucorrcorr} applied with $m+p+1$ insertion points and we have to send the last one, say $x_{m+p+1}$, to $\infty$. This manipulation is straightforward and yields  
\[\begin{split} 
\Big\cjg \prod_{i=1}^{m+p+1} V_{\alpha_i}(x_i) \Big\cjd^{\rm L}_ {\hat{\C}, g_0,\mu_{\rm L}}
=&
 \big(\frac{{\rm v}_{g_0}(\hat{\C})}{\det(\Delta_{g_0})}\big)^{-\frac{1}{2}} 4^{-\Delta^{\rm L}_{\alpha_{m+p+1}}}e^{ -\frac{\chi}{2}( s(\boldsymbol{\alpha})^2-4Q^2) } 
 \prod_{1\leq i\not= i'\leq m+p}|x_i-x_{i'}|^{-\frac{\alpha_i\alpha_{i'}}{2}} 
\\
&\times  \mu_{{\rm L}}^{-s(\boldsymbol{\alpha})}\frac{\Gamma( \frac{s(\boldsymbol{\alpha})}{2b})}{2b} \Big(\prod_{i=1}^{m+p}g_0(x_i)^{ -\Delta^{\rm L}_{\alpha_i}}\Big)   \E[M^{g_0}_{2b}(X_{g_0}+\kappa,\Sigma)^{-\frac{ s(\boldsymbol{\alpha})}{2b}}]
\end{split}\]
where the new parameters are now 
\[\kappa(z) :=\sum_{i=1}^{m+p}\alpha_i\log\frac{1}{|z-x_i|}+
(\frac{1}{2b}-\frac{1}{4}\sum_{i=1}^{m+p+1}\alpha_i)\log g_0(z) \quad \textrm{ and } s(\boldsymbol{\alpha}):=\sum_{i=1}^{m+p+1}\alpha_i-2Q.\] 
The key remark here is the summation index in the function $\kappa$ for the terms  $\log\frac{1}{|z-x_i|}$, which only ranges from $1$ to $m+p$ whereas the summation in $s(\boldsymbol{\alpha})$ ranges from $1$ to   $m+p+1$. This will give us some room to adjust the parameters with those of the $\mathbb{H}^3$ model. This way, we have obtained two reduced expressions for the correlations of both models.

\medskip
The next step is to rewrite the function $\Gamma_{\boldsymbol{\mu},{\bf z}}(z)=\frac{1 }{\pi}  \sum_\ell  \frac{\mu_\ell}{z_\ell-z}$, appearing  in the expression of the correlation of the $\mathbb{H}^3$-model,  under the form \eqref{polyndeg}, if $\boldsymbol{\mu}\not=0$. To land on such an expression, we first write $\Gamma_{\boldsymbol{\mu},{\bf z}}$ as 
$$\sum_{\ell=1}^m  \frac{\mu_\ell}{z_\ell-z}= \frac{\sum_{\ell=1}^m \mu_\ell\prod_{\ell'\not=\ell }(z_{\ell'}-z)}{\prod_{\ell=1}^m(z_\ell-z)}.$$
Now we want to identify the polynomial in the numerator. Using a Vandermonde type argument, there exists a smallest nonnegative integer $r\in [ 1,m-1 ]$ such that $\sum_\ell\mu_\ell z_\ell^{r}\not=0$, since  the insertion points $z_1,\dots, z_m$ are assumed to be distinct. 
By some elementary  algebra based on the Newton identities for polynomials, one can see that the numerator is a polynomial in the variable $z$ of degree exactly $m-1-r$ with coefficient for the (highest order) monomial $ z^{m-1-r}$ given  by $(-1)^{m-1}(\sum_{\ell=1}^m\mu_\ell z_\ell^r)$. Let $y_1,\dots,y_{p}\in \C$ be the distinct zeros of this polynomial, and    $n_1,\dots,n_p\in\N\setminus\{0\}$ their respective multiplicities, thus  satisfying $n_1+\dots+n_p=m-r-1$. This leads to the formula \eqref{polyndeg}, and therefore
\[\Big|\sum_{\ell=1}^m  \frac{\mu_\ell}{z_\ell-z}\Big|^2
= 
 \Big|\sum_{\ell=1}^m\mu_\ell z_\ell^r\Big|^2 \frac{\prod_{\ell=1}^p|z-y_\ell|^{2n_\ell}}{\prod_{\ell=1}^m|z-z_\ell|^{2}}.\]
Plugging this relation into \eqref{H3S}, we get 
\begin{equation}\label{rayas}
\begin{split} 
\Big\langle   \prod_{\ell=1}^mV_{j_\ell}(z_\ell) \Big\rangle_{\Sigma,g_0}^{\mathbb{H}^3}
 =&\big(\frac{k-2}{32\pi}\big)^{\frac{1}{3}}(\frac{\pi}{k})^{-\frac{2}{3}}\Big(\frac{\det(\Delta_{g_0})}{{\rm v}_{g_0}(\hat{\C})}\Big)^{-\frac{3}{2}} 
\frac{\Gamma(s_0)}{2b}  e^{-2b^2\chi (s_0^2-1) }\Big(\prod_{\ell=1}^mg_0(z_\ell)^{-  \Delta_{j_\ell}}\Big) 
 \\
& \times \prod_{\ell\not=\ell'} |z_\ell-z_{\ell'}|^{-2b^2(j_\ell+1)(j_{\ell'}+1)}\mu_{\rm L}^{-s_0}  \E[M^{g_0}_{2b}(X_{g_0} +\iota,\Sigma)^{-s_0}] \delta_{V^0}
\end{split}\end{equation}
with $\mu_{\rm L}:=\frac{ |\sum_{\ell=1}^m\mu_\ell z_\ell^r|^2}{\pi k}$ and the function 
\[\iota(z):=\sum_{i=1}^{m+p}\alpha_i\log\frac{1}{|z-x_i|} +\Big(\frac{1}{2b} -\frac{1}{4} \sum_{i=1}^{m+p+1} \alpha_i\Big)\log g_0(z)\]
with $\alpha_i=2b(j_\ell+1)+\frac{1}{b}$ and $x_i=z_\ell$ for $i=\ell\in\{1,\dots,m\}$, and $\alpha_i=-\frac{n_\ell}{b}$ and $x_i=y_\ell$ for $i=\ell+m$ and $\ell\in\{1,\dots,p\}$, and finally $\alpha_{m+p+1}=-\frac{r-1}{b}$.  Notice that, with this choice of $x_i,\alpha_i$, 
\[s_0=s(\boldsymbol{\alpha})/2b, \qquad \kappa(z)=\iota(z).\]
The expression \eqref{rayas} then corresponds to the Liouville correlation functions, up to a prefactor: 
\[
\Big\langle   \prod_{\ell=1}^mV_{j_\ell,\mu_\ell}(z_\ell) \Big\rangle_{\Sigma,g_0}^{\mathbb{H}^3}
 =
 \mc{F}_{{\bf z},\boldsymbol{j}}({\bf y}) \, \Big\cjg \prod_{i=1}^{m+p+1} V_{\alpha_i}(x_i) \Big\cjd^{\rm L}_ {\Sigma, g_0} \delta_{V^0}
\]
where $\mu_{\rm L}=\frac{ |\sum_{\ell=1}^m\mu_\ell z_\ell^k |^2}{\pi k}$, the $x_i,\alpha_i$'s are defined just above and the overall factor $\mc{F}_{{\bf z},\boldsymbol{j}}$ (for the case $g=g_0$) is given by
\[\begin{split}   
 \mc{F}_{{\bf z},\boldsymbol{j}}({\bf y}) 
  =& \big(\frac{k-2}{32\pi}\big)^{\frac{1}{3}}\big(\frac{k}{\pi}\big)^{\frac{2}{3}} 
\Big(\frac{\det (\Delta_{g_0})}{{\rm v}_{g_0}(\hat{\C})}\Big)^{-1} e^{\frac{\chi}{2}( s(\boldsymbol{\alpha})^2-4Q^2)-2b^2\chi (s_0^2-1)}  
 4^{\Delta^{\rm L}_{\alpha_{m+p+1}}}\prod_{i\not= i'=1}^{m+p}|x_i-x_{i'}|^{\frac{\alpha_i\alpha_{i'}}{2}}\\
& \times \mu_{\rm L}^{-s_0}\Big(\prod_{i=1}^{m+p}g_0(x_i)^{ \Delta^{\rm L}_{\alpha_i}}\Big)  \Big(\prod_{\ell=1}^mg_0(z_\ell)^{- \Delta_{j_\ell}}\Big) 
\prod_{\ell\not=\ell'} |z_\ell-z_{\ell'}|^{-2b^2(j_\ell+1)(j_{\ell'}+1)}
\end{split}\]
For $i=\ell\in\{1,\dots,m\}$, we have $\Delta_{\alpha_i}^{\rm L}=\Delta_{j_\ell}+\frac{1}{2}+\frac{1}{4b^2}$. Using the conformal anomaly relations \eqref{detpolyakov}, \eqref{LCFT_anomaly} and \eqref{e:anomaly_WZW}, we come back to the expression for the metric $g=e^{\omega}g_0$ and we land on the formula \eqref{calculFdeg}.
\end{proof}

\subsection{Structure constants of the $\mathbb{H}^3$-model}
 Now we study   the 3 point correlation functions  of the $\mathbb{H}^3$-model $\langle \prod_{\ell=1}^3V_{j_\ell,\mu_\ell}(z_\ell) \rangle_{\Sigma,g}^{\mathbb{H}^3}$ on the Riemann sphere $\hat{\C}$, equipped with a smooth metric of the form $g=g(z)|dz|^2$ on the sphere. 
 Let us also write $g_0=g_0(z)|dz|^2$ for  the round metric \eqref{def_g_0}. Proposition \ref{confweight}  entails  that the  correlation functions 
are {\it conformally covariant}. More precisely, if $z_1, \cdots, z_{m}$ are  distinct points in $\hat \C$  then for a M\"obius map $\psi(z)= \frac{az+b}{cz+d}$ (with $a,b,c,d \in \C$ and $ad-bc=1$) 
 \begin{equation}\label{KPZformula}
\Big\langle \prod_{\ell=1}^mV_{j_\ell,\mu_\ell}(\psi(z_\ell)) \Big\rangle_{\hat{\C},g_0}^{\mathbb{H}^3}=  \prod_{\ell=1}^{m} \Big(\frac{|\psi'(z_\ell)|^2g_0(\psi(z_\ell))}{g_0(z_\ell)}\Big)^{-   \triangle_{j_\ell}}     \Big\langle \prod_{\ell=1}^mV_{j_\ell,\mu_\ell}(z_\ell) \Big\rangle_{\hat{\C},g_0}^{\mathbb{H}^3}.
\end{equation}  
Combining with the Weyl anomaly in Proposition \ref{confweight} , the dependence of the  $3$-point correlation functions on the metric or the points $z_\ell$  is uniquely determined by the conformal symmetries
\begin{multline}\label{3ptsym}
 \Big\langle  V^{g}_{j_1,\mu_1}(z_1) V^{g}_{j_2,\mu_2}(z_2) V^{g}_{j_3,\mu_3}(z_3) \Big\rangle_{\hat{\C},g}^{\mathbb{H}^3} = e^{ -\frac{{\bf c}(k)}{96\pi}S^0_{\rm L}(g_0,g)}\Big(\prod_{\ell=1}^3g(z_\ell)^{- \triangle_{j_\ell}}\Big)
 \\
 |z_1-z_3|^{2(\triangle_{j_2}-\triangle_{j_1}-\triangle_{j_3})}|z_2-z_3|^{2(\triangle_{j_1}-\triangle_{j_2}-\triangle_{j_3})}|z_1-z_2|^{2(\triangle_{j_3}-\triangle_{j_1}-\triangle_{j_2})}C^{\mathbb{H}^3}(\boldsymbol{j},\boldsymbol{\mu})\delta_{V^0}
\end{multline}
for some constant $C^{\mathbb{H}^3}(\boldsymbol{j},\boldsymbol{\mu})$, the so-called structure constant,  with  $\boldsymbol{j}$ and $\boldsymbol{\mu}$ triplets collecting respectively the $j_\ell$'s and the $\mu_\ell$'s.
We will find an expression for this quantity using Theorem \ref{Liouvillesphere2} under the condition \eqref{seib}. For this, we use the hypergeometric functions ${_2}F_{1}$, which can be extended holomorphically on $\mathbb{C} \setminus  (1,\infty) $,
 \begin{equation}\label{Fpmdef}
F_{-}(y)= {}_2F_1(\alpha,\beta,\gamma,y), \quad F_{+}(y)= y^{1-{ \gamma}} {}_2F_1(1+\alpha-\gamma,1+\beta-\gamma,2-\gamma,y)
\end{equation}
 where $\alpha,\beta,\gamma$ are given by
\begin{align}\label{defabcfirst}
\alpha&= 2+j_1+j_2+j_3,  \quad \beta= j_1+j_2-j_3+1,  \quad \gamma=2(j_1+1),
\end{align}
and the coefficient $ \mathcal{A}  (\alpha_1,\alpha_2,\alpha_3)$ defined by (with $\bar\alpha=\alpha_1+\alpha_2+\alpha_3$ and $l(x)=\Gamma(x)/\Gamma(1-x)$)
$$ \mathcal{A}  (\alpha_1,\alpha_2,\alpha_3)=\frac{l(-b^{-2})  l(\alpha_1/b) l(\alpha_1/b-b^{-2} )  l(\frac{1}{2b} (\bar{\alpha}-2\alpha_1-1/b) )   }{l( \frac{1}{2b} (\bar{\alpha}-1/b - 2Q)  ) l( \frac{1}{2b} (\bar{\alpha}-2\alpha_3-1/b ))  l( \frac{1}{2b} (\bar{\alpha}-2\alpha_2-1/b )) }$$
with  $\alpha_\ell$   given by $\alpha_\ell=2b(j_\ell+1)+\frac{1}{b}$ for $\ell=1,2,3$. 

Before proceeding to the statement of the result, let us stress that the computation of the structure constants is obvious in the case when all the $\mu_\ell$'s vanish: indeed, this is then a simple GFF computation. Thus the interesting situation occurs when at least one $\mu_\ell$ does not vanish and, up to permuting the $\mu_\ell$'s, we can assume that this is $\mu_3$.

\begin{theorem}Assume $\mu_3\not=0$. The structure constants of the $\mathbb{H}^3$-model are given by
\begin{multline*}
C^{\mathbb{H}^3}(\boldsymbol{j},\boldsymbol{\mu})  =     
      \\
 D(g_0)  
   \Big(\frac{\pi k}{|\mu_3|^2}\Big)^{2+j_1+j_2+j_3}  C_{2b,\mu_{\rm L}=1}^{{\rm DOZZ}} (\alpha_1-\tfrac{1}{b},\alpha_2,\alpha_3 ) \Big(|F_-(-\tfrac{\mu_1}{\mu_3})|^2 - \frac{  l(\frac{4}{\gamma^2})}{(\pi     l(\frac{\gamma^2}{4})  )^{\frac{4}{\gamma^2} }  } 
 \mathcal{A}  (\alpha_1,\alpha_2,\alpha_3)|F_+(-\tfrac{\mu_1}{\mu_3})|^2\Big)
\end{multline*}
where  $D(g_0)$ is the  constant  
\begin{equation}
D(g_0):= \frac{1}{2} \big(\frac{k-2}{32\pi}\big)^{\frac{1}{3}}\big(\frac{k}{\pi}\big)^{\frac{2}{3}} 
\Big(\frac{\det (\Delta_{g_0})}{{\rm v}_{g_0}(\Sigma)}\Big)^{-3/2}e^{-2\chi(Q^2-b^2)}4^{\frac{1}{2}+\frac{1}{4b^2}} e^{-6Q^2 S_{\rm L}^0(\hat\C,g,\omega_0)}   
\end{equation}
with  the function $\omega_0:\hat{\C}\to \R$ defined by $e^{\omega_0} g_0= |z|_+^{-4}|dz|^2$ and $S_{\rm L}^0(\hat\C,g,\omega)$   the Liouville action \eqref{LiouvilleS0}.
\end{theorem}

\begin{proof} Let us first observe that, from \eqref{3ptsym},  $C^{\mathbb{H}^3}(\boldsymbol{j},\boldsymbol{\mu}) $ can be recovered from the limit 
$$ C^{\mathbb{H}^3}(\boldsymbol{j},\boldsymbol{\mu}) = g_0(0)^{ \triangle_{j_1}} g_0(1)^{  \triangle_{j_2}}  \lim_{|z|\to\infty}\langle  V^{g_0}_{j_1,\mu_1}(0) V^{g_0}_{j_2,\mu_2}(1) V^{g_0}_{j_3,\mu_3}(z) \rangle_{\hat{\C},g_0}^{\mathbb{H}^3} .$$
From now on, we will work with the round metric $g_0$ only so that we will mostly skip it from notations. Next, we apply Theorem \ref{Liouvillesphere2} in order to express the 3 point correlation functions in the $\mathbb{H}^3$ model with insertions at $0,1,z$  in terms of Liouville correlation functions, and then send $z\to\infty$.  Note that the set of those $z\in\C$ such that $\mu_2+\mu_3z\not =0$ is open and dense in $\C$ and we will always choose such $z$ in what follows. This guarantees that  the parameter $r$ in Theorem \ref{Liouvillesphere2}) is equal to $1$ and, in the case $m=3$, we have then necessarily $p=1$, $n_1=1$. Then the unique point $y$ has  the explicit expression $y:=- \frac{\sum_{\ell=1}^3\mu_\ell \prod_{\ell'\not=\ell}z_{\ell'}}{\sum_{\ell=1}^3\mu_\ell z_\ell}$ obtained from the formula \eqref{polyndeg}.  We write  ${\bf z}=(0,1,z)$ and by \eqref{Liouvillesphere2}
\begin{equation}\label{correlz}
 \Big \langle  V_{j_1,\mu_1}(0)  V_{j_2,\mu_2}(1) V_{j_3,\mu_3}(z)\Big \rangle_{\hat{\C},g_0}^{\mathbb{H}^3} = \mu_{{\rm L}}^{-s_0}
\mc{F}_{{\bf z},\boldsymbol{j}}(y) \Big\cjg   V_{\alpha_1}(0) V_{\alpha_2}(1) V_{\alpha_3}(z) V_{-\frac{1}{b}}(y(z)) \Big\cjd^{\rm L}_ {\hat{\C}, g_0 } 
\end{equation}
\[\textrm{ with } y=y(z):=-\frac{\mu_1z}{\mu_2+\mu_3z}, \quad \mu_{\rm L}:=\frac{ |\mu_2+\mu_3 z|^2}{\pi k}.\]
Now $s_0=\frac{\sum_{j=1}^4\alpha_j-2Q}{2b}=j_1+j_2+j_3+2$ and we set 
\begin{align*}
F(\boldsymbol{\mu}):=&
\lim_{z\to\infty} \mu_{{\rm L}}^{-s_0}\mc{F}_{{\bf z},\boldsymbol{j}}(y(z))
\\
=&
e^{-2\chi(Q^2-b^2) }(\pi k)^{s_0}4^{\frac{1}{2}+\frac{1}{4b^2}}g_0(0)^{\frac{1}{2}+\frac{1}{4b^2}}g_0(1)^{\frac{1}{2}+\frac{1}{4b^2}}g_0(-\tfrac{\mu_1}{\mu_3})^{\Delta_{-\frac{1}{b}}}|\tfrac{\mu_1}{\mu_3}|^{-2(j_1+1)-\frac{1}{b^2}}|1+\tfrac{\mu_1}{\mu_3}|^{-2(j_2+1)-\frac{1}{b^2}}|\mu_3|^{-2s_0}\\
& \times \big(\frac{k-2}{32\pi}\big)^{\frac{1}{3}}\big(\frac{k}{\pi}\big)^{\frac{2}{3}} 
\Big(\frac{\det (\Delta_{g_0})}{{\rm v}_{g_0}(\hat{\C})}\Big)^{-1}.
\end{align*}
Taking the limit $z\to\infty$ in the relation \eqref{correlz}, we obtain
\begin{equation}
C^{\mathbb{H}^3}(\boldsymbol{j},\boldsymbol{\mu})  = \big(\frac{k-2}{32\pi}\big)^{\frac{1}{3}}\big(\frac{k}{\pi}\big)^{\frac{2}{3}} 
\Big(\frac{\det (\Delta_{g_0})}{{\rm v}_{g_0}(\hat{\C})}\Big)^{-1} g_0(0)^{ \triangle_{j_1}} g_0(1)^{  \triangle_{j_2}} F(\boldsymbol{\mu}) \Big\cjg   V_{\alpha_1}(0) V_{\alpha_2}(1) V_{\alpha_3}(\infty) V_{-\frac{1}{b}}(-\tfrac{\mu_1}{\mu_3}) \Big\cjd^{\rm L}_ {\hat{\C},g_0} .
\end{equation}
Now we exploit the fact that the $4$-point correlation function $ \cjg   V_{\alpha_1}(0) V_{\alpha_2}(1) V_{\alpha_3}(\infty) V_{-\frac{1}{b}}(-\tfrac{\mu_1}{\mu_3}) \cjd^{\rm L}_ {\hat{\C},g_0} $ involves a vertex operator $V_{-\frac{1}{b}}$ with a degenerate weight and thus obeys the BPZ equation \cite{KRV19_local}. The expression for the $4$-point function is then deduced from the BPZ equation in   \cite[Theorem 9.7]{KRV_DOZZ}, when $\alpha_1$ is close to $Q$ but this relation extends to all $\alpha_1$ within the Seiberg bounds \eqref{seiberg} by analytic continuation \cite[Theorem 6.1]{KRV_DOZZ}. It reads\footnote{Beware that the Liouville correlations in  \cite{KRV_DOZZ} differ by an overall multiplicative factor $ \prod_{j=1}^mg(z_j)^{ \triangle_{\alpha_j}} $ compared to the present manuscript.}, for $y\in\C$
\begin{multline*}
  g_0(0)^{ \triangle_{\alpha_1}} g_0(1)^{  \triangle_{\alpha_2}} g_0(y)^{\Delta_{-\frac{1}{b}}}\cjg   V_{\alpha_1}(0) V_{\alpha_2}(1) V_{\alpha_3}(\infty) V_{-\frac{1}{b}}(y) \cjd^{\rm L}_ {\hat{\C},g_0}
 \\=
 |y|^{\frac{\alpha_1}{b}}|y-1|^{\frac{\alpha_2}{b}} \Big(g_0(0)^{\Delta_{\alpha_1-\frac{1}{b}}} g_0(1)^{  \triangle_{\alpha_2}} \Big\cjg   V_{\alpha_1-\frac{1}{b}}(0) V_{\alpha_2}(1) V_{\alpha_3}(\infty)   \Big\cjd^{\rm L}_ {\hat{\C},g_0} |F_-(y)|^2
\\
+ g_0(0)^{\Delta_{\alpha_1+\frac{1}{b}}} g_0(1)^{  \triangle_{\alpha_2}} \Big\cjg   V_{\alpha_1+\frac{1}{b}}(0) V_{\alpha_2}(1) V_{\alpha_3}(\infty)   \Big\cjd^{\rm L}_ {\hat{\C},g_0} |F_+(y)|^2\Big).
\end{multline*}  
We stress that this relation is valid for $\alpha_1<Q$ (i.e. $j_1<-1/2$) and the structure constants in the right-hand side have to be understood as the analytic continuation of $ \cjg   V_{\alpha_1+\frac{1}{b}}(0) V_{\alpha_2}(1) V_{\alpha_3}(\infty)   \cjd^{\rm L}_ {\Sigma,g_0}$ in case $\alpha_1+\frac{1}{b}\geq Q$. 
Now recall the following shift relation for the Liouville 3 point function  \cite[section 1.2]{KRV_DOZZ} 
\begin{align*}
 &g_0(0)^{\Delta_{\alpha_1+\frac{1}{b}}} g_0(1)^{  \triangle_{\alpha_2}} \Big\cjg   V_{\alpha_1+\frac{1}{b}}(0) V_{\alpha_2}(1) V_{\alpha_3}(\infty)   \Big\cjd^{\rm L}_ {\hat{\C},g_0} 
\\
&=
 - \frac{  l(\frac{4}{\gamma^2})}{(\pi     l(\frac{\gamma^2}{4})  )^{\frac{4}{\gamma^2} }  } 
 \mathcal{A}  (\alpha_1,\alpha_2,\alpha_3)
  g_0(0)^{\Delta_{\alpha_1-\frac{1}{b}}} g_0(1)^{  \triangle_{\alpha_2}}\Big\cjg   V_{\alpha_1-\frac{1}{b}}(0) V_{\alpha_2}(1) V_{\alpha_3}(\infty)   \Big\cjd^{\rm L}_ {\hat{\C},g_0} 
\end{align*} 
Combining with \eqref{3pointDOZZ}, we land on the announced formula.   
 \end{proof}


\section{Holomorphic rank-$2$ bundles and WZW action}\label{sec:General_case}

Generally speaking, the Wess-Zumino-Witten theory can be understood as a path integral over the space of gauge transforms on a holomorphic vector bundle $E$ over a surface $\Sigma$. In our case, the bundle will have rank $2$ and trivial determinant, which implies it is isomorphic (as a complex vector bundle) to $\Sigma\times \C^2$, and  that the gauge transforms can be thought of as functions $h:\Sigma \to {\rm SL}(2,\C)$. In the $\mathbb{H}^3$ model, the gauges live in the coset ${\rm SL}(2,\C)/{\rm SU}(2)$, which basically corresponds to the hyperbolic space $\mathbb{H}^3$ (hence the name), in which case the elements in the cosets are represented by positive definite ${\rm SL}(2,\C)$ matrices and  should rather be thought of as   Hermitian metrics on the bundle $E$. 
The path integral then averages over these metrics, with the holomorphic structure of $E$ being fixed. \\

We consider a closed oriented Riemannian surface $(\Sigma,g)$. The metric $g$ induces  a complex structure, i.e. a field $J\in C^\infty(\Sigma;{\rm End}(T\Sigma))$ of endomorphisms of the tangent bundle such that $J^2=-{\rm Id}$. We say that $(\Sigma,J)$ is a closed Riemann surface.
The dual map to $J$ is the Hodge star operator $*_g$ on $1$-forms, which only depends on the conformal class $[g]$. On the complexified cotangent bundle $\C T^*\Sigma$, the Hodge operator induces two subbundles
\[ (T^*\Sigma)^{1,0}:=\ker (*_g+i{\rm Id}) ,\quad (T^*\Sigma)^{0,1}:=\ker (*_g-i{\rm Id})\]
and, if $\pi^{1,0}$ and $\pi^{0,1}$ denote the projection on these subpaces, 
we define $\pl :=\pi^{1,0}\circ d : C^\infty(\Sigma)\to C^\infty(\Sigma;(T^*\Sigma)^{1,0})$ and 
$\bar{\pl} :=\pi^{0,1}\circ d : C^\infty(\Sigma)\to C^\infty(\Sigma;(T^*\Sigma)^{0,1})$. Functions $u$ on an open set $U\subset \Sigma$ satisfying $\bar{\pl}u=0$ (resp. $\pl u=0$) are called holomorphic (resp. anti-holomorphic). There are local coordinates $z:U\to \C$ (called holomorphic coordinates) such that $\bar{\pl}z=0$ and in these coordinates we can write 
$\bar{\pl}u(z)=\pl_{\bar{z}}u d\bar{z}$ and $\pl u(z)=\pl_{z}u dz$. Moreover one can find a covering $(U_j)_j$ of $\Sigma$ with charts $\omega_j:U_j\to \D\subset \C$ such that the transition functions $\omega_i\circ \omega_j^{-1}$ are 
holomorphic where they are defined. The metric $g$ on $\Sigma$ extends as a Hermitian metric on $\C T\Sigma$ by 
$g(v_1+iv_1^*,v_2+iv_2^*)=g(v_1,v_2)+g(v_1^*,v_2^*)+i(g(v_2,v_1^*)-g(v_1,v_2^*))$ (if $v_i,v_i^*$ are real tangent vectors), and this metric induces also a Hermitian metric, denoted by $\cjg u,v\cjd_{g}$,  on $\C T^*\Sigma$ by duality.
There is a natural $L^2$-scalar product on $\C T^*\Sigma$ given by
\[ \cjg u,v\cjd_2:= \frac{1}{2}\int_{\Sigma} u\wedge *_g \bbar{v}= \int_{\Sigma} \cjg u,v\cjd_{g} {\rm dv}_g \]
where 
${\rm v}_g$ denotes  the Riemannian volume measure. For a vector space, we shall denote by $\Lambda^p V$ the $p$-antisymmetric product of $V$ and for a vector bundle $V$ over $\Sigma$ we also write $\Lambda^pV$ for the bundle whose fiber at $x$ is the 
$p$-antisymmetric product $\Lambda^p V_x$.
We denote $\Lambda^{p,q}\Sigma:=\Lambda ^p (T^*\Sigma)^{1,0}\wedge \Lambda ^q (T^*\Sigma)^{0,1}$ for $p+q\leq 2$ (remark that $\Lambda^{0,2}\Sigma=\Lambda^{2,0}\Sigma=\{0\}$). We recall that the Hodge star operator on $2$-forms $*_g: C^\infty(\Sigma,\Lambda^2\Sigma)\to C^\infty(\Sigma)$ is defined by $*_g {\rm v}_g:=1$.

\subsection{Wess-Zumino-Witten action}\label{WZWaction}

We shall now introduce the WZW action for the trivial connection on $\Sigma\times \C^2$. The action will be associated to a pair 
$(h,A)$ where $h:\Sigma \to {\rm SL}(2,\C)$ is a smooth function and $A\in C^\infty(\Sigma; {\rm sl}(2,\C) \otimes \Lambda^1\Sigma)$, where ${\rm SL}(2,\C)$ are the $2\times 2$ complex matrices with determinant $1$, and ${\rm sl}(2,\C)$ its Lie algebra. If $G:={\rm SU}(2)$ then $G^\C={\rm SL}(2,\C)$ is its complexification.\\

\noindent\textbf{WZW action for the trivial ${\rm SL}(2,\C)$ connection.} 
 We use the following standard notation: on a manifold $\Sigma$,  if $\alpha_j: \Sigma \to M_n(\C)$ are $M_n(\C)$-valued (matrices $n\times n$) $1$-forms for $j=1,\dots,j_0$, their 
exterior product is defined by its action on vectors $v_1,\dots,v_{j_0}$ by
\[ (\alpha_1\wedge \dots\wedge \alpha_{j_0})(v_1,\dots, v_{j_0}):=\sum_{\sigma\in \mc{P}_{j_0}}\eps(\sigma)\alpha_1(v_{\sigma(1)})\dots \alpha_J(v_{\sigma(j_0))})\]
where $\eps(\sigma)$ is the signature of the permutation $\sigma$, $\mc{P}_{j_0}$ being the set of permutations of $j_0$ elements. In particular, ${\rm Tr}(\alpha_1\wedge \dots\wedge \alpha_{j_0})$ is a differential $j_0$-form  on $\Sigma$. Remark that for such forms,
\begin{equation}\label{Trwedge} 
{\rm Tr}(\alpha_i\wedge \alpha_j)=-{\rm Tr}(\alpha_j\wedge \alpha_i) , \quad {\rm Tr}(\alpha_i\wedge \alpha_j\wedge \alpha_k)={\rm Tr}(\alpha_k\wedge \alpha_i\wedge \alpha_j).
\end{equation}

\begin{lemma}
Let $(\Sigma,J)$ be a closed Riemann surface and let us define the Wess-Zumino-Witten (WZW) action for the trivial ${\rm SL}(2,\C)$ connection: for $h:\Sigma \to G^\C$ a smooth map, we let 
\begin{equation}\label{WZW_def}
S_\Sigma(h):=\frac{1}{4\pi i}\int_{\Sigma} {\rm Tr}((h^{-1}\pl h)\wedge (h^{-1}\bar{\pl}h))+S^{\rm top}_X(\tilde{h})
\end{equation}
where $X$ is any compact 3-manifold with boundary $\Sigma$, $\tilde{h}:X\to G^\C $ is a smooth extension of $h$ to $X$ and 
\[ S^{\rm top}_X(\tilde{h}):=\frac{1}{12\pi i}\int_{X}{\rm Tr}((\tilde{h}^{-1}d\tilde{h})\wedge (\tilde{h}^{-1}d\tilde{h})\wedge (\tilde{h}^{-1}d\tilde{h})).\]
The value of $S_\Sigma(h)\in \C$ a priori depends on the choice of extension $\tilde{h}$ but 
$h\in C^\infty (\Sigma,G^\C)  \mapsto  e^{S_\Sigma(h)}$ is well defined independently of the extension $\tilde{h}$, i.e. 
$S_\Sigma(h)$ can considered as an element in $\C/2\pi i\Z$.
\end{lemma}

\begin{proof} 
Take two compact manifolds $X_1$ and $X_2$ with boundary $\Sigma$,  and two extensions $\tilde{h}_1$ and $\tilde{h}_2$ of $h$ to $X_1$ and $X_2$ respectively. We can glue $X_1$ to $X_2$ 
along $\Sigma$ to produce a closed $3$ manifold $\tilde{X}$ and consider $\tilde{h}={\bf 1}_{X_1}\tilde{h}_1+{\bf 1}_{X_2}\tilde{h}_2$ (which is a continuous, piecewise smooth function with bounded derivative) as a map $\tilde{X}\to G^\C$. Since $G^\C/G=\mathbb{H}^3$, $G^\C$ is homotopy equivalent to $G$ and there is a homotopy $t\mapsto g_t\in C^\infty(\Sigma,G^\C)$ 
such that $g_0=\tilde{h}$ and $g_1$ takes values in $G$. We thus have
\begin{equation}\label{S_top} 
\begin{split}
S^{\rm top}_{X_1}(\tilde{h}_1)-S^{\rm top}_{X_2}(\tilde{h}_2)=&\frac{1}{12\pi i}\int_{\tilde{X}}{\rm Tr}((g_0^{-1}dg_0)^{\wedge 3})=\frac{1}{12\pi i}\int_{\tilde{X}}{\rm Tr}((g_1^{-1}dg_1)^{\wedge 3})=\frac{1}{12\pi i}\int_{\tilde{X}}g_1^*\omega_{{\rm MC}}\\
=&{\rm deg}(g_1)\frac{1}{12\pi i}\int_{G}\omega_{{\rm MC}}
\end{split}
\end{equation}
where $\omega_{{\rm MC}}(g)={\rm Tr}((g^{-1}dg)^{\wedge 3})$ is the Maurer-Cartan form, which satisfies 
that $\frac{1}{24\pi^2}\int_{{\rm SU}(2)}\omega_{{\rm MC}}=1$, thus \eqref{S_top} belongs to $2\pi i\Z$ (here ${\rm deg}(g_1)$ is the degree of the map $g_1:\tilde{X}\to G^\C$). 
\end{proof}
The value $e^{k S_\Sigma(h)}$ for $k\in \Z$ is thus independent of the choice of the $3$-manifold $X$ and of the extension $\tilde{h}$, it depends only on $h$ defined on $\Sigma$.
 The  following  fundamental identity is due to Polyakov and Wiegmann:
\begin{lemma}[\textbf{Polyakov-Wiegmann anomaly formula}]\label{l:Polyakov_Wiegmann}
For $h_1,h_2:\Sigma\to G^\C$ smooth, the following holds true in $\C/2\pi i\Z$
\begin{equation}\label{PolWi}
S_\Sigma(h_1h_2)=S_\Sigma(h_1)+S_{\Sigma}(h_2)+\Upsilon(h_1,h_2), \quad \Upsilon(h_1,h_2):=\frac{1}{2\pi i}\int_\Sigma {\rm Tr}(h_1^{-1}\bar{\pl}h_1\wedge h_2\pl h_2^{-1}).
\end{equation}
\end{lemma}
\begin{proof} Define $S^1_\Sigma(h):=\frac{1}{4\pi i}\int_{\Sigma} {\rm Tr}((h^{-1}\pl h)\wedge (h^{-1}\bar{\pl}h))$. Then a direct computation using the cyclicity of the trace yields 
\begin{equation}\label{expandS1} 
S^1_\Sigma(h_1h_2)=S^1_\Sigma(h_1)+S^1_\Sigma(h_2)+\frac{1}{4\pi i}\int_{\Sigma}{\rm Tr}(h_2^{-1}h_1^{-1}\pl h_1\wedge \bar{\pl}h_2 )-{\rm Tr}(h_2^{-1}h_1^{-1}\bar{\pl} h_1\wedge \pl h_2 )
\end{equation}
and, using again the cyclicity of the trace, for $\tilde{h}_i$ an extension of $h_i$ to $X$, any compact 
$3$-manifold with boundary $\Sigma$,
\[\begin{split} 
{\rm Tr}((\tilde{h}_2^{-1}\tilde{h}_1^{-1}d\tilde{h}_1\tilde{h}_2+\tilde{h}_2^{-1}d\tilde{h}_2)^{\wedge 3})=& {\rm Tr}((\tilde{h}_2^{-1}d\tilde{h}_2)^{\wedge 3})+{\rm Tr}((\tilde{h}_1^{-1}d\tilde{h}_1)^{\wedge 3})+3{\rm Tr}(d\tilde{h}_2\tilde{h}_2^{-1}\wedge d\tilde{h}_2\tilde{h}_2^{-1}\wedge \tilde{h}_1^{-1}d\tilde{h}_1)\\
&+3{\rm Tr}(\tilde{h}_1^{-1}d\tilde{h}_1\wedge \tilde{h}_1^{-1}d\tilde{h}_1\wedge d\tilde{h}_2 \tilde{h}_2^{-1})\\
=& {\rm Tr}((\tilde{h}_2^{-1}d\tilde{h}_2)^{\wedge 3})+{\rm Tr}((\tilde{h}_1^{-1}d\tilde{h}_1)^{\wedge 3})+3d{\rm Tr}(
d\tilde{h}_2\tilde{h}_2^{-1}\wedge \tilde{h}_1^{-1}d\tilde{h}_1).
\end{split}\]
We write on $\Sigma$
\[\begin{split} 
{\rm Tr}(dh_2h_2^{-1}\wedge h_1^{-1}dh_1 )=& {\rm Tr}(\pl h_2h_2^{-1}\wedge h_1^{-1}\bar{\pl}h_1 )+{\rm Tr}(\bar{\pl} h_2h_2^{-1}\wedge h_1^{-1}\pl h_1 )\\
= & -{\rm Tr}(h_2^{-1} h_1^{-1}\bar{\pl}h_1 \wedge \pl h_2)-{\rm Tr}(h_2^{-1}h_1^{-1}\pl h_1\wedge \bar{\pl}h_2 )
\end{split}\]
from which we deduce that 
\[ S_{X}^{\rm top}(\tilde{h}_1\tilde{h}_2)=S_{X}^{\rm top}(\tilde{h}_1)+S_{X}^{\rm top}(\tilde{h}_2)-\frac{1}{4\pi i}\int_\Sigma {\rm Tr}(h_2^{-1} h_1^{-1}\bar{\pl}h_1 \wedge \pl h_2)-\frac{1}{4\pi i}\int_\Sigma {\rm Tr}(h_2^{-1}h_1^{-1}\pl h_1\wedge \bar{\pl}h_2 ).\]
Combining with \eqref{expandS1}, we get the desired result since ${\rm Tr}(h_2^{-1} h_1^{-1}\bar{\pl}h_1 \wedge \pl h_2)=- {\rm Tr}( h_1^{-1}\bar{\pl}h_1 \wedge h_2 \pl h_2^{-1} )$
\end{proof}
\noindent\textbf{WZW action for general ${\rm SL}(2,\C)$ connections.} 
Consider $E=\Sigma\times \C^2$, viewed as a (trivial) complex vector bundle over $\Sigma$ and let $A\in C^\infty(\Sigma;\mathfrak{g}^\C\otimes \Lambda^{1}\Sigma)$ be a $1$-form with values in the Lie algebra $\mathfrak{g}^\C:={\rm sl}(2,\C)$ of ${\rm SL}(2,\C)$. We take $A$ unitary, in the sense that $A^*=-A$, and we denote by $A^{1,0}\in C^\infty(\Sigma;\mathfrak{g}^\C\otimes \Lambda^{1,0}\Sigma)$ and  
$A^{0,1}\in C^\infty(\Sigma;\mathfrak{g}^\C\otimes \Lambda^{0,1}\Sigma)$ its $(1,0)$ and $(0,1)$ part. As we shall explain later, $A^{0,1}$ induces a holomorphic structure on $E$, by saying that a function $u\in C^\infty(\Sigma,\C^2)$ is holomorphic if $(\bar{\pl}+A^{0,1})u=0$; this is a way to twist the classical holomorphic structure on the trivial bundle $E$. We then define the WZW action with the connection $A$ as follows:

\begin{definition}[\textbf{WZW action for general unitary ${\rm SL}(2,\C)$-connections}]
The WZW action with connection form $A\in C^\infty(\Sigma;\mathfrak{g}^\C \otimes \Lambda^{1}\Sigma)$ on the trivial rank $2$ bundle $\Sigma\times \C^2$ is defined as an element in $\C/2\pi i\Z$ by 
\begin{equation}\label{gaugeWZW}
S_\Sigma (h,A):=S_\Sigma(h)-\frac{1}{2\pi i}\int_\Sigma {\rm Tr}(A^{1,0}\wedge h^{-1}\bar{\pl} h+h\pl h^{-1}\wedge A^{0,1}+hA^{1,0}h^{-1}\wedge A^{0,1}).
\end{equation}
Observe that $S_\Sigma ({\rm Id},A)=-\frac{1}{2\pi i}\int_\Sigma {\rm Tr}(A^{1,0}\wedge A^{0,1})$. 
\end{definition}
If one thinks in terms of connections, the form $A$ is associated to the representation of a connection $\nabla=d+A$ on $E$ in the canonical 
basis $s=(s_1,s_2)$ of $E$ by: for each $f\in C^\infty(\Sigma)$
\[ \nabla (fs)=(df +Af)s, \qquad s_1(x)=\left (\begin{array}{c}
1\\
0 \end{array}\right), \quad s_2(x)=\left (\begin{array}{c}
0\\
1 \end{array}\right).\]
On $E$, one can act by a gauge transform $h\in C^\infty(\Sigma,G^\C)$, this changes the basis to $s':=hs$
and the representation of $\nabla$ in the basis $s'$ is the conjugated connection $h\circ \nabla \circ h^{-1}$ has connection form given by
\[ h.A:=hA{h}^{-1}+hd{h}^{-1}\]
which becomes, at the level of $(1,0)$ and $(0,1)$ forms, 
 \begin{equation}\label{def:A_h} 
 A_{h}^{1,0}:=(h . A)^{1,0}=hA^{1,0}{h}^{-1}+h\pl {h}^{-1}, \quad A_{h}^{0,1}:=(h. A)^{0,1}=hA^{0,1}{h}^{-1}+h\bar{\pl} {h}^{-1}.
 \end{equation}
There is also a left action of $C^\infty(\Sigma,G^\C)\times C^\infty(\Sigma,G^\C)$ 
on the set of pairs $(h,A=A^{1,0}+A^{0,1})$ given by 
\begin{equation}\label{gauge_action} 
(h_1,h_2).(h,A)=(h_2hh_1^{-1},A_{h_1}^{1,0}+A_{h_2}^{0,1})
\end{equation}
and we compute the gauge transformation of the action $S_\Sigma(h,A)$:
\begin{lemma}\label{l:invarianceS}
For $h_1,h_2,h\in C^\infty(\Sigma,G^\C)$ and $A\in C^\infty(\Sigma,\mathfrak{g}^\C\otimes \Lambda^1\Sigma)$, the following identity holds true as an element in $\C/2\pi i\Z$
\[ S_\Sigma ((h_1,h_2).(h,A))=S_{\Sigma}(h,A)-S_{\Sigma}(h_1,A^{1,0})-S_{\Sigma}(h_2^{-1},A^{0,1}).\]
\end{lemma}
\begin{proof} 
We compute 
\[\begin{split}
S_\Sigma ((h_1,h_2).(h,A))=&S_\Sigma(h_2hh_1^{-1})-\frac{1}{2\pi i}\int_\Sigma {\rm Tr}((h_1A^{1,0}h_1^{-1}+h_1\pl h_1^{-1})\wedge h_1h^{-1}h_2^{-1}\bar{\pl} (h_2hh_1^{-1}))\\
& -\frac{1}{2\pi i}\int_\Sigma {\rm Tr}(h_2hh_1^{-1}\pl (h_1h^{-1}h_2^{-1})\wedge (h_2A^{0,1}h_2^{-1}+h_2\bar{\pl}h_2^{-1}))\\
& -\frac{1}{2\pi i}\int_{\Sigma}{\rm Tr}(h(A^{1,0}-h_1^{-1}\pl h_1) h^{-1}\wedge (A^{0,1}-h_2^{-1}\bar{\pl} h_2))\\
 = &S_\Sigma(h_2hh_1^{-1})+\frac{1}{2\pi i}\int_\Sigma {\rm Tr}(h_2^{-1}\pl h_2\wedge A^{0,1})
+\frac{1}{2\pi i}\int_{\Sigma}{\rm Tr}(A^{1,0}\wedge h_1^{-1}\bar{\pl}h_1)\\
&  -\frac{1}{2\pi i}\int_{\Sigma}{\rm Tr}(h\pl h^{-1} \wedge A^{0,1}+A^{1,0}\wedge h^{-1}\bar{\pl}h
+hA^{1,0}h^{-1}\wedge A^{0,1})\\
&-\Upsilon(h_2,hh_1^{-1})-\Upsilon(h,h_1^{-1}) +\Upsilon(h_2,h_2^{-1})+\Upsilon(h_1,h_1^{-1}).
\end{split}\]
From Lemma \ref{l:Polyakov_Wiegmann}, we can expand
\[ \begin{split}
S_\Sigma(h_2hh_1^{-1})=&S_{\Sigma}(h_2)+S_{\Sigma}(h)+S_{\Sigma}(h_1^{-1})+\Upsilon(h_2,hh_1^{-1})+\Upsilon(h,h_1^{-1}) \\
=& -S_{\Sigma}(h_2^{-1})-\Upsilon(h_2,h_2^{-1})-S_\Sigma (h_1)-\Upsilon(h_1,h_1^{-1})+S_{\Sigma}(h)+\Upsilon(h_2,hh_1^{-1})+\Upsilon(h,h_1^{-1}).
\end{split}\]
Gathering everything, we deduce the  gauge covariance for the action $S_{\Sigma}(h,A)$.
\end{proof}
A direct corollary of Lemma \ref{l:invarianceS} and of the relation (by Lemma \ref{l:Polyakov_Wiegmann}) in $\C/2\pi i\Z$
\[ S_\Sigma(h)+S_\Sigma(h^{-1})=\frac{1}{2\pi i}\int_\Sigma {\rm Tr}(h^{-1}\pl h\wedge h^{-1}\bar{\pl}h)\]
is the invariance of $\tilde{S}_\Sigma$ under the diagonal gauge action:
\begin{corollary}[Gauge invariance]\label{gauge_inv}
For all $h',h\in G^{\C}$ and  $A\in C^\infty(\Sigma,\mathfrak{g}^\C\otimes \Lambda^1\Sigma)$, the following identity holds true in $\C/2\pi i\Z$
\begin{equation}\label{invarianceWZW} 
\tilde{S}_\Sigma((h',h').(h,A))=\tilde{S}_\Sigma(h,A).
\end{equation}
where 
\begin{equation}\label{deftildeS}
\tilde{S}_\Sigma (h,A):=S_\Sigma (h,A)-S_\Sigma ({\rm Id},A).
\end{equation}
\end{corollary}

\noindent\textbf{Variation of WZW action.} 
We will check that the WZW action has a simple  variation formula that can be expressed in terms of local quantities on $\Sigma$, namely the curvature of a connection. The critical points of the actions will be interepreted as flat connections, as we  will explain later.
\begin{lemma}\label{l:variation_WZW_action}
Let $h_t\in C^\infty(\Sigma,G^\C)$ be a $1$-parameter smooth family of gauge transforms such that $h_0={\rm Id}$, $A\in C^\infty(\Sigma;\mathfrak{g}^\C\otimes \Lambda^1\Sigma)$ and $h\in C^\infty(\Sigma,G^\C)$. Then we get with $\dot{h}=\pl_th_t|_{t=0}$
\[ \pl_tS_\Sigma (hh_t,A)|_{t=0}= -\frac{1}{2\pi i}\int_\Sigma {\rm Tr}(\dot{h}F_{A^{1,0}+A_{h^{-1}}^{0,1}})\]
where $F_B=dB+B\wedge B$ denotes the curvature of a connection form $B$.
\end{lemma}
\begin{proof} Since one can choose the extensions $\tilde{h}_t$ of $h_t$ to be $C^1$ in $t$, $t\mapsto S_\Sigma (hh_t,A)$ is $C^1$ with values in $\C/2\pi i\Z$, its $\pl_t$-derivative at $t=0$ is thus well-defined independently of the extension.
Using Lemma \ref{l:Polyakov_Wiegmann},
\begin{equation}\label{variationSht}
\begin{split} 
\pl_t S_\Sigma(hh_t)|_{t=0}=& \pl_t S_\Sigma(h)|_{t=0}+\pl_t S_\Sigma(h_t)|_{t=0}+\frac{1}{2\pi i}\pl_t\int_\Sigma{\rm Tr}(h^{-1}\bar{\pl} h\wedge h_t\pl h_t^{-1})|_{t=0}\\
=&- \frac{1}{2\pi i}\int_\Sigma{\rm Tr}(h^{-1}\bar{\pl} h\wedge \pl \dot{h})\\
=& - \frac{1}{2\pi i}\int_\Sigma{\rm Tr}(\dot{h} d(h^{-1} \bar{\pl} h)).
\end{split}\end{equation}
We also compute using Stokes formula in the second equality
\[\begin{split}
(\pl_t(S_\Sigma (hh_t,A)-S_\Sigma(hh_t))|_{t=0}=& -\frac{1}{2\pi i}\int_\Sigma {\rm Tr}(\dot{h}h^{-1}\bar{\pl}h\wedge A^{1,0})+  {\rm Tr}(A^{1,0}\wedge h^{-1}\bar{\pl}(h\dot{h}))\\
 & -\frac{1}{2\pi i}\int_\Sigma {\rm Tr}(\dot{h}\pl h^{-1}\wedge A^{0,1}h)-{\rm Tr}(\pl(\dot{h}h^{-1})\wedge A^{0,1}h)\\
& -  \frac{1}{2\pi i}\int_\Sigma {\rm Tr}(h\dot{h}A^{1,0}h^{-1}\wedge A^{0,1}-hA^{1,0}\dot{h}h^{-1}\wedge A^{0,1}) \\
 =& - \frac{1}{2\pi i}\int_\Sigma {\rm Tr}(\dot{h}h^{-1}\bar{\pl}h\wedge A^{1,0})+{\rm Tr}(\dot{h}d(A^{1,0}h^{-1})h)\\
 & - \frac{1}{2\pi i}\int_\Sigma {\rm Tr}(\dot{h}\pl h^{-1}\wedge A^{0,1}h)+ {\rm Tr}(\dot{h}h^{-1}d(A^{0,1}h))\\
&  -  \frac{1}{2\pi i}\int_\Sigma {\rm Tr}(\dot{h}A^{1,0}\wedge h^{-1}A^{0,1}h+\dot{h} h^{-1}A^{0,1}h\wedge A^{1,0}) \\
 =&- \frac{1}{2\pi i}\int_\Sigma {\rm Tr}(\dot{h} d(A^{1,0}+h^{-1}A^{0,1}h))\\
 & - \frac{1}{2\pi i}\int_\Sigma {\rm Tr}(\dot{h}(A^{1,0}+A_{h^{-1}}^{0,1})\wedge (A^{1,0}+A_{h^{-1}}^{0,1}))
\end{split}\]
and this proves the Lemma by adding \eqref{variationSht}.
\end{proof}

\subsection{The WZW action on the coset space ${\rm SL}(2,\C)/{\rm SU}(2)=\mathbb{H}^3$}\label{WZW_coset}
We will now consider the WZW action on a space of cosets of $G^\C={\rm SL}(2,\C)$ by its maximal compact subgroup $G={\rm SU}(2)$.
 Working in the ${\rm SL}(2,\C)$ setting has the disadvantage that the action $S_\Sigma(h)$ is not bounded above, making the path integral definition from probability not (a priori) possible. By moding out by $G$, we obtain a Riemannian symmetric space, namely the hyperbolic $3$-space $\mathbb{H}^3$. 
Restricting the field $h$ to be  $G^\C/G$ valued turn the action into an functional of $h$ whose formal measure 
$e^{S_\Sigma(h)}Dh$ has the flavour of a Gaussian measure.\\ 

\noindent \textbf{The action.} The quotient ${\rm SL}(2,\C)/{\rm SU}(2)$ is the symmetric space $\mathbb{H}^3$, i.e. the $3$-dimensional hyperbolic space. We wish to describe the cosets in a convenient way. Using the polar decomposition, any matrix $h\in G^\C={\rm SL}(2,\C)$ can be uniquely written as $h=h_+u$ where $h_+$ is a positive definite $2\times 2$ matrix and $u\in {\rm U}(2)$. 
Moreover, since $\det(h)=1$, we deduce that $\det(h_+)=1=\det(u)$, i.e. $u\in {\rm SU}(2)$. We can write $h_+$ uniquely under the form
\begin{equation}\label{coset_p}
h_+=\left(\begin{array}{cc}
e^{\phi}(1+|\nu|^2) & \nu \\
\bar{\nu} &e^{-\phi}
\end{array}\right)
\end{equation}
for some $\phi\in \R$ and $\nu \in \C$. Observe that $h_+=qq^*$ with $q=an$ where 
\begin{equation}\label{def_an}
 a=\left(\begin{array}{cc}
e^{\phi/2}& 0 \\
0 & e^{-\phi/2}
\end{array}\right),\quad n=\left(\begin{array}{cc}
1 & \nu \\
0 & 1
\end{array}\right).
\end{equation}
We can thus represent the coset space as  
\[ {\rm SL}(2,\C)/{\rm SU}(2)\simeq \{ (e^{\phi}, \nu)\in \R^+\times \C\}=\mathbb{H}^3\]
and $\phi, \nu$ will be the coordinates used for this space. 
Consider the trivial vector bundle $E=\Sigma\times \C^2$ equipped with a Hermitian product $g_E$.
Any unitary connection with trivial determinant has the form $d+A$, in the canonical basis $(s_1,s_2)$, 
for some connection $1$-forms $A\in C^\infty(\Sigma;\mathfrak{g}^\C\otimes \Lambda^{0,1}\Sigma)$ satisfying $A^*=-A$ where the adjoint is with respect to $g_E$. Consider the set of pairs 
\[  \mc{Z}:= \{ (h,A) \in C^\infty(\Sigma;G^{\C})\times C^\infty(\Sigma;\mathfrak{g}^{\C} \otimes \Lambda^{0,1}\Sigma) \,|\, h \textrm{ positive definite}, \,  A^*=-A\}.\]
We claim that the restriction of the action \eqref{gauge_action} to smooth functions with values in  $(G^\C)_d:=\{ ((h^*)^{-1},h)\in G^\C \times G^\C\}$ preserves the set $\mc{Z}$. Indeed, for $(h,A)\in \mc{Z}$ and $h_0\in C^\infty(\Sigma,G^\C)$, the section $h_0hh_0^*$ takes values in positive definite matrices in $G^\C$ and 
\[ \begin{split}
(A_{{h_0^*}^{-1}}^{1,0}+A_{h_0}^{0,1})^*= & h_0(A^{1,0})^*h_0^{-1}+{h_0^*}^{-1}(A^{0,1})^*h_0^*
+(\pl h_0^*)^*h_0^{-1}+ (\bar{\pl} h_0^{-1})^*h_0^* \\
=& - h_0A^{0,1}h_0^{-1}-{h_0^*}^{-1}A^{1,0}h_0^*+\bar{\pl} h_0 h_0^{-1}+\pl {h_0^*}^{-1}h_0^*\\
 =& -(A_{{h_0^*}^{-1}}^{1,0}+A_{h_0}^{0,1})^*.
\end{split}\]

\noindent \textbf{Absence of quantization condition for $\mathbb{H}^3$-WZW.} Let us make an important remark about the possible values of multiple $k$ of the action 
to define the formal measure $e^{kS_\Sigma(h)}Dh$.
Proceeding as in \eqref{S_top}, if $h\in C^\infty(\Sigma;G^\C)$ is positive definite (thus of the form \eqref{coset_p} with $(e^{\phi},\nu)\in \mathbb{H}^3$) and $\tilde{h}_i\in C^\infty(X;G^\C)$ two extensions to two $3$-manifolds $X_1,X_2$ with boundary $\Sigma$ we see that, if $X_1\# X_2$ is the connected sum of $X_1$ with $X_2$ along their boundary,
\[S^{\rm top}_{X_1}(\tilde{h}_1)-S^{\rm top}_{X_2}(\tilde{h}_2)=\frac{1}{12\pi i}\int_{X_1\# X_2}{\rm Tr}((\tilde{h}^{-1}d\tilde{h})^{\wedge 3})=0\]
since $\mathbb{H}^3$ is contractible. We then see that $S^{\rm top}_{X}(\tilde{h})$ and $S_\Sigma(h,A)$ as elements in $\C$ (rather than $\C\setminus 2\pi i\Z$ when $h$ takes values in ${\rm SL}_2(\C)$) are independent of $X$ and of the extension of $h$ to a $3$-manifold $X$. The functional $e^{kS_\Sigma(h)}$ is then well-defined for all $k\in \C$ independently of the extension $\tilde{h}$ to $X$ as a positive definite matrix, if $h\in C^\infty(\Sigma;G^\C)$ is positive definite.\\

\noindent \textbf{Computation of the action.} 
Let us compute the action $S_\Sigma(h)$ where $h=ann^*a$ is a representative of the coset. We represent $a,n$ as in \eqref{def_an} with $\phi:\Sigma \to \R$ and $\nu:\Sigma\to \C$ two smooth functions. As we shall see in the probabilistic approach, it will be convenient to consider the variable $\gamma:=e^{\phi}\nu$. 
\begin{lemma}\label{lem:S(h)}
For $h=a nn^*a$ with the notation \eqref{def_an} and  with $\gamma:=e^{\phi}\nu$, the following holds true in $\R$
\begin{equation}\label{actionS(h)}
 \begin{split}
 S_\Sigma(h)
 &= \frac{1}{2\pi i} \int_\Sigma \pl \phi \wedge \bar{\pl}\phi + (\pl+\pl \phi)\bar{\nu}\wedge (\bar{\pl}+\bar{\pl}\phi)\nu\\ 
& = \frac{1}{2\pi i} \int_\Sigma \pl \phi \wedge \bar{\pl}\phi + e^{-2\phi}\pl \bar{\gamma} \wedge \bar{\pl}\gamma.
 \end{split}
 \end{equation}
\end{lemma}
\begin{proof}
Since we choose the extension of $h$ to a $3$-manifold $X$ with $\pl X=\Sigma$  to be in $\mathbb{H}^3$ (positive definite), it has the form $\tilde{h}=(\tilde{a}\tilde{n})(\tilde{a}\tilde{n})^*$ with $\tilde{a},\tilde{n}$  
of the form \eqref{def_an} with $\tilde{\phi},\tilde{\nu}$ some extension of $\phi,\nu$ to $X$. 
We first observe that $S_{\Sigma}(n)=S_{\Sigma}(n^*)=0$ since $n^{-1}dn$  and ${n^*}^{-1}dn^*$ are nilpotent matrices,
and, a direct computation gives $S_{\rm top}(\tilde{a})=0$, implying that 
\[S_\Sigma(a)=\frac{1}{4\pi i} \int_\Sigma \pl \phi \wedge \bar{\pl}\phi. \]
By Lemma \ref{l:Polyakov_Wiegmann} and using that $S_{\Sigma}(n)=S_{\Sigma}(n^*)=0$ (since $n^{-1}dn$  and ${n^*}^{-1}dn^*$ are nilpotent matrices), one obtains
\[ S_\Sigma(a nn^*a )=2S_\Sigma (a)+\Upsilon(a,nn^*a)+\Upsilon(n,n^*a)+\Upsilon(n^*,a).\] 
We also have $\Upsilon(n^*,a)= 0$, and a direct computation  by taking the extension $\tilde{h}=\tilde{a}\tilde{n}$ gives
\[\Upsilon(n,n^*a)=-\frac{1}{2\pi i}\int_\Sigma \bar{\nu}\bar{\pl}\nu \wedge \pl \phi+\bar{\pl}\nu\wedge \pl \bar{\nu}\]
\[\Upsilon(a,nn^*a)= -\frac{1}{2\pi i}\int_{\Sigma}  (\frac{1}{2}+|\nu|^2) \bar{\pl}\phi\wedge \pl \phi+\nu\bar{\pl}\phi \wedge \pl \bar{\nu}.\qedhere\]
\end{proof}

\noindent \textbf{Discussion on a probabilistic definition of the path integral for non-trivial connections.} Using the expression \eqref{actionS(h)} for the action $S_\Sigma(h)$, we have seen in Section \ref{sub:PIsphere} that the path integral for the WZW action with $A=0$ can be done using $\phi=c+X_g$ a GFF and $\gamma_g$ a witten field with covariance written in terms of GMC integrals. When $A\not=0$, there are terms in the action \eqref{gaugeWZW} that cause problems. On the other hand, if we consider the special case where $A^{0,1}$ is upper triangular (and $A^{1,0}=-(A^{0,1})^*$)
\[ A^{0,1}=\left(\begin{array}{cc}
a & \beta  \\
0 & -a
\end{array}\right), \quad  a,\beta \in C^\infty(\Sigma;\Lambda^{0,1}\Sigma),\]
a direct computation gives\footnote{Lemma \ref{lem:S(h,A)_formula_stable} gives a more general proof.}, for $h$ of the form \eqref{coset_p} and $\gamma=e^{\phi}\nu$, 
\[\begin{split} 
\tilde{S}_{\Sigma}(h,A)=& \frac{1}{2\pi i} \int_\Sigma \pl \phi \wedge \bar{\pl}\phi +e^{-2\phi}( (\pl-2\bar{a})  (e^{\phi}\bar{\nu})+\bar{\beta})\wedge  ((\bar{\pl}+2a)(e^{\phi}\nu)+ \beta)-  2\phi da-\bar{\beta}\wedge \beta \\
=&  -\frac{1}{\pi }\int_\Sigma (\frac{1}{4}|d\phi|^2_g +e^{-2\phi}|(\bar{\pl}+2a)\gamma+ \beta|^2_{g}-|\beta|^2_{g}){\rm dv}_g -\frac{1}{\pi i}\int_{\Sigma }\phi da.
\end{split}\]
The probabilistic definition of the path integral for such action can be performed in a way quite similar to Section \ref{sub:PIsphere}, but replacing
the operator  $\caD_\phi =\bar{\pl}^\ast e^{-2\phi}\partial $ by 
\[ \mc{D}_\phi= (\bar{\pl}+2a)^*e^{-2\phi}(\bar{\pl}+2a)\]
and using its inverse as a covariance kernel for the random variable $\gamma_g$, still with $\phi=c+X_g$ a GFF. The $\beta$ term is a shift and can be dealt with by Cameron-Martin. 

When $A^{0,1}$ is not upper (or lower) triangular, the action contains some terms of the form 
$e^{\phi}|\nu|^2 \pl\nu$ and $e^{\phi}\nu \pl|\nu|^2$ that appear (with $\nu=e^{-\phi}\gamma$), and finding a proper probabilistic definition for these terms is not clear to us, 
because of the regularity of $\gamma_g$ and $X_g$. We shall rather use a geometric approach to set a definition for the path integral in this general case. The very definition of the fields $\phi$ and $\gamma$ is associated to the splitting $\C\oplus \C$ of $\C^2$ using the basis 
$s_1,s_2$. When $A^{0,1}$ is triangular superior, it means that $\C s_1$ is a holomorphic line subbundle  (this notion is explained in the next section) of $E=\Sigma \x \C^2$, and in general we will use that there always exists a holomorphic subbundle of $E$. Such subbundles do not a priori admit a global non-vanishing section. This procedure needs the general theory of rank $2$ holomoprhic bundles and its classification using the theory of extensions. This is explained in the next section.

\subsection{Holomorphic vector bundles}\label{sec:holo_vect_bundle}

In this section, we recall for completeness some standard facts about (rank 2) vector bundles on surfaces, how they define a principal ${\rm SL}(2,\C)$ bundle, their holomorphic structure and Hermitian metrics/unitary connections, as well as their description using extensions.
A reader familiar with these concepts can skip directly to Section \ref{s:WZWaction_E}.\\

\noindent\textbf{Complex vector bundles.} A smooth complex vector bundle $E\to \Sigma$ with projection $\pi_E$ is said to have (complex) \textbf{rank} $n$ if $\dim_{\C} E_x=n$ for each $x$ and $E_x:=\pi_E^{-1}(x)$ is the fiber over $x$. A smooth morphism $F:E\to E'$ between two rank $n$ complex vector bundle is a smooth map such that there is $f:\Sigma\to \Sigma$ smooth with $\pi_{E'}\circ F=f\circ \pi_E$ and such that $F|_{E_x}:E_x\to E'_{f(x)}$ is a complex linear map. We say it is an isomorphism if it admits a smooth inverse morphism.  Equivalence classes of isomorphic complex vector bundles are characterized by their rank and their degree. The \textbf{degree} of $E$ is denoted ${\rm deg}(E)$ and is defined to be 
\[{\rm deg}(E):=\int_\Sigma c_1(E)\in \Z\]
if $c_1(E)\in H^2(\Sigma,\Z)$ is the first Chern class of $E$ (the Chern class can be computed by choosing a connection $\nabla$ on $E$ and then  by using the formula \eqref{Chern_class}) and $H^k(\Sigma,\Z)$ is the $k$-th cohomology group of $\Sigma$ with value in $\Z$.
Two smooth complex vector bundles $E$ and $E'$ with (complex) rank $n$ are isomorphic if and only if ${\rm deg}(E)={\rm deg}(E')$.\\

\noindent\textbf{Principal bundles associated to $E$.} There is a (smooth) principal ${\rm GL}(n,\C)$ bundle $\pi_P:P\to \Sigma$ associated to $\pi_E:E\to \Sigma$ defined by setting the fibers $P_x=\pi^{-1}_P(x)$ at $x\in \Sigma$ to be 
\[ P_x=\{ p: \C^n \to E_x \textrm{ complex linear isomorphism}\}\]
and with action 
\begin{equation}\label{rightaction} 
(h,(x,p_x)) \in {\rm GL}(n,\C)\times  P \mapsto (x, p_x\circ h). 
\end{equation} 
The determinant bundle $\Lambda^nE$ of $E$ is the rank $1$ complex vector bundle over $\Sigma$ with fiber at $x$ given by $\Lambda^nE_x$. If the determinant bundle is trivial, there is a global non-vanishing section $s:\Sigma \to \Lambda^nE$ and one can define a principal ${\rm SL}(n,\C)$ bundle $\pi_Q:Q\to \Sigma$ by setting the fiber at $x$ to be 
\[ Q_x=\{ p: \C^n \to E_x \textrm{ complex linear isomorphism}\,|\, p_*(e_1\wedge \dots \wedge e_n)=s\}\]
if $(e_1,\dots, e_n)$ is the canonical basis of $\C^n$ and $p_*$ the linear action $\Lambda^n\C^n\to \Lambda^nE_x$ induced by $p:\C^n\to E_x$. The action \eqref{rightaction} with ${\rm GL}(n,\C)$ replaced by  ${\rm SL}(n,\C)$ produces an ${\rm SL}(n,\C)$-principal bundle structure on $Q$.
Any principal  ${\rm SL}(2,\C)$-bundle over a closed surface $\Sigma$ is trivial, thus has a global section. Indeed, there is a bijection between the isomorphism classes of ${\rm SL}(2,\C)$-principal bundles and the homology group $H^4(\Sigma,\Z)$, which is equal to $0$ as $\dim \Sigma=2$ (the bijection is given by taking the second Chern class $c_2(E)\in H^4(\Sigma,\Z)$ of $E$, where $E$ is a rank-$2$ complex vector bundle with trivial determinant associated to $P$).
This implies that the bundle $\pi_E:E\to \Sigma$ is trivial if its determinant bundle is trivial.\\

\noindent\textbf{Holomorphic vector bundles.} Here we review classical material about holomorphic vector bundles on surfaces, we refer to \cite{Kobayashi,Schaffhauser} and \cite[Appendix]{Wells} for more details.
A holomorphic structure on a smooth complex vector bundle $\pi_E: E\to \Sigma$ of rank $n$  over a Riemann surface $\Sigma$  is an atlas of smooth local trivializations $h_i: \pi^{-1}(U_i)\to U_i\times \C^n$, $i\in I$, such that the transition functions  
$h_{ij}=h_i\circ h_j^{-1}: U_i\cap U_j\to {\rm GL}(n,\C)$ are holomorphic with respect to the holomorphic structure on 
$\Sigma$. They satisfy the cocycle relation $h_{ij}h_{jk}=h_{ik}$ on $U_i\cap U_j\cap U_k$.
We say that a  section $s: U\subset \Sigma\to E$ on an open set $U$ is holomorphic if for all $i$, $h_i\circ s$ is holomorphic as a map from 
$U_i\cap U$ to $\C^n$. Tensor products of holomorphic vector bundles are also holomorphic vector bundles and the dual of a holomorphic vector bundle is also a holomorphic vector bundle. The bundle dual to a complex (or holomorphic) vector bundle $E$
will be denoted by $E^*$. A smooth complex vector subbundle 
$F\subset E$ (of dimension $k<n$) is said to be a \textbf{holomorphic subbundle} of $E$ if one can find trivializations of $E$ that satisfy 
$h_i:F \to U_i\times \C^{k}\times \{0\}$ and the transition functions $h_{ij}$ preserve $\C^k\otimes \{0\}$, i.e. they have upper triangular block 
form in ${\rm GL}(n,\C)$, over $U_i\cap U_j$. The $h_i|_{F}$ thus induce a structure of holomorphic vector bundle on $F$.
If $F\subset E$ is a holomorphic subbundle, then $E/F$, the complex vector bundle with fiber $E_x/F_x$ at $x\in \Sigma$, also inherits a structure of holomorphic vector bundle. But in general one can not find a complement holomorphic vector bundle 
$F'\subset E$ such that $F\oplus F'=E$.

The determinant bundle $\Lambda^nE$ inherits a holomorphic structure: a local section $s:U_i\to \Lambda^nE$ is holomorphic if for each local holomorphic trivialization $h_i$ of $E$, one has $(h_i)_*s=f_i e_1\wedge \dots \wedge e_n$ for some holomorphic function $f_i:U_i\to \C$, with $(e_1,\dots,e_n)$ the canonical basis of $\C^n$.
If $\Lambda^nE$ is holomorphically trivial, there is a global holomorphic section $s:\Sigma \to \Lambda^nE$ and 
$(h_i)_*s=f_i(x)e_1\wedge \dots \wedge e_n$ for some non-vanishing holomorphic functions $f_i:U_i\to \C$ , satisfying $f_i(x)f_j(x)^{-1}= \det(h_{ij}(x))$. We can modify the trivialization $h_i$ to $\tilde{h}_i=(f_i^{-1}h_{i,1},h_{i,2},\dots,h_{i,n})$ if $h_{i}=(h_{i,1},\dots,h_{i,n})$, and we see that 
$\tilde{h}_{ij}:=\tilde{h}_i\circ \tilde{h}_j^{-1}$ belongs to ${\rm SL}(n,\C)$. Conversely, if there is a covering by local trivializations such that $h_{ij}: U_i\cap U_j\to {\rm SL}(n,\C)$, the determinant bundle is trivial, since the local sections 
$s_i:=(h_i^{-1})_*(e_1\wedge \dots \wedge e_n)$ glue together to a global holomorphic non-vanishing section of $\Lambda^nE$. We shall say in this case that $E$ has \textbf{trivial determinant}.\\

\noindent\textbf{Dolbeault operators}. Given a smooth complex vector bundle $E$, a Dolbeault operator is a $\C$-linear differential operator of order $1$ acting on $E$-valued forms and 
satisfying $\bar{\pl}_E \circ \bar{\pl}_E=0$ and $\bar{\pl}_E(fs)= s\otimes (\bar{\pl}f)+ f\bar{\pl}_Es$ for any $f\in C^\infty(\Sigma)$ and $s\in  C^\infty(\Sigma;E)$. In fact, defining $\bar{\pl}_E: C^\infty(\Sigma;E)\to C^\infty(\Sigma;E\otimes \Lambda^{0,1}\Sigma)$ is sufficient, as it induces naturally an operator on all forms, satisfying $\bar{\pl}_E \circ \bar{\pl}_E=0$.
As an example, if $\nabla:C^\infty(\Sigma;E)\to C^\infty(\Sigma;E\otimes \Lambda^1\Sigma)$ is a connection on the bundle 
$E$, we can split the connection under the form $\nabla^{1,0}+\nabla^{0,1}$ with 
\[\nabla^{0,1}: C^\infty(\Sigma;E)\to C^\infty(\Sigma;E\otimes \Lambda^{0,1}\Sigma), \quad \nabla^{1,0}: C^\infty(\Sigma;E)\to C^\infty(\Sigma;E\otimes \Lambda^{1,0}\Sigma)\]
and we denote $\bar{\pl}^\nabla:=\nabla^{0,1}$. This operator $\bar{\pl}^\nabla$ can be extended to acting on all differential forms and it satisfies $\bar{\pl}^\nabla \circ \bar{\pl}^\nabla=0$; it is thus a Dolbeault operator. 
On a holomorphic vector bundle,  there is a Dolbeault  operator 
\[ \bar{\pl}_E: C^\infty(\Sigma;E\otimes  \Lambda^{p,q}\Sigma)\to C^\infty(\Sigma;E\otimes \Lambda^{p,q+1}\Sigma)\]
given in local trivializations $h_i$ over $U_i$ by 
\[\forall f\in C^\infty(U_i;\Lambda^{p,q}\C),\quad  \bar{\pl}_E (fh_i):= h_i\otimes \bar{\pl}f\]
and whose elements in $\ker \bar{\pl}_E\cap C^1(U;E)$ (for $U\subset \Sigma$ open set) are exactly the holomorphic sections of $E$ in $U$.
Conversly, by \cite[Proposition 3.7]{Kobayashi}, for a given Dolbeault operator $\bar{\pl}_E$ there is a holomorphic vector bundle on $E$ compatible with $\bar{\pl}_E$ (i.e. whose local sections killed by $\bar{\pl}_E$ are holomorphic), and it is unique up to isomorphism. Riemann-Roch theorem allows to relate the dimensions of the cohomologies
\[H^0(\Sigma;E):=\ker \bar{\pl}_E, \quad H^1(\Sigma;E):=C^\infty(\Sigma;E\otimes\Lambda^{0,1}\Sigma)/ \{ \bar{\pl}_Ef \,|\, f\in C^\infty(\Sigma;E)\} \] 
by the formula, if $r={\rm rank}(E)$,
\begin{equation}\label{RRE} 
\dim H^0(\Sigma;E)-\dim H^1(\Sigma;E)={\rm deg}(E)-r(g(\Sigma)-1).
\end{equation}
The group $\mc{G}_E$ of automorphisms of $E$ acts on the affine space of Dolbeault operators by 
\[ h.\bar{\pl}_E(s) := h\bar{\pl}_E(h^{-1}s)\]
and if ${\rm Dol}(E)$ denotes the space of Dolbeault operators, we have an isomorphism 
\[ {\rm Dol}(E)/\mc{G}_E \to \{\textrm{Holomorphic structures on }E\}/ {\rm isomorphism}.\]

This space of equivalence classes is studied in \cite{Narasimhan-Seshadri}, it is in general non-Hausdorff but there is a subclass of complex vector bundles, which is generic, for which the equivalence class is an analytic variety, the so-called stable bundles; 
let us give a very brief summary of the picture. 
This discussion will not be used in the current article but 
it is still instructive. Indeed,  it is natural to investigate how the path integral reacts to variations of the holomorphic vector bundle inside the moduli space, giving rise to the Kac-Moody symmetries.   
To define stable bundles, we need the notion of slope of a vector bundle $E$, given by $\mu(E):=\deg(E)/{\rm rank}(E)$.
A holomorphic vector bundle $(E,\bar{\pl}_E)$ on $\Sigma$ is called \textbf{stable} if for all non-trivial holomorphic vector subbundle $F$
\[ \mu(F)<\mu(E)\]
while it is called \textbf{semi-stable} if for all non-trivial holomorphic vector subbundle $F$
\[ \mu(F)\leq \mu(E).\] 
A stable bundle is indecomposable, i.e. can not be decomposed as a direct sum of holomorphic subbundles. 
In our case of interest, $E$ has trivial determinant (thus ${\rm deg}(E)=0$) and rank $2$, thus if it is stable (resp. semi-stable) then 
we see that any holomorphic line subbundle $\mc{L}$ satisfies
\begin{equation}\label{stable_deg}
{\rm deg}(\mc{L})<0 \qquad  (\textrm{resp. } {\rm deg}(\mc{L})\leq 0.) 
\end{equation}
The space of stable vector bundles is denoted $\mc{C}_s:=\{ \bar{\pl}_E\,|\, (E,\bar{\pl}_E) \textrm{ is stable}\}$ and the moduli space of stable vector bundle of rank $n$ and degree $d$ is the set 
\[\mc{M}_s(n,d):= \mc{C}_s/\mc{G}_E.\] 
 It was proved by Mumford and Seshadri that $\mc{M}_s(n,d)$ admits a structure of smooth complex manifold of dimension
 $1+n^2(g_\Sigma-1)$ if $g_\Sigma\geq 2$ denotes the genus of $\Sigma$. In the case we study in this article, $n=2$ and we assume in addition that the determinant is trivial, the space of stable bundles of rank $2$ with trivial determinant then has complex dimension $3g(\Sigma)-3$.\\

\noindent\textbf{Connection form.} 
Let $(E,\bar{\pl}_E)$ be a holomorphic vector bundle of rank $n$. If $U\subset \Sigma$ is a (simply connected) local chart, we can trivialize $E$ over $U$ using sections $S=(s_1,\dots, s_n)$.  Then the Dolbeault operator $\bar{\pl}_E$ is represented in the trivialization by: if $f_j$ are smooth functions on $U$  
\[\begin{split} 
\bar{\pl}_E(\sum_{j=1}^n f_j s_j)& =\sum_{j=1}^n (s_j\otimes \bar{\pl} f_j +f_j \bar{\pl}_Es_j) =\sum_{j=1}^n (\bar{\pl}_zf_jd\bar{z}+\sum_{i=1}^n A_{ji}d\bar{z} f_i)s_j 
\end{split}\] 
where  $A_S:=(A_{ij}d\bar{z})_{ij}$  is the  $(0,1)$-form over $U$, with values in the Lie algebra ${\rm GL}(n,\C)$ of ${\rm GL}(n,\C)$, defined by  $\bar{\pl}_E s_i=\sum_{j=1}^nA_{ji}s_jd\bar{z}$. This form is called the \textbf{connection form} of $\bar{\pl}_E$ in the frame $S$. 
Changing the basis of sections to $S'=(s'_1,\dots, s'_n)$ is encoded in a local gauge transform $h\in C^\infty(U,{\rm GL}(n,\C))$ with $S'=Sh$, i.e. $s'_i=\sum_{j=1}^n  h_{ji} s_j$,  and one has 
$$ A_{S'}=h^{-1}A_Sh  +  h^{-1} (\bar{\pl}h).$$

If $E$ is a trivial complex vector bundle (i.e. ${\rm deg}(E)=0$), it admits a global trivialization and the connection form $A_S$ in a trivialization $S$ is defined globally (depending on $S$). If in addition the holomorphic line bundle has trivial determinant, then the local connection $1$-form in a frame $S$ has valued in the Lie algebra ${\rm SL}(n,\C)$ of ${\rm SL}(n,\C)$. In particular, holomorphic rank $2$ bundles with trivial determinant admits a global smooth trivialization and thus can be described by a global connection $1$-form $A^{0,1}\in C^\infty(\Sigma;{\rm SL}(2,\C))$ (in the trivialization).\\

\noindent\textbf{Hermitian metric and unitary connections.} Consider $E$ a  complex vector bundle on a surface $\Sigma$. A Hermitian metric $g_E$ on $E$ is a smooth family $g_E\in C^\infty(\Sigma; S^2E^*)$ of Hermitian metrics on $E$ (here $S^2E^*$ denotes symmetric tensor product). A connection $\nabla$ on $E$ is said \textbf{unitary} if for each vector field $v$ and sections $u_1,u_2\in C^\infty(\Sigma;E)$
\[ v(g_E(u_1,u_2))=g_E(\nabla_{v}u_1,u_2)+g_E(u_1,\nabla_vu_2).\]
 It is a classical result (\cite[Chap. 3, Sec. 2]{Wells}) that if $(E,\bar{\pl}_E)$ is a holomorphic vector bundle equipped with a Hermitian metric $g_E$, there is a unique unitary connection $\nabla$ such that $\nabla^{0,1}=\bar{\pl}_E$. 
 The curvature $F_{\nabla}\in C^\infty(\Sigma; {\rm End}(E)\otimes \Lambda^2\Sigma)$ is called central if 
 \begin{equation}\label{central_connection}
F_{\nabla}=-i\mu {\rm Id}_E \omega_\Sigma
 \end{equation} 
 where $\omega_\Sigma$ is the canonical K\"ahler form on $\Sigma$ normalized so that $\int_\Sigma \omega_\Sigma=2\pi$. A connection with central curvature is called projectively flat, and flat if the curvature is $0$. 
 Since the first Chern class is 
 \begin{equation}\label{Chern_class}
 c_1(E)=\frac{i}{2\pi}{\rm Tr}(F_\nabla),
 \end{equation} 
 we see that $\mu=\mu(E)$ is the slope of $E$.\\

\noindent\textbf{Decomposition of rank $2$ bundles in terms of extensions.} Let us now focus on holomorphic vector bundles of rank $2$ and trivial determinant. A classical way to describe such bundles is to decompose them in terms of line subbundles, via a method called extension. 
First, any holomorphic rank $2$ vector bundle $E$ (with holomorphically trivial determinant in our case) contains a holomorphic line subbundle, i.e. of rank $1$, see \cite[Section 1.2]{Donaldson2021}. For convenience of notations for later, we shall denote by $\mc{L}^{-1}$ this holomorphic line subbundle of $E$ and by $\mc{L}$ its dual line bundle (note that $\mc{L}$ is not a priori a holomorphic subbundle of $E$). We also denote more generally by $\mc{L}^k:=\otimes^{k}\mc{L}$ and $\mc{L}^{-k}:=\otimes^{k}\mc{L}^{-1}$ the tensor products for $k>0$ (which are also holomorphic line bundles).  Our holomorphic line subbundle of $E$ is thus $\mc{L}^{-1}$, which means that there is an exact sequence of holomorphic vector bundles
\begin{equation}\label{exact_sequence} 
0 \rightarrow \mc{L}^{-1}  \rightarrow E  \rightarrow Q \rightarrow 0 
\end{equation}
where $Q:=E/\mc{L}^{-1}$ is a holomorphic line bundle, the arrow $\mc{L}^{-1}\to E$ is simply the injection and the map $E  \rightarrow Q $ is the natural quotient map, both of them being holomorphic in the sense that they map local holomorphic sections to local holomorphic sections.
Since $E$ has trivial determinant, there is a global holomorphic section $s$ of $\det(E)$ and this induces an isomorphism 
\[ \psi: Q \to  \mc{L} ,\quad   (\ell^-\wedge q)=\psi(q)(\ell^-)s, \quad \forall \ell^- \in \mc{L}^{-1}.\]
Here, we have viewed $\mc{L}=(\mc{L}^{-1})^*$ as the dual of $\mc{L}^{-1}$, we  
have used that for $\ell^- \in \mc{L}^{-1}$, $q\in E\mapsto \ell^-\wedge q\in \Lambda^2E$ 
descends to $Q=E/\mc{L}^{-1}$, and that any smooth section of $\Lambda^2E$ can be written as $fs$ in a unique way for some smooth function 
$f$ since $s$ never vanishes.
Let now $g_E$ be a Hermitian metric on $E$.
We will thus identify $Q$ to $\mc{L}$ and the bundle ${\rm Hom}(Q,\mc{L}^{-1})$ of 
endomorphisms from $Q$ to $\mc{L}^{-1}$ identifies to $\mc{L}^{-1}\otimes \mc{L}^*=\mc{L}^{-2}$. 
Since $E$ is equipped with a Hermitian metric $g_E$, one can split the sequence \eqref{exact_sequence} and write $E= \mc{L}^{-1}\oplus Q_0$ (which is not a holomorphic splitting) where $Q_0=(\mc{L}^{-1})^{\perp}$. Note that $g_E$ induces a Hermitian metric $g_{\mc{L}^{-1}}$ on $\mc{L}^{-1}$ and thus a metric $g_{\mc{L}}$ on $\mc{L}$ by duality.
The metric $g_E$ also induces a Hermitian metric on $Q_0$, which is isomorphic to $Q$ by $v\in E/\mc{L}^{-1}\mapsto v^{\perp}\in Q_0$ where $v^\perp$ is the orthogonal projection of $v$ on $Q_0$; using the isomorphism $\psi$, $\mc{Q}_0$ is isomorphic to $\mc{L}$ as holomorphic vector line bundle, and $g_E$ also induces a Hermitian metric on $\mc{L}$ via this isomorphism. 
Up to scaling the global section $s$ defining $\psi$ by a positive function, we can assume that these two metrics on $\mc{L}$ induced by $g_E$ are equal (note that this fixes $\psi$).  We can and shall identify $Q_0$ to $\mc{L}$ as complex vector bundles as described above. 
In the smooth splitting $E=\mc{L}^{-1}\oplus \mc{Q}_0$, the operator $\bar{\pl}_E$ then has the upper triangular matrix form
\begin{equation}\label{dbar_upper_triang} 
\bar{\pl}_E= \left(\begin{array}{cc} 
\bar{\pl}_{\mc{L}^{-1}} & \beta \\
0 & \bar{\pl}_{Q_0}
\end{array}\right) 
\end{equation}
with $\beta\in C^\infty(\Sigma;\Lambda^{0,1}\otimes {\rm Hom}(Q_0,\mc{L}^{-1}))$ (this readily follows from the fact that $\mc{L}^{-1}$ is a line subbundle using local trivialisations so that transition maps are upper triangular).

Using that $\bar{\pl}_E\circ \bar{\pl}_E=0$, we see that $\bar{\pl}_{Q_0}\circ \bar{\pl}_{Q_0}=0$, which means that $\bar{\pl}_E$ induces a holomorphic structure on $Q_0$. Using that the section $s$ of $\det(E)$ is holomorphic,
we obtain directly that our identification $Q_0\simeq \mc{L}$ described above transports $\bar{\pl}_{Q_0}$ to $\bar{\pl}_{\mc{L}}$. We also have 
from $\bar{\pl}_E\circ \bar{\pl}_E=0$ that $\bar{\pl}_{{\rm Hom}(Q_0,\mc{L}^{-1})}\beta=0$ where $\bar{\pl}_{{\rm Hom}(Q_0,\mc{L}^{-1})}$ is the Dolbeault operator on the complex vector bundle ${\rm Hom}(Q_0,\mc{L}^{-1})$ induced by $\bar{\pl}_{\mc{L}^{-1}}$ and $\bar{\pl}_{E}$. 
Thus $\beta$ defines a class  $[\beta]$  in the cohomology $H^1(\Sigma,{\rm Hom}(\mc{L},\mc{L}^{-1}))=H^1(\Sigma,\mc{L}^{-2})$, and this vector space turns out to be the set of equivalence classes of such extensions \cite[Proposition 1]{Donaldson2021}.
There is a family of gauge transform of $E$ which preserve the splitting $\mc{L}^{-1}\oplus \mc{L}^{-1}$ (i.e. the extension class) in the sense that the $h\circ \bar{\pl}_E \circ h^{-1}$ has a triangular form similar to \eqref{dbar_upper_triang} for such $h$; this is explained are used in Section \ref{sec:gauge_transform}.

To summarize, $\mc{L}$ describes a way to write $E$ as an extension of the form \eqref{exact_sequence}  and $\beta$ (or its cohomology class) describes the equivalence class of extension representing $E$. These two data 
\[ (\mc{L},\beta)\]
will be instrumental in our probabilistic construction below, but we emphasize that there are in general several ways to write $E$ as an extension, with either different subbundles $\mc{L}^{-1}$ but sometimes also with the same $\mc{L}^{-1}$ but different cohomology classes of $\beta$. 
In order to apprehend the situations that can appear, 
we gather a couple of (standard) important facts about line subbundles $\mc{L}^{-1}\subset E$ and the space of holomorphic sections of $\mc{L}^{-n}$ and $\mc{L}^{-n}\otimes \Lambda^{0,1}\Sigma$  for $n\geq 1$. For $E$ with rank $2$ and trivial determinant, any holomorphic  subbundle $\mc{L}^{-1}$ satisfies ((1), (2), (3) are easy, (4) can be found in 
\cite[Lemma 10.30]{Mukai}):
\begin{enumerate}
\item for $n\geq 1$, $\ker \bar{\pl}_{\mc{L}^{-n}}=H^{0}(\Sigma,\mc{L}^{-n})=0$ if ${\rm deg}(\mc{L}^{-1})<0$, and in that case (by \eqref{RRE})
\[\dim \ker \bar{\pl}^*_{\mc{L}^{-n}}=\dim H^1(\Sigma,\mc{L}^{-n}\otimes \Lambda^{0,1}\Sigma)=g(\Sigma)-1+n\, {\rm deg}(\mc{L}).\]
In particular if $E$ is stable, any holomorphic subbundle $\mc{L}^{-1}$ satisfies $H^0(\Sigma,\mc{L}^{-2})=0$. 
\item for $n\geq 1$, if ${\rm deg}(\mc{L}^{-1})>2(g(\Sigma)-1)$ then $\ker \bar{\pl}^*_{\mc{L}^{-n}}\simeq H^1(\Sigma,\mc{L}^{-n}\otimes \Lambda^{0,1}\Sigma)=0$ 
\item for $n\geq 1$, ${\rm deg}(\mc{L}^{-1})=0$ and $H^{0}(\Sigma,\mc{L}^{-n})\not=0$, then  $\mc{L}^{-n}$ is holomorphically equivalent to  $\Sigma\times \C$,
\item for $\Sigma=\hat{\C}$ being the sphere, $E$ is holomorphically equivalent to a direct sum $\mc{L}^{-1}\oplus \mc{L}$ of two dual line subbundles, with $\mc{L}^{-1}=\mc{O}(k)$ and $\mc{L}=\mc{O}(-k)$ where $\mc{O}(-k)$ for $k\in \Z$ is the $k$-th power of the tautological bundle on $\hat{\C}$, satisfying ${\rm deg}(\mc{O}(k))=k$ and 
\[ H^0(\hat{\C},\mc{O}(k))=\left \{ \begin{array}{ll}
0 & \textrm{ if }k<0\\
1+k & \textrm{ if }k\geq 0
\end{array} \right., \qquad H^1(\hat{\C},\mc{O}(k)) \left \{ \begin{array}{ll}
-k-1 & \textrm{ if }k<0\\
0 & \textrm{ if }k\geq 0
\end{array} \right..
\]
\end{enumerate}
 Finally, let us briefly discuss the possible subbundles that can appear in a rank-$2$ bundle $E$ with trivial determinant: the space of 
 holomorphic morphisms $\mc{L}^{-1}\to E$ is $H^0(\Sigma, E\otimes \mc{L})$ and by Riemann-Roch \eqref{RRE}
 \[ \dim H^0(\Sigma, E\otimes \mc{L})=\dim H^1(\Sigma, E\otimes \mc{L})+2\, {\rm deg}(\mc{L})+2(1-g(\Sigma))\geq 2({\rm deg}(\mc{L})+1-g(\Sigma))\]
 thus there exist holomorphic subbundles of $E$ isomorphic to $\mc{L}^{-1}$ if ${\rm deg}(\mc{L})>g(\Sigma)-1$.\\

\noindent\textbf{Equivalence between stable bundles and irreducible projectively flat connections.} Let $(E,g_E)$ be a complex vector bundle of degree $d$ and rank $n$ equipped with  Hermitian metric $g_E$, and let
\[\mc{A}_0:=\{\nabla \textrm{ unitary on }(E,g_E)\,|\, \Omega_\nabla=-i\mu(E) {\rm Id}_E \omega_\Sigma\} \]
 be the set of unitary connections on $(E,g_E)$ with central curvature, let $\mc{G}$ the group of automorphism preserving $g_E$ and let $\mc{A}_0^*\subset \mc{A}_0$ be the subset of those $\nabla$ which are irreducible, in the sense that $(E,g_E,\nabla)$ can not be split in a holomorphic way into a direct  sum $(E^1,g^1_E,\nabla^1)\oplus (E^2,g^2_E,\nabla^2)$. 
 Narasimhan-Seshadri \cite{Narasimhan-Seshadri} and Donaldson \cite{Donaldson} proved an homeomorphism
 \[\mc{A}_0^*/\mc{G} \to \mc{M}_s(n,d)=\mc{C}_s/\mc{G}_E.\] 
Concretely, this means that for a stable holomorphic bundle $(E,\bbar{\pl}_E)$ of rank $n$ and degree $d$ and $g_E$ a Hermitian metric, there is a unique (up to gauge) unitary connection $\nabla$ such that $\nabla^{0,1}=\bbar{\pl}_E$ and 
curvature $F_\nabla=-i\mu(E) {\rm Id}_E \omega_\Sigma$. In particular, if the bundle has degree $0$, i.e. is trivial, then there is a unique flat unitary connection up to gauge that corresponds to the holomorphic structure $\bbar{\pl}_E$.

\subsection{WZW actions of $(E,\bar{\pl}_E,g_E)$}\label{s:WZWaction_E}

For $E$ be a rank $2$ holomorphic bundle we denote by ${\rm SL}(2,E)$ the space of automorphism of $E$ with trivial determinant. 
Using \eqref{invarianceWZW} and the material just introduced, we set:  
\begin{definition}[WZW action of Hermitian holomorphic rank-$2$ bundles with trivial determinant]\label{def:actionWZW}
Let $(E,\bar{\pl}_E)$ be a holomorphic vector bundle of rank $2$ with trivial determinant and $g_E$ a Hermitian metric. 
The WZW action of $(\bar{\pl}_E,g_E)$ is defined as an element in $\C/2\pi i\Z$ by: for $h\in C^\infty(\Sigma,{\rm SL}(2,E))$ 
\[ \tilde{S}_{\Sigma,\bar{\pl}_E,g_E}(h):=\tilde{S}_\Sigma(s\circ h\circ s^{-1},A_s)\]
where $A_s\in C^\infty(\Sigma;\mathfrak{g}^{\C}\otimes \Lambda^{0,1}\Sigma)$ is the connection form of the unitary connection $\nabla$ associated to $(\bar{\pl}_E,g_E)$ in a trivialisation $s:E\to \Sigma\times \C^2$. By \eqref{invarianceWZW}, the right hand side is independent of the choice of trivialisation $s$.
\end{definition}

\noindent \textbf{Geometric interpretation of the $\mathbb{H}^3$-WZW action.} 
Let us fix a Hermitian metric $g_E$ on the rank $2$ complex vector bundle $E$ with trivial determinant and let $s:E\to \Sigma\times \C^2$ be a global  trivialization in which $g_E$ becomes the canonical Hermitian product $\cjg \cdot,\cdot\cjd_{\C^2}$ on $\C^2$. A holomorphic structure is given by a Dolbeault operator
$\bar{\pl}_E$, which in the trivialization $s$ is written under the form $\bar{\pl}_E=\bar{\pl}+A^{0,1}$ for some $\mathfrak{g}^\C$-valued  
$(0,1)$-form. 
The unique unitary connection associated to $(\bar{\pl}_E,g_E)$ has connection form (in the trivialization $s$)
given by $A=A^{1,0}+A^{0,1}:=A^{0,1}-(A^{0,1})^*$ (the adjoint is for $\cjg \cdot,\cdot\cjd_{\C^2}$).  
As noted above, any element in ${\rm SL}(2,\C)/{\rm SU}(2)$ can be represented by $h_+=hh^*$ for some $h=an\in {\rm SL}(2,\C)$ (see \eqref{def_an}). Now, for a gauge transform $h\in C^\infty(\Sigma,{\rm SL}(2,\C))$, the connection form $A^{0,1}_{h^{-1}}$ is the connection form of 
the Dolbeault operator $h^{-1}\circ \bar{\pl}_E\circ h$ which is equivalent to $\bar{\pl}_E$. The unique unitary connection for the metric $g_E$ that is associated to $A^{0,1}_{h^{-1}}$ has connection form (in the trivialization $s$)
\begin{equation} \label{defAh}
A_h:=A^{0,1}_{h^{-1}}+A^{1,0}_{h^*}.
\end{equation}
This unitary connection $d+A_h$ is gauge equivalent to $(h^*)^{-1}(d+A_h)h^*$ which has connection form 
\[ A^{1,0}+A^{0,1}_{h_+^{-1}}\]
with $h_+=hh^*$. Thus, $F_{A_h}=h^*F_{A^{1,0}+A^{0,1}_{h_+^{-1}}}(h^*)^{-1}$ and $F_{A_h}=0$ if and only if $A^{1,0}+A^{0,1}_{h_+^{-1}}$ is flat. In particular, by Lemma \ref{l:variation_WZW_action},  the critical points of $h_+ \in C^\infty(\Sigma, G^\C/G) \mapsto S_{\Sigma}(h_+,A)$ 
are exactly the metrics $h_+=hh^*$ with gauges $h$ such that the unitary  connection $A_h$ of \eqref{defAh} is flat. This is also equivalent to saying that the metric $h_+$ on the bundle $\Sigma\times \C^2$ is such that the $h_+$-unitary connection $A^{1,0}_{h_+}+A^{0,1}$ associated to $\bar{\pl}+A^{0,1}$ is flat. Forgetting the trivialization, this means that the metric $g_E(s^{-1} h_+ s \cdot,\cdot)$ on $E$ is such that the unitary connection associated to $\bar{\pl}_E$ and this metric is flat. The classical $\mathbb{H}^3$-WZW action $\tilde{S}_{\Sigma,\bar{\pl}_E,g_E}$ is an action on the space of metrics with determinant $1$ on the bundle $E$ with critical points the metrics $h_+$ 
for which the unique $h_+$-unitary connection associated to $\bar{\pl}_E$ are flat. We have seen above that for stable bundles, the work of Donaldson shows  the existence and uniqueness of such metric $h_+$ (up to unitary gauge).\\

\noindent \textbf{General case of a rank $2$-bundle.} In order to define a probabilistic path integral as we did for the sphere in the case of the trivial bundle, it is convenient to be able to isolate a real valued field $\phi$ and a complex valued field $\gamma$ as entries of the matrix $h$ representing our ${\rm SL}(2,\C)/{\rm SU}(2)$-valued field, for which the WZW action has the form \eqref{actionS(h)} (or some perturbation of this form). For a general rank $2$ holomorphic bundle $(E,\bar{\pl}_E)$ on a surface $\Sigma$, the coordinates $\phi$ and $\gamma$ are associated to a 
 splitting of the holomorphic bundle $E$ into two subbundles. Such a splitting may not always be a holomorphic splitting, but one can always
 split $E$ as a direct sum in a way that $\bar{\pl}_E$ becomes upper triangular. Indeed, as explained in Section \ref{sec:holo_vect_bundle}, we can write our holomorphic rank-$2$ bundle with trivial determinant under the form $E=\mc{L}^{-1}\oplus \mc{L}$ with  $\mc{L}^{-1}$ being a holomorphic line subbundle. We also consider a Hermitian metric $g_E$ so that the splitting is orthogonal, and denote by $g_{\mc{L}}$ the metric induced on $\mc{L}$ and more generally by $g_{\mc{L}^n}$ on $\mc{L}^n$ by   for $n\in \Z$.  We use a global trivialisation $s:E\to \Sigma\times \C^2$  in which the Hermitian product $g_E$ becomes the standard one on $\C^2$. The Dolbeault operator $\bar{\pl}_E$ has the block form on $\mc{L}^{-1}\oplus \mc{L}$
 \[ \bar{\pl}_E= \left(\begin{array}{cc} 
\bar{\pl}_{\mc{L}^{-1}} & \beta \\
0 & \bar{\pl}_{\mc{L}}
\end{array}\right) \] 
with  $\beta\in C^\infty(\Sigma;\mc{L}^{-2}\otimes \Lambda^{0,1}\Sigma)$. The operator then becomes $s\circ \bar{\pl}_E\circ s^{-1}=\bar{\pl}+A^{0,1}$ on the trivial bundle $\C^2$ for some $A^{0,1}\in C^\infty(\Sigma;\mathfrak{g}^\C\otimes \Lambda^{0,1}\Sigma)$.
 The associated unitary connection form is $A=A^{0,1}-(A^{0,1})^*=A^{1,0}+A^{0,1}$.
 Consider $h$ a positive definite matrix  automorphism of $E=\mc{L}^{-1}\oplus \mc{L}$ with trivial determinant. It must then be of the   form (for some $\hat{h}$ with values in ${\rm SL}(2,\C)$)  
\begin{equation}\label{handw}
 h:=s^{-1}\hat{h}s= ww^*,\quad \textrm{ with }  w:=a n
 \end{equation}
and, in the orthogonal splitting $E=\mc{L}^{-1}\oplus \mc{L}$, $a$ and $n$ have the block form
\begin{equation}\label{alphaandn}
 a=\left(\begin{array}{cc}
e^{\phi/2}& 0 \\
0 & e^{-\phi/2}
\end{array}\right),\quad n=\left(\begin{array}{cc}
1 & \nu \\
0 & 1
\end{array}\right)
\end{equation}
where $\phi\in C^\infty(\Sigma)$  and $\nu\in C^\infty(\Sigma;\mc{L}^{-2})$ are uniquely defined,
 similarly to  \eqref{def_an}. We want to express the WZW action  $\tilde{S}_{E,\bar{\pl}_E,g_E}(h)$ 
 in terms of $\phi$ and $\nu$.
We denote by $\nu^*\in C^\infty(\Sigma;\mc{L}^2)$ the adjoint of $\nu$, defined by $\nu^*(\nu):=|\nu|^2_{\mc{L}^{-2}}$, and similarly for the adjoint $\beta^*$  of $\beta$.

\begin{lemma}\label{lem:S(h,A)_formula_stable}
For $h=an(an)^*\in {\rm End}(E)$ of the form \eqref{handw}$+$\eqref{alphaandn} 
\[\begin{split} 
\tilde{S}_{E,\bar{\pl}_E,g_E}(h)=& \frac{1}{2\pi i} \int_\Sigma \pl \phi \wedge \bar{\pl}\phi +e^{-2\phi}( \pl_{\mc{L}^2}  (e^{\phi}\nu^*)+\beta^*)\wedge  (\bar{\pl}_{\mc{L}^{-2}}(e^{\phi}\nu)+ \beta)+  2F_{\mc{L}}\phi-\beta^*\wedge \beta \\
=&  -\frac{1}{\pi }\int_\Sigma (\frac{1}{4}|d\phi|^2_g +e^{-2\phi}|\bar{\pl}_{\mc{L}^{-2}}(e^{\phi}\nu)+ \beta|^2_{\mc{L}^{-2},g}-|\beta|^2_{\mc{L}^{-2},g}){\rm dv}_g +\frac{1}{\pi i}\int_{\Sigma }F_{\mc{L}}\phi
\end{split}\]
where $F_{\mc{L}}$ is the curvature $2$-form of $\mc{L}$, normalized by $\frac{i}{2\pi}\int_\Sigma F_{\mc{L}}={\rm deg}(\mc{L})$, and $\phi$ and $\nu$ are the sections in \eqref{alphaandn} attached to $h$, while $\pl_{\mc{L}^2}$ is the $(1,0)$ part of the unitary connection $\nabla_{\mc{L}^2}$ on $\mc{L}^2$.
 \end{lemma}
\begin{proof} We compute the derivative of $t\mapsto \tilde{S}_\Sigma(\hat{h}_t,A)$ using Lemma \ref{l:variation_WZW_action} where $t\in [0,1]\mapsto \hat{h}_t\in C^\infty(\Sigma,{\rm SL}(2,\C))$ is a smooth path, with $\hat{h}_t(x)$ taking values in the positive definite ${\rm SL}(2,\C)$ matrices for each $x\in \Sigma$. 
Let us write $h_t:=s^{-1}\hat{h}_ts=(a_t n_t)(a_t n_t)^*$ with $a_t,n_t$ under the form \eqref{alphaandn} for some path $\nu_t\in C^\infty(\Sigma;\mc{L}^{-2})$ and $\phi_t\in C^\infty(\Sigma)$. We get, with $w_t=a_tn_t$,
\[\begin{split}
\pl_t\tilde{S}_\Sigma (\hat{h}_t,A)& =-\frac{1}{2\pi i}\int_\Sigma {\rm Tr}\Big(\hat{h}_t^{-1}\pl_t\hat{h}_tF_{A^{1,0}+ A_{\hat{h}_t^{-1}}^{0,1}}\Big)\\
& = -\frac{1}{2\pi i}\int_\Sigma {\rm Tr}(h_t^{-1}\pl_th_t\, s^{-1}F_{A^{1,0}+ A_{\hat{h}_t^{-1}}^{0,1}}s)\\
\end{split}
\]
with $h_t^{-1}\pl_th_t =(w_t^*)^{-1}w_t^{-1}\dot{w}_tw_t^*+(w_t^*)^{-1}\dot{w}_t^*$ and $\dot{w}_t=\pl_tw_t$.
In particular we get (we remove the $t$-subscript for notational simplicity)
\[\begin{split}
{\rm Tr}\Big(h_t^{-1}\pl_th_t\, s^{-1}F_{A^{1,0}+ A_{\hat{h}_t^{-1}}^{0,1}}s\Big)=&{\rm Tr}((w^{-1}\dot{w} +\dot{w}^*{w^*}^{-1}) (w^*s^{-1})F_{A^{1,0}+ A_{\hat{h}^{-1}}^{0,1}}(s{w^*}^{-1}))\\
=& {\rm Tr}((w^{-1}\dot{w} +\dot{w}^*{w^*}^{-1})F_{B(p)}) 
\end{split}\]
where $B(p)$ is the connection form 
\[ \begin{split}
B(p)=&(s{w^*}^{-1})^{-1}(A^{1,0}+ A_{\hat{h}^{-1}}^{0,1})s{w^*}^{-1}+(s{w^*}^{-1})^{-1}d(s{w^*}^{-1})\\
=& (s{w^*}^{-1})^{-1}A^{1,0}(s{w^*}^{-1})+(s{w^*}^{-1})^{-1}\pl (s{w^*}^{-1})+w^{-1}\Big(s^{-1}A^{0,1}s+s^{-1}\bar{\pl} s+\bar{\pl}w.w^{-1}\Big)w\\
=&  w^*B^{1,0}{w^*}^{-1}+w^*\pl({w^*}^{-1}) +(w^{-1}B^{0,1}w+w^{-1}\bar{\pl}w)
\end{split}\]
where $B^{1,0}=s^{-1}A^{1,0}s+s^{-1}\pl s$ and $B^{0,1}=s^{-1}A^{0,1}s+s^{-1}\bar{\pl}s$. Using a local orthonormal trivialization $(s_1,s_2)$ of $\mc{L}^{-1}\oplus \mc{L}$ with 
$s_1$ section of $\mc{L}^{-1}$ and $s_2$ of $\mc{L}$, we can write $B^{1,0},B^{0,1}$ under the matrix form (in the splitting $\mc{L}^{-1}\oplus \mc{L}$) 
\[ B^{0,1}=\left(\begin{array}{cc}
a& b \\
0 &-a
\end{array}\right), \quad B^{1,0}=\left(\begin{array}{cc}
-\bar{a}& 0 \\
-b^* &\bar{a}
\end{array}\right)\]
for some local $(0,1)$-form $a\in C^\infty(\Sigma;\Lambda^{0,1}\Sigma)$, and $b\in C^\infty(\Sigma;\Lambda^{0,1}\Sigma)$ is the representation in $(s_1,s_2)$ of the global form $\beta \in C^\infty(\Sigma;\mc{L}^{-2}\otimes \Lambda^{0,1}\Sigma)$  ($\beta^*$ will be its dual represented by $\bar{b}$ in the frame).
We compute 
\[ \begin{split} 
w^{-1}\dot{w}=\left(\begin{array}{cc}
\frac{1}{2}\dot{\phi}& \dot{\nu}+\dot{\phi}\nu \\
0 &-\frac{1}{2}\dot{\phi}
\end{array}\right),\quad \dot{w}^*{w^*}^{-1}=\left(\begin{array}{cc}
\frac{1}{2}\dot{\phi}& 0 \\
\dot{\nu}^*+\dot{\phi}\nu^* &-\frac{1}{2}\dot{\phi}
\end{array}\right)
\end{split}\]
\[ w^{-1}B^{0,1}w+w^{-1}\bar{\pl}w=
\left(\begin{array}{cc}
a+\frac{1}{2}\bar{\pl}\phi& e^{-\phi}((\bar{\pl}+2a)(e^{\phi}\nu)+b) \\
0 &-a-\frac{1}{2}\bar{\pl}\phi
\end{array}\right)
\]
\[w^*B^{1,0}{w^*}^{-1}+w^*\pl({w^*}^{-1})=\left(\begin{array}{cc}
-\bar{a}-\frac{1}{2}\pl \phi & 0 \\
-e^{-\phi}((\pl+2\bar{a})(e^{\phi}\nu^*)+\bar{b}) &\bar{a}+\frac{1}{2}\pl\phi
\end{array}\right).
\]
This implies that there are (explicit) $2$-forms $Z_1,Z_2$ (that won't play any role) such that
\[ \begin{split} 
{\rm Tr}((w^{-1}\dot{w} +\dot{w}^*{w^*}^{-1})dB(p))=& 2\dot{\phi}\Big(\pl a+\pl\bar{\pl}\phi-\bar{\pl}\bar{a}-\nu Z_1+\nu^*Z_2
\Big)\\
& - \dot{\nu}\bar{\pl}(e^{-\phi}((\pl+2\bar{a})(e^{\phi}\nu^*)+\bar{b}))+\dot{\nu}^*\pl(e^{-\phi}((\bar{\pl}+2a)(e^{\phi}\nu)+b))
\end{split}\]
\[\begin{split} 
{\rm Tr}((w^{-1}\dot{w} +\dot{w}^*{w^*}^{-1})(B(p)\wedge B(p)))=& 2\dot{\phi}e^{-2\phi}((\pl +2\bar{a})(e^{\phi}\nu^*)+\bar{b})\wedge((\bar{\pl} +2a)(e^{\phi}\nu)+b)
\\
& -2\dot{\nu}e^{-\phi}((\pl+2\bar{a})(e^{\phi}\nu^*) +\bar{b})\wedge (a+\frac{1}{2}\bar{\pl}\phi)\\
& -2\dot{\nu}^*(\bar{a}+\frac{1}{2}\pl \phi) \wedge e^{-\phi}((\bar{\pl}+2a)(e^{\phi}\nu)+b).
\end{split}\]
For $w=an$, we can choose a path $w_t=a_tn_t$ for $t\in [0,1]$ satisfying $w_0={\rm Id}$ and $w_1=w$ that is obtained as 
$w_t=an_t$ for $t\in [1/2,1]$ with $n_{1/2}={\rm Id}$ and $n_1=n$, and $w_t=a_t$ for $t\in [0,1/2]$ 
with $a_0={\rm Id}$ and $a_{1/2}=a$. 
First, consider the $t\in [1/2,1]$-variation where $\phi_t=\phi_{1/2}$ is fixed (i.e. $\dot{\phi}=0$) and $\nu_t$ depends on $t$: 
\[\begin{split}
{\rm Tr}((w^{-1}\dot{w} +\dot{w}^*{w^*}^{-1})F_{B(p)}) = &- \dot{\nu}e^{\phi}(\bar{\pl}-2a)(e^{-2\phi}((\pl+2\bar{a})(e^{\phi}\nu^*)+\bar{b}))\\
& +\dot{\nu}^*e^{\phi}(\pl-2\bar{a})(e^{-2\phi}((\bar{\pl}+2a)(e^{\phi}\nu)+b))\\
=& - \dot{\nu}e^{\phi}(\bar{\pl}_{\mc{L}^2}(e^{-2\phi}(\pl_{\mc{L}^2}(e^{\phi}\nu^*)+\beta^*))\\
& +\dot{\nu}^*e^{\phi}(\pl_{\mc{L}^{-2}}(e^{-2\phi}(\bar{\pl}_{\mc{L}^{-2}}(e^{\phi}\nu)+\beta)).
\end{split}\]
The last line being global, i.e. does not depend on $s_1,s_2$, this expresses globally the variation. Now we compute using Stokes formula  in the last line
 \begin{align*}
& \pl_t \int_\Sigma \pl \phi \wedge \bar{\pl}\phi +e^{-2\phi}( \pl_{\mc{L}^2}  (e^{\phi}\nu^*_t)+\beta^*)\wedge  (\bar{\pl}_{\mc{L}^{-2}}(e^{\phi}\nu_t)+ \beta)+  2F_{\mc{L}}\phi\\
& =\int_\Sigma e^{-2\phi} \pl_{\mc{L}^2} (e^{\phi}\dot{\nu}^*)\wedge (\bar{\pl}_{\mc{L}^{-2}}(e^{\phi}\nu)+\beta) 
+\int_\Sigma e^{-2\phi} (\pl_{\mc{L}^2}  (e^{\phi}\nu^*)+\beta^*)\wedge \bar{\pl}_{\mc{L}^{-2}}(e^{\phi}\dot{\nu})\\
&= -\int_\Sigma {\rm Tr}((w^{-1}\dot{w} +\dot{w}^*{w^*}^{-1})F_{B(p)}).
\end{align*}
It remains to check the variation for $t\in [0,1/2]$, where $\nu_t=0$ but $\phi_t$ is not constant. In that case, using $F_{\mc{L}}=-\pl a+\bar{\pl}\bar{a}$ in the local frame, we have 
\[\begin{split}
{\rm Tr}((w^{-1}\dot{w} +\dot{w}^*{w^*}^{-1})F_{B(p)})=&  2\dot{\phi}(\pl a+\pl\bar{\pl}\phi-\bar{\pl}\bar{a}+e^{-2\phi}(\beta^*\wedge \beta))\\
 =& 2\dot{\phi}(\pl\bar{\pl}\phi +e^{-2\phi}(\beta^*\wedge \beta)- F_{\mc{L}})
\end{split}\]
and we compute using Stokes formula
\[\begin{split} 
\pl_t \int_\Sigma \pl \phi_t \wedge \bar{\pl}\phi_t +e^{-2\phi_t}b^*\wedge b+  2F_{\mc{L}}\phi_t=&\int_{\Sigma} 
\pl\dot{\phi}\wedge \bar{ \pl} \phi+\pl{\phi}\wedge \bar{\pl}\dot{\phi}-2\dot{\phi}e^{-2\phi}\beta^*\wedge \beta+2F_{\mc{L}} \dot{\phi}\\
=& -\int_{\Sigma} 2\dot{\phi}(\pl\pl\bar \phi+e^{-2\phi}\beta^*\wedge \beta-F_{\mc{L}}).
\end{split}\]
This shows the result by integrating the variation, since the claimed 
identity is obviously valid at $t=0$.
\end{proof}

\section{Probabilistic construction of the $\mathbb{H}^3$-WZW model for general rank-$2$ bundles on surfaces}\label{H3surface}





Since this Section contains a certain number of notations, we have gathered the most important notations at the end of each subsection
 in a way that they can be used as an index by the reader.
 
\subsection{Geometric setup}\label{setup_E}
In this Section, we introduce the setup that will be used in the remaining sections of the article. 

Consider a closed Riemann surface $\Sigma$ of genus $g(\Sigma)\geq 0$ with a Riemannian metric $g$ compatible with the complex structure.  
We consider $(E,\bar{\pl}_E,g_E)$ a rank-$2$ holomorphic vector bundle $E$ with trivial determinant with $\bar{\pl}_E$ its Dolbeault operator and $g_E$ its Hermitian metric. We represent $E$ as an extension 
\begin{equation}\label{short_exact_seq} 
0\longrightarrow \mc{L}^{-1} \longrightarrow E\longrightarrow \mc{L}\longrightarrow 0,
\end{equation}
with $\mc{L}^{-1}$ a holomorphic line subbundle with Dolbeault operator $\bar{\pl}_{\mc{L}^{-1}}=\bar{\pl}_E|_{\mc{L}^{-1}}$ and $\mc{L}$ its dual holomorphic bundle. Concretely, there is an orthogonal (but not holomorphic) splitting $E=\mc{L}^{-1}\oplus \mc{L}$ 
such that $g_E=g_{\mc{L}^{-1}}\oplus g_{\mc{L}}$ where the Hermitian metrics on $\mc{L}$ and $\mc{L}^{-1}$ are dual ones, and the Dolbeault operator on $E$ decomposes as  
\begin{equation}\label{dbar_Esetup} \bar{\pl}_E= \left(\begin{array}{cc} 
\bar{\pl}_{\mc{L}^{-1}} & \beta \\
0 & \bar{\pl}_{\mc{L}}
\end{array}\right) \end{equation}
for some  $\beta\in C^\infty(\Sigma;\mc{L}^{-2}\otimes \Lambda^{0,1}\Sigma)$. We call the pair 
\[ (\mc{L},\beta)\]
the \textbf{parameters of the extension} describing $E$. Since, when $E$ is a stable bundle (which is generic in the moduli space of holomorphic bundles), one has 
$H^0(\Sigma,\mc{L}^{-2})=0$, we will say that we are in the 
\begin{align*} 
&\textbf{ Generic case }  \textrm{ if } H^0(\Sigma,\mc{L}^{-2})=0,\\
& \textbf{ Non generic case }  \textrm{ if } H^0(\Sigma,\mc{L}^{-2})\not=0.
\end{align*}

The induced metric on $\mc{L}^n$ for $n\in \Z$ is denoted $g_{\mc{L}^n}$.
For a complex vector bundle $F$ over $\Sigma$, we shall denote by $\mc{H}^s(\Sigma;F)$ the sections with $\mc{H}^s$-Sobolev regularity for $s\in \R$ (this can be defined in local charts for example or as the domain of $(1+\Delta_{F})^{s/2}$ where $\Delta_{F}$ is a fixed connection Laplacian). 

\subsection{Witten Laplacian, its Green function and determinant}

We shall use the $L^2$ pairing on the bundle $\mc{L}^{n}\otimes \Lambda^1\Sigma$ and $\mc{L}^{n}$ for $n\in \Z$
\begin{align} 
&\forall s_1,s_2\in C^\infty(\Sigma;\mc{L}^{n}), \quad \cjg s_1,s_2\cjd_{2}=\int_\Sigma  \cjg s_1,s_2\cjd_{g_{\mc{L}^n}}{\rm dv}_g \label{norme2} \\
&  \forall u_1,u_2\in C^\infty(\Sigma;\mc{L}^{n}\otimes \Lambda^1\Sigma), \quad  \cjg u_1,u_2\cjd_{2}=\int_\Sigma  \cjg u_1,u_2\cjd_{g_{\mc{L}^n}\otimes g}{\rm dv}_g\label{norme2bis} 
 \end{align}
(here $g_{\mc{L}^n}\otimes g$ is a metric on $\mc{L}^{n}\otimes \Lambda^1\Sigma$).
For $n\in \Z$, we denote by $\bar{\pl}^*_{\mc{L}^n}: C^\infty(\Sigma; \mc{L}^n\otimes  \Lambda^{0,1} \Sigma)\to C^\infty(\Sigma; \mc{L}^n)$ the $L^2$-adjoint of the Dolbeault operator on $\mc{L}^{n}$. More generally, for a background measure $\theta$ on $\Sigma$, we shall denote $L^2(\theta)$ for the $L^2$ space with respect to $\theta$ and $\cjg \cdot,\cdot\cjd_{L^2(\theta)}$ for the associated pairing, keeping $g_{\mc{L}^n}$ and $g$ as pointwise products on $\mc{L}^n$ and $\Lambda^k\Sigma$ as in \eqref{norme2} and \eqref{norme2bis}. 

To construct the field $\gamma_g$ as in the case of the Riemann sphere and trivial bundle, we will need to consider the Green fucntion of a Witten type Laplacian on a certain line bundle. We first need to introduce some basic material about these operators.
Consider for $n\in \Z$, $\phi\in C^\infty(\Sigma,\C)$ complex valued, the Witten Laplacian (later we will be interested in the case $n=-2$)
\begin{equation}\label{def_Witten_Delta}
\tilde{\mc{D}}_\phi:=(e^{-\phi}\bar{\pl}_{\mc{L}^n}e^{\phi})^*(e^{-\phi}\bar{\pl}_{\mc{L}^n}e^{\phi})
\end{equation}
and the related operator ($*$ denotes the $L^2$-adjoint)
\begin{equation}\label{def_D_phi}
 \mc{D}_{\phi} :=    (e^{-\phi}\bar{\pl}_{\mc{L}^n})^*(e^{-\phi}\bar{\pl}_{\mc{L}^n})=e^{-\bar{\phi}}\tilde{\mc{D}}_\phi e^{-\phi}. 
 \end{equation}
These operators depend on the choice of metric $g$, and we will later sometimes write $\mc{D}_{g,\phi}$ and $\tilde{\mc{D}}_{g,\phi}$ 
instead of $\mc{D}_\phi$ and $\tilde{\mc{D}}_\phi$ to emphasize this dependence when this plays an important role. 
In parallel to the sphere case, the operator $\tilde{\mc{D}}_\phi$ is used as a mean to understand the covariance of the field $\nu_g$ while $\mc{D}_\phi$ will be used to define the covariance of $\gamma_g$.
The Riemann-Roch theorem applied to 
$\bar{\pl}_{\mc{L}^n}$ reads
\begin{equation}\label{Riemann_RochL^k}
\dim \ker \bar{\pl}^*_{\mc{L}^n}-\dim \ker \bar{\pl}_{\mc{L}^n}+\frac{\chi(\Sigma)}{2}+n\, {\rm deg}(\mc{L})=0.
\end{equation}
Let us make a quick comment on $\ker \bar{\pl}^*_{\mc{L}^{n}}$. By identifying $\Lambda^{1,0}\Sigma$ with the dual 
bundle $(\Lambda^{0,1}\Sigma)^*$ via the map $(v_1,v_2)\in \Lambda^{1,0}\Sigma \times \Lambda^{0,1}\Sigma\mapsto \frac{1}{2} *_g(v_1\wedge v_2)\in \C$,
and $(\Lambda^{1,1}\Sigma)^*$ with $\C$ via  $\frac{1}{2}*_g$,
we define the maps 
\begin{equation}\label{dualitys^*}
\begin{gathered}
s\in \mc{L}^{-n}\otimes \Lambda^{1,0}\Sigma\to s^*\in \mc{L}^{n}\otimes \Lambda^{0,1}\Sigma, \quad s^*(u):=\cjg u,s\cjd_{g_{\mc{L}^{-n}}\otimes g}, \, \forall u \in  \mc{L}^{-n}\otimes \Lambda^{1,0}\\
s\in \mc{L}^n \mapsto s^*\in \mc{L}^{-n}\otimes \Lambda^{1,1}\Sigma, \quad  s^*(u)=\cjg u,s\cjd_{g_{\mc{L}^{-n}}}
\end{gathered}
\end{equation}
It is then easily checked (using local trivialisations) that 
\begin{equation}\label{dbar^*-dbar} 
 \forall s \in C^\infty(\Sigma;\mc{L}^{-n}\otimes \Lambda^{1,0}\Sigma), \, \bar{\pl}_{\mc{L}^{-n}}s=-(\bar{\pl}^*_{\mc{L}^n}s^*)^*
 \end{equation}
thus $s\in H^0(\Sigma,\mc{L}^{-n}\otimes \Lambda^{1,0}\Sigma)\mapsto s^* \in \ker \bar{\pl}^*_{\mc{L}^{n}} \simeq H^1(\Sigma,\mc{L}^{n})$ is an isomorphism, which is exactly Serre duality. To summarize, 
\[ \ker \bar{\pl}_{\mc{L}^n}=H^0(\Sigma,\mc{L}^n), \qquad \ker \bar{\pl}^*_{\mc{L}^{n}} \simeq H^1(\Sigma,\mc{L}^{n}).\]

\noindent \textbf{Projections on co-kernels and kernels.} Let us consider two orthogonal projections 
\begin{equation}\label{def_orthogonal_proj_coker}
\tilde{\Pi}_\phi: L^2(\Sigma;  \mc{L}^n\otimes \Lambda^{0,1} \Sigma) \to  \ker (e^{-\phi}\bar{\pl}_{\mc{L}^n}e^{\phi})^*, \qquad  
\Pi_\phi: L^2(\Sigma;  \mc{L}^n\otimes \Lambda^{0,1} \Sigma) \to  \ker \bar{\pl}^*_{\mc{L}^n}
\end{equation}
where $\tilde{\Pi}_\phi$ is the orthogonal projection with respect to the $L^2({\rm v}_g)$ product and 
$\Pi_\phi$ is the orthogonal projection with respect to the $L^2(e^{2{\rm Re}(\phi)}{\rm v}_g)$ product. Notice first that $\Pi_\phi=e^{-\bar{\phi}}\tilde{\Pi}_\phi e^{\bar{\phi}}$. Then,
if $(e_p)_{p=1,\dots,N}$ is a basis of $\ker \bar{\pl}_{\mc{L}^n}^*$, then $(e^{\bar{\phi}}e_p)_p$ is a basis for $\ker (e^{-\phi}\bar{\pl}_{\mc{L}^n}e^{\phi})^*$ and we can write $\tilde{\Pi}_\phi$ and $\Pi_\phi$  as follows
\begin{equation}\label{tildePiphi}
\tilde{\Pi}_\phi u:= \sum_{p,q=1}^N (\mc{G}_\phi^{-1})_{pq} e^{\bar{\phi}}e_p \cjg u,e^{\bar{\phi}}e_q\cjd_2, \qquad \Pi_\phi u:= \sum_{p,q=1}^N (\mc{G}_\phi^{-1})_{pq} e_p \cjg u,e_q\cjd_{L^2(e^{2{\rm Re}(\phi)}{\rm v}_g)}\\
\end{equation}
in terms of the Gram matrix $\mc{G}_\phi$ 
\begin{equation}\label{Gphi} 
(\mc{G}_\phi)_{pq}:= \cjg e^{\bar{\phi}}e_p,e^{\bar{\phi}}e_q\cjd_{2}= \cjg e_p,e_q\cjd_{L^2(e^{2{\rm Re}(\phi)}{\rm v}_g)}.
\end{equation}
Similarly, we define 
\begin{equation}\label{def_orthogonal_proj_ker}
\tilde{P}_\phi:L^2(\Sigma;\mc{L}^n)\to \ker (\bar{\pl}_{\mc{L}^n}e^{\phi}) , \qquad  
P_0:L^2(\Sigma;\mc{L}^n)\to \ker \bar{\pl}_{\mc{L}^n}
\end{equation}   
to be  the orthogonal projections with respect to the $L^2({\rm v}_g)$ pairing. Notice that $\tilde{P}_0=P_0$ and that $\tilde{P}_\phi$ can be expressed as 
\begin{equation}\label{PPhi} 
\tilde{P}_\phi w:= \sum_{i,j=1}^d (\mc{N}_\phi^{-1})_{ij} e^{-\phi}f_i\cjg w,e^{-\phi}f_j\cjd_2
\end{equation}
if $(f_j)_j$ is a basis of $\ker \bar{\pl}_{\mc{L}^n}$ and $\mc{N}_\phi$ is the matrix 
\begin{equation}\label{NPhi}  
(\mc{N}_\phi)_{ij}:= \cjg e^{-\phi}f_i,e^{-\phi}f_j\cjd_{2}.
\end{equation}

\begin{lemma}[\textbf{Inverses for $\mc{D}_\phi$ and $\tilde{\mc{D}}_\phi$}]\label{inverseDphi}
The following holds true:\\
1) The operators 
\[\mc{D}_{\phi}:\mc{H}^2(\Sigma;\mc{L}^n)\to L^2(\Sigma;\mc{L}^n), \qquad \tilde{\mc{D}}_{\phi}:\mc{H}^2(\Sigma;\mc{L}^n)\to L^2(\Sigma;\mc{L}^n)\]
are elliptic, self-adjoint Fredholm operators, and have a unique inverse operator $R_{\phi}$ and $\tilde{R}_\phi$,  bounded as maps 
$\mc{H}^{s}(\Sigma;\mc{L}^n)\to \mc{H}^{s+2}(\Sigma; \mc{L}^n)$ for all $s\in \R$, 
with range in $(\ker \bar{\pl}_{\mc{L}^n})^\perp$, such that 
\begin{equation} \label{DphiRphi}
\begin{gathered}
\mc{D}_{\phi}R_{\phi}={\rm Id}-P_0, \qquad \tilde{\mc{D}}_{\phi}\tilde{R}_{\phi}={\rm Id}-\tilde{P}_\phi \\
\quad\, {\rm Ran}(R_{\phi})\subset (\ker \bar{\pl}_{\mc{L}^n})^\perp, \qquad {\rm Ran}(\tilde{R}_{\phi})\subset (\ker \bar{\pl}_{\mc{L}^n}e^{\phi})^\perp
\end{gathered}\end{equation}
where $P_0,\tilde{P}_\phi$  are the orthogonal projections \eqref{def_orthogonal_proj_ker}. The integral kernel of $R_\phi$ and $\tilde{R}_{\phi}$ are smooth outside the diagonal.
Moreover, the operator $\tilde{R}_\phi$ also satisfies 
\begin{equation}\label{tildeRphi}
 (e^{-\phi}\bar{\pl}_{\mc{L}^{-2}}e^{\phi})\tilde{R}_\phi(e^{-\phi}\bar{\pl}_{\mc{L}^{-2}}e^{\phi})^*={\rm Id}-\tilde{\Pi}_\phi
\end{equation}
if $\tilde{\Pi}_\phi$ is the orthogonal projection of \eqref{def_orthogonal_proj_coker}.\\
2) The operators $\tilde{T}_\phi:=e^{-\phi}\bbar{\pl}_{\mc{L}^n}e^{\phi} \tilde{R}_{\phi}$ and $T_\phi:=e^{-\phi}\bbar{\pl}_{\mc{L}^n}R_{\phi}$, bounded on 
\[\tilde{T}_\phi:\mc{H}^s(\Sigma;\mc{L}^n)\to \mc{H}^{s+1}(\Sigma, \mc{L}^n\otimes  \Lambda^{0,1}\Sigma), \quad  T_\phi:\mc{H}^s(\Sigma;\mc{L}^n)\to \mc{H}^{s+1}(\Sigma, \mc{L}^n\otimes  \Lambda^{0,1}\Sigma)\]
satisfy
\begin{equation}\label{Prop_T_phi}
(e^{-\phi}\bbar{\pl}_{\mc{L}^n}e^{\phi})^*\tilde{T}_\phi={\rm Id}-\tilde{P}_\phi, \qquad (e^{-\phi}\bbar{\pl}_{\mc{L}^n})^*T_\phi={\rm Id}-P_0
\end{equation} 
and, near the diagonal, in complex coordinates $z$ on a neighborhood $U\subset \Sigma$ of a point $z_0$, the integral kernel of $\tilde{T}_\phi$ has the form 
\[\tilde{T}_\phi(z,z')=\frac{e^{a(z)+\bar{\phi}(z)-a(z')-\bar{\phi}(z)}}{2\pi(\bar{z}'-\bar{z})}d\bar{z}+H_\phi (z,z')d\bar{z}, \quad H_\phi \in C^\infty(U^2)\]
for some smooth function $a\in C^\infty(U,\C)$ independent of $\phi$. The operator $T_\phi$, $\tilde{T}_\phi$ and $T_0:=T_{\phi=0}$ are related by
\begin{equation}\label{TphiT0}
T_\phi=({\rm Id}-\tilde{\Pi}_\phi)e^{\bar{\phi}}T_0, \qquad  \tilde{T}_\phi=({\rm Id}-\tilde{\Pi}_\phi)e^{\bar{\phi}}T_0e^{-\bar{\phi}}({\rm Id}-\tilde{P}_\phi)
\end{equation} 
and the operators $R_\phi$, $\tilde{R}_\phi$ and $T_\phi$, $\tilde{T}_\phi$ are related by 
\begin{equation}\label{relRphi}
R_\phi= T_\phi^*T_\phi, \qquad \tilde{R}_{\phi}=\tilde{T}_\phi^*\tilde{T}_\phi,  
\end{equation}
and $R_\phi$ can be written in terms of $T_0$ as follows
\begin{equation}\label{relRphi2}
R_{\phi}= T_0^*e^{2{\rm Re}(\phi)}({\rm Id}-\Pi_\phi)T_0.
\end{equation}
Finally, the following identities also hold 
\begin{equation}\label{Rphidbar*}
\begin{split} 
\tilde{R}_\phi(e^{-\phi}\bar{\pl}_{\mc{L}^{n}}e^{\phi})^*= \tilde{T}_\phi^*({\rm Id}-\tilde{\Pi}_\phi), \qquad (e^{-\phi}\bar{\pl}_{\mc{L}^{n}}e^{\phi})\tilde{R}_\phi(e^{-\phi}\bar{\pl}_{\mc{L}^{n}}e^{\phi})^*=({\rm Id}-\tilde{\Pi}_\phi).
\end{split}\end{equation} 
\end{lemma}
\begin{proof} 
We take a cover of $\Sigma$ by open sets $U_i$ and local holomorphic trivialisations $s_i$ of $\mc{L}^n$ and the Riemannian metric on $U_i$ is given in complex coordinates by $g=e^{\sigma_i}|dz|^2$. The operators
$e^{-\phi}\bar{\pl}_{\mc{L}^n}$ and its adjoint have the form in local holomorphic coordinates: for $f\in C_c^\infty(U_i)$ and $u=u(z)dz\in C_c^\infty(U_i;\Lambda^{0,1}\Sigma)$, 
\[e^{-\phi}\bar{\pl}_{\mc{L}^n}(fs_i)=e^{-\phi}(\bar{\pl}f)s_i,\quad \bar{\pl}_{\mc{L}^n}^*(e^{-\bar{\phi}}u s_i)=-e^{-\sigma_i}(\pl_z(e^{-\bar{\phi}}u) + e^{-\bar{\phi}}u \pl_z\log(|s_i|^2_{\mc{L}^n}))s_i\]
thus the associated Laplacian $\mc{D}_\phi$ has the form  
\[ \mc{D}_\phi(fs_i)=-e^{-\sigma_i}(\pl_z(e^{-2{\rm Re}(\phi)}\pl_{\bar{z}}f) + e^{-2{\rm Re}(\phi)}\pl_{\bar{z}}f \pl_z\log(|s_i|^2_{\mc{L}^n}))s_i.\]
This implies that $\mc{D}_\phi$ and $\tilde{\mc{D}}_\phi$ are elliptic with principal symbol $\sigma_{\mc{D}_\phi}(x,\xi)=\frac{1}{4}e^{2{\rm Re}(\phi(x))}|\xi|_{g_x}^2$ and $\sigma_{\tilde{\mc{D}}_\phi}(x,\xi)=\frac{1}{4}|\xi|_{g_x}^2$.
Being also self-adjoint, they are thus Fredholm of index $0$ as a bounded map $\mc{H}^2(\Sigma;\mc{L}^n)\to L^2(\Sigma;\mc{L}^n)$. 
Moreover,  we have  $\ker \mc{D}_\phi=\ker \bar{\pl}_{\mc{L}^n}$.
By standard pseudo-differential methods \cite{Shubin}, there are unique bounded self-adjoint operator $R_\phi$ and $\tilde{R}_\phi$
mapping $\mc{H}^s(\Sigma;\mc{L}^n)\to \mc{H}^{s+2}(\Sigma;\mc{L}^n)$ for all $s\in \R$, with range being orthogonal respectively to 
$\ker \bar{\pl}_{\mc{L}^n}$ and to $\ker (e^{-\phi}\bar{\pl}_{\mc{L}^n}e^{\phi})$ and 
such that 
\[ \mc{D}_\phi R_\phi={\rm Id}-P_0, \qquad \tilde{\mc{D}}_\phi \tilde{R}_\phi={\rm Id}-\tilde{P}_\phi \]
where $P_0$ and $\tilde{P}_\phi$ being respectively the orthogonal projection on $\ker \bar{\pl}_{\mc{L}^n}$ and $\ker (e^{-\phi}\bar{\pl}_{\mc{L}^n}e^{\phi})$ (note that $P_0$ and $\tilde{P}_\phi$ are finite rank and smoothing and that $\tilde{P}_0=P_0$).
The operator
\[\begin{gathered}
 \tilde{T}_\phi:=(e^{-\phi}\bbar{\pl}_{\mc{L}^n}e^{\phi}) \tilde{R}_{\phi} : \mc{H}^s(\Sigma;\mc{L}^n)\to \mc{H}^{s+1}(\Sigma;  \mc{L}^n  \otimes \Lambda^{0,1}\Sigma) \\
 T_\phi:=e^{-\phi}\bbar{\pl}_{\mc{L}^n}R_{\phi} : \mc{H}^s(\Sigma;\mc{L}^n)\to \mc{H}^{s+1}(\Sigma;  \mc{L}^n  \otimes \Lambda^{0,1}\Sigma)
 \end{gathered}\]
satisfy 
\[(e^{-\phi}\bbar{\pl}_{\mc{L}^n}e^{\phi})^*\tilde{T}_\phi={\rm Id}-\tilde{P}_\phi, \quad (e^{-\phi}\bbar{\pl}_{\mc{L}^n})^*T_\phi={\rm Id}-P_0\] 
Moreover, its range of $\tilde{T}_\phi$ is orthogonal to $\ker (e^{-\phi}\bar{\pl}_{\mc{L}^n}e^{\phi})^*$ and the range of 
$T_\phi$ is orthogonal to $\ker (e^{-\phi}\bar{\pl}_{\mc{L}^n})^*$.
 Let us describe the integral kernel of this operator near the diagonal. Let $a_i:=\log(|s_i|^2_{\mc{L}^n}))$ and rewrite 
$\bar{\pl}_{\mc{L}^n}^*(u s_i)=-e^{-\sigma_i}e^{-a_i}\pl_z(e^{a_i}u)s_i$ in the chart $U_i$. Let us define the operator on $U_i$
\[ \tilde{T}_{\phi,i} (fs_i)(w):=\frac{s_i(w)e^{-a_i(w)-\bar{\phi}(w)+\bar{\phi}(z)}}{2\pi}\int_{\C} e^{a_i(z)}f(z)(\frac{1}{\bar{z}-\bar{w}})e^{\sigma_i(z)}\dd {\rm Re}(z)\dd{\rm Im}(z)\]
This operator satisfies, for $f\in C_c^\infty(U_i)$,  
\[(e^{-\phi}\bbar{\pl}_{\mc{L}^n}e^{\phi})^*(\tilde{T}_{\phi,i} (fs_i))=fs_i.\]
If $\sum_i\chi_i=1$ is a partition of unity associated to $U_i$ and $\tilde{\chi}_i\in C_c^\infty(U_i)$ equals $1$ on support of $\chi_i$
we can construct an approximation of the inverse of $\bar{\pl}_{\mc{L}^n}^*$ by setting
\[ \mc{Q}_\phi=\sum_{i} \tilde{\chi}_i\tilde{T}_{\phi,i}\chi_i.\]
We get $(e^{-\phi}\bbar{\pl}_{\mc{L}^n}e^{\phi})^* \mc{Q}_\phi={\rm Id}+K_\phi$
where $K_\phi:C^{-\infty}(\Sigma;\mc{L}^n)\to C^\infty(\Sigma;\mc{L}^n\otimes \Lambda^{0,1}\Sigma)$ is smoothing. 
We thus have $(e^{-\phi}\bbar{\pl}_{\mc{L}^n}e^{\phi})^*(\tilde{T}_\phi-\mc{Q}_\phi)=-\tilde{P}_\phi-K_\phi$ thus for each $f\in C^{-\infty}(\Sigma;\mc{L}^n)$ we have 
$(e^{-\phi}\bbar{\pl}_{\mc{L}^n}e^{\phi})^*(T_\phi-\mc{Q}_\phi)f\in C^\infty(\Sigma;\mc{L}^n)$. Since $(e^{-\phi}\bbar{\pl}_{\mc{L}^n}e^{\phi})^*$ is elliptic on $\mc{L}^{n}\otimes \Lambda^{0,1}\Sigma$ we obtain that $(\tilde{T}_\phi-\mc{Q}_\phi)f$ is smooth with continuous dependence on $f$, and thus $\tilde{T}_\phi-\mc{Q}_\phi$ is a smoothing operator. Note that $\tilde{T}_\phi^*=\tilde{R}_\phi(e^{-\phi}\bbar{\pl}_{\mc{L}^n}e^{\phi})^*$ thus
\begin{equation}\label{dbar*barT_0} 
\begin{split}
(e^{-\phi}\bbar{\pl}_{\mc{L}^n}e^{\phi})^*(e^{-\phi}\bbar{\pl}_{\mc{L}^n}e^{\phi}) \tilde{T}_\phi^* \tilde{T}_\phi=& (e^{-\phi}\bbar{\pl}_{\mc{L}^n}e^{\phi})^*(e^{-\phi}\bbar{\pl}_{\mc{L}^n}e^{\phi})\tilde{R}_\phi(e^{-\phi}\bbar{\pl}_{\mc{L}^n}e^{\phi})^* \tilde{T}_\phi\\
 =& ({\rm Id}-\tilde{P}_\phi)(e^{-\phi}\bbar{\pl}_{\mc{L}^n}e^{\phi})^* \tilde{T}_\phi={\rm Id}-\tilde{P}_\phi.
\end{split}\end{equation}
Since $\tilde{T}_\phi^*$ maps to $\ker (e^{-\phi}\bar{\pl}_{\mc{L}^n}e^{\phi})^\perp$ (as $\tilde{R}_\phi$ does), 
this  shows that $\tilde{R}_\phi= \tilde{T}_\phi^*\tilde{T}_\phi$. 
Now can apply the same reasoning for $\mc{D}_\phi$: we obtain 
\begin{equation}\label{zidane} 
 R_\phi = (e^{-\phi}\bar{\pl}_{\mc{L}^n}R_\phi)^*(e^{-\phi}\bar{\pl}_{\mc{L}^n}R_\phi).
 \end{equation}
We also have $(e^{-\phi}\bar{\pl}_{\mc{L}^n})^*=\bar{\pl}_{\mc{L}^n}^*e^{-\bar{\phi}}$ and thus  
\[
(e^{-\phi}\bar{\pl}_{\mc{L}^n})^*e^{\bar{\phi}}T_0={\rm Id}-P_0, \quad (e^{-\phi}\bar{\pl}_{\mc{L}^n}e^{\phi})^*e^{\bar{\phi}}T_0e^{-\bar{\phi}}={\rm Id}-P_0
\]
Now we observe that 
\begin{equation}\label{computation_of_Tphi} 
(e^{-\phi}\bar{\pl}_{\mc{L}^n})^*(e^{\bar{\phi}}T_0- T_\phi)=0, \quad (e^{-\phi}\bar{\pl}_{\mc{L}^n}e^{\phi})^*(e^{\bar{\phi}}T_0e^{-\bar{\phi}}-\tilde{T}_\phi)=\tilde{P}_\phi-e^{\bar{\phi}}P_0e^{-\bar{\phi}} 
\end{equation}
Since  $T_\phi$ maps to $(\ker (e^{-\phi}\bar{\pl}_{\mc{L}^n})^*)^\perp$, by the first equation of \eqref{def_orthogonal_proj_coker} we get for $\tilde{\Pi}_\phi$ the orthogonal projection of \eqref{def_orthogonal_proj_coker} that 
\[ e^{\bar{\phi}}T_0- T_{\phi}= \tilde{\Pi}_\phi(e^{\bar{\phi}}T_0- T_\phi)=\tilde{\Pi}_\phi e^{\bar{\phi}}T_0,\]
which can be rewritten as  $T_\phi=({\rm Id}-\tilde{\Pi}_\phi)e^{\bar{\phi}}T_0$, thus proving the first identity of \eqref{TphiT0}.
Therefore, combining with \eqref{zidane}, we have (recall \eqref{tildePiphi})
\begin{equation}\label{RphiTphi}
R_\phi = (e^{\bar{\phi}}T_0)^*({\rm Id}-\tilde{\Pi}_\phi)e^{\bar{\phi}}T_0= T_0^*e^{2{\rm Re}(\phi)}({\rm Id}-\Pi_\phi)T_0,
\end{equation}
which is \eqref{relRphi2}. 
We have, using $\tilde{P}_\phi (e^{-\phi}\bar{\pl}_{\mc{L}^{n}}e^{\phi})^*=(e^{-\phi}\bar{\pl}_{\mc{L}^{n}}e^{\phi}\tilde{P}_\phi)^*=0$, that 
\[ (e^{-\phi}\bar{\pl}_{\mc{L}^{n}}e^{\phi})^*\tilde{T}_\phi (e^{-\phi}\bar{\pl}_{\mc{L}^{n}}e^{\phi})^*=({\rm Id}-\tilde{P}_\phi)(e^{-\phi}\bar{\pl}_{\mc{L}^{n}}e^{\phi})^*=(e^{-\phi}\bar{\pl}_{\mc{L}^{n}}e^{\phi})^*\]
which shows, using that $\tilde{T}_\phi$ maps to $(\ker (e^{-\phi}\bbar{\pl}_{\mc{L}^n}e^{\phi})^*)^\perp$, that 
\begin{equation}\label{barTdbar*}
\tilde{T}_\phi (e^{-\phi}\bar{\pl}_{\mc{L}^{n}}e^{\phi})^*=({\rm Id}-\tilde{\Pi}_\phi).
\end{equation} 
Now, applying $({\rm Id}-\tilde{P}_\phi)$ on the right of the second equation of \eqref{computation_of_Tphi} gives 
\[e^{\bar{\phi}}(e^{-\phi}\bar{\pl}_{\mc{L}^n})^*(e^{\bar{\phi}}T_0e^{-\bar{\phi}}-\tilde{T}_\phi)({\rm Id}-\tilde{P}_\phi)=-e^{\bar{\phi}}P_0e^{-\bar{\phi}}({\rm Id}-\tilde{P}_\phi).\]
Multiply by $P_0e^{-\bar{\phi}}$ on the left, the left hand side then gives $0$ since $P_0\bar{\pl}_{\mc{L}^{-2}}^*=0$, and thus the 
right hand side $P_0e^{-\bar{\phi}}({\rm Id}-\tilde{P}_\phi)$ also vanishes, giving  
\[ (e^{-\phi}\bar{\pl}_{\mc{L}^n}e^{\phi})^*(e^{\bar{\phi}}T_0e^{-\bar{\phi}}-\tilde{T}_\phi)({\rm Id}-\tilde{P}_\phi)=0.\]
Then we apply $\tilde{T}_\phi$ on the left and use  \eqref{barTdbar*} to get 
\[ \tilde{T}_\phi=({\rm Id}-\tilde{\Pi}_\phi)e^{\bar{\phi}}T_0e^{-\bar{\phi}}({\rm Id}-\tilde{P}_\phi)\]
showing \eqref{TphiT0}.  Taking adjoints of \eqref{barTdbar*}, 
$(1-\tilde{\Pi}_\phi)=(e^{-\phi}\bar{\pl}_{\mc{L}^{n}}e^{\phi})\tilde{T}_\phi^*$.
Using \eqref{RphiTphi} and \eqref{barTdbar*} in the second line, we get 
\[
\begin{split} 
\tilde{R}_\phi(e^{-\phi}\bar{\pl}_{\mc{L}^{n}}e^{\phi})^*= \tilde{T}_\phi^*\tilde{T}_\phi (e^{-\phi}\bar{\pl}_{\mc{L}^{n}}e^\phi)^*= \tilde{T}_\phi^*({\rm Id}-\tilde{\Pi}_\phi).
\end{split}\]
which is the first identity of \eqref{Rphidbar*}.
Applying the operator $(e^{-\phi}\bar{\pl}_{\mc{L}^{n}}e^{\phi})$ on the left gives 
\[(e^{-\phi}\bar{\pl}_{\mc{L}^{n}}e^{\phi})\tilde{R}_\phi(e^{-\phi}\bar{\pl}_{\mc{L}^{n}}e^{\phi})^*=(e^{-\phi}\bar{\pl}_{\mc{L}^{n}}e^{\phi})\tilde{T}_\phi^*({\rm Id}-\tilde{\Pi}_\phi)=({\rm Id}-\tilde{\Pi}_\phi).\]
This ends the proof of \eqref{Rphidbar*}.
\end{proof}

Let us consider the determinant of  
\[\tilde{\mc{D}}_\phi:= e^{\bar{\phi}}\mc{D}_\phi e^{\phi}= (e^{-\phi}\bar{\pl}_{\mc{L}^n}e^{\phi})^*(e^{-\phi}\bar{\pl}_{\mc{L}^n}e^{\phi})\] 
defined using its spectral zeta function, where $\phi\in C^\infty(\Sigma,\C)$ as above. We denote it by $\det(\tilde{\mc{D}}_\phi)$, with the convention that when $\tilde{\mc{D}}_\phi$ has a kernel, one needs to remove the kernel in the definition. Notice that $e^{-\phi}\bar{\pl}_{\mc{L}^n}e^{\phi}=(\bar{\pl}_{\mc{L}^n}+\bar{\pl}\phi)$ is a Dolbeault operator that is gauge equivalent to $\bar{\pl}_{\mc{L}^n}$.  
\begin{lemma}[\textbf{Polyakov anomaly for Witten Laplacian}]\label{lem:detwitten}
The determinant of $\tilde{\mc{D}}_{\phi}$ can be expressed as 
\begin{equation}\label{det_of_Dphi}
 \det(\tilde{\mc{D}}_{\phi})=e^{-\frac{1}{2\pi}\int_{\Sigma}(|d {\rm Re}(\phi)|^2_g-\frac{1}{2}K_g{\rm Re}(\phi)){\rm dv}_g-\frac{1}{\pi i}\int_\Sigma {\rm Re}(\phi)F_{\mc{L}^n}}\det(\mc{D}_{0})\frac{\det(\mc{G}_\phi)}{\det(\mc{G}_0)}\frac{\det(\mc{N}_\phi)}{\det(\mc{N}_0)}
 \end{equation}
where $\mc{G}_\phi$ and $\mc{N}_\phi$ are the finite dimensional matrices defined  by \eqref{Gphi} and \eqref{NPhi} depending on a choice of basis $(e_j)_j$  of $\ker \bar{\pl}_{\mc{L}^n}^*$ and $(f_j)_j$ of $\ker \bar{\pl}_{\mc{L}^n}$
If these bases are orthonormal, then $\det(\mc{G}_0)=\det(\mc{N}_0)=1$. Furthermore, the determinant obeys:\\
1) For $\la>0$, one has 
\begin{equation}\label{scalingdet}
\det(\la\tilde{\mc{D}}_\phi)=\det(\tilde{\mc{D}}_\phi)\la^{ \frac{\chi(\Sigma)}{6}-\dim \ker \bar{\pl}_{\mc{L}^n}+\frac{1}{2}{\rm deg}(\mc{L}^n)}.
\end{equation}
2) If $g'=e^{\omega}g$ is a metric conformal to $g$  with $\omega\in C^\infty(\Sigma)$ and $\mc{N}_{0,g'}$ is the matrix \eqref{NPhi} 
for $\phi=0$ with respect to the background metric $g'$, one has 
\begin{equation}\label{anomalydetD0} 
\frac{\det(\mc{D}_{0,g'})}{\det(\mc{N}_{0,g'})}=\frac{\det(\mc{D}_{0,g})}{\det(\mc{N}_{0,g})}e^{ -\frac{1}{48\pi}\int_\Sigma(|d\omega|_g^2+2K_g\omega){\rm dv}_g+\frac{1}{4\pi i}\int_\Sigma \omega F_{\mc{L}^n}}.
\end{equation}
\end{lemma}
\begin{proof}
In this proof only, we will use the shortcut notation  $\bar{\pl}_{\phi}:=e^{-\phi}\bar{\pl}_{\mc{L}^n}e^{\phi}$.
We use the asymptotics of the local trace of the integral kernel $e^{-t\tilde{\mc{D}}_{\phi}}$ on the diagonal as $t\to 0$: 
since $\bar{\pl}_\phi$ is gauge equivalent to $\bar{\pl}_{\mc{L}^n}$ we have by \cite[Formula (1.5.4)]{Bost_Bourbaki}
\begin{equation}\label{heatkernel} 
{\rm Tr}(e^{-t\tilde{\mc{D}}_{\phi}}(x,x)){\rm dv}_g(x)= \frac{2}{\pi t}{\rm dv}_g+\frac{1}{24\pi}K_g{\rm dv}_g -\frac{1}{4\pi i}F_{\mc{L}^n}-\frac{1}{4\pi}\Delta_g{\rm Re}(\phi) {\rm dv}_g+\mc{O}(t)
\end{equation}
where $K_g$ is the scalar curvature of $g$ and $F_{\mc{L}^n}$ is the curvature $2$-form of $\mc{L}^n$.  Let $\tilde{P}_\phi$ be the orthogonal projection on $\ker \bar{\pl}_{\phi}=e^{-\phi}\ker \bar{\pl}_{\mc{L}^n}$ and $d:=\dim \ker \bar{\pl}_{\phi}$ (recall \eqref{def_orthogonal_proj_ker}
if $(f_j)_j$ is a basis of $\ker \bar{\pl}_{\mc{L}^n}$).
We differentiate $u\mapsto \log(\det \tilde{\mc{D}}_{u\phi})$. In that aim we differentiate 
\[ \begin{split}
\pl_u \Big(\frac{1}{\Gamma(s)}\int_0^\infty {\rm Tr}(e^{-t\tilde{\mc{D}}_{u\phi}}-\tilde{P}_{u\phi})t^{s}\frac{dt}{t}\Big)=&-\frac{2}{\Gamma(s)}\int_0^\infty ({\rm Tr}({\rm Re}(\phi) \tilde{\mc{D}}_{u\phi} e^{-t\tilde{\mc{D}}_{u\phi}}))t^s dt\\
& +\frac{2}{\Gamma(s)}\int_0^\infty  ({\rm Tr}(\bar{\pl}^*_{u\phi}{\rm Re}(\phi) \bar{\pl}_{u\phi}e^{-t\tilde{\mc{D}}_{u\phi}}))t^{s} dt
\end{split}\]
where we used the cyclicity of the trace.
The first term gives, using integration by parts (for $s\gg 1$)  
\[-\frac{2}{\Gamma(s)}\int_0^\infty {\rm Tr}({\rm Re}(\phi) \tilde{\mc{D}}_{u\phi}e^{-t\tilde{\mc{D}}_{u\phi}})t^sdt=-\frac{2s}{\Gamma(s)}\int_0^\infty {\rm Tr}({\rm Re}(\phi) (e^{-t\tilde{\mc{D}}_{u\phi}}-\tilde{P}_{u\phi}))t^{s-1} dt.\]
The function $s/\Gamma(s)$ has a second order $0$ at $s=0$, thus by \eqref{heatkernel} and \eqref{PPhi} the only term surviving in this expression after differentiating at $s=0$ is 
\[ -2\int_\Sigma {\rm Re}(\phi)( \frac{K_g}{24\pi} -{\rm Tr}_{\mc{L}^n}(\tilde{P}_{u\phi})){\rm dv}_g +\frac{1}{2\pi i}\int {\rm Re}(\phi)F_{\mc{L}^n}+\frac{1}{2\pi}\int_{\Sigma}u{\rm Re}(\phi)\Delta_g{\rm Re}(\phi) {\rm dv}_g.\]
For the second term, we can use the cyclicity of the trace and the fact that 
$e^{-t\tilde{\mc{D}}_{u\phi}}\bar{\pl}_{u\phi}^*=\bar{\pl}_{u\phi}^*
e^{-t\bar{\pl}_{u\phi}\bar{\pl}^*_{u\phi}}$ on $(0,1)$-forms: this term becomes 
\begin{equation}\label{second_term}
\begin{split}
\frac{1}{\Gamma(s)}\int_0^\infty 2{\rm Tr}({\rm Re}(\phi)\bar{\pl}_{u\phi} \bar{\pl}^*_{u\phi}(e^{-t\bar{\pl}_{u\phi}\bar{\pl}^*_{u\phi}}-\tilde{\Pi}_{u\phi}))t^{s} dt&=-\frac{2}{\Gamma(s)}\int_0^\infty \pl_t\Big({\rm Tr}_{\Lambda^{0,1}}({\rm Re}(\phi) (e^{-t\bar{\pl}_{u\phi}\bar{\pl}^*_{u\phi}}-\tilde{\Pi}_{u\phi}))\Big)t^{s} dt\\
& =\frac{2s}{\Gamma(s)}\int_0^\infty{\rm Tr}_{\Lambda^{0,1}}({\rm Re}(\phi)(e^{-t\bar{\pl}_{u\phi}\bar{\pl}^*_{u\phi}}-\tilde{\Pi}_{u\phi}))t^{s-1} dt\
\end{split}
\end{equation}
where $\tilde{\Pi}_{u\phi}$ is the orthogonal projector defined in \eqref{def_orthogonal_proj_coker}.
From \cite[Formula (1.5.4)]{Bost_Bourbaki}, we have 
\[
\begin{split}
{\rm Tr}(e^{-t\bar{\pl}_{u\phi}\bar{\pl}^*_{u\phi}}(x,x)){\rm dv}_g(x)= &\frac{2}{\pi t}{\rm dv}_g+\frac{1}{24\pi}K_g{\rm dv}_g -\frac{1}{4\pi i}F_{\mc{L}^{-n}\otimes \Lambda^{1,0}}+\frac{u}{4\pi}\Delta_g{\rm Re}(\phi) {\rm dv}_g+\mc{O}(t)\\
=& \frac{2}{\pi t}{\rm dv}_g-\frac{1}{12\pi}K_g{\rm dv}_g +\frac{1}{4\pi i}F_{\mc{L}^n}+\frac{u}{4\pi}\Delta_g{\rm Re}(\phi)+\mc{O}(t).
\end{split}
\]
The only term surviving in \eqref{second_term} after differentiating at $s=0$ is thus
\[ 2\int_\Sigma {\rm Re}(\phi)(- \frac{K_g}{12\pi}+\frac{u}{4\pi}\Delta_g{\rm Re}(\phi) -{\rm Tr}_{\mc{L}^n}\tilde{\Pi}_{u\phi}){\rm dv}_g +\frac{1}{2\pi i}\int {\rm Re}(\phi)F_{\mc{L}^n}.\]
Gathering everything, we obtain
\[ \pl_u \log \det(\tilde{\mc{D}}_{u\phi})=\int_\Sigma {\rm Re}(\phi)( \frac{K_g}{4\pi}-\frac{u}{\pi}\Delta_g{\rm Re}(\phi)+2{\rm Tr}_{\mc{L}^n}\tilde{\Pi}_{u\phi}-2{\rm Tr}_{\mc{L}^n}\tilde{P}_{u\phi}){\rm dv}_g -\frac{1}{\pi i}\int {\rm Re}(\phi)F_{\mc{L}^n}.\] 
Integrating in $u\in [0,1]$ gives the result since (notice that $\mc{G}_{\phi}=\mc{G}_{{\rm Re}(\phi)}$ and $\mc{N}_\phi=\mc{N}_{{\rm Re}(\phi)}$) 
\[2\int_\Sigma {\rm Re}(\phi) {\rm Tr}_{\mc{L}^n}\tilde{\Pi}_{u\phi} {\rm dv}_g=\frac{\pl_u \det(\mc{G}_{u\phi})}{ \det(\mc{G}_{u\phi})}, \quad -2\int_\Sigma {\rm Re}(\phi) {\rm Tr}_{\mc{L}^n\otimes \Lambda^{0,1}}\tilde{P}_{u\phi} {\rm dv}_g=\frac{\pl_u \det(\mc{N}_{u\phi})}{ \det(\mc{N}_{u\phi})}.\]
Finally, scaling $\tilde{\mc{D}}_\phi$ by a factor of $\la>0$ gives $\det(\la\tilde{\mc{D}}_\phi)=\det(\tilde{\mc{D}}_\phi)\la^{\zeta_{\tilde{\mc{D}}_\phi}(0)}$
where 
\[\begin{split}
\zeta_{\tilde{\mc{D}}_\phi}(0)=&{\rm Res}_{s=0}\Big(\int_0^\infty {\rm Tr}(e^{-t\tilde{\mc{D}}_{\phi}}-\tilde{P}_{\phi})t^{s}\frac{dt}{t}\Big)={\rm Res}_{s=0}\Big(\int_0^1 {\rm Tr}(e^{-t\tilde{\mc{D}}_{\phi}}-\tilde{P}_{\phi})t^{s}\frac{dt}{t}\Big)\\
=& \int_\Sigma  (\frac{K_g}{24\pi} -{\rm Tr}_{\mc{L}^n}(\tilde{P}_{\phi})){\rm dv}_g -\frac{1}{4\pi i}\int F_{\mc{L}^n}\\
=& \frac{\chi(\Sigma)}{6}-\dim \ker \bar{\pl}_{\mc{L}^n}+\frac{1}{2}{\rm deg}(\mc{L}^n).
\end{split}\]
We have used Gauss-Bonnet in the last equality. The proof of \eqref{anomalydetD0} goes along the same line by differentiating with respect to $u\in [0,1]$ the heat trace $e^{-t\mc{D}_{0,g_u}}$ with $g_u:=e^{u\omega}g$ using \eqref{heatkernel}, and observing that both $\ker \bar{\pl}_{\mc{L}^n}$ and $\ker \bar{\pl}_{\mc{L}^n}^*$ do not depend on $u$. Moreover the matrix $\mc{G}_{0,g_u}$ (defined using $g_u$ as background metric) is independent of $u$ since the scalar product on $1$-forms is conformally invariant. We do not write the full details as it is standard. 
\end{proof}
Notice that the formula \eqref{det_of_Dphi} applied with $\phi(x)=\la>0$ constant allows us to recover the Riemann Roch formula \eqref{Riemann_RochL^k} as the left hand side is $\la$-invariant.\\

\textbf{Summary:} For readability, we gather here the notation for the important objects introduced just above and that will be often used below. 
\begin{itemize}
\item $\mc{D}_\phi=(e^{-\phi}\bar{\pl}_{\mc{L}^n})^*(e^{-\phi}\bar{\pl}_{\mc{L}^n})$ and $\tilde{\mc{D}}_\phi=(e^{-\phi}\bar{\pl}_{\mc{L}^n}e^{\phi})^*(e^{-\phi}\bar{\pl}_{\mc{L}^n}e^{\phi})$, with adjoint being with respect to the $L^2({\rm v}_g):=L^2(\Sigma; \mc{L}^n\otimes \Lambda^k\Sigma,{\rm v}_g)$ pairing, which uses a Hermitian metric $g_{\mc{L}^n}$  on $\mc{L}^n$ and the Riemannian metric $g$ on differential forms $\Lambda^k\Sigma$ for $k=0,1$ to measure pointwise scalar products before integration. 
\item $P_0$ and $\tilde{P}_\phi$ are respectively orthogonal projections on $\ker \mc{D}_\phi=\ker \mc{D}_0=\ker \bar{\pl}_{\mc{L}^n}$ and  $\ker \tilde{\mc{D}}_\phi$ with respect to the $L^2({\rm v}_g)$ product; defined in \eqref{def_orthogonal_proj_ker}. 
\item The operators $\tilde{R}_\phi$ and $R_\phi$, defined in Lemma \ref{inverseDphi}, are inverses of $\mc{D}_\phi$ and $\mc{D}_\phi$, i.e. satisfying 
$\tilde{\mc{D}}_\phi \tilde{R}_\phi={\rm Id}-\tilde{P}_\phi$ and $\mc{D}_\phi R_\phi={\rm Id}-P_0$.
\item  The operators $\tilde{T}_\phi$ and $T_\phi$, defined in Lemma \ref{inverseDphi}, are right inverses of $(e^{-\phi}\bar{\pl}_{\mc{L}^n}e^{\phi})^*$ and $(e^{-\phi}\bar{\pl}_{\mc{L}^n})^*$, i.e. they satisfy $(e^{-\phi}\bbar{\pl}_{\mc{L}^n}e^{\phi})^*\tilde{T}_\phi={\rm Id}-\tilde{P}_\phi$ and   $(e^{-\phi}\bbar{\pl}_{\mc{L}^n})^*T_\phi={\rm Id}-P_0$ 
\item $\tilde{\Pi}_\phi$ and $\Pi_\phi$ are respectively orthogonal projections on $\ker (e^{-\phi}\bar{\pl}_{\mc{L}^n}e^{\phi})^*$ with respect 
to the $L^2({\rm v}_g)$ product and on $\ker (e^{-\phi}\bar{\pl}_{\mc{L}^n})^*$ with respect to the $L^2(e^{2{\rm Re}(\phi)}{\rm v}_g)$ product; defined in \eqref{def_orthogonal_proj_coker}.
\item The matrices $\mc{G}_\phi$ and $\mc{N}_\phi$, defined in \eqref{Gphi} and \eqref{NPhi},  are respectively 
Gram matrices for a basis $(e_p)_p$ of $\ker \bar{\pl}_{\mc{L}^n}^*$ and $(f_p)_p$ of $\ker \bar{\pl}_{\mc{L}^n}$ for the $L^2({\rm v}_g)$ product.
 \end{itemize}

 \subsection{Witten pairs}
 
 As for the sphere, we define  the notion of Witten pairs for a holomorphic vector bundle $(E,\bar{\pl}_E)$ equipped with a Hermitian metric $g_E$ and represented as an extension with parameter $(\mc{L},\beta)$ as in the setting of Section \ref{setup_E}.
As above $g_{\mc{L}^n}$ denotes the induced Hermitian metric on $\mc{L}^n$ for $n\in \Z$ and we shall now use the case $n=-2$.

A measure $\theta$ on $\Sigma$ such that $\theta(U)>0$ for each ball $U\subset \Sigma$ is called \emph{admissible}. 
We use the notation for $u,u'\in L^2(\Sigma;\mc{L}^{-2}\otimes \Lambda^j\Sigma)$ for $j=0,1,2$
\begin{equation}\label{defltheta}
 \|u\|^2_{L^2(\theta)}=\int_\Sigma |u(x)|^2_{g_{\mc{L}^{-2}}\otimes g}\dd \theta(x), \quad \cjg u,u'\cjd_{L^2(\theta)}=
\int_\Sigma \cjg u(x),u'(x)\cjd_{g_{\mc{L}^{-2}}\otimes g}\dd\theta(x).
\end{equation}
Let 
$\Pi^{\theta}: L^2(\Sigma;\mc{L}^{-2}\otimes\Lambda^{0,1}\Sigma)\to \ker (\bar{\pl}_{\mc{L}^{-2}}^*)$ be the operator (recall \eqref{defltheta})
\begin{equation}\label{def_proj} 
\Pi^\theta u:= \sum_{p,q=1}^N (({\mc{G}^\theta})^{-1})_{pq} e_p \cjg u,e_q\cjd_{L^2(\theta)}
\end{equation}
where $(e_p)_p$ is any  basis of $\ker \bar{\pl}_{\mc{L}^{-2}}^*$ and $\mc{G}^\theta$ is the Gram matrix 
\begin{equation}\label{def_GramG} 
(\mc{G}^{\theta})_{pq}= \cjg e_p,e_q\cjd_{L^2(\theta)}.
\end{equation}
We emphasize that $\Pi^\theta$ also depends on the choice of metric $g$ used for the scalar product on $1$-forms in \eqref{defltheta}, we shall thus sometime write $\Pi_g^{\theta}$ when the dependence on $g$ is important, in particular when comparing two conformally related metrics.
Notice that for $\la>0$, we have scale invariance 
\[ \Pi^{\la \theta}=\Pi^\theta, \quad \mc{G}^{\la\theta}=\la \mc{G}^{\theta}.\]
\textbf{Remark:} Notice that if $\phi\in C^\infty(\Sigma,\R)$ and if we choose $\theta(\dd x)=e^{2\phi(z)}{\rm v}_g(\dd x)$, the operator  $\Pi_\phi$ of \eqref{tildePiphi} is equal to $\Pi_\phi=\Pi^\theta$. Later, we shall typically take $\phi(x)=X_g(x)$ with $X_g$ the standard GFF, and $\theta(\dd x)=M^g_{2b}(X_g,\dd x)$.

\begin{lemma}
For $\theta$ an admissible measure on $\Sigma$ and $g$ a fixed Riemannian metric on $\Sigma$, the sesquilinear map 
\begin{equation}\label{defpositive} 
(u,u')\mapsto    \cjg u,u'\cjd_{L^2(\theta)} 
\end{equation}
restricts to a positive definite quadratic form on $\ker \bar{\pl}_{\mc{L}^{-2}}^*$,  $\mc{G}^{\theta}$ is its Gram matrix in the chosen basis $(e_p)_p$ of $\ker \bar{\pl}_{\mc{L}^{-2}}^*$, and is thus positive definite, and $\Pi^\theta$ restricted to $C^\infty(\Sigma;\mc{L}^{-2}\otimes\Lambda^{0,1}\Sigma)$ is a projection on $\ker \bar{\pl}_{\mc{L}^{-2}}^*$ which satisfies for each $u,u'\in C^\infty(\Sigma;\mc{L}^{-2}\otimes \Lambda^{0,1}\Sigma)$
\[ \cjg \Pi^\theta u,u'\cjd_{L^2(\theta)}=\cjg u,\Pi^\theta u'\cjd_{L^2(\theta)}.\]
\end{lemma}
\begin{proof} Let  $u\in \ker \bar{\pl}_{\mc{L}^{-2}}^*$ not identically $0$. If $\|u\|_{L^2(\Sigma,\theta)}^2=0$, 
then  $|u(x)|_{ g_{\mc{L}^{-2}}\otimes g }=0$ for $x$ in some small 
ball $B$. But this is not possible by the unique continuation property for elliptic operators. This shows that  \eqref{defpositive} is a positive definite quadratic form on the finite dimensional space $\ker \bar{\pl}_{\mc{L}^{-2}}^*$. The fact that $\Pi^\theta$ is a projector is direct. Finally, 
\[  \cjg \Pi^\theta u,u'\cjd_{L^2(\theta)}=\sum_{p,q=1} (({\mc{G}^\theta})^{-1})_{pq} \cjg e_p,u'\cjd_{L^2(\theta)}\cjg u,e_q\cjd_{L^2(\theta)}=
\cjg u,\Pi^\theta u'\cjd_{L^2(\theta)}.\qedhere\] 
\end{proof}
In what follows, we will consider the GMC $M^{g}_{2b}(X_g ,{\rm v}_g,\dd x)$ (for some $b\in(0,1)$) with respect to the GFF $X_g$ on $\Sigma$ and we will shortcut our notations by dropping the dependence on the measure ${\rm v}_g$, i.e. we will just write $M^{g}_{2b}(X_g ,\dd x)$.
\begin{definition}\label{witten_pair_Sigma}
Let $\theta$ be an admissible measure on $\Sigma$ equipped with a Riemannian metric $g$ and $k>0$  some real valued parameter. A  rescaled Witten field $\gamma_g$ on $\Sigma$ associated to $(g,\theta)$ is a centered $\mc{H}^{-s}(\Sigma;\mc{L}^{-2})$-valued Gaussian process for $s>0$, with covariance (recall the definition \eqref{norme2})
\begin{align*}
& \E[ \cjg \gamma_g,f\cjd_2 \overline{\cjg \gamma_g,f'\cjd_2}] =     \frac{\pi}{k}  \cjg (1-\Pi^\theta)T_0f',
(1-\Pi^\theta)T_0 f\cjd_{L^2(\theta)}\\  
& \E[\cjg \gamma_g,f\cjd_2 \cjg \gamma_g,f'\cjd_2]=0
\end{align*}
for any pair of smooth sections $f,f'\in C^\infty(\Sigma;\mc{L}^{-2})$
where $\Pi^{\theta}: L^2(\Sigma;\mc{L}^{-2}\otimes \Lambda^{0,1}\Sigma)\to \ker (\bar{\pl}_{\mc{L}^{-2}}^*)$ is the projection defined in \eqref{def_proj} using the metric $g$, the $\cjg \cdot,\cdot\cjd_2$ pairing is with respect to the Riemannian volume measure ${\rm v}_g$ and the pointwise scalar product on $\mc{L}^{-2}\otimes \Lambda^{0,1}$ in the right hand side is with respect to $g_{\mc{L}^{-2}}\otimes g$.\\
Let $c\in \R$, $k>2$ and set
\[ b:=(k-2)^{-1/2}.\]
A random distribution  $(X_g,\gamma_g)$ on $\mc{H}^{-s}(\Sigma;\R)\times \mc{H}^{-s}(\Sigma;\mc{L}^{-2})$  (with $s>0$) such that $X_g$ is the GFF with background metric $g$ and, conditionally on $X_g$, $\gamma_g$ is a rescaled Witten field with respect to $(g,M^{g}_{2b}(c+X_g ,\dd x))=(g,e^{2bc}M^{g}_{2b}(X_g ,\dd x))$, is called a GFF-Witten pair with parameters $(g,M^{g}_{2b}(c+X_g ,\dd x))$. 
\end{definition}
When $\theta_{bc}=e^{2bc}\theta(X_g,\dd x)$ with $\theta:=M^{g}_{2b}(X_g ,\dd x)$, note that we have
\begin{equation}\label{PibXg_GbX_g} 
\Pi^{\theta_{bc}}=\Pi^{\theta}, \quad e^{-2bc}\mc{G}^{\theta_{bc}}=\mc{G}^{\theta}.
\end{equation}

If $\psi: \Sigma\to \Sigma'$ is a diffeomorphism and $g$ is Riemannian metric on $\Sigma$, and if $E'$ is a rank-$2$ holomorphic vector bundle with trivial determinant on $\Sigma'$ represented as an extension with parameters $(\mc{L}',\beta')$, by using the pullback by $\psi$ one obtains a rank-$2$ holomorphic vector bundle $E=\psi^*E'$ with trivial determinant on $\Sigma$ represented as an extension with parameters 
$(\mc{L}:=\psi^*\mc{L}',\beta:=\psi^*\beta')$. Since $\psi^{-1}$ transports the $\bar{\pl}_{E'}$ and $\bar{\pl}_{{\mc{L}'}^{-2}}$ to $\bar{\pl}_E$ 
and $\bar{\pl}_{\mc{L}^{-2}}$,  as well as all the corresponding operators appearing in Lemma \ref{inverseDphi}, we obtain:
\begin{lemma}[\textbf{Diffeomorphism invariance of $(X_g,\gamma_g)$}]\label{l:diffeos}
If $\psi: \Sigma\to \Sigma'$ is a diffeomorphism, the following identity holds in law for the Witten pairs
\begin{equation}\label{equalityinlaw}
(\psi^*X_{\psi_*g},\psi^*\gamma_{\psi_*g}) = (X_g,\gamma_g),
\end{equation}
in the sense that $\psi^*X_{\psi_*g} =X_{g}$ in law and the conditional law of $\psi^*\gamma_{\psi_*g}$ with respect to $X_{\psi_*g}$ 
is equal to the conditional law of $\gamma_g$ with respect to $X_g$.
\end{lemma}

The following also holds true:
\begin{lemma}[\textbf{Conformal invariance of $\gamma_g$}]\label{lemmainvar_Sigma}
Let  $g'=e^\omega g$ be a metric  conformal to  $g$ with $\omega\in C^\infty(\Sigma)$ and $\theta$ be an admissible measure on $\Sigma$. Let  $\gamma_g$ be a  
rescaled Witten field  on $\Sigma$ associated to the pair $(g, \theta)$ and $\gamma_{g'}$ a 
rescaled Witten field on $\Sigma$ associated to the pair $(g', e^{\omega}\theta)$.
Then the process $\gamma_{g'}$ has the same law as $\gamma_g-P'_0\gamma_g$ where $P_0'$ is the orthogonal projector on $\ker \bar{\pl}_{\mc{L}^{-2}}$ with respect to the $L^2({\rm v}_{g'})$ scalar product.
\end{lemma}
\begin{proof} Let us write $T_0',P_0',\Pi_0'$ for the operator $T_0,P_0,\Pi_0$ of Lemma \ref{inverseDphi} with respect to the background metric $g'$ instead of $g$ and write $\theta':=e^{\omega}\theta$. The operator $\bar{\pl}$ is the same for $g$ and $g'$, but the adjoint with respect to $g'$ is $e^{-\omega}\bar{\pl}^*$ where $\bar{\pl}^*$ is the adjoint with respect to $g$. We have  
\[ e^{-\omega}\bar{\pl}^*(T_0'-T_0e^{\omega})=e^{-\omega}P_0e^{\omega}-P_0'\]
and applying $({\rm Id}-P_0')$ on the right this gives 
\[ e^{-\omega}\bar{\pl}^*(T_0'-T_0e^{\omega})({\rm Id}-P_0')=e^{-\omega}P_0e^{\omega}({\rm Id}-P_0')=0\]
where the vanishing of the right hand side comes from the fact that $P_0 \bar{\pl}^*=0$. We apply $T'_0$ on the left and use 
\eqref{barTdbar*} with $\phi=0$ and the background metric to be $g'$, to obtain 
$({\rm Id}-\Pi'_0)(T_0'-T_0e^{\omega})({\rm Id}-P_0')=0$. Using that $({\rm Id}-\Pi'_0)T_0'({\rm Id}-P_0')=T_0'$, we finally get
\[T_0'=({\rm Id}-\Pi'_0)T_0e^{\omega}({\rm Id}-P_0').\]
We now use $\Pi^\theta=\Pi^{\theta}_g$ using the background metric $g$ and notice that $\Pi_g^\theta=\Pi_{g'}^{\theta'}$.
Since $\Pi^{\theta}$ is the identity on ${\rm Ran}(\Pi'_0)$, we have $\Pi^{\theta}\Pi'_0=\Pi'_0$ and thus $({\rm Id}-\Pi^\theta)({\rm Id}-\Pi'_0)=({\rm Id}-\Pi^\theta)$, which gives 
\[ ({\rm Id}-\Pi^\theta)T_0'=({\rm Id}-\Pi^\theta)T_0e^{\omega}({\rm Id}-P_0').\]
Now, notice that $\cjg \gamma_{g}, ({\rm Id}-P_0')f\cjd_{L^2({\rm v}_{g'})}=\cjg \gamma_{g}, e^{\omega}({\rm Id}-P_0')f\cjd_{L^2({\rm v}_g)}$ thus 
\begin{align*}
&\E[ \cjg \gamma_g,({\rm Id}-P_0')f\cjd_{L^2({\rm v}_{g'})}\bbar{\cjg \gamma_g,({\rm Id}-P_0')f'\cjd}_{L^2({\rm v}_{g'})}]\\
&= 
 \frac{\pi}{k}  \int_\Sigma \cjg ({\rm Id}-\Pi^\theta)T_0(e^{\omega}({\rm Id}-P_0')f'),
({\rm Id}-\Pi^\theta)T_0 (e^{\omega}({\rm Id}-P_0')f)\cjd_{g_{\mc{L}^{-2}}\otimes g} \theta(\dd x)\\
& = \frac{\pi}{k}  \int_\Sigma \cjg ({\rm Id}-\Pi_{g'}^{\theta'})T_0' f, ({\rm Id}-\Pi_{g'}^{\theta'})T_0' f'\cjd_{g_{\mc{L}^{-2}}\otimes g'}  \theta'(\dd x)\\
& =\E[ \cjg \gamma_{g'},f\cjd_{L^2({\rm v}_{g'})}\bbar{\cjg \gamma_{g'},f'\cjd}_{L^2({\rm v}_{g'})}]
\end{align*}
which shows that $({\rm Id}-P_0')(\gamma_g)$ has the same law as $\gamma_{g'}$.
\end{proof}

Concerning the construction of a GFF-Witten pair, the proof of Proposition \ref{existpair} can be adapted to  give in the general setting
\begin{proposition}\label{exist_Witten_pairs_Sigma}
Let $b\in (0,1)$ and $c\in \R$. Then, there exists a GFF-Witten pair with parameters  $(g,e^{2bc}M^{g}_{2b}(X_g ,\dd x))$.
\end{proposition}
\begin{proof} 
Consider $(\alpha_n)_{n\in \N}$ a family of i.i.d. real centered Gaussian random variables with variance $1$ representing the GFF  $X_g$  as 
in \eqref{hsxg}.
If $(f_n)_{n\in \N}$ is an orthonormal basis of eigenfunctions for $\mc{D}_0=\bar{\pl}_{\mc{L}^{-2}}^*\bar{\pl}_{\mc{L}^{-2}}$, conditionally on $X_g$, we define the law of the sequence $(\beta_n)_{n\in \N}$ to be that of complex Gaussians with covariances with $\theta:=M^{g}_{2b}(X_g ,\dd x)$)
\begin{equation}\label{betacov_Sigma}
\E[\beta_n\bar{\beta}_m|X_g]= \frac{\pi}{k}\cjg (1-\Pi^{\theta})T_0 f_m,(1-\Pi^{\theta})T_0 f_n\cjd_{L^2(\theta)} , \quad  \E[\beta_n \beta_m]=0.
\end{equation}

For $\theta=M_{2b}^g(X_g,\dd x)$, the matrix $\mc{G}^{\theta}$ is defined almost surely and  \eqref{defpositive}
is a positive definite quadratic form on $\ker \bar{\pl}_{\mc{L}^{-2}}^*$ almost surely since, almost surely, $M_{2b}^g(X_g,B)>0$  for every  ball $B$. Its Gram matrix for the basis $e_1,\dots,e_N$ is precisely $\mc{G}^\theta$.
If $f_0,\dots,f_{d-1}$ are the basis elements in $\ker \bar{\pl}_{\mc{L}^{-2}}$, we set
\[\gamma_g^0:=\sqrt{\frac{\pi}{k}}\sum_{  n\geq d} \beta_n f_n, \quad \gamma_g=e^{bc}\gamma_g^0.\]
 The Sobolev norm of $\gamma_g$ (using $\mc{D}_0$ for the Sobolev norms) is then 
\[\|\gamma_g\|^2_{\mc{H}^{-s}(\Sigma;\mc{L}^{-2})}=e^{2bc}\frac{\pi}{k}\sum_{n\geq d}\kappa_n^{-s}|\beta_n|^2.\]
if $\mc{D}_0f_n=\kappa_nf_n$. Therefore using that 
\[ \E\Big[ \sum_{p,q=1}^N (\mc{G}^{\theta})^{-1}_{pq}\cjg e_p,T_0f_n\cjd_{L^2(\theta)}
\cjg T_0f_n,e_q\cjd_{L^2(\theta)}\Big]\geq 0\]
we obtain
\begin{align*}
\E[\|\gamma_g\|^2_{\mc{H}^{-s}(\Sigma;\mc{L}^{-2})}]
=&e^{2bc}\frac{\pi}{k}\E\Big[\sum_{n\geq d}\kappa_n^{-s}\E[|\beta_n|^2|X_g]\Big]
\leq e^{2bc}\frac{\pi}{k} \E\Big[\sum_{n\geq d}\kappa_n^{-s}\|T_0 f_n\|^2_{L^2(\theta)}\Big]
\\
\leq & Ce^{2bc} \sum_{n\geq d}\kappa_n^{-s} \|T_0 f_n\|_{L^2({\rm dv}_g)}^2=Ce^{2bc} \sum_{n\geq d}\kappa_n^{-s-1 }
\end{align*}
where we used that $\E[M^{g}_{2b}(X_g,\dd x)]\leq C {\rm dv}_g$ and that, when $\phi=0$, $\int_\Sigma |T_0 f_n|^2 {\rm dv}_g=\cjg 
\mc{D}_{0}^{-1}f_n,f_n\cjd_{2}=\kappa_n^{-1}$. Therefore the last series converges for $s>0$ by the Weyl law applied to $\mc{D}_0$.
The process $\gamma_g$ then converges in Sobolev spaces $\mc{H}^{-s}(\Sigma;\mc{L}^{-2})$ for $s>0$.
\end{proof}

We will need to regularize the pair. If $(X_g,\gamma_g)$ is a GFF-Witten pair with parameters $(g,e^{2bc}M^{g}_{2b}(X_g ,\dd x))$, we 
let $\rho$ be a smooth radial function with compact support on $\R$, nonnegative with $\int \rho=1$. 
We set $\rho_{\epsilon,x}(y):=\epsilon^{-2}\rho(d_g(x,y)/\epsilon)$ for $y\in \Sigma$ and
 \[\gamma_{g, \epsilon}(x):=\int_\Sigma \gamma_g(y)\rho_{\epsilon,x}(y){\rm dv}_g(y) \in \mc{L}^{-2}_x.\]
 The mapping $x\mapsto \rho_{\epsilon,x}$ is continuous with values in $\mc{H}^{s}(\Sigma)$ for each $s>0$, the process $\gamma_{g,\epsilon}$ is then continuous. The regularized GFF-Witten pair is the continuous process $(X_{g,\epsilon},\gamma_{g,\epsilon})$. By Definition \ref{witten_pair_Sigma}, the conditional covariance of $\gamma_{g,\eps}$ is, with $\theta:=M_{2b}^g(X_g,\dd x)$
\[\E[\gamma_{g,\epsilon}(x)\overline{\gamma_{g,\epsilon}}(x')|X_g]=e^{2bc}\frac{\pi}{k}  \cjg (1-\Pi^{\theta})T_0\rho_{\epsilon,x},
(1-\Pi^{\theta}) T_0\rho_{\epsilon,x'}\cjd_{L^2(\theta)}  \in \mc{L}_x^{-2}\otimes \mc{L}^{-2}_{x'}. \]

\noindent\textbf{Summary:} For readability, we gather here the notation for the important objects introduced in this section.
\begin{itemize} 
\item For $\theta$ an admissible measure, $\Pi^\theta: C^\infty(\Sigma;\mc{L}^{-2}\otimes\Lambda^{0,1}\Sigma)\to \ker \bar{\pl}_{\mc{L}^{-2}}^*$ is the orthogonal projection in $L^2(\theta)=L^2(\Sigma,\theta)$.
\item For a basis $(e_p)_p$ of $\ker \bar{\pl}_{\mc{L}^{-2}}^*$, $\mc{G}^\theta$ is the Gram matrix \eqref{def_GramG} of $(e_p)_p$ for the scalar product on $L^2(\theta)$. 
\end{itemize}
\subsection{Probabilistic construction of the path integral}

In this section, we consider a rank $2$ holomorphic vector bundle $(E,\bar{\pl}_E,g_E)$ with trivial determinant, $g_E$ being a Hermitian metric. We represent it as an extension with parameters $(\mc{L},\beta)$ as explained in Section \ref{setup_E}, thus $E=\mc{L}^{-1}\oplus \mc{L}$ where $\mc{L}^{-1}$ is a holomorphic subbundle of $E$ and $\bar{\pl}_E$ has the triangular form \eqref{dbar_Esetup}.

Before we give the probabilistic definition of the path integral, let us explain the heuristics of the definition below. 
First recall the formula in Lemma \ref{lem:S(h,A)_formula_stable} for the action $\tilde{S}_{E,\bar{\pl}_E,g_E}(\phi,\nu)$ and let us fix $\phi$ to be a smooth function and view $\nu$ as our variable of integration.\\ 

\noindent \textbf{Heuristics for the probabilistic definition of the path integral.} Assume first that $(E,\bar{\pl}_E, g_E)$ 
is represented as an extension with parameters $(\mc{L},\beta)$ and we are in the generic case, i.e. $\dim\ker  \bar{\pl}_{\mc{L}^{-2}}=0$ (then $\tilde{\mc{D}}_\phi$ is injective). 
Let $\nu_g$ be a complex Gaussian variable with 
covariance kernel ($\tilde{R}_\phi$ being defined in Lemma \ref{inverseDphi})
\[ \E[ \cjg \nu_g,f\cjd_2 \bbar{\cjg \nu_g, f'\cjd}_2]=\frac{\pi}{k}\cjg \tilde{R}_{\phi}f',f\cjd_2 ,\qquad  \E[\cjg \nu_g,f\cjd_2 \cjg \nu_g,f'\cjd_2 ]=0.\]
 Then using  the relation
\[  \|e^{-\phi}\bar{\pl}_{\mc{L}^{-2}}(e^{\phi}\nu)\|^2_2=\cjg \tilde{\mc{D}}_{\phi} \nu,\nu\cjd_2\]
and Lemma \ref{lem:detwitten} for the expression of the determinant of $\tilde{\mc{D}}_{\phi}$ and \eqref{Gphi} for the definition of $\mc{G}_\phi$, we define the following path integral in $\nu$ by
\begin{align*}
\int & F( \nu )e^{ -\frac{k}{\pi }\| e^{-\phi}(\bar{\pl}_{\mc{L}^{-2}}(e^{\phi}\nu)+\beta)\|^2_2} D\nu
\\&:=(\frac{\pi}{k})^{\frac{\chi(\Sigma)}{6}+2{\rm deg}(\mc{L})}\frac{e^{\frac{1}{2\pi}\int_\Sigma (|d\phi|^2_g-\tfrac{1}{2}K_g\phi){\rm dv}_g-\frac{2}{\pi i}\int_\Sigma \phi F_{\mc{L}}}}{\det(\mc{D}_0)\det(\mc{G}_\phi)\det(\mc{G}_0)^{-1}}  \E\Big[F\big(\nu_g \big)e^{-2\tfrac{k}{\pi}{\rm Re}\cjg \nu_g,e^\phi \bar{\pl}^*_{\mc{L}^{-2}}(e^{-2\phi}\beta) \cjd_2-\frac{k}{\pi}\|e^{-\phi}\beta\|^2_2}\Big].
\end{align*}
Now we want to treat the exponential of the linear term in $\nu_g$  using Cameron-Martin, and for this we need to compute covariances:
using the second identity of \eqref{tildeRphi} to compute the first variance and \eqref{Rphidbar*} to compute the second covariance, we get 
\[ \begin{split}
\E[(2{\rm Re}\cjg \nu_g,e^{\phi}\bar{\pl}^*_{\mc{L}^{-2}}(e^{-2\phi}\beta) \cjd_2)^2]=& 2\frac{\pi}{k} \cjg \tilde{R}_\phi (e^{-\phi}\bar{\pl}_{\mc{L}^{-2}}e^{\phi})^*(e^{-\phi}\beta) , (e^{-\phi}\bar{\pl}_{\mc{L}^{-2}}e^{\phi})^*(e^{-\phi}\beta)\cjd_2  \\
=& 2\frac{\pi}{k}\|(1-\tilde{\Pi}_\phi) (e^{-\phi}\beta)\|_2^2,\\ 
\E[2 \nu_g {\rm Re}\cjg \nu_g,e^{\phi}\bar{\pl}^*_{\mc{L}^{-2}}(e^{-2\phi}\beta) \cjd_2]=& \frac{\pi}{k}\tilde{R}_\phi (e^{-\phi}\bar{\pl}_{\mc{L}^{-2}}e^{\phi})^*(e^{-\phi}\beta)\\
=&\frac{\pi}{k}e^{-\phi}\tilde{T}_\phi^*({\rm Id}-\tilde{\Pi}_\phi)(e^{-\phi}\beta)\\
= &\frac{\pi}{k}e^{-\phi}T_0^*e^{2\phi}({\rm Id}-\Pi_\phi)(e^{-2\phi}\beta)
\end{split}\]
with $\tilde{\Pi}_\phi$ and $\Pi_\phi$ the projections defined in \eqref{tildePiphi}. This gives, using the notations  of Lemma \ref{inverseDphi}, 
\begin{align*}
& \E\Big[F\big(\nu_g \big)e^{-2\frac{k}{\pi}{\rm Re}\cjg \nu_g,e^\phi \bar{\pl}^*_{\mc{L}^{-2}}(e^{-2\phi}\beta) \cjd_2-\frac{k}{\pi}\|e^{-\phi}\beta\|^2_2}\Big]\\
 &= \E\Big[F\big(\nu_g +e^{-\phi}T_0^*e^{2\phi}({\rm Id}-\Pi_\phi)(e^{-2\phi}\beta)\big)e^{\frac{k}{\pi}\|(1-\tilde{\Pi}_\phi) (e^{-\phi}\beta)\|_2^2 -\frac{k}{\pi}\|e^{-\phi}\beta\|^2_2}\Big]\\
&=  \E\Big[F\big(\nu_g+e^{-\phi}T_0^*e^{2\phi}({\rm Id}-\Pi_\phi)(e^{-2\phi}\beta) \big)e^{ -\frac{k}{\pi}\|\tilde{\Pi}_\phi (e^{-\phi}\beta)\|^2_2}\Big].
\end{align*}
Now, switching to the $\gamma_g=e^{\phi}\nu_g$ field, this leads to   the following probabilistic definition of the $\gamma$-path integral (conditional on $\phi$) 
\begin{align*}
\int & F( e^{\phi}\nu )e^{ -\frac{k}{\pi }\| e^{-\phi}(\bar{\pl}_{\mc{L}^{-2}}(e^{\phi}\nu)+ \beta)\|^2_2} D\nu\\&:=(\frac{\pi}{k})^{\frac{\chi(\Sigma)}{6}+2{\rm deg}(\mc{L})}\frac{e^{\frac{1}{2\pi}\int_\Sigma (|d\phi|^2_g-\tfrac{1}{2}K_g\phi){\rm dv}_g-\frac{2}{\pi i}\int_\Sigma \phi F_{\mc{L}}}}{\det(\mc{D}_0)\det(\mc{G}_\phi)\det(\mc{G}_0)^{-1}}  \E\Big[F(\gamma_g +T_0^*({\rm Id}-e^{2\phi}\hat{\Pi}_\phi)\beta)e^{ -\frac{k}{\pi}\|e^{\phi}\hat{\Pi}_\phi \beta\|^2_2}\Big].
\end{align*}
where 
\[
\hat{\Pi}_\phi u:= e^{-\phi}\tilde{\Pi}_\phi (e^{-\phi}u)=\sum_{j,i=1}^N (({\mc{G}_\phi})^{-1})_{pq} e_p  \cjg u,e_q\cjd_2.
\]
The case where $\dim \ker \bar{\pl}_{\mc{L}^{-2}}=d>0$  is somewhat similar with a twist that we have now to incorporate the zero modes $\gamma_0\in \ker \bar{\pl}_{\mc{L}^{-2}}$, distributed as the Lebesgue measure. Finally we can integrate out the $\phi$ variable as in \eqref{PIcondphi2} and we see that $\phi$ has the law of 
$b(c+X_g)$ where $X_g$ is a GFF and $c$ a variable   distributed according to the Lebesgue measure $\dd c$, and $b=(k-2)^{-1/2}$ as before.\\

\noindent \textbf{Probabilistic definition of the path integral.} The discussion above motivates the probabilistic definition of the full path integral (in the variables $(\phi,\gamma)$) as follows.
Set \[b=(k-2)^{-1/2}\in (0,1), \quad N=\dim \ker \bar{\pl}^*_{\mc{L}^{-2}}=\dim H^1(\Sigma,\mc{L}^{-2})\] 
and let $(X_g,\gamma_g)$ be a GFF-Witten pair in $\mc{H}^{-s}(\Sigma,\R)\times \mc{H}^{-s}(\Sigma,\mc{L}^{-2})$ (for some fixed $s>0$), with parameters  $(g,M^g_{2b}(c+X_g,\dd x))$, defined on some probability space $(\Omega,\mc{F},\P)$ (with expectation $\E$).
Notice that for $\theta=M_{2b}^g(X_g,dx)$, the GMC measure on $\Sigma$, conditionally on $X_g$, the matrix $\det(\mc{G}^\theta)^{-1}>0$ almost surely. Indeed, as a $N$-dimensional Gram matrix, we have for each $v=(v_1,\dots,v_N)\in \C^N\setminus\{0\}$ and $(e_p)_p$ a basis of $\ker \bar{\pl}^*_{\mc{L}^{-2}}$
\[ \sum_{p,q=1}^N \bar{v}_pv_q \mc{G}^{\theta}_{pq}= \|\sum_{p}v_pe_p\|^2_{L^2(\theta)}>0,\]
the positivity coming from the fact that $\sum_{p}v_pe_p(x)$ does not vanish on some open ball if $v\not=0$. Now we state two  lemmas that will be useful for the convergence of the correlation functions. Their proofs are postponed to the appendix.
\begin{lemma}\label{det_negative_moment}
Let ${\bf z}=(z_1,\dots,z_m)$ be a set of $m$ disjoint points and $b\in (0,1)$. 
Let $\theta=M_{2b}^g(X_g,\dd x)$, let $u\in C^\infty(\Sigma\setminus \{\bf z\},\R)$  such that  there is $C>0$ such that for all $x$ near $z_\ell$,
$u(x)< bQ'\log(1/d_g(x,z_\ell))+C$, with $Q'<b+1/b $.
Let $\theta_u:=e^{2u}\theta$, then the random variable $\det(\mc{G}^{\theta_u})$ satisfies 
\[\E[ \det(\mc{G}^{\theta_u})^{-1}]<\infty.\] 
\end{lemma}
 
Let $\theta$ be an admissible measure on $\Sigma$ and $g$ a Riemannian metric on $\Sigma$.
Define the operator
$\hat{\Pi}^{\theta}$ on $C^\infty(\Sigma;\mc{L}^{-2}\otimes \Lambda^{0,1}\Sigma)$ by
\begin{equation}\label{hatPi} 
\hat{\Pi}^{\theta}(u):=\sum_{p,q=1}^N ((\mc{G}^{\theta})^{-1})_{pq} e_p  \cjg u,e_q\cjd_2
\end{equation}
and the section of $\mc{L}^{-2}$
\begin{equation}\label{def_of_Y}  
Y^{\theta}(\beta)(x):=T_0^*\beta(x) - \int_{\Sigma}T_0^*(x,y)(\hat{\Pi}^{\theta}\beta)(y)\theta(\dd y).
\end{equation}

\begin{lemma}\label{lemmaYtheta}
Assume the map $f\in \mc{H}^{s+1}(\Sigma;\R)\mapsto \theta(f)$ is continuous for some $s>0$. Then the map $f\in \mc{H}^{s}(\Sigma;\mc{L}^{-2})\mapsto \cjg Y^{\theta}(\beta),f\rangle_2$ is continuous. In particular this holds if $\theta$ is the GMC measure and $\theta_u$ as defined in Lemma \ref{det_negative_moment}.
\end{lemma}

For a holomorphic line bundle $\mc{L}$ on $(\Sigma,g)$ and $k\in \R$, it is useful to define the $2$-form 
\begin{equation}\label{defK_L}
K_{\mc{L}}:=\frac{K_g}{4\pi}{\rm v}_g-\frac{k-2}{\pi i}F_{\mc{L}}
\end{equation}  
For $d:=\dim \ker \bar{\pl}_{\mc{L}^{-2}}=\dim H^0(\Sigma,\mc{L}^{-2})$ and 
$N:=\dim \ker  \bar{\pl}_{\mc{L}^{-2}}^*=\dim H^1(\Sigma,\mc{L}^{-2})$,  
we have $N=2{\rm deg}(\mc{L})-\chi(\Sigma)/2+d$ by Riemann-Roch formula \eqref{Riemann_RochL^k}, and combining 
with Gauss-Bonnet formula $\int_\Sigma K_g{\rm v}_g=4\pi \chi(\Sigma)$, one gets 
\begin{equation}\label{RR+GB} 
-2N-\int_\Sigma \mc{K}_{\mc{L}}=-2k\, {\rm deg}(\mc{L})-2d.
\end{equation}

Based on the discussion above and recalling the expression of the action in Lemma \ref{lem:S(h,A)_formula_stable}, one can set:
\begin{definition}[\textbf{Definition of the path integral}]\label{defH3sigma}
Let $E$ be a rank $2$ holomorphic vector bundle on $(\Sigma,g)$, that we represent as an extension with parameters $(\mc{L},\beta)$ as in Section \ref{setup_E}.
Pick a basis $(e_p)_{p=1,\dots,N}$ of $\ker \bar{\pl}^*_{\mc{L}^{-2}}$, let $(X_g,\gamma_g)$ be a GFF-Witten pair with parameters  $(g,M^g_{2b}(c+X_g,\dd x))$ and let $\theta(\dd x):=M^g_{2b}(X_g,\dd x)$. For any nonnegative measurable $F$ on $\mc{H}^{-s}(\Sigma,\R)\times \mc{H}^{-s}(\Sigma;\mc{L}^{-2})$, we set:\\
1) \textbf{Generic case:} i.e. when $H^0(\Sigma,\mc{L}^{-2})=0$, 
\begin{equation}\label{def_PI_case_noker}
\begin{split}
\langle F(\phi,\gamma) \rangle_{\Sigma,g,\mc{L},\beta}:=& C_k \Big(\frac{\det(\mc{D}_0)}{\det(\mc{G}_0)}\Big)^{-1}\Big(\frac{\det(\Delta_{g})}{{\rm v}_g(\Sigma)}\Big)^{-1/2}e^{\frac{k}{\pi}\|\beta\|_{2}^2} \\
& \times \int_{\R}e^{-2k{\rm deg}(\mc{L})bc}\E\Big[\frac{F(b(c+X_g), \gamma_g+Y^{\theta}(\beta))}{\det(\mc{G}^{\theta})}e^{-b\int_\Sigma \mc{K}_{\mc{L}}X_g-\frac{k}{\pi}e^{-2bc} \|\hat{\Pi}^{\theta}\beta\|^2_{L^2(\theta)}}\Big]\,\dd c.
\end{split}
\end{equation}
where $\mc{G}_0$, $\mc{G}^{\theta}$, $\hat{\Pi}^{\theta}$, $Y^{\theta}(\beta)$, $\mc{K}_{\mc{L}}$ are defined by \eqref{Gphi}, \eqref{def_GramG}, \eqref{hatPi}, \eqref{def_of_Y} and \eqref{defK_L}, while $C_k$ is the topological constant
\begin{equation}\label{constant_Ck} 
C_k=\big(\frac{k-2}{2\pi}\big)^{-\frac{\chi(\Sigma)}{12}+\frac{1}{2}}\big(\frac{\pi}{k}\big)^{\frac{\chi(\Sigma)}{6}+2{\rm deg}(\mc{L})}.
\end{equation}
\noindent 2) \textbf{Non generic case}: if $W_0:=\ker \bar{\pl}_{\mc{L}^{-2}}=H^0(\Sigma,\mc{L}^{-2})$ is non trivial, then we set 
\begin{equation}\label{def_PI_case_ker}
\begin{split}
\langle &F(\phi,\gamma) \rangle_{\Sigma,g,\mc{L},\beta}:=C_k \Big(\frac{\det(\mc{D}_0)}{\det(\mc{G}_0)}\Big)^{-1}\Big(\frac{\det(\Delta_{g})}{{\rm v}_g(\Sigma)}\Big)^{-1/2}e^{\frac{k}{\pi}\|\beta\|_{2}^2} \\
& \times \int_{\R}\int_{W_0}e^{-(2k{\rm deg}(\mc{L})+2d)bc}\E\Big[\frac{F(b(c+X_g), \gamma_0+\gamma_g+Y^{\theta}(\beta))}{\det(\mc{G}^{\theta})}e^{-b\int_\Sigma \mc{K}_{\mc{L}}X_g-\frac{k}{\pi}e^{-2bc}\|\hat{\Pi}^{\theta}\beta\|^2_{L^2(\theta)} }\Big]\, \dd \gamma_0\,\dd c
\end{split}\end{equation}
where $\dd \gamma_0$ is the volume measure associated to the Hermitian metric given by the $L^2({\rm v}_g)$ pairing on $W_0$.
\end{definition}
In the definition above, we have gathered (in \eqref{def_PI_case_noker}) the terms involving exponentials of the zero mode $c$ of $\phi_g=c+X_g$: 
there is a term $e^{-2Nbc}$ coming from $(\det(\mc{G}_\phi))^{-1}=e^{-2Nbc}(\det(\mc{G}^\theta))^{-1}$ and a term $e^{-bc \int_\Sigma \mc{K}_{\mc{L}}}$, then \eqref{RR+GB} gives that the sum of these terms become $e^{-(2k{\rm deg}(\mc{L})+2d)bc}$.

\begin{proposition}{\bf (Weyl anomaly and diffeomorphism invariance)}\label{prop:weyl_surface}  Consider a metric $g'=e^{\omega}g$ conformal to the metric $g$ for some smooth $\omega\in C^\infty(\Sigma)$. Then, for nonnegative measurable $F$ on $\mc{H}^{-s}(\Sigma,\R)\times \mc{H}^{-s}(\Sigma,\C)$,
\[\langle F(\phi,\gamma) \rangle_{\Sigma,g',\mc{L},\beta}=e^{ \frac{{\bf c}}{96\pi}S^0_{\rm L}(g,\omega)}\langle F(\phi-\tfrac{b^2}{2}\omega,\gamma) \rangle_{\Sigma,g,\mc{L},\beta}\]
with ${\bf c}(k):=3k/(k-2)$ is the central charge of the $\mathbb{H}^3$-model and $S^0_{\rm L}$ is the Liouville action \eqref{LiouvilleS0}.

If $\psi:\Sigma\to\Sigma$ is a diffeomorphism and $\psi^*g$ the pullback metric then
\[\langle F(\psi^*\phi ,\psi^*\gamma) \rangle_{\Sigma,\psi^*g,\psi^*\mc{L},\psi^*\beta}= \langle F(\phi ,\gamma) \rangle_{\Sigma,g} .\]
\end{proposition}
\begin{proof} The diffeomoprhism invariance is direct by using the equality in law \eqref{equalityinlaw}. Let us then consider 
the confomal covariance (weyl anomaly).
Consider the case $\ker \bar{\pl}_{\mc{L}^{-2}}\not=0$, the case $\ker \bar{\pl}_{\mc{L}^{-2}}=0$ being simpler. 
The scalar product on $1$-forms depends only on the conformal class, thus the 
the matrix $\mc{G}_0$ and the norm $\|\beta\|^2_2$ and  do not depend on the metric $g$ in a fixed conformal class. The law of $c+X_g$ also does not depend on the metric $g$ in a fixed conformal class, so we can replace $X_{g'}$ by $X_g$ in the expression defining 
$\langle F(\phi,\gamma) \rangle_{\Sigma,g',\mc{L},\beta}$. Now, we remark that if we fix an orthonormal basis $(f_p)_{p=1,\dots,d}$ of $W_0=\ker \bar{\pl}_{\mc{L}^{-2}}$ for $L^2({\rm v}_g)$ (not orthonormal for $L^2({\rm v}_{g'})$), the $W_0$-integral in \eqref{def_PI_case_ker} with background metric $g'$ can be rewitten as 
\[ \det(\mc{N}_{0,g'})\int_\C \int_{W_0}\E\Big[\frac{F(b(c+X_g),\gamma_0+Y^{\theta}(\beta))}{e^{2Nbc}\det(\mc{G}^{\theta'})}e^{-b\int_\Sigma \mc{K}'_{\mc{L}}(X_g+c)-\frac{k}{\pi}e^{-2bc}\|\hat{\Pi}^{\theta'}\beta\|^2_{L^2(\theta')} }\Big]\, \dd \gamma_0 \]
where $\dd \gamma_0$ is the volume measure associated to the $L^2({\rm v}_g)$ Hermitian product, $\mc{N}_{0,g'}$ is the matrix \eqref{NPhi} (with $\phi=0$ and background metric $g'$), $\mc{K}'_{\mc{L}}$ defined by \eqref{defK_L} with $g'$ instead of $g$ and $\theta'(\dd x)=M_{2b}^{g'}(X_g,{\rm v}_{g'},\dd x)$. For convenience we have also rewritten the $e^{-(2k{\rm deg}(\mc{L})+2d)bc}$ term as  $e^{-2Nbc-bc\int_{\Sigma}\mc{K}'_{\mc{L}}}$ using \eqref{RR+GB}.
We are thus left with considering the following quantity
\begin{equation}\label{anomaly}
C(g')\int_{\R\times W_0}\E\Big[\frac{F(b(c+X_g), \gamma_0+\gamma_{g'})}{e^{2Nbc}\det(\mc{G}^{\theta'})}e^{-b\int_\Sigma \mc{K}'_{\mc{L}}(X_{g}+c) -\frac{k}{\pi}e^{-2bc}\|\hat{\Pi}^{\theta'}\beta\|^2_{L^2(\theta')}}\Big]\,\dd c\, \dd \gamma_0
\end{equation}
with 
\[C(g'):=\frac{\det(\mc{N}_{0,g'}){\rm v}_{g'}(\Sigma)^{1/2}}{\det(\mc{D}_{0,g'})\det(\Delta_{g'})^{1/2}}.\]
Here we emphasize that $\gamma_{g'}$ depends on $c+X_{g'}$ and we shall thus write $\gamma_{g'}[c+X_{g'}]$ to keep track of the dependence, and by the discussion above,  $\gamma_{g'}[c+X_{g'}]=\gamma_{g'}[c+X_{g}]$ in law.
By \eqref{curv}, the term $\int_\Sigma K_{g'}(c+X_{g})\dd {\rm v}_{g'}$ can be written as 
\[\int_\Sigma K_{g'}(c+X_{g})\dd {\rm v}_{g'}=\int_\Sigma K_{g}(c+X_{g})\dd {\rm v}_{g}+\int_\Sigma \Delta_{g}\omega X_{g}\dd {\rm v}_{g}.\]
Define the Gaussian random variable $Y:= -\frac{b}{4\pi}\int_\Sigma \Delta_{g}\omega X_{g}\dd {\rm v}_{g}$. We have
\begin{equation}\label{rayas2005}
\demi \E[Y^2]= \frac{b^2}{16\pi}\int_\Sigma|d\omega|_g^2{\rm dv}_g, \quad \E[YX_{g}]= -\frac{b}{2}(\omega-m_g(\omega))
\end{equation}
where $m_g(\omega)$ stands for the orthogonal projection of $\omega$ on constants in $L^2({\rm v}_g)$.
Therefore, we can apply Cameron-Martin to the random variable $Y$: using the notation $\phi_g:=c+X_g$, Cameron-Martin amounts to replacing 
$\phi_g$ by $\tilde{\phi}_g:=\phi_{g} -\frac{b}{2}(\omega-m_g(\omega))$  in \eqref{anomaly}. Using that $\theta'(\dd x)=e^{(1+b^2)\omega(x)}M_{2b}^g(X_g,{\rm v}_g,\dd x)$ (by \eqref{relationentrenorm} and ${\rm v}_{g'}=e^{\omega}{\rm v}_g$), we see that the term $e^{-2bc}\|\hat{\Pi}^{\theta'}\beta\|^2_{L^2(\theta')}$ and $e^{2Nbc}\det(\mc{G}^{\theta'})$ in \eqref{anomaly} are replaced respectively by 
\[e^{-2b(c+\frac{b}{2}m_g(\omega))} \int_{\Sigma} |\hat{\Pi}^{\theta}\beta|_{g}^2 \theta(\dd x)\quad \textrm{and }\, e^{2Nb(c+\frac{b}{2}m_g(\omega))}\det(\mc{G}^{\theta}).\]
We can thus rewrite, using a change of variable $c\mapsto c-\frac{b}{2}m_g(\omega)$  and the Polyakov formulas \eqref{anomalydetD0} and \eqref{detpolyakov},
\begin{align*}
 \langle & F(\phi,\gamma) \rangle_{\Sigma,g',\mc{L},\beta}
\\
&=
C(g')e^{ \frac{b^2}{16\pi}\int_\Sigma|d\omega|_g^2{\rm dv}_g}\int_{\R\times W_0}\E\Big[\frac{F(b\tilde{\phi}_g, \gamma_0+ \gamma_{g'}[\tilde{\phi}_{g}])}{ e^{2Nb(c+\frac{b}{2}m_g(\omega))}\det(\mc{G}^{\theta})}
e^{-b\int_\Sigma \mc{K}_{\mc{L}}\tilde{\phi}_{g} -\frac{k}{\pi}e^{-2b(c+\frac{b}{2}m_g(\omega))}\|\hat{\Pi}^{\theta}\beta\|^2_{L^2(\theta)} }\Big]\,\dd c\, \dd \gamma_0
\\
&=C(g)e^{ \frac{3+6b^2}{96\pi}S^0_{\rm L}(g,\omega)}
\int_{\R\times W_0}\E\Big[\frac{F(b\phi_g -\tfrac{b^2}{2}\omega,  \gamma_0+\gamma_{g'}[\phi_{g} -\tfrac{b}{2}\omega])}{ e^{2Nbc}\det(\mc{G}^{\theta})}e^{-b\int_\Sigma \mc{K}_{\mc{L}}\phi_{g} -\frac{k}{\pi}e^{-2bc}\|\hat{\Pi}^{\theta}\beta\|^2_{L^2(\theta)}}\Big]\,\dd c\, \dd \gamma_0
\end{align*}
Here, we emphasize that the $\exp(-\frac{1}{2\pi i}\int_\Sigma F_{\mc{L}}\omega)$ in the Polyakov anomaly of \eqref{anomalydetD0} (applied with $n=-2$) cancels the term 
$\exp(-\frac{b^2}{2}\int_\Sigma \frac{k-2}{\pi i}F_{\mc{L}}\omega)$ coming from $\exp(-b\int_\Sigma \mc{K}_{\mc{L}}\tilde{\phi}_{g})$.
The GMC measure describing the parameters of the Witten field becomes, using \eqref{relationentrenorm},
\[M^{g'}_{2b}(c+X_g-\tfrac{b}{2}\omega,{\rm v}_{g'},\dd x) =e^{\omega}M^{g}_{2b}(c+X_g,{\rm v}_{g} ,\dd x).\]
Then by Lemma \ref{lemmainvar_Sigma} applied with $\theta(\dd x)=M^{g}_{2b}(c+X_g,{\rm v}_{g} ,\dd x)$, conditionally on $X_g$, one has $\gamma_{g'}[\phi_g-\tfrac{b}{2}\omega]=\gamma_{g}[\phi_{g}]-P_0'(\gamma_{g}[\phi_{g}])$ in law (with $P_0'$ the projection on $\ker \bar{\pl}_{\mc{L}^{-2}}$ with respect to $L^2({\rm v}_{g'})$). This gives
\begin{align*}
\langle & F(\phi,\gamma) \rangle_{\Sigma,g',\mc{L},\beta}\\
& =C(g)e^{ \frac{3+6b^2}{96\pi}\int_\Sigma(|d\omega|_g^2+2K_g\omega){\rm dv}_g}\int_{\R\times W_0}\E\Big[F(b\phi_g -\tfrac{b^2}{2}\omega, \gamma_0+\gamma_{g}-P_0'(\gamma_{g}))e^{-\frac{b}{4\pi}\int_\Sigma K_{g}\phi_{g} \dd {\rm v}_{g}}\Big]\,\dd c\, \dd \gamma_0.
\end{align*}
We can absorb the $P_0'(\gamma_{g})$ term by using the invariance of the Lebesgue measure $\dd \gamma_0$ under 
translations.
\end{proof}

\noindent \textbf{Summary:} for readability, we gather here the notation for the new objects defined in this section:
\begin{itemize}
\item For $\theta$ an admissible measure and a basis $(e_p)_p$ of $\ker \bar{\pl}_{\mc{L}^{-2}}^*$,  $\hat{\Pi}^{\theta}$ and $Y^{\theta}(\beta)$ are defined as the operator and the random variable (see \eqref{def_of_Y} and \eqref{hatPi})
\[
\hat{\Pi}^{\theta}(u):=\sum_{p,q=1}^N ((\mc{G}^{\theta})^{-1})_{pq} e_p  \cjg u,e_q\cjd_2,\quad 
Y^{\theta}(\beta)(x):=T_0^*\beta(x) - \int_{\Sigma}T_0^*(x,y)(\hat{\Pi}^{\theta}\beta)(y)d\theta(y).
\]
\item $\mc{K}_{\mc{L}}=\frac{K_g}{4\pi}{\rm v}_g-\frac{k-2}{\pi i}F_{\mc{L}}$ is the $2$-form introduced in \eqref{defK_L} satisfying $\int_\Sigma \mc{K}_{\mc{L}}=-2N+2k\, {\rm deg}(\mc{L})$.
\item $W_0:=\ker \bar{\pl}_{\mc{L}^{-2}}=H^0(\Sigma,\mc{L}^{-2})$, equipped with the Hermitian product induced by the $L^2({\rm v}_g)$ pairing defined in \eqref{norme2} (with $n=-2$).
\item We use the The GMC $\theta(\dd x):=M^g_{2b}(X_g,\dd x)$ for the definition of the path integral.
\end{itemize}

\subsection{Gauge covariance of the path integral}\label{sec:gauge_transform}
In this section, we compare our path integral construction in two different gauges representing gauge equivalent line subbundles $\mc{L}^{-1}$, and prove that they are related by a simple expression involving the Polyakov-Wiegman anomaly.

We consider as in Section \ref{setup_E} a rank $2$ holomorphic bundle with trivial determinant represented as an extension with parameters $(\mc{L},\beta)$. The Dolbeault operator $\bar{\pl}_E$ has theupper triangular form \eqref{dbar_Esetup} in $\mc{L}^{-1}\oplus \mc{L}$. 
Consider a smooth gauge in this decomposition
\begin{equation}\label{matrix_h_0}
h_0^{-1}:= \left(\begin{array}{cc} 
e^{\frac{u}{2}} & e^{-\frac{u}{2}}v \\
0 & e^{-\frac{u}{2}} 
\end{array}\right)\end{equation}
for $u\in C^\infty(\Sigma;\C)$, $v\in C^\infty(\Sigma;\mc{L}^{-2})$. It maps 
the decomposition $\mc{L}_u^{-1}\oplus \mc{L}_u$ to the decomposition  $\mc{L}^{-1}\oplus \mc{L}$, if $\mc{L}_u$ is the complex vector bundle bundle $\mc{L}$ but equipped with the conjugated Dolbeault operator 
$\bar{\pl}_{\mc{L}_u}:=e^{\frac{u}{2}}\bar{\pl}_{\mc{L}}e^{-\frac{u}{2}}$. Similarly for $n\in \Z$,
$\mc{L}_u^n$ is a holomorphic line bundle equipped with the Dolbeault operator $\bar{\pl}_{\mc{L}_u^n}:=e^{n\frac{u}{2}}\bar{\pl}_{\mc{L}^n}e^{-n\frac{u}{2}}$. 
The Dolbeault operator $\bar{\pl}_E$ can be conjugated by $h_0$: a direct calculation using local trivialisations of $\mc{L}$ gives  
\[ h_0\circ \bar{\pl}_E\circ h_0^{-1}=  \left(\begin{array}{cc} 
e^{-\frac{u}{2}}\bar{\pl}_{\mc{L}^{-1}}e^{\frac{u}{2}} & e^{-u}(\beta+\bar{\pl}_{\mc{L}^{-2}}(v)) \\
0 & e^{\frac{u}{2}}\bar{\pl}_{\mc{L}}e^{-\frac{u}{2}} 
\end{array}\right)= \left(\begin{array}{cc} 
\bar{\pl}_{\mc{L}_{u}^{-1}} & e^{-u}\beta+\bar{\pl}_{\mc{L}_{u}^{-2}}(e^{-u}v) \\
0 & \bar{\pl}_{\mc{L}_{u}}
\end{array}\right).\]
The $\beta$ form then becomes in the new gauge 
\begin{equation}\label{beta_uv} 
\beta \mapsto \beta_{u,v}:=e^{-u}(\beta+\bar{\pl}_{\mc{L}^{-2}}(v)).
\end{equation}
As explained in Section \ref{WZW_coset}, a positive definite quadratic form represented by 
\[h=(an)(an)^*=  \left(\begin{array}{cc} 
e^{\frac{\phi}{2}} & \nu e^{\frac{\phi}{2}} \\
0 &  e^{-\frac{\phi}{2}}
\end{array}\right)  \left(\begin{array}{cc} 
e^{\frac{\phi}{2}} & 0 \\
\nu^* e^{\frac{\phi}{2}} &  e^{-\frac{\phi}{2}}
\end{array}\right)  \in {\rm SL}(2,E) \] 
in the decomposition $E=\mc{L}^{-1}\oplus \mc{L}$ with $a,n$ of the form \eqref{alphaandn} will become $h_0hh_0^*=(h_0an)(h_0an)^*=(a'n')(a'n')^*$ in the decomposition $\mc{L}_u^{-1}\oplus \mc{L}_u$ under the gauge change by $h_0$, with 
\[ a':=  \left(\begin{array}{cc} 
e^{\frac{\phi-{\rm Re}(u)}{2}} & 0 \\
0 &  e^{-\frac{(\phi-{\rm Re}(u))}{2}}
\end{array}\right), \quad n':=\left(\begin{array}{cc} 
1&  \nu e^{-i{\rm Im}(u)}-ve^{-\phi-i{\rm Im}(u)}\\
0 &  1
\end{array}\right)\]
which means that the $(\phi,\nu)$ coordinates are becoming in the gauge $h_0$ 
\begin{equation}\label{phi'u'} 
(\phi',\nu')=(\phi-{\rm Re}(u),  \nu e^{-i{\rm Im}(u)}-ve^{-\phi-i{\rm Im}(u)}).
\end{equation} 
The $(\phi,\gamma)$ coordinates (recall $\gamma=e^{\phi}\nu$) change under $h_0$ by 
\[  (\phi,\gamma)\mapsto (\phi',\gamma')=(\phi-{\rm Re}(u), e^{-u}(\gamma-v)).\]
Let us now compare the WZW action in the gauge induced by $h_0$ and our original gauge. 
\begin{lemma}\label{l:relation_action}
The following identity holds: 
\[\tilde{S}_{\Sigma,h_0\circ \bar{\pl}_E\circ h_0^{-1},g_E}(h_0hh_0^*)=\tilde{S}_{\Sigma, \bar{\pl}_E,g_E}(h)+\mc{W}(u,\beta,v)\]
where, if $\beta_{u,v}$ is given by \eqref{beta_uv}, the anomaly $\mc{W}(u,\beta,v)$ is given by 
\begin{equation}\label{Wubetav}
\mc{W}(u,\beta,v):=\frac{1}{\pi}\int_{\Sigma}(\frac{1}{4}|d{\rm Re}(u)|_g^2+|\beta_{u,v}|^2_{\mc{L}^{-2},g}-|\beta|^2_{\mc{L}^{-2},g}){\rm dv}_g
-\frac{1}{\pi i}\int_\Sigma {\rm Re}(u)F_{\mc{L}}
\end{equation}
\end{lemma} 
\begin{proof} We use the expression computed in Lemma \ref{lem:S(h,A)_formula_stable} together with \eqref{phi'u'} and \eqref{beta_uv}:
\[\begin{split} 
\tilde{S}_{\Sigma,h_0\circ \bar{\pl}_E\circ h_0^{-1},g_E}(h_0hh_0^*)=& -\frac{1}{\pi }\int_\Sigma (\frac{1}{4}|d\phi'|^2_g +e^{-2\phi'}|\bar{\pl}_{\mc{L}_u^{-2}}(e^{\phi'}\nu')+ \beta_{u,v}|^2_{\mc{L}^{-2},g}){\rm dv}_g\\
&  +\frac{1}{\pi }\int_{\Sigma }(\frac{1}{i}{\rm Tr}(F_{\mc{L}_u})\phi'+|\beta_{u,v}|^2_{\mc{L}^{-2},g}){\rm v}_g\\
=&-\frac{1}{\pi }\int_\Sigma (\frac{1}{4}(|d\phi|^2+|d{\rm Re}(u)|^2_g)-\frac{1}{2}\cjg d\phi,d{\rm Re}(u)\cjd_g 
+e^{-2\phi}|\bar{\pl}_{\mc{L}^{-2}}(e^{\phi}\nu)+ \beta|^2_{\mc{L}^{-2},g}){\rm dv}_g\\
&  +\frac{1}{\pi }\int_{\Sigma }(\phi-{\rm Re}(u))(\frac{1}{i}F_{\mc{L}}-\frac{1}{2}\Delta_g {\rm Re}(u){\rm v}_g)+|\beta_{u,v}|^2_{\mc{L}^{-2},g}{\rm v}_g\\
=& \tilde{S}_{\Sigma, \bar{\pl}_E,g_E}(h)+\mc{W}(u,\beta,v)
\end{split}\]
which ends the proof. 
\end{proof}
Let $\hat{h}_0:=sh_0s^{-1}\in {\rm SL}(2,\C)$.
Notice that if $A$ is the unitary connection form associated to $(\bar{\pl}_E,g_E)$ represented in an orthonormal frame $s$, 
we get by using Lemma \ref{l:invarianceS} (recall also Definition \ref{def:actionWZW}): 
\[ \begin{split}
\tilde{S}_{\Sigma,h_0\circ \bar{\pl}_E\circ h_0^{-1},g_E}(h_0hh_0^*)= &S_\Sigma(( \hat{h}_0^{*-1}, \hat{h}_0). ( shs^{-1},A))-S_{\Sigma}({\rm Id},A_{\hat{h}^{-1}_0})\\
=& S_\Sigma(shs^{-1},A)-S_{\Sigma}( \hat{h}_0^{*-1},A^{1,0})-S_{\Sigma}( \hat{h}_0^{-1},A^{0,1})-S_{\Sigma}({\rm Id},A_{\hat{h}^{-1}_0})\\
=& \tilde{S}_{\Sigma,\bar{\pl}_E,g_E}(h)-S_{\Sigma}( \hat{h}_0^{*-1},A^{1,0})-S_{\Sigma}(\hat{h}_0^{-1},A^{0,1})-S_{\Sigma}({\rm Id},A_{\hat{h}^{-1}_0})+S_{\Sigma}({\rm Id},A).
\end{split}\]
This means that $\mc{W}(u,\beta,v)$ can be rewritten as:
\[\mc{W}(u,\beta,v)=-S_{\Sigma}( \hat{h}_0^{*-1},A^{1,0})-S_{\Sigma}(\hat{h}_0^{-1},A^{0,1})-S_{\Sigma}({\rm Id},A_{\hat{h}^{-1}_0})+S_{\Sigma}({\rm Id},A).\]
The relation of Lemma \ref{l:relation_action} on the action lets us expect the following relation on the path integral: 
\[\langle  F(\phi,\gamma) \rangle_{\Sigma,g,\mc{L}_u,\beta_{u,v}}=e^{k\mc{W}(u,\beta,v)}\Big\langle  F(\phi-{\rm Re}(u),e^{-u}(\gamma+v)) \Big\rangle_{\Sigma,g,\mc{L},\beta}.\]
The rest of the section will focus on proving this relation.
\begin{lemma}\label{lem:relation_laws}
The covariance, conditional on $X_g$,  of the random field $\gamma^u_g$ 
associated to the line bundle $\mc{L}_{u}$ as in Definition \ref{witten_pair_Sigma} is  given by 
\begin{equation} \label{covariance_gamma_u}
 \E[ \cjg \gamma_g^{u},f\cjd_2\bbar{\cjg \gamma_g^{u},f'\cjd}_2]= e^{2bc}\cjg (1-\Pi^{\theta_u})T_0e^{-\bar{u}}(1-\tilde{P}_u)f',(1-\Pi^{\theta_u})T_0e^{-\bar{u}}(1-\tilde{P}_u)f\cjd_{L^2(\theta_u)}.
 \end{equation}
where $\theta_u:=e^{2{\rm Re}(u)}\theta$ with $\theta=M_{2b}^g(X_g,dx)$ and $\Pi_{u}^{\theta}$ is the projector \eqref{def_proj} associated to $\ker \bar{\pl}^*_{\mc{L}_{u}^{-2}}$.
The relation between the integral kernels $\mc{T}_{X_g}$ and $\mc{T}_{X_g}^{u}$ of the covariances of $\gamma_g$ and $e^{u}\gamma_g^{u}$ 
(conditional on $X_g$) is then given by
\begin{equation}\label{relation_kernels_cov} 
\mc{T}_{X_g}^{u}(x,x')=  \mc{T}_{X_g+b^{-1}{\rm Re}(u)}(x,x').
\end{equation}
Consequently, for each $F\in C^0(\mc{H}^{-s}(\Sigma)\times \mc{H}^{-s}(\Sigma;\mc{L}^{-2}))$ non-negative, one has 
\begin{align*}
& e^{-\frac{1}{4\pi b^2}\int_{\Sigma}|d{\rm Re}(u)|_g^2{\rm dv}_g}\int_{\R}\E[ e^{-\frac{1}{2\pi b}\int_\Sigma \Delta_g{\rm Re}(u)X_g{\rm dv}_g} F(b(c+X_g),\gamma_g^{u})]\dd c \\
&=\int_{\R}\E[ F(b(c+X_g)-{\rm Re}(u),e^{-u}\gamma_g)]\dd c.
 \end{align*}
\end{lemma}
\begin{proof}
The operator $T_0$ inverting $\bar{\pl}_{\mc{L}^{-2}}$ (modulo kernel) in \eqref{Prop_T_phi} becomes the operator $\tilde{T}_u$ for $\bar{\pl}_{\mc{L}_{u}^{-2}}=e^{-u}\bar{\pl}_{\mc{L}^{-2}}e^{u}$  on $\mc{L}_{u}^{-2}$, by Lemma \ref{inverseDphi}. We have, by \eqref{TphiT0}
\[ \tilde{T}_u=(1-\tilde{\Pi}_{u})e^{\bar{u}}T_0e^{-\bar{u}}(1-\tilde{P}_u)\]
 where $\tilde{\Pi}_{u}$ is defined by \eqref{tildePiphi} (setting $\phi=u$ there). We also observe that 
  \begin{equation}\label{T0u^*} 
(1-\tilde{\Pi}_{u})e^{\bar{u}}T_0e^{-\bar{u}}=e^{\bar{u}}(1-\Pi^{{\rm v}_{g_u}})T_0e^{-\bar{u}},\quad (1-\Pi^{\theta}_{u})e^{\bar{u}}=e^{\bar{u}}(1-\Pi^{\theta_u}) \end{equation}
where $\theta_u:=e^{2{\rm Re}(u)}\theta$, 
$\Pi_{u}^{\theta}$ is the projector \eqref{def_proj} associated to $\ker \bar{\pl}^*_{\mc{L}_{u}^{-2}}$ and $\Pi^{{\rm v}_{g_u}}$ the projector \eqref{def_proj} associated to $\ker \bar{\pl}^*_{\mc{L}^{-2}}$ with $\theta$ replaced by the volume measure ${\rm v}_{g_u}:=e^{2{\rm Re}(u)}{\rm v}_g$ of $g_u=e^{2{\rm Re}(u)}g$. In particular 
\[ (1-\Pi^{\theta}_{u})T_u=e^{\bar{u}}(1-\Pi^{\theta_u})(1-\Pi^{{\rm v}_{g_u}})T_0e^{-\bar{u}}.\]
Since $\Pi^{{\rm v}_{g_u}}$ projects to $\ker \bar{\pl}^*_{\mc{L}_{u}^{-2}}$ and $\Pi^{\theta_u}$ is equal to the identity on this space, we have 
$\Pi^{\theta_u}\Pi^{{\rm v}_{g_u}}=\Pi^{{\rm v}_{g_u}}$ so 
\[ (1-\Pi^{\theta_u})(1-\Pi^{{\rm v}_{g_u}})=1-\Pi^{\theta_u}-\Pi^{{\rm v}_{g_u}}+\Pi^{{\rm v}_{g_u}}= (1-\Pi^{\theta_u}),\]
which implies  the formula \eqref{covariance_gamma_u} on the covariance conditional on $X_g$ (recall Definition \ref{witten_pair_Sigma}), and the relations on the integral kernels. Let us write $u_0:=u-m_g(u)$ where $m_g(u):={\rm v}_g(\Sigma)^{-1}\int_\Sigma u{\rm dv}_g$. We notice that $\gamma_g^u=\gamma_g^{u_0}$. We shift the GFF $X_g$ by $-\frac{1}{b}{\rm Re}(u_0)=-\frac{1}{2\pi b}\E[X_g\int_\Sigma \Delta_g{\rm Re}(u_0)X_g{\rm dv}_g]$ 
by using Cameron-Martin  to obtain (we write $\gamma_g^{u}[c+X_g]$ to emphasize the measurable dependance of $\gamma_g^{u}$ wrt $c+X_g$)
\begin{align*} 
& e^{-\frac{1}{4\pi b^2}\int_{\Sigma}|du_0|_g^2{\rm dv}_g}\int_{\R}\E[ e^{-\frac{1}{2\pi b}\int_\Sigma \Delta_g{\rm Re}(u_0)X_g{\rm dv}_g}F(b(c+X_g),\gamma^{u}_g[c+X_g])]\dd c\\ 
&= \int_{\R}\E[ F(b(c+X_g-\frac{1}{b}{\rm Re}(u_0)),\gamma_g^{u}[c+X_g-\frac{1}{b}{\rm Re}(u_0)]) ]\dd c\\
&=  \int_{\R}\E[ F(b(c+X_g-\frac{1}{b}{\rm Re}(u)),e^{-u}\gamma_g) ]\dd c
\end{align*}
where in the last line we did a change of variable $c\mapsto c+m_g(u)$ and we used  that, conditionally on $X_g$, the law of $\gamma_g^{u}[c+X_g-\frac{1}{b}{\rm Re}(u)]$ is equal to the law of $e^{-u}\gamma_g[c+X_g]=e^{-u}\gamma_g$. 
\end{proof}

To compare the definition of the path integral in the two gauges, let us express the random variable
$Y_u^{\theta}(\beta_{u,v})$  defined in \eqref{def_of_Y} with $\beta_{u,v}$ the form \eqref{beta_uv} in the gauge $h_0$ and 
using the operator $\bar{\pl}_{\mc{L}^{-2}_u}$ instead of $\bar{\pl}_{\mc{L}^{-2}}$ (in particular the basis 
$(e_p)_p$ of $\ker \bar{\pl}_{\mc{L}^{-2}}^*$ is replaced by the basis  $(e^{\bar{u}}e_p)_p$ of $\bar{\pl}_{\mc{L}^{-2}_u}^*$).
\begin{lemma}\label{lem:relationYbeta}
The following identity holds true almost surely 
\[Y_u^{\theta}(\beta_{u,v})=e^{-u}Y^{\theta_u}(\beta)+e^{-u}(v-P_0v),\]
where $P_0$ is the orthogonal projection on $\ker\bar{\pl}_{\mc{L}^{-2}}$.
\end{lemma}
\begin{proof}
Using \eqref{T0u^*}, we have
\begin{equation}\label{Ybeta_u}
\begin{split}
Y_u^{\theta}(\beta_{u,v})(x):=&\tilde{T}_u^*\beta_{u,v}(x) - \int_{\Sigma}\tilde{T}_u^*(x,y)(\hat{\Pi}_u^{\theta}\beta_{u,v})(y)d\theta(y)\\
=& e^{-u(x)}\Big(T_0^*(\beta+\bar{\pl}_{\mc{L}^{-2}}v)(x)-\int_{\Sigma}(T_0^*(x,y)e^{u(y)}(\hat{\Pi}_u^{\theta}\beta_{u,v})(y)d\theta(y)\Big)\\
& -e^{-u(x)}\Big((\Pi^{{\rm v}_{g_u}}T_0)^*(\beta+\bar{\pl}_{\mc{L}^{-2}}v)(x)-\int_{\Sigma}(\Pi^{{\rm v}_{g_u}}T_0)^*(x,y)e^{u(y)}(\hat{\Pi}_u^{\theta}\beta_{u,v})(y)d\theta(y)\Big)
\end{split}
\end{equation}
where $\hat{\Pi}^{\theta}_u$ is defined in \eqref{hatPi} but with the basis $(e_p)_p$ of $\ker \bar{\pl}_{\mc{L}^{-2}}^*$ replaced by the basis $e^{\bar{u}}e_p$ of $\ker \bar{\pl}_{\mc{L}^{-2}_u}^*$. Now, notice that $\bar{\pl}^*_{\mc{L}^{-2}}\Pi^{{\rm v}_{g_u}}=0$ so 
$(\Pi^{{\rm v}_{g_u}})^*\bar{\pl}_{\mc{L}^{-2}}v=0$, and by \eqref{dbar*barT_0}, $T_0^*\bar{\pl}_{\mc{L}^{-2}}v=v-P_0v$ where $P_0$ is the orthogonal projection on $\ker \bar{\pl}_{\mc{L}^{-2}}$. This gives 
\begin{equation}\label{First_term}
e^{-u}T_0^*(1-\Pi^{{\rm v}_{g_u}})^*(\bar{\pl}_{\mc{L}^{-2}}v)=e^{-u}(v-P_0v).
\end{equation}
Now we also observe, using that $\cjg \bar{\pl}_{\mc{L}^{-2}}v,e_q\cjd_2=0$ for each $q$, that
\begin{equation}\label{relationPihat}
\hat{\Pi}^{\theta}_u\beta_{u,v}=e^{\bar{u}}\hat{\Pi}^{\theta_u}\beta,
\end{equation}
and therefore 
\begin{equation}\label{Second_term}
\int_{\Sigma}T_0^*(x,y)e^{u(y)}(\hat{\Pi}_u^{\theta}\beta_{u,v})(y)d\theta(y)=\int_{\Sigma}T_0^*(x,y)(\hat{\Pi}^{\theta_u}\beta)(y)d\theta_u(y).
\end{equation}
This shows that the term in the second line of \eqref{Ybeta_u} is equal to
\[e^{-u(x)}(T_0^*(\beta)(x)+(v-P_0v)(x))-e^{-u(x)}\int_{\Sigma}T_0^*(x,y)(\hat{\Pi}^{\theta_u}\beta)(y)d\theta_u(y).\]
The integral kernel $(\Pi^{{\rm v}_{g_u}}T_0)^*(x,y)$ of $(\Pi^{{\rm v}_{g_u}}T_0)^*$ with respect to the ${\rm v}_g$ measure can be written as
\[(\Pi^{{\rm v}_{g_u}}T_0)^*(x,y)=\sum_{p,q}(\mc{G}^{{\rm v}_{g_u}})^{-1}_{qp}\int_{\Sigma}T_0^*(x,z)e_q(z){\rm dv}_g(z)\bbar{e_p}(y)\]
and we also have
\[\begin{split}
\cjg e^{2{\rm Re}(u)}\hat{\Pi}^{\theta_u}(\beta),e_p\cjd_{L^2(\theta_u)}=&\sum_{i,j}(\mc{G}^{\theta_u})^{-1}_{ij}\cjg e_i,e_p\cjd_{L^2(\theta_u)}
\cjg \beta,e_j\cjd_{2}= \cjg \beta,e_p\cjd_{2},
\end{split}\]
thus these two identities give (recall \eqref{relationPihat})
\[\int_{\Sigma}(\Pi^{{\rm v}_{g_u}}T_0)^*(x,y)e^{u(y)}(\hat{\Pi}_u^{\theta}\beta_{u,v})(y)d\theta(y)=
((\Pi^{{\rm v}_{g_u}}T_0)^*\beta)(x).\]
We conclude from this that the third line of \eqref{Ybeta_u} is equal to $0$, and therefore
\[\begin{split}
Y_u^{\theta}(\beta_{u,v})(x)&=e^{-u(x)}\Big(T_0^*(\beta)(x)+(v-P_0v)(x)-\int_{\Sigma}T_0^*(x,y)(\hat{\Pi}^{\theta_u}\beta)(y)d\theta_u(y)\Big)\\
&= (e^{-u}Y^{\theta_u}(\beta)+e^{-u}(v-P_0v))(x).\qedhere
\end{split}\]
\end{proof}
Using these results, we can deduce the following gauge invariance:
\begin{proposition}[\textbf{Gauge covariance}]\label{gauge_invariance}
For each $F: \mc{H}^{-s}(\Sigma)\times \mc{H}^{-s}(\Sigma;\mc{L}^{-2})\to \C$ continuous, if $s>0$ is fixed, the following identity holds for triangular gauge  changes 
\begin{equation}\label{gauge_change_formula} 
\langle  F(\phi,\gamma) \rangle_{\Sigma,g,\mc{L}_u,\beta_{u,v}}=e^{k\mc{W}(u,\beta,v)}\Big\langle  F(\phi-{\rm Re}(u),e^{-u}(\gamma+v)) \Big\rangle_{\Sigma,g,\mc{L},\beta}
\end{equation}
where the $\mc{W}(u,\beta,v)$ anomaly is defined in \eqref{Wubetav}.
\end{proposition}
\begin{proof} Let us consider the generic case, i.e. $\ker \bar{\pl}_{\mc{L}^{-2}}=0$. 
We want to express the path integral \eqref{def_PI_case_noker} for 
the splitting $E=\mc{L}_u^{-1}\oplus \mc{L}_u$ in terms of the path integral 
\eqref{def_PI_case_noker} for the splitting $E=\mc{L}^{-1}\oplus \mc{L}$.
First, if $\mc{G}_u^{\theta}$ is the Gram matrix defined in \eqref{def_GramG} but with the basis 
$(e^{\bar{u}}e_p)_p$ of $\ker \bar{\pl}_{\mc{L}_u^{-2}}^*$ instead of the basis $(e_p)_p$ of $\ker \bar{\pl}_{\mc{L}^{-2}}^*$ (which we take to be orthonormal for the $L^2({\rm v}_g)$ product \eqref{norme2bis}) and the measure $\theta=M_{2b}^g(X_g,\dd x)$, we readily see that 
\begin{equation}\label{relation_no1}
\mc{G}_u^{\theta}=\mc{G}^{\theta_u}.
\end{equation} 
We also have $\mc{K}_{\mc{L}_u}=\mc{K}_{\mc{L}}+\frac{1}{2\pi b^2}\Delta_g{\rm Re}(u){\rm v}_g$ and ${\rm deg}(\mc{L}_u)={\rm deg}(\mc{L})$, and by Lemma \ref{lem:detwitten}
\begin{equation}\label{relation_no2} 
\det(\tilde{\mc{D}}_u)=e^{-\frac{1}{2\pi}\int_{\Sigma}(|d {\rm Re}(u)|^2_g-\frac{1}{2}K_g{\rm Re}(u)){\rm dv}_g+\frac{2}{\pi i}\int_\Sigma {\rm Re}(u)F_{\mc{L}}}\det(\mc{D}_{0})\frac{\det(\mc{G}_u)}{\det(\mc{G}_0)}.
\end{equation}
Now, if $\hat{\Pi}^{\theta}_u$ is the operator \eqref{hatPi} associated to the basis $(e^{\bar{u}}e_p)_p$ rather than $(e_p)_p$, we can use \eqref{relationPihat} to get 
\begin{equation}\label{relation_no3}
\|\hat{\Pi}_u^{\theta}\beta\|^2_{L^2(\theta)}=\|\hat{\Pi}^{\theta_u}\beta\|_{L^2(\theta_u)}.
\end{equation}
Combining  \eqref{relation_no1}, \eqref{relation_no2} and \eqref{relation_no3} with Lemma \ref{lem:relation_laws} and Lemma \ref{lem:relationYbeta}, we obtain 
\[\begin{split}
& \langle  F(\phi,\gamma) \rangle_{\Sigma,g,\mc{L}_u,\beta_{u,v}}\\
& \,\, =\frac{(\frac{\pi}{k})^{\frac{\chi(\Sigma)}{6}+2{\rm deg}(\mc{L}_u)}\det{\mc{G}_u}}{\det(\tilde{\mc{D}}_u)}\Big(\frac{\det(\Delta_{g})}{{\rm v}_g(\Sigma)}\Big)^{-1/2}e^{\frac{k}{\pi}\|\beta_{u,v}\|_{2}^2}\\
& \qquad \times \int_{\R}e^{-2k{\rm deg}(\mc{L}_u)bc}\E\Big[\frac{F(b(c+X_g), \gamma_g^u+Y_u^{\theta}(\beta))}{\det(\mc{G}_u^{\theta})}e^{-b\int_\Sigma \mc{K}_{\mc{L}_u}X_g-\frac{k}{\pi}e^{-2bc} \|\hat{\Pi}_u^{\theta}\beta_{u,v}\|^2_{L^2(\theta_0)}}\Big]\,\dd c\\
 &\,\, = \frac{(\frac{\pi}{k})^{\frac{\chi(\Sigma)}{6}+2{\rm deg}(\mc{L})}\det(\mc{G}_0)}{\det(\mc{D}_0)}\Big(\frac{\det(\Delta_{g})}{{\rm v}_g(\Sigma)}\Big)^{-1/2}e^{\frac{k}{\pi}\|\beta_{u,v}\|_{2}^2}
e^{\frac{1}{2\pi}\int_{\Sigma}(|d{\rm Re}(u)|^2_g-\frac{1}{2}K_g{\rm Re}(u)){\rm dv}_g-\frac{2}{\pi i}\int_\Sigma {\rm Re}(u)F_{\mc{L}}}
\\
&\qquad \times  \int_{\R}e^{-2k{\rm deg}(\mc{L})bc}\E\Big[\frac{F(b(c+X_g), \gamma_g^u+e^{-u}(Y^{\theta_u}(\beta)+v))}{\det(\mc{G}^{\theta_u})e^{\frac{1}{2\pi b}\int_\Sigma \Delta{\rm Re}(u)X_g{\rm dv}_g}}e^{-b\int_\Sigma \mc{K}_{\mc{L}}X_g-\frac{k}{\pi}e^{-2bc} \|\hat{\Pi}^{\theta_u}\beta\|^2_{L^2(\theta_u)}}\Big]\,\dd c.
\end{split}\]
We can then use Lemma \ref{lem:relation_laws} to deduce \eqref{gauge_change_formula}. The proof is exactly the same in the non generic case 
$\ker \bar{\pl}_{\mc{L}^{-2}}\not=0$ except that one uses a translation in $\gamma_0$ in the last step to get rid of the $P_0v$ term coming from 
\ref{lem:relationYbeta}.
\end{proof}

\subsection{Correlation functions}
Here again we consider a rank $2$ holomorphic bundle $(E,\bar{\pl}_E,g_E)$ with Hermitian metric, and we represent it as an extension with parameters $(\mc{L},\beta)$ as in Section \ref{setup_E}. Define 
\begin{equation}\label{def_W_0_W_1} 
W_0:=\ker \bar{\pl}_{\mc{L}^{-2}}=H^0(\Sigma,\mc{L}^{-2}), \qquad W_1=\ker \bar{\pl}_{\mc{L}^{-2}}^*\simeq H^1(\Sigma,\mc{L}^{-2}).
\end{equation}
These two vector spaces are equipped with a Hermitian metric given by the $L^2$ product using the metric $g_{\mc{L}^{-2}}$ and the Riemannian metric $g$. 

Similarly to the case of the sphere, we consider the following observables, called vertex operators, formally given for $z\in \Sigma$ fixed by 
\begin{equation}\label{V_jmu_Sigma}  
V_{j,\mu}(z):=e^{2b(j+1)\phi(z)}e^{\mu.\gamma(z)-\bar\mu.\bar\gamma(z)}
\end{equation}
for $j\in\R$ and $\mu\in\mc{L}^2(z)$ is an element in the fiber of the bundle $\mc{L}^2$ at the point $z$. As for the sphere, we regularize the vertex by 
\begin{equation}\label{Vgeps}
V^{g, \epsilon}_{j,\mu}(z):=\epsilon^{2b^2(j+1)^2}e^{2b(j+1)(c+X_{g,\epsilon}(z))}e^{\mu.(\gamma_0(z)+\gamma_{g,\eps}(z))-\bar\mu.(\bar{\gamma}_0(z)+\bar\gamma_{g,\eps}(z))}
\end{equation}
where $(X_{g,\epsilon}, \gamma_{g,\epsilon})$ is a regularized  GFF-Witten pair as in Subsection  \ref{regpair}, $c\in \R$ and $\gamma_0\in W_0$ (recall that in the generic case $W_0=0$).  By Riemann-Roch theorem, each element $\gamma_0\in W_0$ has $-2{\rm deg}(\mc{L})$ zeros counted with multiplicity. 
Notice that $V^{g, \epsilon}_{j,\mu}(z)$ is a random variable that is bounded on all interval $c\in [-T,T]$. Let $\chi_T\in C_c^\infty(\R,[0,1])$ with support in $[-T-1,T+1]$ and equal to $1$ on $[-T,T]$, and consider $\chi_T(c+X_g):=\chi_T(c)$ as a functional of the field $\phi_g=c+X_g$. 
For any family of $m$ distinct points $z_1,\dots,z_m\in\Sigma$ and, for $\ell=1,\dots,m$, $j_\ell\in\R$ and $\mu_\ell\in\mc{L}^{2}_{z_\ell}$, the correlation functions of the $\mathbb{H}^3$-model are defined as the following limit as distributions in the $\boldsymbol{\mu}$ variable in $\mc{D}'(\mc{L}^{2}_{z_1}\times \dots \times \mc{L}^{2}_{z_m})$, if the limit exists, 
\begin{equation}\label{correl_Sigma}
\langle \prod_{\ell=1}^mV^{g}_{j_\ell,\mu_\ell}(z_\ell) \rangle_{\Sigma,g,\mc{L},\beta}:=\lim_{T\to \infty}\lim_{T'\to \infty}\lim_{\epsilon\to 0}\langle \prod_{\ell=1}^mV^{g, \epsilon}_{j_\ell,\mu_\ell}(z_\ell)\chi_T(c)\chi_{T'}(|\gamma_0|) \rangle_{\Sigma,g,\mc{L},\beta}.
\end{equation}
We will see that in the generic case (then there is no $\gamma_0$), the limit exists both pointwise outside $\boldsymbol{\mu}=0$ and as distribution in $\mc{L}^{2}_{z_1}\times \dots \times \mc{L}^{2}_{z_m}$.
The correlation functions will be expressed in terms of a particular section of $\mc{L}^{-2}\otimes \Lambda^{0,1}\Sigma$ defined by
 \begin{equation}\label{Gammae(la)}
\Gamma_{\boldsymbol{\mu},{\bf z}}:= T_0(\sum_{\ell=1}^m\bbar{\mu}_\ell \delta_{z_\ell})
 \end{equation}
and satisfying $\bar{\pl}^*_{\mc{L}^{-2}}\Gamma_{\boldsymbol{\mu},{\bf z}}=\sum_{\ell=1}^m\bar{\mu}_\ell \delta_{z_\ell}$. We will say that the weights ${\bf j}=(j_1,\dots, j_m)$ satisfy the admissibility bounds for the correlation function to be defined if the following holds:
\begin{equation}\label{seib_Sigma} 
\forall \ell=1,\dots,m,\quad j_\ell<-1/2\quad \text{ and } -k\, {\rm deg}(\mc{L})+\sum_{\ell}(j_{\ell}+1)>d=\dim H^0(\Sigma,\mc{L}^{-2}).
\end{equation}
We show the following existence result for the correlation functions:
\begin{proposition}[\textbf{Correlations, generic case}]\label{exists_correl_Sigma}
Consider a rank $2$ holomorphic bundle  $(E,\bar{\pl}_E,g_E)$ with a Hermitian metric on a compact Riemannian surface $(\Sigma,g)$, and we represent $E$ as an extension with parameters $(\mc{L},\beta)$ as described in Section \ref{setup_E}. Assume that $b=(k-2)^{-1/2}\in (0,1)$ and 
\[  \ker \bar{\pl}_{\mc{L}^{-2}}=H^0(\Sigma,\mc{L}^{-2})=0.\]
Let ${\bf z}=(z_1,\dots,z_m)$ be $m$ distinct marked points on $\Sigma$, $\boldsymbol{j}=(j_1,\dots,j_m)$ be weights satisfying the admissiblity 
bounds \eqref{seib_Sigma} and let $s_0:= -(k-2){\rm deg}(\mc{L})-\frac{\chi(\Sigma)}{2}+\sum_{\ell}(j_{\ell}+1)$.
For each  $\boldsymbol{\mu}=(\mu_1,\dots,\mu_m)\in \mc{L}^{2}_{z_1}\times \dots \times \mc{L}^{2}_{z_m}$ non zero, the limit \eqref{correl_Sigma} exists and is equal to
\[\begin{split} 
\Big\langle \prod_{\ell=1}^mV_{j_\ell,\mu_\ell}(z_\ell)\Big\rangle_{\Sigma,g,\mc{L},\beta}=&\frac{C_{k,\mc{L},\beta}(s_0)}{\det(\mc{D}_0)}\Big(\frac{\det(\Delta_{g})}{{\rm v}_g(\Sigma)}\Big)^{-\frac{1}{2}}e^{B_{\H^3}({\bf z},\boldsymbol{j})} \int_{W_1}e^{-2i{\rm Im}\cjg \Gamma_{\boldsymbol{\mu},{\bf z}}+ \omega,\beta \cjd_2} 
 \E\Big[ \|\Gamma_{\boldsymbol{\mu},{\bf z}}+\omega\|_{L^2(\theta_u)}^{-2s_0} \Big] \dd \omega.
\end{split}\]
with $\dd \omega$ the volume measure associated to the Hermitian metric given by the $L^2({\rm v}_g)$ pairing on $W_1=\ker \bar{\pl}^*_{\mc{L}^{-2}}\simeq H^1(\Sigma,\mc{L}^{-2})$, $\Gamma_{\boldsymbol{\mu},{\bf z}}$ is given by \eqref{Gammae(la)}, $u$ and $\theta_u$ are given by 
\begin{equation}\label{def_of_u}
u(x):= \sum_{\ell=1}^m 2b^2(j_\ell+1)G_{g}(x,z_\ell)-b^2 \int_\Sigma G_g(x,y)\mc{K}_{\mc{L}}(y), \qquad \theta_u:= e^{2u}M^g_{2b}(X_g,\dd x)
\end{equation}
 and the constants $C_{k,\mc{L}}(s_0)$ and $B({\bf z},\boldsymbol{j})$ are given by the expressions (recall \eqref{constant_Ck}, \eqref{defK_L} and \eqref{varYg})
\begin{equation}\label{Cklbeta} 
C_{k,\mc{L},\beta}(s_0)= \frac{C_k}{k^{2{\rm deg}(\mc{L})-\chi(\Sigma)/2}}\frac{\Gamma(s_0)}{2b}(\frac{\pi}{k})^{-s_0}e^{\frac{k}{\pi}\|\beta\|_{2}^2}
\end{equation}
\begin{equation}\label{Bzj}
\begin{split}
B_{\H^3}({\bf z},\boldsymbol{j}) :=&  2b^2\sum_{\ell=1}^m(j_\ell+1)^2W_g(z_\ell)+2b^2\sum_{\ell\not=\ell'}(j_\ell+1)(j_{\ell'}+1)G_g(z_\ell,z_{\ell'})\\
&+\frac{b^2}{2}\int_{\Sigma^2} G_g(x,y)\mc{K}_{\mc{L}}(x)\mc{K}_{\mc{L}}(y)
-2b^2\sum_{\ell=1}^m(j_\ell+1)\int_\Sigma G_g(x,z_\ell)\mc{K}_{\mc{L}}(x).
\end{split}
\end{equation}
Moreover, if in addition $\sum_{\ell=1}^mj_\ell<k\, {\rm deg}(\mc{L})$, then $\langle \prod_{\ell=1}^mV_{j_\ell,\mu_\ell}(z_\ell)\rangle_{\Sigma,g,\mc{L},\beta}$ extends as a distribution in the $\boldsymbol{\mu}$ variable in $\mc{L}^{2}_{z_1}\times \dots \times \mc{L}^{2}_{z_m}$ and the limit \eqref{correl_Sigma}  also exists in $\mc{D}'(\mc{L}^{2}_{z_1}\times \dots \times \mc{L}^{2}_{z_m})$.
\end{proposition}
\begin{proof} We take $\boldsymbol{\mu}\not=0$. Pick $(e_p)_{p=1,\dots,N}$ to be an northonormal basis of $W_1$ (then $\mc{G}_0={\rm Id}$). 
For $\mu_j\in\mc{L}^2(z_j)$, $\eps>0$ and $T>0$, and under the condition \eqref{seib_Sigma} we have 
(recall $-2N-\int_\Sigma \mc{K}_{\mc{L}}=-2k\, {\rm deg}(\mc{L})$ by \eqref{RR+GB})
\begin{equation}\label{corel_eps}
\langle \prod_{\ell=1}^mV^{g, \epsilon}_{j_\ell,\mu_\ell}(z_\ell)\chi_T \rangle_{\Sigma,g,\mc{L},\beta}=\frac{C_k}{\det(\mc{D}_0)}\Big(\frac{\det(\Delta_{g})}{{\rm v}_g(\Sigma)}\Big)^{-1/2}e^{\frac{k}{\pi}\|\beta\|_{2}^2} \int_{\R}\chi_T(c)e^{(-k{\rm deg}(\mc{L})+\sum_{\ell}(j_{\ell}+1))2bc}A_\eps(c,\boldsymbol{\mu},{\bf z})\,\dd c
\end{equation}
with $C_k$ the constant in \eqref{constant_Ck} and 
\begin{align*}
&A_\eps(c,\boldsymbol{\mu},{\bf z})\\
&:=\E\Big[\prod_{\ell=1}^m\epsilon^{2b^2(j_\ell+1)^2}e^{2b(j_\ell+1)(X_{g,\epsilon}(z_\ell))}e^{2i{\rm Im}(\mu_\ell (\gamma_{g,\eps}(z_\ell)+Y^{\theta}_\eps(\beta,z_\ell)))}e^{-b\int_\Sigma K_{\mc{L}} X_g}\frac{e^{-\frac{k}{\pi}e^{-2bc}\int_\Sigma |\hat{\Pi}^{\theta}\beta|^2_gM^g_{2b}(X_g,dx) }}{\det(\mc{G}^{\theta})}\Big].
\end{align*}
where $Y^{\theta}_\eps(\beta,x):=\cjg Y^{\theta}(\beta),\rho_{\eps,x}\cjd_2$,  $Y^{\theta}(\beta)$ being defined by \eqref{def_of_Y}.
First, by Lemma \ref{det_negative_moment} and by bounding the exponential term involving  the $\mu_\ell$ by $1$, we see that this expectation converges absolutely.
Then, we use Cameron-Martin, conditionally on $X_g$, to shift the Gaussian variable ${\rm Im}(\mu_\ell (\gamma_{g,\epsilon}(z_\ell)))$ and  the expectation above becomes 
\[\begin{split} 
A_\eps(c,\boldsymbol{\mu},{\bf z})=\E\Big[& \prod_{\ell=1}^m\epsilon^{2b^2(j_\ell+1)^2}e^{2b(j_\ell+1)(X_{g,\epsilon}(z_\ell))+2i{\rm Im}(\mu_\ell Y^{\theta}_\eps(\beta,z_\ell))}\\
&\times e^{-b\int_\Sigma K_{\mc{L}} X_g}e^{-e^{2bc}\frac{\pi}{k}\sum_{\ell,\ell'=1}^m\mu_\ell\bar\mu_{\ell'}R_\epsilon(z_\ell,z_{\ell'})}\frac{e^{-\frac{k}{\pi}e^{-2bc}\|\hat{\Pi}^{\theta}\beta\|^2_{L^2(\theta)}}}{\det(\mc{G}^{\theta})}\Big] 
 \end{split}\]
where $\frac{\pi}{k}R_\epsilon(z,z'):=e^{-2bc}\E[\gamma_{g,\epsilon}(z)\overline{\gamma_{g,\epsilon}}(z')|X_g]$. Let us compute, using \eqref{relRphi},
\[ \sum_{\ell,\ell'=1}^m\mu_\ell\bar\mu_{\ell'}R_\epsilon(z_\ell,z_{\ell'})= \|(1-\Pi^{\theta})T_0\sum_{\ell}\bbar{\mu}_\ell \rho_{\epsilon,z_\ell}\|^2_{L^2(\theta)}.\]
Let us denote the $\mc{L}^{-2}$ valued $1$-forms  
\[
\Gamma_{\boldsymbol{\mu},{\bf z}}:=T_0(\sum_{\ell=1}^m\bbar{\mu}_\ell \delta_{z_\ell}), \quad \Gamma^\eps_{\boldsymbol{\mu},{\bf z}} :=T_0\sum_{\ell=1}^m\bbar{\mu}_\ell \rho_{\epsilon,z_\ell}.
\]
The form $\Gamma^\eps_{\boldsymbol{\mu},{\bf z}} $ converges as $\eps\to 0$ to $\Gamma_{\boldsymbol{\mu},{\bf z}} $ which, by Lemma \ref{inverseDphi}, is smooth outside the points $z_\ell$ and is bounded near $z_j$ by $|\Gamma_{\boldsymbol{\mu},{\bf z}} |_{g\otimes g_{\mc{L}^{-2}}}\leq Cd_g(x,z_\ell)^{-1}$. 
Next we use Cameron-Martin to the exponential term  $e^{\sum_\ell 2b(j_\ell+1)X_{g,\eps}(z_\ell)-b\int_\Sigma K_{\mc{L}}X_g}$: as $\eps \to 0$   we get
\[\begin{split}
A_\eps(c,\boldsymbol{\mu},{\bf z}) =&(1+o(1))e^{B({\bf z},\boldsymbol{j})}\E\Big[\frac{e^{-e^{2bc}\frac{\pi}{k} \|(1-\Pi^{\theta_{u_\eps}})\Gamma^\eps_{\boldsymbol{\mu},{\bf z}} \|^2_{L^2(\theta_{u_\eps})}-\frac{k}{\pi}e^{-2bc}\|\hat{\Pi}^{\theta_{u_\eps}}\beta\|^2_{L^2(\theta_{u_\eps})}+2i\sum_\ell {\rm Im}(\mu_\ell Y^{\theta_{u_\eps}}_\eps(\beta,z_\ell))}}{\det(\mc{G}^{\theta_{u_\eps}})}\Big]
\end{split}\]
with $W_g$ defined in \eqref{varYg},  the function $u_\epsilon$ and measure $\theta_{u_\eps}$ are given by 
\begin{equation}\label{ueps}
u_\epsilon(x):=\sum_{\ell=1}^m 2b^2(j_\ell+1)G_{g,\epsilon}(x,z_\ell)-b^2 \int_\Sigma G_{g,\eps,\eps}(x,y)\mc{K}_{\mc{L}}(y), \qquad \theta_{u_\eps}:=e^{2u_\eps}\theta
\end{equation}
with $G_{g,\eps}$ (resp. $G_{g,\eps,\eps}$) the $g$-regularisation of the Green function in its second variable (resp. both variables),
and finally $B({\bf z},\boldsymbol{j})$ is the deterministic function \eqref{Bzj} (coming from the variance in Cameron-Martin).
The function $u_\epsilon$ converges to $u_0$
 expressed as in \eqref{ueps} but with $G_{g,\eps}(x,z_\ell)$ replaced by $G_{g}(x,z_\ell)$. 
The singularity of  $e^{2u(x)}|\Gamma_{\boldsymbol{\mu},{\bf z}}|_{g\otimes g_{\mc{L}^{-2}}}^2$ at $z_\ell$ is of the  form $d_g(z_\ell,x)^{-2-4b^2(j_\ell+1)}$, which is integrable with respect to the measure $\theta=M^{g}_{2b}(X_g  ,\dd x)$ provided that $2+4b^2(j_\ell+1)<2+2b^2$, namely $j_\ell<-1/2$. Setting $\theta_u:=e^{2u}\theta$, we have by Lemma \ref{det_negative_moment} 
that $\E[(\det(\mc{G}^{\theta_u}))^{-1}]<\infty$. 
The dominated convergence can then be used to deduce that 
\begin{equation}\label{limAeps}
\begin{split}
\lim_{\eps\to 0}A_\eps(c,\boldsymbol{\mu},{\bf z}) = e^{B_{\H^3}({\bf z},\boldsymbol{j})} \E\Big[\frac{e^{-e^{2bc}\frac{\pi}{k} \|(1-\Pi^{\theta_u})\Gamma_{\boldsymbol{\mu},{\bf z}}\|^2_{L^2(\theta'_0)}-\frac{k}{\pi}e^{-2bc}\|\hat{\Pi}^{\theta_u}\beta\|^2_{L^2(\theta'_0)}+2i\sum_\ell {\rm Im}(\mu_\ell Y^{\theta_u}(\beta,z_\ell))}}{\det(\mc{G}^{\theta_u})}\Big]
\end{split}
\end{equation}
with $\theta_u=e^{2u}\theta$ the limit of $\theta_{u_\eps}$ as $\eps\to0$. We next rewrite (recall \eqref{def_proj}, \eqref{def_of_Y}  and \eqref{hatPi})
\begin{align*}
&-e^{2bc}\frac{\pi}{k} \|(1-\Pi^{\theta_u})\Gamma_{\boldsymbol{\mu},{\bf z}}\|^2_{L^2(\theta_u)}-\frac{k}{\pi}e^{-2bc}\|\hat{\Pi}^{\theta_u}\beta\|^2_{L^2(\theta_u)}+2i\sum_{\ell=1}^m{\rm Im}(\mu_\ell Y^{\theta_u}(\beta,z_\ell))\\
& =  -\sum_{p,q=1}^N (\mc{G}^{\theta_u})^{-1}_{pq}\Big( \frac{k}{\pi}e^{-2bc}\beta_{p}\bbar{\beta_{q}}-\frac{\pi}{k}e^{2bc} \Gamma^p_{\boldsymbol{\mu},{\bf z}} \bbar{\Gamma^q_{\boldsymbol{\mu},{\bf z}}}+\beta_q\bbar{\Gamma^p_{\boldsymbol{\mu},{\bf z}}}-\bbar{\beta_p}\Gamma^q_{\boldsymbol{\mu},{\bf z}}\Big)\\
 &\quad  -e^{2bc}\frac{\pi}{k} \|\Gamma_{\boldsymbol{\mu},{\bf z}}\|^2_{L^2(\theta)}+2i\sum_{\ell=1}^m {\rm Im}\cjg \beta,\Gamma_{\boldsymbol{\mu},{\bf z}}\cjd_2
 \end{align*}
where 
\[\beta_p:=\cjg \beta,e_p\cjd_2,\quad  \Gamma^p_{\boldsymbol{\mu},{\bf z}}:= \cjg \Gamma_{\boldsymbol{\mu},{\bf z}}, e_p\cjd_{L^2(\theta_u)}.\]
If $A$ is a Hermitian matrix in $\C^N$ and $u,v\in \C^N$, we recall the expression 
\[ \int_{\C^N} e^{-\cjg A\la,\la\cjd +i\sum_{p=1}^N((\bar{u}_p-\bar{v}_p)\lambda_p+(u_p+v_p)\bar{\la}_p)}d\la=\pi^N\frac{
e^{-\cjg A^{-1}u,u\cjd+\cjg A^{-1}v,v\cjd+\sum_{p,q} (A^{-1})_{pq}(u_q\bar{v}_p-\bar{u}_pv_q)}}{\det(A)}\]
and use this to  write (with $u=-\beta e^{-bc}\sqrt{k/\pi}$ and $v=\Gamma_{\boldsymbol{\mu},{\bf z}}e^{bc}\sqrt{\pi/k}$)
\begin{equation}\label{FourierT}
\begin{gathered}
 \frac{1}{\det(\mc{G}^{\theta_u})}e^{-\sum_{p,q=1}^N (\mc{G}^{\theta_u})^{-1}_{pq}\Big( \frac{k}{\pi}e^{-2bc}\beta_{p}\bbar{\beta_{q}}-\frac{\pi}{k}e^{2bc} \Gamma^p_{\boldsymbol{\mu},{\bf z}} \bbar{\Gamma^q_{\boldsymbol{\mu},{\bf z}}}+\beta_q\bbar{\Gamma^p_{\boldsymbol{\mu},{\bf z}}}-\bbar{\beta_p}\Gamma^q_{\boldsymbol{\mu},{\bf z}}\Big)}\\
=k^{-N} \int_{\C^N}e^{-\frac{\pi}{k}\cjg \mc{G}^{\theta_u}\la,\la\cjd}e^{i\frac{\pi}{k}\sum_p (-\frac{k}{\pi}e^{-bc}\bbar{\beta_p}
-e^{bc}\bbar{\Gamma^p_{\boldsymbol{\mu},{\bf z}}})\la_p+(-\frac{k}{\pi}e^{-bc}\beta_p
+e^{bc}\Gamma^p_{\boldsymbol{\mu},{\bf z}})\bbar{\la}_p}\dd \la
\end{gathered}
\end{equation}
where the integral is absolutely convergent $\mathbb{P}$ almost surely. 
Notice also that, if $e(\la):=i\sum_{p=1}^N\la_pe_p$,
\begin{equation}\begin{gathered}\label{halfsquare}
\cjg \mc{G}^{\theta_u}\la,\la\cjd =  \|\sum_{p}i\la_pe_p\|^2_{L^2(\theta_u)}=  \|e(\la)\|^2_{L^2(\theta_u)} ,
\\
  ie^{bc}\sum_p( \Gamma^p_{\boldsymbol{\mu},{\bf z}}\bbar{\la}_p
-\bbar{\Gamma^p_{\boldsymbol{\mu},{\bf z}}}\la_p) =
-2{\rm Re}\cjg e^{bc}\Gamma_{\boldsymbol{\mu},{\bf z}},e(\la)\cjd_{L^2(\theta_u)},
\\
 -i\sum_p( \bbar{\beta}_p\la_p+\beta_p\bbar{\la}_p
) = -2i{\rm Im}\cjg  e(\lambda),\beta\cjd_{2}
\end{gathered}
\end{equation}
Combining \eqref{limAeps}, \eqref{FourierT} and \eqref{halfsquare}, we thus obtain 
\begin{align*}
\lim_{\eps\to 0}A_\eps(c,\boldsymbol{\mu},{\bf z})
=& 
e^{B_{\H^3}({\bf z},\boldsymbol{j})}k^{-N} \int_{\C^N} \E\Big[  \exp\Big(-\frac{\pi}{k}\|e(\la)\|^2_{L^2(\theta_u)} -2i e^{-bc}{\rm Im}\cjg  e(\lambda),\beta\cjd_{2}-e^{2bc}\frac{\pi}{k} \|\Gamma_{\boldsymbol{\mu},{\bf z}}\|^2_{L^2(\theta)}
\\
& \qquad\qquad\qquad\qquad\qquad\quad  +2i\sum_{\ell=1}^m {\rm Im}\cjg \beta,\Gamma_{\boldsymbol{\mu},{\bf z}}\cjd_2 -2 \frac{\pi}{k} {\rm Re}\cjg e^{bc}\Gamma_{\boldsymbol{\mu},{\bf z}}, e(\la)\cjd_{L^2(\theta_u)}\Big)\Big]  d\la
\\  
=&
 k^{-N} e^{B_{\H^3}({\bf z},\boldsymbol{j})+2Nbc} \int_{\C^N}e^{-2i{\rm Im}\cjg \Gamma_{\boldsymbol{\mu},{\bf z}}+e(\la),\beta \cjd_2
} \E\Big[e^{-\frac{\pi}{k} e^{2bc}\|\Gamma_{\boldsymbol{\mu},{\bf z}}+e(\la)\|^2_{L^2(\theta_u)}}\Big]d\la
\end{align*}
where the second equality was obtained by packing the terms and performing the change of variables $\lambda\to e^{bc}\lambda$.
Here we notice that this integral in $ \dd \la \, \dd \mathbb{P}$ is absolutely convergent: indeed, 
it suffices to show the convergence when $\beta=0$, in which case the integrals are positive and convergent since given by  
\eqref{limAeps} with $\beta=0$.
Therefore, we input this expression into \eqref{corel_eps} and use Riemann-Roch \eqref{Riemann_RochL^k} (which gives $N-k\,{\rm deg}(\mc{L})=-\tfrac{{\rm deg}(\mc{L})}{b^2}-\frac{\chi(\Sigma)}{2}$) to obtain 
\begin{equation}\label{Correl_final}
\begin{split} 
 \Big\langle \prod_{\ell=1}^mV^{g, \epsilon}_{j_\ell,\mu_\ell}(z_\ell)& \chi_T \Big\rangle_{\Sigma,g}=\frac{C_k}{k^N\det(\mc{D}_0)}\Big(\frac{\det(\Delta_{g})}{{\rm v}_g(\Sigma)}\Big)^{-1/2}   e^{\frac{k}{\pi}\|\beta\|^2_{2}}e^{B_{\H^3}({\bf z},\boldsymbol{j})}\\
& \times \int_{\C^N}e^{-2i{\rm Im}\cjg \Gamma_{\boldsymbol{\mu},{\bf z}}+e(\la),\beta \cjd_2
}\int_{\R} \E \Big[ \chi_T(c)e^{(-\frac{{\rm deg}(\mc{L})}{b^2}-\frac{\chi(\Sigma)}{2}+\sum_{\ell}(j_{\ell}+1))2bc} e^{-\frac{\pi}{k} e^{2bc}\|\Gamma_{\boldsymbol{\mu},{\bf z}}+e(\la)\|^2_{L^2(\theta_u)}}\Big]
\,\dd c \dd\la.
\end{split}\end{equation}
We claim that the $\dd c \dd \la\, \dd \mathbb{P}$ integral has a finite limit as $T\to \infty$ if $s_0:=-\frac{{\rm deg}(\mc{L})}{b^2}-\frac{\chi(\Sigma)}{2}+\sum_{\ell}(j_{\ell}+1)>N$. To prove this we can use dominated convergence: we bound $|e^{-2i{\rm Im}\cjg \Gamma_{\boldsymbol{\mu},{\bf z}}+e(\la),\beta \cjd_2}\chi_T(c)|\leq 1$ and check that the integral 
with $e^{-2i{\rm Im}\cjg \Gamma_{\boldsymbol{\mu},{\bf z}}+e(\la),\beta \cjd_2}\chi_T(c)$ replaced by $1$ converges.
To prove this, we make the change variables $y=\frac{\pi}{k} e^{2bc}\|\Gamma_{\boldsymbol{\mu},{\bf z}}+e(\la)\|^2_{L^2(\theta_u)}$ in the $\dd c$ integral, one obtains the formula 
\begin{equation}\label{change_var}
 \int_{\R}\E\Big[e^{(-\frac{{\rm deg}(\mc{L})}{b^2}-\frac{\chi(\Sigma)}{2}+\sum_{\ell}(j_{\ell}+1))2bc} e^{-\frac{\pi}{k} e^{2bc}\|\Gamma_{\boldsymbol{\mu},{\bf z}}+e(\la)\|^2_{L^2(\theta_u)}}\Big]\,\dd c=k^{-N}\Big(\frac{k}{\pi}\Big)^{s_0}\frac{\Gamma(s_0)}{2b} 
 \E\Big[ \|\Gamma_{\boldsymbol{\mu},{\bf z}}+e(\la)\|_{L^2(\theta_u)}^{-2s_0} \Big].
\end{equation}
As in the proof of Proposition \ref{existcorrel}, the expectation is finite if $s_0>0$. Now, as in the proof of Lemma \ref{det_negative_moment} (see the appendix), there is a constant $C>0$ and an open ball $B\subset \Sigma$ such that for all $|(\boldsymbol{\mu},\la)|=1$,
\[
 \E\Big[ \|\Gamma_{\boldsymbol{\mu},{\bf z}}+e(\la)\|_{L^2(\theta_u)}^{-2s_0} \Big]\leq C\E[M_{2b}^g(X_g+\frac{u}{b},B)^{-2s_0}]<\infty
\]
which implies that for the same $C>0$
\[ \E\Big[ \|\Gamma_{\boldsymbol{\mu},{\bf z}}+e(\la)\|_{L^2(\theta_u)}^{-2s_0} \Big]\leq C |\boldsymbol{\mu}|^{-2s_0}\big(1+\frac{|\la|}{|\boldsymbol{\mu}|}\big)^{-2s_0}.\]
Integrating in $\la$ after a change of variables $\la\to |\boldsymbol{\mu}|$ gives, assuming $s_0>N$,
\[ \int_{\C^N}\E\Big[ \|\Gamma_{\boldsymbol{\mu},{\bf z}}+e(\la)\|_{L^2(\theta_u)}^{-2s_0} \Big]\dd \lambda\leq C |\boldsymbol{\mu}|^{-2(s_0-N)}.\]
This completes the proof of the pointwise convergence. Notice that (recall $N=2\, {\rm deg}(\mc{L})-\chi(\Sigma)/2$) 
\[s_0-N=-k\, {\rm deg}(\mc{L})+\sum_{\ell=1}^m(j_\ell+1) <m/2-k\,{\rm deg}(\mc{L}).\]
The existence of the correlation as a distribution in $\boldsymbol{\mu}$ (including at $\boldsymbol{\mu}=0$) is exactly as in the proof of Proposition  \ref{existcorrel},but requires that $|\boldsymbol{\mu}|^{-2(s_0-N)}$ is integrable in the unit ball $\{|\boldsymbol{\mu}|\leq 1\}$, i.e. that $s_0-N<m$. Note that these conditions are equivalent to asking $\sum_{\ell}j_{\ell}<k\, {\rm deg}(\mc{L})<m+\sum_{\ell=1}^mj_\ell$
\end{proof}
In the next Proposition, we describe the correlation in the non-generic case where $W_0\not=0$. We first need to introduce the definition of a certain measure. Fix $(f_1,\dots,f_d)$ an orthonormal basis of $W_0=\ker \bar{\pl}_{\mc{L}^{-2}}$,  $V_{\bf z}:=\mc{L}^{2}_{z_1}\times \dots \times \mc{L}^{2}_{z_m}\simeq \C^m$ and let $M_{\bf z}$ be the linear map
\begin{equation}\label{def_Mz} 
M_{{\bf z}}: \boldsymbol{\mu}=(\mu_1,\dots,\mu_m)\in V_{\bf z} \mapsto (\sum_{\ell=1}^m\mu_\ell f_1(z_\ell),\dots,\sum_{\ell=1}^m\mu_\ell f_d(z_\ell))\in \C^d.
\end{equation}
We then let $\delta_{V_{\bf z}^0}$ denote the measure on $V_{\bf z}:=\mc{L}^{2}_{z_1}\times \dots \times \mc{L}^{2}_{z_m}\simeq \C^m$,  supported on 
\begin{equation}\label{Vz0}
V_{\bf z}^0:=\ker M_{\bf z}
\end{equation}  
and defined by 
\begin{equation}\label{defVz}
\forall F\in C_c^\infty(V_{\bf z}),\quad  \cjg \delta_{V^0_{\bf z}},F\cjd=(2\pi)^d|\det (M_{\bf z}|_{(V^0_{\bf z})^\perp})|^{-1}\int_{V_{\bf z}^0}F(\boldsymbol{\mu}^0,0)\dd \boldsymbol{\mu}^0
\end{equation}
where $\dd \boldsymbol{\mu}^0$ is the Lebesgue measure on $V_{\bf z}^0$ in the decomposition $V_{\bf z}=V_{\bf z}^0\oplus (V_{\bf z}^0)^\perp$.
Notice that $V_{\bf z}^0$ does not depend on the basis $f_1,\dots,f_d$ of $W_0$: it is the kernel of the map 
$\boldsymbol{\mu}\mapsto \sum_{\ell=1}^m\mu_\ell \delta_{z_\ell}|_{W_0}$ where the right hand side is considered as an element in the dual $W_0^*$ of $W_0$. 

\begin{proposition}[\textbf{Correlations, non-generic case}]\label{exists_correl_Sigma_NG}
Consider a rank $2$ holomorphic bundle  $(E,\bar{\pl}_E,g_E)$ with a Hermitian metric on a compact Riemannian surface $(\Sigma,g)$, and we represent $E$ as an extension with parameters $(\mc{L},\beta)$ as described in Section \ref{setup_E}. Assume that $b=(k-2)^{-1/2}\in (0,1)$ and 
\[  W_0=H^0(\Sigma,\mc{L}^{-2})\not=0.\]
Let ${\bf z}=(z_1,\dots,z_m)$ be $m$ distinct marked points on $\Sigma$, $\boldsymbol{j}=(j_1,\dots,j_m)$ be weights satisfying the admissiblity 
bounds \eqref{seib_Sigma}.
The limit \eqref{correl_Sigma} exists as a distribution in $\boldsymbol{\mu}$, i.e. in $\mc{D}'((\mc{L}^{2}_{z_1}\times \dots \times \mc{L}^{2}_{z_m})\setminus\{0\})$, and is equal to
\[
\Big\langle \prod_{\ell=1}^mV_{j_\ell,\mu_\ell}(z_\ell)\Big\rangle_{\Sigma,g,\mc{L},\beta}=\delta_{V_{\bf z}^0}\frac{C_{k,\mc{L},\beta}(s_0)}{\det(\mc{D}_0)}\Big(\frac{\det(\Delta_{g})}{{\rm v}_g(\Sigma)}\Big)^{-\frac{1}{2}}e^{B_{\H^3}({\bf z},\boldsymbol{j})} \int_{W_1}e^{-2i{\rm Im}\cjg \Gamma_{\boldsymbol{\mu},{\bf z}}+ \omega,\beta \cjd_2} 
 \E\Big[ \|\Gamma_{\boldsymbol{\mu},{\bf z}}+\omega\|_{L^2(\theta_u)}^{-2s_0} \Big] \dd \omega  
\]
where we used the same notations as in Proposition \ref{exists_correl_Sigma} and $\delta_{V_{\bf z}^0}$ is the measure introduced in \eqref{defVz}. If in addition $\sum_{\ell=1}^mj_\ell<k\, {\rm deg}(\mc{L})$, the limit holds in $\mc{D}'(\mc{L}^{2}_{z_1}\times \dots \times \mc{L}^{2}_{z_m})$.
\end{proposition}
\begin{proof}
Let $(f_1,\dots,f_d)$ be an orthonormal basis of $W_0=\ker \bar{\pl}_{\mc{L}^{-2}}$. We proceed similarly as the proof of Proposition \ref{exists_correl_Sigma} except that we smear the correlation functions against a test function $\eta\in C_c^\infty(\mc{L}^{2}_{z_1}\times \dots \times \mc{L}^{2}_{z_m})$ in the variable $\boldsymbol{\mu}$, so that 
  \eqref{corel_eps} gets replaced by 
\[\begin{split}
& \int  \langle \prod_{\ell=1}^mV^{g, \epsilon}_{j_\ell,\mu_\ell}(z_\ell) \chi_T(c)\chi_{T'}(|\gamma_0|) \rangle_{\Sigma,g,\mc{L},\beta} \eta(\boldsymbol{\mu})\dd \boldsymbol{\mu}
\\
&=\frac{C_k}{\det(\mc{D}_0)}\Big(\frac{\det(\Delta_{g})}{{\rm v}_g(\Sigma)}\Big)^{-1/2}e^{\frac{k}{\pi}\|\beta\|_{2}^2}\\
& \times   \int_{\C^d}\chi_{T'}(|\kappa|)\Big( \int e^{i{\rm Im}(\sum_{p=1}^d \kappa_p\sum_{\ell=1}^m\mu_\ell.f_p(z_\ell))} \int_{\R}\chi_T(c)e^{(-k{\rm deg}(\mc{L})-d+\sum_{\ell}(j_{\ell}+1))2bc}A_\eps(c,\boldsymbol{\mu},{\bf z})\,\dd c\,  \eta(\boldsymbol{\mu})\dd \boldsymbol{\mu}\Big)  \dd \kappa 
\end{split}\]
where $\dd \boldsymbol{\mu} $ is the Lebesgue measure on $V_{\bf z}:=\mc{L}^{2}_{z_1}\times \dots  \mc{L}^{2}_{z_m}\simeq \C^{m}$ and $\dd\kappa$ the Lebesgue measure on $\C^d$ identified to $W_0$ by the map $\kappa\in \C^d\mapsto \gamma_0=\sum_{p=1}^d\kappa_p f_p$. 
This expression makes sense for $\epsilon,T,T'$ fixed, as the $ \dd \kappa$ integral of  a Fourier transform in $\boldsymbol{\mu}$, which we are now going to make explicit. We apply the same strategy as in the proof of Proposition \ref{existcorrel}.
Let     $M_{\bf z}$ be the linear map \eqref{def_Mz}
and decompose $V_{\bf z}=\ker M_{\bf z}\oplus (\ker M_{\bf z})^\perp=:V_{\bf z}^0\oplus (V_{\bf z}^0)^\perp$ using the metric on $V_{\bf z}$ induced by $g_{\mc{L}^{-2}}$.
We compute that, in the distributional sense in $\mc{D}'(\mc{L}^{2}_{z_1}\times \dots \times \mc{L}^{2}_{z_m})$,
\[\int_{\C^d}e^{i{\rm Im}(\sum_{p=1}^d \kappa_p\sum_{\ell=1}^m\mu_\ell.f_p(z_\ell))}\dd \kappa= 
\delta_{V^0_{\bf z}}.\]
where $  \delta_{V^0_{\bf z}}$ is the measure on $\C^m$ defined by $ \delta_{V^0_{\bf z}}(F) =(2\pi)^d|\det (M_{\bf z}|_{(V^0_{\bf z})^\perp})|^{-1}\int_{V_{\bf z}^0}F(\boldsymbol{\mu}^0,0)\dd \boldsymbol{\mu}^0$ in the decomposition $V_{\bf z}=V_{\bf z}^0\oplus (V_{\bf z}^0)^\perp$.
We can then argue as to get \eqref{Correl_final} to obtain in the distributional sense
\begin{multline} \label{case2identity}
\lim_{T'\to \infty}\lim_{\eps\to 0}\langle \prod_{\ell=1}^m V^{g, \epsilon}_{j_\ell,\mu_\ell}(z_\ell) \chi_T \rangle_{\Sigma,g}=\frac{C_k}{k^N\det(\mc{D}_0)}\Big(\frac{\det(\Delta_{g})}{{\rm v}_g(\Sigma)}\Big)^{-1/2}e^{\frac{k}{\pi}\|\beta\|^2_{2}}  e^{B_{\H^3}({\bf z},\boldsymbol{j})} 
 \\
 \times\Big( \int_{\C^N}e^{-2i{\rm Im}\cjg \Gamma_{\boldsymbol{\mu},{\bf z}}+e(\la),\beta \cjd_2
}\int_{\R} \E \Big[ \chi_T(c)e^{(-\frac{{\rm deg}(\mc{L})}{b^2}-\frac{\chi(\Sigma)}{2}+\sum_{\ell}(j_{\ell}+1))2bc} e^{-\frac{\pi}{k} e^{2bc}\|\Gamma_{\boldsymbol{\mu},{\bf z}}+e(\la)\|^2_{L^2(\theta_u)}}\Big]
\,\dd c\Big) \delta_{V^0_{\bf z}}
\end{multline} 
and when $s_0>N$ the limit $T\to \infty$ converges and using the same change of variables as the one used to get \eqref{change_var}, we obtain the desired result. Here $N=2{\rm deg}(\mc{L})-\chi(\Sigma)/2+d$ by Riemann-Roch, and $s_0>N$ is equivalent to the second admissibility bound \eqref{seib_Sigma}.
\end{proof}

We then  use Proposition  \ref{prop:weyl_surface} to obtain the following:
\begin{lemma}\label{confweightgeneral}
We consider the setup of Proposition  \ref{exists_correl_Sigma}. 
Let $g'=e^\omega g$ be a metric conformal to $g$ on $\Sigma$ for some smooth $\omega\in C^\infty(\Sigma)$. Then
\[\Big\langle \prod_{\ell=1}^mV^{g'}_{j_\ell,\mu_\ell}(z_\ell)\Big\rangle_{\Sigma,g',\mc{L},\beta} =  e^{ \frac{{\bf c}(k)}{96\pi}S^0_L(g,\omega)-\sum_{\ell=1}^m \triangle_{j_\ell}\omega(z_\ell)}\Big\langle\prod_{\ell=1}^mV^{g}_{j_\ell,\mu_\ell}(z_\ell) \Big\rangle_{\Sigma,g}\]
where ${\bf c}(k)=3+6/(k-2)$ is   the central charge,   $S^0_{\rm L}$ is the Liouville action \eqref{LiouvilleS0}, and $\triangle_{j_\ell} $ is the  conformal weight\footnote{Note that the conformal weight does not depend on the $\mu_\ell$'s.} of $V^{g}_{j_\ell,\mu_\ell} $, which is  given by
\begin{equation}\label{CWH3general}
\triangle_{j_\ell} =-\frac{j_\ell(j_\ell+1)}{k-2}.
\end{equation}
\end{lemma}

\begin{proof}  First note that changing the metric regularisation yields (using Landau notation)
\[
\Big\langle\prod_{\ell=1}^mV^{g', \epsilon}_{j_\ell,\mu_\ell}(z_\ell)\chi_T \Big\rangle_{\Sigma,g',\mc{L},\beta}= (1+o_\epsilon(1))\prod_{\ell=1}^me^{b^2(j_\ell+1)^2\omega(z_\ell)}\Big\langle\prod_{\ell=1}^mV^{g, \epsilon}_{j_\ell,\mu_\ell}(z_\ell) \chi_T\Big\rangle_{\Sigma,g',\mc{L},\beta}.
\]
We can then use Proposition  \ref{prop:weyl_surface}     to get
\[\Big\langle\prod_{\ell=1}^mV^{g, \epsilon}_{j_\ell,\mu_\ell}(z_\ell) \Big\rangle_{\Sigma,g',\mc{L},\beta}=(1+o_\epsilon(1))  e^{ \frac{{\bf c}(k)}{96\pi}S^0_{\rm L}(g,\omega)}\prod_{\ell=1}^me^{b^2j_\ell(j_\ell+1)\omega(z_\ell)}\Big\langle\prod_{\ell=1}^mV^{g, \epsilon}_{j_\ell,\mu_\ell}(z_\ell) \Big\rangle_{\Sigma,g,\mc{L},\beta}.\]
Combining and passing to the limit  as $\epsilon\to 0$, then $T\to \infty$,   we complete the proof, using $b=(k-2)^{-1/2}$.
\end{proof}

\section{Correspondence with Liouville theory}
In this section, we relate the $\mathbb{H}^3$-WZW correlation functions to Liouville correlation functions. In Proposition 
\ref{exists_correl_Sigma}, we have seen that the correlations are expressed in terms of negative moments of  $\|\Gamma_{\boldsymbol{\mu},{\bf z}}+\omega\|^2_{L^2(\theta_u)}$ where $u$ is expressed in terms of Green's function $G_g$ and $\omega\in W_1=\ker \bar{\pl}^*_{\mc{L}^{-2}}\simeq H^1(\Sigma;\mc{L}^{-2})$. 
In order to recover some Liouville correlations we have to express the pointwise squared norm $|\Gamma_{\boldsymbol{\mu},{\bf z}}+\omega|^2$ in terms of Green's function $G_g$. 
We shall show that this is indeed the case: $\Gamma_{\boldsymbol{\mu},{\bf z}}^*+\omega^*$ is a meromorphic section of $\mc{L}^{-2}\otimes \Lambda^{1,0}\Sigma$, with poles at $z_i$ and with zeros $y_\ell(\omega)$ (with multiplicity $n_\ell(\omega)$) that will play 
the role of new insertions in the Liouville CFT side with degenerate weights $V_{-n_\ell(\omega)/b}(y_\ell(\omega))$.\\

We still consider the geometric setup of Section \ref{setup_E}: $(E,\bar{\pl}_E,g_E)$ is represented as an extension with parameters 
$(\mc{L},\beta)$ and we let $W_0:=\ker \bar{\pl}_{\mc{L}^{-2}}=H^0(\Sigma,\mc{L}^{-2})$ and $W_1=\ker \bar{\pl}^*_{\mc{L}^{-2}}\simeq H^1(\Sigma;\mc{L}^{-2})$. 
\begin{lemma}[\textbf{New insertions}]\label{expression_norm_Gamma}
Let ${\bf z}=(z_1,\dots,z_m)$ be $m$ distinct points on $\Sigma$, and 
$\boldsymbol{\mu}=(\mu_1,\dots,\mu_m)\in V_{{\bf z}}^0$ be fixed, with  $V_{{\bf z}}^0$ defined in \eqref{Vz0}.
Consider $\Gamma_{\boldsymbol{\mu},{\bf z}}$ the section of $\mc{L}^{-2}\otimes \Lambda^{0,1}\Sigma$ defined in \eqref{Gammae(la)} satisfying 
 \[\bar{\pl}^*_{\mc{L}^{-2}}\Gamma_{\boldsymbol{\mu},{\bf z}}=\sum_{\ell=1}^m\bar{\mu}_\ell \delta_{z_\ell}.\]
 1) For each $\omega\in W_1$, using the duality map \eqref{dualitys^*}, $\Gamma_{\boldsymbol{\mu},{\bf z}}^*+\omega^*$ is a meromorphic section of 
$\mc{L}^{-2}\otimes \Lambda^{1,0}\Sigma$, with zeros called \textbf{new insertion} and denoted by $y_\ell(\omega)$ and multiplicity $n_\ell(\omega)$ for $\ell=1,\dots,m(\omega)$ such that
\begin{equation}\label{nb_of_zeros}
m-\sum_{\ell=1}^{m(\omega)} n_\ell(\omega) -\chi(\Sigma)+2{\rm deg}(\mc{L})=0.
\end{equation}
2) The pointwise squared norm   $| \Gamma_{\boldsymbol{\mu},{\bf z}}+\omega|^2:=| \Gamma_{\boldsymbol{\mu},{\bf z}}+\omega|^2_{g_{\mc{L}^{-2}}\otimes g}$ with respect to  $g_{\mc{L}^{-2}}\otimes g$ satisfies 
\begin{equation}\label{log|Gamma|} 
\begin{split}
\log | \Gamma_{\boldsymbol{\mu},{\bf z}}+\omega|_{\mc{L}^{-2}\otimes \Lambda^{0,1}\Sigma}^2=&C(\omega)+2(\sum_{\ell=1}^mG_g(\cdot,z_j)
 -\sum_{\ell=1}^{m(\omega)} n_\ell(\omega) G_g(\cdot,y_\ell(\omega))) \\
& -\frac{1}{2\pi}\int_\Sigma G_g(\cdot,y)K_g(y){\rm dv}_g(y)-\frac{2}{\pi i}\int_\Sigma G_g(\cdot,y)F_{\mc{L}}(y)
\end{split}\end{equation}
where $C(\omega)$ is the constant defined by 
\begin{equation}\label{defC(la)}
C(\omega):= \int_\Sigma \log | \Gamma_{\boldsymbol{\mu},{\bf z}}+\omega|^2{\rm dv}_g.
\end{equation}
\end{lemma}
\begin{proof}
We first consider the pointwise squared norm $ | \Gamma_{\boldsymbol{\mu},{\bf z}}+\omega|^2$ and rewrite it in terms of Green's function of the Laplacian.  Using the maps \eqref{dualitys^*} and the relation \eqref{dbar^*-dbar}, we remark that when $W_0=0$ or when  $\boldsymbol{\mu}\in V^0_{{\bf z}}$, then 
$\Gamma_{\boldsymbol{\mu},{\bf z}}^*$ is a meromorphic 
section of $\mc{L}^{2}\otimes \Lambda^{1,0}\Sigma$ since, as a $2$-current, we have for $(f_j)_j$ and orthonormal basis of $W_0$
\[ \bar{\pl}_{\mc{L}^{2}\otimes \Lambda^{1,0}\Sigma}\Gamma_{\boldsymbol{\mu},{\bf z}}^*=-( \bar{\pl}_{\mc{L}^{-2}}^*\Gamma_{\boldsymbol{\mu},{\bf z}}^*)^*=-\Big(\sum_{\ell=1}^m\mu_\ell \delta_{z_\ell}{\rm v}_g- \sum_{j=1}^d\sum_{\ell=1}^m\mu_\ell f_j(z_\ell)\bar{f}_j{\rm v}_g\Big)=-\sum_{\ell=1}^m\mu_\ell \delta_{z_\ell}{\rm v}_g\] 
where the last inequality follows from the condition $\boldsymbol{\mu}\in V^0_{{\bf z}}$.

If $s$ is a local holomorphic trivialization of $\mc{L}^{2}\otimes \Lambda^{1,0}\Sigma$ in a small open set $U\subset \Sigma$, 
the curvature of the holomorphic line bundle $\mc{L}^{2}\otimes \Lambda^{1,0}\Sigma$ can then be expressed by the formula (see \cite[eq 2.4.1]{Bost_Bourbaki})
\[ \bar{\pl}\pl \log|s|^2= {\rm Tr}(F_{\mc{L}^{2}\otimes \Lambda^{1,0}\Sigma})=2F_{\mc{L}}+\frac{i}{2}K_g{\rm v}_g\]
where $F_{\mc{L}^{2}\otimes \Lambda^{1,0}\Sigma}$ denotes the curvature $(1,1)$-form of $\mc{L}^{2}\otimes \Lambda^{1,0}\Sigma$ and $\frac{i}{2}K_g{\rm v}_g$ is the curvature of the canonical bundle $\Lambda^{1,0}\Sigma$. Since $|\Gamma_{\boldsymbol{\mu},{\bf z}}^*+\omega^*|_{g_{\mc{L}^2}\otimes g}=|\Gamma_{\boldsymbol{\mu},{\bf z}}+\omega|_{g_{\mc{L}^2}\otimes g}$, this means that outside the zeros and poles of $\Gamma_{\boldsymbol{\mu},{\bf z}}^*+\omega^*$, we have 
\[\Delta_g\log (| \Gamma_{\boldsymbol{\mu},{\bf z}}+\omega|^2){\rm v}_g=-K_g{\rm v}_g+4iF_{\mc{L}}.\]
Now, we consider a neighborhood $U$ of a pole $z_j$ (resp. a zero $x_j$) of  $(\Gamma_{\boldsymbol{\mu},{\bf z}}+\omega)^*$ and a local complex coordinate $z$ centered at $z_j$ (resp. $x_j$). We can write $(\Gamma_{\boldsymbol{\mu},{\bf z}}+\omega)^*(z)=(z-z_j)^{-1}s(z)$ in $U$
for some local holomorphic non-vanishing section $s$ of $\mc{L}^{2}\otimes \Lambda^{1,0}\Sigma$ by Lemma \ref{inverseDphi}. Therefore we have in $U$
\[\begin{split}
\Delta_g\log (| \Gamma_{\boldsymbol{\mu},{\bf z}}+\omega|^2){\rm v}_g=&(-\Delta_g\log |z-z_j|^2+\Delta_g\log |s|^2){\rm v}_g\\
 =& 4\pi \delta_{z_j} -K_g{\rm v}_g+4iF_{\mc{L}}.\end{split}
 \]
In the case of a zero $y_j$ or order $n_j$, the same argument gives the same result but with $4\pi \delta_{z_j}$ replaced by $-4\pi n_j\delta_{y_j}$, and 
we obtain in the distributional sense 
\[ \Delta_g \log | \Gamma_{\boldsymbol{\mu},{\bf z}}+\omega|^2 {\rm v}_g= 4\pi (\sum_{\ell=1}^m \delta_{z_\ell} -\sum_{\ell=1}^{p} n_\ell \delta_{y_\ell})-K_g{\rm v}_g+4iF_{\mc{L}}.\]
This implies that the right hand side integrates to $0$ on $\Sigma$ and integrating this identity against the Green function  gives \eqref{nb_of_zeros} (this is also Riemann-Roch theorem).
\end{proof}
Notice that the zeros $y_\ell(\omega)$ and their multiplicity also depends on $({\bf z},\boldsymbol{\mu})$. 

We can now write the correspondence between $\mathbb{H}^3$-WZW correlation functions and Liouville correlation functions \eqref{muL=1}, 
with $\mu_L=1$. 
\begin{theorem}[\textbf{WZW--Liouville correspondence, generic case}]\label{Th_corresp_G}
Let  $(E,\bar{\pl}_E,g_E)$ be a holomorphic vector bundle of rank $2$ and trivial determinant, equipped with a Hermitian metric $g_E$ on a closed Riemannian surface $(\Sigma,g)$. Consider $E$ as an extension with parameters $(\mc{L},\beta)$ as described in Section \ref{setup_E} and assume that
\[  \ker \bar{\pl}_{\mc{L}^{-2}}=H^0(\Sigma,\mc{L}^{-2})=0, \qquad b=(k-2)^{-\frac{1}{2}}\in (0,1)\]
Let ${\bf z}=(z_1,\dots,z_m)$ be $m$ distinct marked points on $\Sigma$, $\boldsymbol{j}=(j_1,\dots,j_m)$ be weights satisfying the admissiblity 
bounds \eqref{seib_Sigma}.
For $\boldsymbol{\mu}=(\mu_1,\dots,\mu_m)\in \mc{L}^{2}_{z_1}\times \dots \times \mc{L}^{2}_{z_m}$ the following WZW-Liouville correspondence holds
\begin{equation}\label{identitycorrespondence}
\Big\langle \prod_{\ell=1}^mV_{j_\ell,\mu_\ell}(z_\ell)  \Big\rangle^{\mathbb{H}^3}_{\Sigma,g,\mc{L},\beta}=
\int_{W_1}e^{-2i{\rm Im}\cjg \Gamma_{\boldsymbol{\mu},{\bf z}}+\omega,\beta \cjd_2
-s_0 C(\omega)}\mc{F}_{{\bf z},\boldsymbol{j}}({\bf y}(\omega))\, \Big\cjg \prod_{\ell=1}^m V_{\alpha_\ell}(z_\ell)\prod_{\ell=1}^{m(\omega)}V_{-\frac{n_\ell(\omega)}{b}}(y_\ell(\omega))\Big\cjd_{\Sigma,g}^{\rm L} \dd\omega
\end{equation}
with $W_1=\ker \bar{\pl}_{\mc{L}^{-2}}^*\simeq H^{1}(\Sigma,\mc{L}^{-2})$, $\dd \omega$ is the volume measure associated to the $L^2$-product \eqref{norme2bis}, with $\alpha_\ell:=2b(j_\ell+1)+\frac{1}{b}$ and ${\bf y}(\omega)=(y_1(\omega),\dots,y_{m(\omega)}(\omega))$ 
is the set of new insertions defined by Lemma \ref{expression_norm_Gamma}, i.e. the zeros of $\Gamma_{\boldsymbol{\mu},{\bf z}}+\omega$, with multiplicity ${\bf n}(\omega)=(n_1(\omega),\dots,n_{m(\omega)}(\omega))$, and for each family of disjoint points
${\bf y}=(y_1,\dots,y_{p})$ with multiplicities ${\bf n}=(n_1,\dots,n_{p})$,  
the function $\mc{F}_{{\bf z},\boldsymbol{j}}({\bf y})$ is given by 
\begin{equation}\label{formula_F_lastTh}
\mc{F}_{{\bf z},\boldsymbol{j}}({\bf y})=\frac{2b\, C_{k,\mc{L},\beta}(s_0)}{\det(\mc{D}_0)\Gamma(2bs_0)}e^{B_{\H^3}({\bf z},\boldsymbol{j})-B_{\rm L}({\bf z},{\bf y},\boldsymbol{\alpha},-{\bf n}/b)}
\end{equation}
where $\boldsymbol{\alpha}=(\alpha_1,\dots,\alpha_m)$ and the constants $C_{k,\mc{L},\beta}(s_0)$, $C(\omega)$, $B_{\H^3}$, $B_{\rm L}$ are defined in \eqref{Cklbeta}, \eqref{defC(la)}, \eqref{Bzj} and \eqref{BL(x,alpha)}.
\end{theorem}
\begin{proof}  Thanks to Lemma \ref{confweightgeneral} and \eqref{LCFT_anomaly}, it suffices to consider the case where the curvature 
$K_g$ is constant, and to simplify exposition we assume so. 
By  \eqref{log|Gamma|}, we observe that for $u$ given by \eqref{def_of_u}
\[  \E\Big[ \|\Gamma_{\boldsymbol{\mu},{\bf z}}+\omega\|_{L^2(\theta_u)}^{-2s_0} \Big]= e^{-s_0 C(\omega)}\E\Big[ M_{2b}^g(X_g+u_{\boldsymbol{\alpha}},\Sigma)^{-s_0} \Big]\]
where the function $u_{\boldsymbol{\alpha}}$ is given by\footnote{Recall that $\int_\Sigma G_g(x,x')K_g(x'){\rm v}_g(\dd x')=0$ if $K_g$ is constant, and we notice that the $F_{\mc{L}}$ term coming from $|\Gamma_{\boldsymbol{\mu},{\bf z}}+\omega|^2$ in \eqref{log|Gamma|} cancels that appearing in $e^{2u}\theta$ from  \eqref{def_of_u}.}
\[ u_{\boldsymbol{\alpha}}(x):= \sum_{\ell=1}^m \alpha_\ell G_g(x,z_\ell)-\sum_{\ell=1}^{m(\omega)}\frac{n_\ell(\omega)}{b}G_g(x,y_\ell(\omega)) \]
with $\alpha_\ell:=2b(j_\ell+1)+\frac{1}{b}$.
Combining this with Proposition \ref{exists_correl_Sigma}, we obtain
\[ \begin{split}
\Big\langle \prod_{\ell=1}^mV_{j_\ell,\mu_\ell}(z_\ell) \Big\rangle^{\H^3}_{\Sigma,g,\mc{L},\beta}=& \frac{C_{k,\mc{L},\beta}(s_0)}{\det(\mc{D}_0)}\Big(\frac{\det(\Delta_{g})}{{\rm v}_g(\Sigma)}\Big)^{-\frac{1}{2}}e^{B_{\H^3}({\bf z},\boldsymbol{j})}\\
& \times \int_{W_1}e^{-2i{\rm Im}\cjg \Gamma_{\boldsymbol{\mu},{\bf z}}+\omega,\beta \cjd_2
-s_0 C(\omega)}\E\Big[ M_{2b}^g(X_g+u_{\boldsymbol{\alpha}},\Sigma)^{-s_0} \Big]\ \dd\omega
\end{split}
\]
Using the definition of $s_0$ and \eqref{nb_of_zeros}, we also have (with $Q:= b+\frac{1}{b}$)
\[ 2s_0=\frac{1}{b^2}\Big(m-\sum_{\ell}n_\ell(\omega)-\chi(\Sigma)\Big)-\chi(\Sigma)+2\sum_{\ell=1}^m(j_\ell+1) =\frac{1}{b}\Big(\sum_{\ell=1}^m\alpha_\ell-\sum_{\ell=1}^{m(\omega)}\frac{n_\ell(\omega)}{b}-Q\chi(\Sigma)\Big)=\frac{s(\boldsymbol{\alpha}')}{b}\]
where $s(\boldsymbol{\alpha}')$ is the exponent in the Liouville correlation function \eqref{zidaneilamarque} with 
$\boldsymbol{\alpha}'=\boldsymbol{\alpha}'(\omega)=(\boldsymbol{\alpha},-{\bf n}(\omega)/b)$.
Comparing the expression in the Liouville correlation function \eqref{zidaneilamarque}, one obtains
\begin{align*}
& \Big\langle \prod_{\ell=1}^mV_{j_\ell,\mu_\ell}(z_\ell) \Big\rangle^{\H^3}_{\Sigma,g,\mc{L},\beta}\\
&=\frac{2b\, C_{k,\mc{L},\beta}e^{B_{\H^3}({\bf z},\boldsymbol{j})}}{\det(\mc{D}_0)\Gamma(2bs_0)}\int_{W_1}e^{-2i{\rm Im}\cjg \Gamma_{\boldsymbol{\mu},{\bf z}}+\omega,\beta \cjd_2
-s_0 C(\omega)-B_{\rm L}({\bf x}(\omega),\boldsymbol{\alpha}'(\omega))}\Big\langle \prod_{\ell=1}^{n(\omega)}V_{\alpha_\ell}(x_\ell(\omega)) \Big\rangle_{\Sigma,g}^{{\rm L}} \dd \omega\end{align*}
with ${\bf x}(\omega)=({\bf z}, {\bf y}(\omega))$ the union of the WZW insertions and of the $m(\omega)$ new insertions (here $n(\omega)=m+m(\omega)$) given by Lemma \ref{expression_norm_Gamma}. The Liouville correlations in the right hand side have parameter $\mu_{\rm L}=1$.
\end{proof}
Now, the same proof as above combined with Proposition \ref{exists_correl_Sigma_NG} gives the following:
\begin{theorem}[\textbf{WZW--Liouville correspondence, non-generic case}]\label{Th_corresp_NG}
Let  $(E,\bar{\pl}_E,g_E)$ be a holomorphic vector bundle of rank $2$ and trivial determinant, equipped with a Hermitian metric $g_E$ on a closed Riemannian surface $(\Sigma,g)$. Consider $E$ as an extension with parameters $(\mc{L},\beta)$ as described in Section \ref{setup_E} and assume that
\[  \ker \bar{\pl}_{\mc{L}^{-2}}=H^0(\Sigma,\mc{L}^{-2})\not=0, \qquad b=(k-2)^{-\frac{1}{2}}\in (0,1)\]
Let ${\bf z}=(z_1,\dots,z_m)$ be $m$ distinct marked points on $\Sigma$, $\boldsymbol{j}=(j_1,\dots,j_m)$ be weights satisfying the admissiblity 
bounds \eqref{seib_Sigma}.
The following WZW-Liouville correspondence holds in the distribution sense in $\mc{D}'(\mc{L}^{2}_{z_1}\times \dots \times \mc{L}^{2}_{z_m})$ in the variable $\boldsymbol{\mu}=(\mu_1,\dots,\mu_m)$
\begin{equation}\label{identitycorrespondence2}
\Big\langle \prod_{\ell=1}^mV_{j_\ell,\mu_\ell}(z_\ell)  \Big\rangle^{\mathbb{H}^3}_{\Sigma,g,\mc{L},\beta}=\delta_{V_{\bf z}^0}
\int_{W_1}e^{-2i{\rm Im}\cjg \Gamma_{\boldsymbol{\mu},{\bf z}}+\omega,\beta \cjd_2
-s_0 C(\omega)}\mc{F}_{{\bf z},\boldsymbol{j}}({\bf y}(\omega))\, \Big\cjg \prod_{\ell=1}^m V_{\alpha_\ell}(z_\ell)\prod_{\ell=1}^{m(\omega)}V_{-\frac{n_\ell(\omega)}{b}}(y_\ell(\omega))\Big\cjd_{\Sigma,g}^{\rm L} \dd\omega
\end{equation}
using the same notations as in Theorem \ref{Th_corresp_NG} and $\delta_{V_{\bf z}^0}$ is the measure supported on $V_{\bf z}^0$ defined in \eqref{defVz}.
\end{theorem}

\appendix
\section{Auxiliary results}
In this section, we gather the proofs of several auxiliary lemmas stated in Section \ref{H3surface}.

\subsection{Proof of Lemma   \ref{det_negative_moment}}
We first note that, if $N=\dim \ker \bar{\pl}^*_{\mc{L}^{-2}}$, 
\[ \det(\mc{G}^{\theta_u})> \Big(\inf_{|\la|=1} \int_{\Sigma} |\sum_{p}\la_pe_p(x)|^2_{g}d\theta_u\Big)^N.\]
For each $\la^0$ with $|\la^0|=1$, the function  $F(\la,x):=|\sum_{p}\la_pe_p(x)|^2_{g}$ is bounded below by some $c(\la^0)>0$ for all $x\in B(\la^0),\la\in V(\la_0)$
where $B(\la^0) \subset \Sigma$ is an open ball and $V(\la^0)$ a neighborhood of $\la^0$. By compactness, we then get that there is a ball $B\subset \Sigma$ and $C>0$, both independent of $X_g$, such that 
\[ \det(\mc{G}^{\theta_u})> C M_{2b}(X_g+\tfrac{u}{b},B)^N.\]
Since $M_{2b}(X_g+\frac{u}{b},B)$ has negative moments of all order by \cite[Th. 2.12]{rhodes2014_gmcReview}, the desired result follows.

\subsection{Proof of Lemma   \ref{lemmaYtheta}}
 
 Let us focus on the term  $$ \cjg \int_{\Sigma}T_0^*(\cdot,y)(\hat{\Pi}^{\theta}\beta)(y)d\theta(y),f\cjd= \int_{\Sigma}T_0^*f(y)(\hat{\Pi}^{\theta}\beta)(y)d\theta(y).$$ The first statement then follows from the continuity of the map  $T_0^*:\mc{H}^{s}(\Sigma,\R)\to\mc{H}^{s+1}(\Sigma,\R)$. 

Next we prove the continuity of $f\in \mc{H}^{s+1}(\Sigma,\R)\mapsto \theta_0'(f)$. We claim that, almost surely, 
\begin{equation}\label{GMCest}
\int_\Sigma\int_\Sigma G_g(x,y)\theta_0'(\dd x)\theta_0'(\dd y)<+\infty.
\end{equation}
This will proved thereafter. This implies that, if $(\lambda_n)_n$ and $(e_n)_n$ are respectively the (positive) eigenvalues and eigenfunctions of the Laplacian $\Delta_g$,  
\begin{align*}
\sum_n \lambda_n^{-s-1}|\theta_0'(e_n)|^2\leq \sum_n \lambda_n^{-1}\iint_{\Sigma^2}e_n(x)e_n(y)\theta_0'(\dd x)\theta_0'(\dd y)=\iint_{\Sigma^2} G_g(x,y)\theta_0'(\dd x)\theta_0'(\dd y)<+\infty.
\end{align*}
And this ensures the continuity of the map $f\in \mc{H}^{s+1}(\Sigma,\R)\mapsto \theta_0'(f)$, for $s>0$, by Cauchy-Schwartz.

Now we prove the statement \eqref{GMCest}, and actually a stronger version.
\begin{lemma}\label{lem:GMC}
Let $b\in (0,1)$, $\alpha\in (0,Q)$ with $Q=b+\frac{1}{b}$, and $\beta,p>0$. If the following conditions hold
$$\frac{\beta}{2}+\alpha<Q,\quad \text{ and }\quad p<\frac{1}{2b}(Q-\frac{\beta}{2}-\alpha)\wedge 1 \wedge \frac{1}{2b^2}\wedge \frac{1}{4b}(Q-\frac{\beta}{2})+ \frac{1}{4b}\sqrt{(Q-\frac{\beta}{2})^2-4}$$
then
$$\E\Big[\Big(\iint_{\Sigma^2}e^{2b \alpha G_g(z_0,x)+2b \alpha G_g(z_0,x')+2b\beta G_g(x,x')}M_{2b}^g(X_g,{\rm v}_g,\dd x)M_{2b}^g(X_g,{\rm v}_g,\dd x')\Big)^p\Big]<+\infty.$$
\end{lemma}

\begin{proof} By Kahane's convexity inequalities (see \cite{rhodes2014_gmcReview}), it is enough to prove
$$\E\Big[\Big(\iint_{\Sigma^2}|x|^{-2b\alpha}|x'|^{-2b \alpha}|x-x'|^{-2b\beta} M_{2b}(\dd x )M_{2b}(\dd x' )\Big)^p\Big]<+\infty$$
with $\Sigma=[0,1]^2$ and $M_{2b}^g(X_g,{\rm v}_g,\dd x)$  replaced 
 by 
 $$M_{2b}(\dd x):=\lim_{\epsilon\to 0}M^\epsilon_{2b}(\dd x)\quad \text{ and }\quad M^\epsilon_{2b}(\dd x):= e^{2b X_\epsilon-2b^2\E[X_\epsilon(x)^2]}\,\dd x$$
  where   $X$ is  the log-correlated field  with covariance $\E[X(x)X(y)]=\ln_{+} \frac{1}{|x-y|}$  and $X_\epsilon$ is the translation invariant  cut-off approximation  to $X$ such that for all $\lambda\in (0,1)$, $(X_{\lambda \epsilon} (\lambda x ))_{|x| \leq 1} \overset{(Law)}{=} (X_{\epsilon} (x ))_{|x| \leq 1}+\Omega_\lambda $ where $\Omega_\lambda$ is an independent centered Gaussian variable with variance $\ln \frac{1}{\lambda}$ (see \cite{rhodes2014_gmcReview} for instance). Using this cut-off approximation, it is proved in \cite[Lemma 3.10]{DKRV16}  that   for all $\alpha\in (0,Q)$ and  $0<p < \frac{1}{b}(Q-\alpha)\wedge 1$
\begin{equation}\label{estGMC1}
\sup_{\epsilon>0}\E\Big[\Big(\int_{\Sigma}(|x|+\epsilon)^{-2b\alpha}  M^\epsilon_{2b}(\dd x )\Big)^p\Big]<+\infty.
\end{equation}
Also, we make the following claims. Set $r_+:=  \frac{1}{2b}(Q-\frac{\beta}{2})+ \frac{1}{2b}\sqrt{(Q-\frac{\beta}{2})^2-4}$, then
\begin{itemize}
\item for $0<\frac{\beta}{2}<Q-2$ and $0<p<\frac{1}{2b}(Q-\frac{\beta}{2})\wedge 1\wedge \frac{1}{2b^2} \wedge \frac{r_+}{2}$,
\begin{equation}\label{estGMC2}
\sup_{\epsilon>0}\E\Big[\Big(\iint_{\Sigma^2}(|x-x'|+\epsilon)^{-2b\beta}  M^\epsilon_{2b}(\dd x )M^\epsilon_{2b}(\dd x ')\Big)^p\Big]<+\infty.
\end{equation}
\item for $0<\frac{\beta}{2}<Q-2$, $\alpha\in (0,Q)$ and $0<p<\frac{1}{2b}(Q-\frac{\beta+\alpha}{2})\wedge  \frac{1}{2b}(Q-\alpha)\wedge 1\wedge \frac{1}{2b^2}\wedge     \frac{r_+}{2}$,
\begin{equation}\label{estGMC3}
\sup_{\epsilon>0}\E\Big[\Big(\iint_{\Sigma^2}|x|^{-2b\alpha} (|x-x'|+\epsilon)^{-2b\beta}  M^\epsilon_{2b}(\dd x )M^\epsilon_{2b}(\dd x ')\Big)^p\Big]<+\infty.
\end{equation}
\end{itemize}
Note that the range of $p$ is non empty in all cases. We postpone for now the proof of these two claims.  Let us introduce the quantity, for arbitrary measurable sets $A,B\subset [0,1]^2$,  
 $$E_\epsilon(A,B):=\E\Big[\Big(\iint_{A\times B}(|x|+\epsilon)^{-2b\alpha}(|x'|+\epsilon)^{-2b\alpha}(|x-x'|+\epsilon)^{-2b\beta} M^\epsilon_{2b}(\dd x )M^\epsilon_{2b}(\dd x' )\Big)^p\Big].$$
We split the square $\Sigma$ into 2 pieces $\Sigma_1:=[0,\frac{1}{2}]^2$, $\Sigma_2:=[0,1]^2\setminus [0,\frac{1}{2}]^2$. Next, because $p<1$, we argue
 $$E_\epsilon(\Sigma^2)\leq \sum_{i,j=1}^2E_\epsilon(\Sigma_i,\Sigma_j).$$
In this sum,   the term  $E_\epsilon(\Sigma_2,\Sigma_2)$ is bounded independently of $\epsilon$ using \eqref{estGMC2}, and same conclusion for the terms  $E_\epsilon(\Sigma_1,\Sigma_2)$ and  $E_\epsilon(\Sigma_2,\Sigma_1)$ using \eqref{estGMC3}. Now we use the scaling relation with $\lambda=1/2$  to estimate $E_\epsilon(\Sigma_1,\Sigma_1)$:
\begin{align*}
E_\epsilon(\Sigma_1,\Sigma_1)=&2^{4b\alpha p+2b\beta p-4p}\E[e^{4bp\Omega_{1/2}-4b^2p\E[\Omega_{1/2}^2]}]
\\
&\E\Big[\Big(\iint_{\Sigma^2}(|x|+2\epsilon)^{-2b\alpha}(|x'|+2\epsilon)^{-2b\alpha}(|x-x'|+2\epsilon)^{-2b\beta} M^{2\epsilon}_{2b}(\dd x )M^{2\epsilon}_{2b}(\dd x' )\Big)^p\Big]
\\
=&2^{-f_{\alpha+\beta/2}(2p)} E_{2\epsilon}(\Sigma^2)
\end{align*}
where we have set $f_\theta(p):=2b(Q-\theta)p-2b^2p^2$. Therefore, if $Q-\alpha-\frac{\beta}{2}>0$ and $p\in (0,\frac{1}{2b}(Q-\alpha-\frac{\beta}{2}))$, we have $f_{\alpha+\beta/2}(2p)>0$. We deduce easily that $\sup_nE_{2^{-n} }(\Sigma^2)<+\infty$ and, using Kahane's convexity inequality, that $\sup_\epsilon E_{\epsilon }(\Sigma^2)<+\infty$.

 Now we prove \eqref{estGMC2}. For this, we first claim that
 \begin{itemize}
\item for $0<\frac{\beta}{2}<Q$ and $0<p<\frac{1}{2b}(Q-\frac{\beta}{2})\wedge 1\wedge \frac{1}{2b^2}$,
\begin{equation}\label{estGMC4}
\sup_{\epsilon>0}\E\Big[\Big(\iint_{[0,\frac{1}{2}]^2\times [\frac{1}{2},1]^2}(|x-x'|+\epsilon)^{-2b\beta}  M^\epsilon_{2b}(\dd x )M^\epsilon_{2b}(\dd x ')\Big)^p\Big]<+\infty.
\end{equation}
\item for $0<\frac{\beta}{2}<Q-\sqrt{2}$ and $0<p<\frac{1}{2b}(Q-\frac{\beta}{2})\wedge 1\wedge \frac{1}{2b^2}\wedge  \frac{r_+}{2}$,
\begin{equation}\label{estGMC5}
\sup_{\epsilon>0}\E\Big[\Big(\iint_{[0,\frac{1}{2}]^2\times ([0,\frac{1}{2}]\times [\frac{1}{2},1])} (|x-x'|+\epsilon)^{-2b\beta}  M^\epsilon_{2b}(\dd x )M^\epsilon_{2b}(\dd x ')\Big)^p\Big]<+\infty.
\end{equation}
\end{itemize}
These two claims are proven similarly as before. We introduce the quantity 
$$E'_\epsilon(A,B):=\E\Big[\Big(\iint_{A\times B} (|x-x'|+\epsilon)^{-2b\beta} M^\epsilon_{2b}(\dd x )M^\epsilon_{2b}(\dd x' )\Big)^p\Big].$$
In the first case, we split the square $\Sigma_1=[0,\frac{1}{2}]^2$ as the union of $\Sigma'_1:=[0,\frac{1}{4}]^2$  and $\Sigma'_2:=\Sigma_1\setminus \Sigma'_1$, and the square   $\Sigma_3:=[\frac{1}{2},1]^2$ as the union of $\Sigma'_3:=[\frac{1}{2},\frac{3}{4}]^2$  and $\Sigma'_4:=\Sigma_3\setminus \Sigma'_3$. The same argument then produces the bound $E'_\epsilon(\Sigma_1,\Sigma_3)\leq  2^{-f_{\beta/2}(2p)} E'_{2\epsilon}(\Sigma_1,\Sigma_3)+C $. And we conclude as before.

In the second case, we split both squares $\Sigma_1$ and $\Sigma_5:=[0,\frac{1}{2}]\times [\frac{1}{2},1]$ into 4 equally sized squares so as to get the relation (using \eqref{estGMC4} in the process and translational invariance to get the additional factor $2$)
 $$E'_\epsilon(\Sigma_1,\Sigma_5)\leq 2^{1-f_{\beta/2}(2p)}E'_{2\epsilon}(\Sigma_1,\Sigma_5)+C.$$
The polynomial $f_{\beta/2}(p)-1$ has a positive discriminant   if $0<\frac{\beta}{2}<Q-\sqrt{2}$, in which  case the roots are $\frac{1}{2b}(Q-\frac{\beta}{2})\pm \frac{1}{2b}\sqrt{(Q-\frac{\beta}{2})^2-2}$, which means that, when $2p$ is located strictly between these roots, $f_{\beta/2}(2p)-1>0$. For such $p$, we deduce that \eqref{estGMC5} holds.

We can now prove \eqref{estGMC2}. The same scaling argument, dividing $\Sigma$ into 4 equally sized squares, produces $E_\epsilon(\Sigma,\Sigma)\leq 2^{2-f_{\beta/2}(2p)}E_{2\epsilon}(\Sigma,\Sigma)+C$, using  \eqref{estGMC4} and \eqref{estGMC5} in the process, which leads to our claim given the fact that the polynomial $f_{\beta/2}(p)-2$ has positive discriminant for  $0<\frac{\beta}{2}<Q-2$ with  roots  $r_{\pm}:=\frac{1}{2b}(Q-\frac{\beta}{2})\pm \frac{1}{2b}\sqrt{(Q-\frac{\beta}{2})^2-4}$. We stress that $\frac{r_-}{2}<1\wedge \frac{1}{2b^2}\wedge \frac{1}{2b }(Q-\frac{\beta}{2})$ in such a way that we can find some $p$ fulfilling our conditions.

Finally we prove \eqref{estGMC3} and, for this, we consider the functional
$$E''_\epsilon(A,B):=\E\Big[\Big(\iint_{A\times B}(|x|+\epsilon)^{-2b\alpha} (|x-x'|+\epsilon)^{-2b\beta} M^\epsilon_{2b}(\dd x )M^\epsilon_{2b}(\dd x' )\Big)^p\Big]$$
and the same considerations as before lead to the relation $E''_\epsilon(\Sigma,\Sigma)\leq 2^{-f_{\alpha/2+\beta/2}(2p) }E''_{2\epsilon}(\Sigma,\Sigma)+C$. Therefore, if $Q- \frac{\beta+\alpha}{2}>0$ and $p\in (0,\frac{1}{2b}(Q-\frac{\alpha+\beta}{2}))$, we have $f_{\alpha/2+\beta/2}(2p)>0$.  Hence our claim.

\end{proof}


\bibliographystyle{alpha}
\bibliography{H3}

\end{document}